\documentclass[a4paper,10pt]{article}

\usepackage{hyperref}
\hypersetup{
  colorlinks=true,
  linkcolor=blue,
  filecolor=magenta,
  urlcolor=cyan,
  citecolor=blue,
  bookmarksopen=true,    
}

\usepackage[cm]{fullpage}

\usepackage{amsthm}
\usepackage{thmtools}
\usepackage{mathtools}

\declaretheorem[name=Theorem,numberwithin=section]{theorem}
\declaretheorem[name=Lemma,sibling=theorem]{lemma}
\declaretheorem[name=Corollary,sibling=theorem]{corollary}

\declaretheorem[style=definition,name=Definition,sibling=theorem]{definition}
\declaretheorem[style=definition,name=Example,sibling=theorem]{example}

\declaretheorem[style=remark,name=Remark,sibling=theorem]{remark}
\declaretheorem[style=remark,name=Notation,numbered=no]{notation}
\usepackage{thm-restate}

\usepackage[T1]{fontenc} 
\usepackage[utf8]{inputenc} 

\usepackage{amsmath} 
\allowdisplaybreaks
\usepackage{cite} 
\usepackage{cmll}
\usepackage{multicol}
\usepackage{proof}
\usepackage{amssymb}
\usepackage{amsfonts}
\usepackage{stmaryrd}
\SetSymbolFont{stmry}{bold}{U}{stmry}{m}{n} 

\usepackage{quiver}
\usepackage{tikz}
\usetikzlibrary{cd}
\newcommand\ket[1]{\ensuremath{|{#1}\rangle}}

\newcommand\OC{\ensuremath{{\mathcal L}_!^{\mathcal S}}}
\newcommand\fv{\ensuremath{\mathsf{fv}}}

\newcommand\zero{\ensuremath{\mathfrak 0}}
\newcommand\one{\ensuremath{\mathfrak 1}}
\newcommand\abstr[1]{#1.}

\newcommand\inl{\mbox{\small \it inl}}
\newcommand\inr{\mbox{\small \it inr}}
\newcommand\elimone{\delta_{\one}}
\newcommand\elimzero{\delta_{\zero}}
\newcommand\elimwith{\delta_{\with}}
\newcommand\elimplus{\delta_{\oplus}}

\newcommand\elimtens{\delta_{\otimes}}
\newcommand\elimbang{\delta_\oc}

\DeclareRobustCommand{\plus}{%
  \mathbin{%
    \begingroup
      \dimen0=1.4ex\relax 
      \dimen2=0.35ex\relax 
      \dimen4=\dimen0 \advance\dimen4 by -\dimen2 \divide\dimen4 by 2 
      \vcenter{\hbox{%
        \rlap{\hbox to \dimen0{\hss\vrule width \dimen2 height \dimen0\hss}}%
        \raise\dimen4\hbox to \dimen0{\hss\rule{\dimen0}{\dimen2}\hss}%
      }}%
    \endgroup
  }%
}
\newcommand\dotprod{\ensuremath{\mathrel{\raisebox{0.15ex}{\scalebox{0.8}{$\bullet$}}}}}

\newcommand\SN{\ensuremath{\mathsf{SN}}}

\newcommand\pair[2]{\langle #1, #2 \rangle}

\newcommand\topintro{\langle\rangle}

\newcommand\lra{\longrightarrow}
\newcommand\lla{\longleftarrow}
\newcommand\lras{\longrightarrow^*}
\newcommand\llas{\mathrel{{}^*{\longleftarrow}}}

\newcommand\hatmultimap{\ensuremath{\mathbin{\hat\multimap}}}
\newcommand\hatotimes{\ensuremath{\mathbin{\hat\otimes}}}
\newcommand\hatand{\ensuremath{\mathbin{\hat\with}}}
\newcommand\hatoplus{\ensuremath{\mathbin{\hat\oplus}}}
\newcommand\hatbang{\ensuremath{\mathop{\hat\oc}}}

\newcommand\citeintitle[2]{\cite[#1]{#2}}

\newcommand\scalars{\mathcal S}
\newcommand\abscat{\mathbf C_\oc^\scalars}
\newcommand\mono[1]{\llparenthesis {#1} \rrparenthesis}
\newcommand\interpretation[1]{\left\llbracket{#1}\right\rrbracket}
\newcommand\ruleinterpretation[1]{\interpretation{\vcenter{#1}}}
\newcommand\bang{\oc}
\newcommand\sm{\mathsf{SM}_\scalars}
\newcommand\set{\mathsf{Set}}
\newcommand\singleton{{\{\ast\}}}
\newcommand\naturalnumbers{\mathbb N}
\newcommand\complexnumbers{\mathbb C}

\title{An Algebraic Extension of Intuitionistic Linear Logic:\\
The \texorpdfstring{$\OC$}{L-S-!}-Calculus and Its Categorical Model}

\author{Alejandro D\'{\i}az-Caro$^{1,2}$ \and Malena Ivnisky$^{3,4,5}$ \and Octavio Malherbe$^6$}

\date{
  \small
  $^1$Universit\'e de Lorraine, CNRS, Inria, LORIA, Nancy, France\\
  $^2$Universidad Nacional de Quilmes, DCyT, Bernal, PBA, Argentina\\
  $^3$Universidad de Buenos Aires, DC, FCEyN, Buenos Aires, Argentina\\
  $^4$Universidad de Buenos Aires-CONICET, ICC, Buenos Aires, Argentina\\
  $^5$Universidad de la Rep\'ublica--MEC, PEDECIBA, Montevideo, Uruguay\\
  $^6$Universidad de la Rep\'ublica, IMERL, FIng, Montevideo, Uruguay
}

\begin{document}
\maketitle

\abstract{
  We introduce the $\OC$-calculus, a linear lambda-calculus extended with
  scalar multiplication and term addition, that acts as a proof language for
  intuitionistic linear logic (ILL). These algebraic operations enable the
  direct expression of linearity at the syntactic level, a property not
  typically available in standard proof-term calculi. Building upon previous
  work, we develop the $\OC$-calculus as an extension of the
  $\mathcal{L^S}$-calculus with the $\oc$ modality. We prove key
  meta-theoretical properties—subject reduction, confluence, strong
  normalisation, and an introduction property—as well as preserve the
  expressiveness of the original $\mathcal{L^S}$-calculus, including the
  encoding of vectors and matrices, and the correspondence between proof-terms
  and linear functions. A denotational semantics is provided in the framework
  of linear categories with biproducts, ensuring a sound and adequate
  interpretation of the calculus. This work is part of a broader programme
  aiming to build a measurement-free quantum programming language grounded in
  linear logic.
}

\section{Introduction}\label{sec:intro}
The design of proof languages plays a central role in both logic and
programming languages, particularly when reasoning about resource usage and
algebraic computation. Linear Logic~\cite{Girard87} provides a framework for
such reasoning, but its standard proof-term calculi do not make linear
properties—such as distribution over addition or scalar multiplication—explicit
at the syntactic level. This lack of syntactic expressiveness limits its
applicability in settings like quantum computing, where linearity is not just a
property but a fundamental constraint.

Linear Logic is named as such because it is modelled by vector spaces and
linear maps, and more generally by monoidal
categories~\cite{EilenbergKelly66,Benabou63,Seely89,Barr91}. These types of
categories also include the so-called Cartesian categories, generating a
formal place of interaction between purely algebraic structures and purely
logical structures, i.e.~between algebraic operations and the exponential
connective ``$\oc$''.  In the strictly linear fragment (without $\oc$),
functions between two propositions are linear functions. However,
expressing this linearity within the proof-term language itself is
challenging. Properties such as $f(u+v) = f(u)+f(v)$ and $f(a \cdot u) = a
\cdot f(u)$, for some scalar $a$, require operations like addition and
scalar multiplication, which are typically absent in the proof language.

In~\cite{DiazcaroDowekMSCS24} this challenge has been addressed by introducing
the $\mathcal{L^S}$-calculus, a proof language for the strictly linear fragment
of intuitionistic linear logic (IMALL) extended with syntactic constructs for
addition ($\plus$) and scalar multiplication ($\dotprod$). While the extension does not change
provability, it enables the expression of linear properties at the term level,
such as distributivity of functions over sums and scalar multiplication. For
example, any proof-term $t$ of $A\multimap B$ satisfies that $t~(u\plus v)$ is observationally equivalent to $t~u \plus t~v$.
Similarly, $t~(a\dotprod u)$ is equivalent to $a\dotprod t~u$.

This extension involves changing the proof-term $\star$ of proposition $\one$ into a
family of proof-terms $a.\star$, one for each scalar $a$ in a given fixed commutative semiring
$\mathcal S$:
\[
  \infer[\mbox{\small $\one_i$($a$)}]{\vdash a.\star:\one}{}
\]

The following two deduction rules have also been added:
\[
  \infer[{\mbox{\small sum}}]{\Gamma\vdash t\plus u:A}{\Gamma\vdash t:A & \Gamma\vdash u:A}
  \qquad\qquad
  \infer[{\mbox{\small prod($a$)}}]{\Gamma\vdash a\dotprod t:A}{\Gamma\vdash t:A}
\]

Incorporating these rules requires adding commuting rules to preserve
cut-elimination. Indeed, the new rules may appear between an introduction and
an elimination of some connective. For example, consider the following
derivation.
\[
  \infer[\mbox{\small $\with_{e1}$}]{\Gamma\vdash C}
  {
    \infer[\mbox{\small prod($a$)}]{\Gamma\vdash A\with B}
    {
      \infer[\mbox{\small $\with_i$}]{\Gamma\vdash A\with B}
      {
	\Gamma\vdash A
	&
	\Gamma\vdash B
      }
    }
    &
    \Gamma,A\vdash C
  }
\]

To achieve cut-elimination, we must commute the rule prod($a$) either with the
introduction
or with the elimination, as follows:

\[
  \infer[\mbox{\small$\with_{e1}$}]{\Gamma\vdash C}
  {
    \infer[\mbox{\small$\with_i$}]{\Gamma\vdash A\with B}
    {
      \infer[\mbox{\small prod($a$)}]{\Gamma\vdash A}{\Gamma\vdash A}
      &
      \infer[\mbox{\small prod($a$)}]{\Gamma\vdash B}{\Gamma\vdash B}
    }
    &
    \Gamma,A\vdash C
  }
  \qquad
  \qquad
  \infer[\mbox{\small prod($a$)}]{\Gamma\vdash C}
  {
    \infer[\mbox{\small$\with_{e1}$}]{\Gamma\vdash C}
    {
      \infer[\mbox{\small$\with_i$}]{\Gamma\vdash A\with B}
      {
	\Gamma\vdash A
	&
	\Gamma\vdash B
      }
      &
      \Gamma,A\vdash C
    }
  }
\]

Both of these are reducible. We refer to the sum and prod($a$) rules
as \emph{interstitial rules}, as they can appear in the interstice between an
introduction and an elimination. We choose to commute these rules with the
introductions as much as possible.  This means we introduce the following
commutation rule
\(
  a\dotprod\pair{t}{u}\lra\pair{a\dotprod t}{a\dotprod u}
\)
instead of the alternative rule
\(
  \elimwith^1(a\dotprod t,x^A.u) \lra a\dotprod\elimwith^1(t,x^A.u)
\).
This choice provides a better introduction property, for example, a closed irreducible proof-term
of a proposition $A\with B$ is a pair.

In the conference version of this paper \cite{DiazcaroDowekIvniskyMalherbeWoLLIC2024} we extended the proof system
to second-order intuitionistic linear logic, adding the exponential connective
and a universal quantifier. We proved that the linearity result still holds for
second order. 

In this journal version of the paper, polymorphism has been removed to allow for a thorough study of the ILL fragment. We provide a denotational semantics for this fragment of the language, which we refer to as the $\OC$-calculus.

While our primary focus is on introducing a minimal extension to the proof
language within the realm of intuitionistic linear logic, our work draws
inspiration from various domains, particularly quantum programming languages.
These languages were trailblazers in merging programming constructs with
algebraic operations, such as addition and scalar multiplication.

QML~\cite{AltenkirchGrattageLICS05} introduced the concept of superposition of terms through the
$\mathsf{if}^\circ$ constructor, allowing the representation of linear
combinations $a.u+b.v$ by the expression $\mathsf{if}^\circ\ a.\ket 0 + b.\ket
1\ \mathsf{then}\ u\ \mathsf{else}\ v$. The linearity (and even unitarity)
properties of QML were established through a translation to quantum circuits.

The ZX calculus~\cite{ZXBook17}, a graphical language based on a categorical
model, lacks direct syntax for addition or scalar multiplication but defines a
framework where such constructs can be interpreted. This language is extended
by the Many Worlds Calculus~\cite{ChardonnetDevismeValironVilmartLMCS25} which allows for
linear combinations of diagrams.

The algebraic lambda-calculus~\cite{Vaux2009} and Lineal~\cite{Lineal}
exhibit syntax similarities with $\mathcal L^{\mathcal S}$-calculus. However,
the algebraic lambda-calculus lacks a proof of linearity in its simple
intuitionistic type system. In contrast, Lineal is untyped and axiomatizes the linearity,
relying on explicit definitions like $f(u+v) = f(u) +
f(v)$ and $f(a.u) = a.f(u)$.  Several type systems have been proposed for
Lineal~\cite{ArrighiDiazcaroLMCS12,ArrighiDiazcaroValironIC17,LambdaS,DiazcaroGuillermoMiquelValironLICS19,DiazcaroMalherbeLMCS22}.
However, none of these systems are related to linear logic, and their purpose
is not to prove linearity, as we do, but rather to axiomatize it.

One of the objectives of our line of research is to define a measurement-free quantum programming language rooted in Linear Logic. Or, another way to state it, is to extend a proof language of Intuitionistic Linear Logic with the necessary constructs to represent quantum algorithms. The $\mathcal L^{\mathcal S}$-calculus was the first step in this direction~\cite{DiazcaroDowekMSCS24}.
In our conference paper~\cite{DiazcaroDowekIvniskyMalherbeWoLLIC2024}, for
which the current paper is its journal version, we introduced the $\mathcal
L_{\oc\forall}^{\mathcal S}$-calculus as the second step, extending the
language to Second-Order Intuitionistic Linear Logic, where we can encode
natural numbers, and iterators, and in general, classical computing, apart from
vectors and matrices for quantum computing. The $\OC$-calculus is a minimal
extension of the $\mathcal L^{\mathcal S}$-calculus, adding just the
exponential connective.

The reason to discard polymorphism is just to study its categorical semantics
without needing to appeal to linear hyperdoctrines~\cite{Maneggia}, which would
introduce unnecessary complexity to the model. The $\OC$-calculus is a step
towards a quantum programming language based on Linear Logic. Next steps
include to add measurement, and to restrict the language to unitary maps
following the techniques from~\cite{DiazcaroGuillermoMiquelValironLICS19}.

Our contributions are as follows:
\begin{itemize}
  \item We extend the $\mathcal L^{\mathcal S}$-calculus to intuitionistic linear logic, resulting in the $\OC$-calculus (Section~\ref{sec:secls}).
  \item We prove its correctness (Section~\ref{sec:correctness}), namely,
    Subject Reduction (Theorem~\ref{thm:SR}), Confluence
    (Theorem~\ref{thm:Confluence}), Strong Normalisation 
    (Theorem~\ref{thm:ST}), and the Introduction Property
    (Theorem~\ref{thm:introductions}). In particular, the proof of strong
    normalisation involves applying a technique due to Girard, called ultra-reduction. It is mostly a straightforward extension of the proof for the $\mathcal L^{\mathcal S}$-calculus.
  \item Since it is an extension, the encodings for vectors
    (Section~\ref{sec:secvectors}) and matrices
    (Section~\ref{sec:secmatrices}), already
    present in the $\mathcal L^{\mathcal S}$-calculus are still valid. We
    provide detailed explanations of these encodings for self-containment.
  \item The syntactic linearity proofs
    (Section~\ref{sec:seclinearity}) extend the
  linearity proofs for the $\OC$-calculus.

  \item The main contribution of this journal version, with respect to the
    conference paper~\cite{DiazcaroDowekIvniskyMalherbeWoLLIC2024}, is a
    refinement on the reduction rules and the denotational semantics of the
    $\OC$-calculus
    (Section~\ref{sec:denotationalsemantics}).

  We provide a categorical characterisation of the $\OC$-calculus, proving that
  it is sound (Theorem~\ref{thm:soundness}) and
  adequate (Theorem~\ref{thm:semanticsadequacy}) with
  respect to any linear category with biproducts such that there exists a
  monomorphism from the semiring of scalars to the semiring $\mathrm{Hom}(I,
  I)$ (Definition~\ref{def:abscat}).
\end{itemize}

\section{The \texorpdfstring{$\OC$}{LS!}-calculus}
\subsection{Syntax and Typing Rules}
\label{sec:secls}

The $\OC$-logic is intuitionistic linear logic (we follow the presentation of
$\mathsf{DILL}$~\cite{DILL}, extended with the additive linear
connectives).
\[
  A =  \one \mid A \multimap A \mid A \otimes A \mid \top \mid \zero
  \mid A \with A \mid A \oplus A \mid \oc A 
\]

Let ${\mathcal S}$ be a commutative semiring of {\it scalars}, for instance $\{\star\}$,
$\mathbb N$, ${\mathbb Q}$, ${\mathbb R}$, or ${\mathbb C}$.  The proof-terms
of the $\OC$-calculus are given in Figure~\ref{fig:proofterms}, where $a$ is a scalar
in $\mathcal S$.
\begin{figure}[t]
  \[
    \begin{array}{rllc}
      t = x \mid t \plus u \mid a \dotprod t &\\
      & \mid a.\star & \mid \elimone(t,u) & (\one)\\  
      & \mid \lambda x^A.t & \mid t~u &(\multimap)\\
      &\mid t \otimes u & \mid \elimtens(t, x^A y^B.u) & (\otimes)\\
      & \mid \langle \rangle  && (\top)\\
      & &\mid\elimzero(t) & (\zero)\\
      &\mid \pair{t}{u} &\mid \elimwith^1(t,x^A.u) \mid \elimwith^2(t,x^B.u) & (\with)\\
      &\mid \inl(t)\mid \inr(t) & \mid \elimplus(t,x^A.u,y^B.v) & (\oplus)\\
      & \mid \oc t & \mid \elimbang(t, x^A.u) & (\oc)\\
      &\textrm{Introductions} & \textrm{Eliminations} & \textrm{Connective}\\
    \end{array}
  \]
  \caption{The proof-terms of the $\OC$-calculus.\label{fig:proofterms}}
\end{figure}

The $\alpha$-equivalence relation and the free and bound variables of a
proof-term are defined as usual, we write as $\fv(t)$ the set of free variables
of $t$.  Proof-terms are defined modulo $\alpha$-equivalence.  A proof-term is
closed if it contains no free variables. 
We write $(u/x)t$ for the
substitution of $u$ for $x$ in $t$. 

A judgement has the form $\Upsilon;\Gamma\vdash t:A$, where $\Upsilon$ is the
intuitionistic
context and $\Gamma$ the linear one.  The deduction rules are those of
Figure~\ref{fig:figuretypingrules}.  These rules are exactly the deduction rules of
intuitionistic linear natural deduction, with proof-terms, with
two differences: the interstitial rules and the scalars (see
Section~\ref{sec:intro}). This typing system admits
weakening in the intuitionistic context, which will sometimes be used
implicitly.

\begin{figure*}[t]
  \[
    \infer[{\mbox{\small {lin-}ax}}]{{\Upsilon;} x^A \vdash x:A}{}
    \qquad
    \infer[{\mbox{\small ax}}]{\Upsilon, x^A;\varnothing \vdash x:A}{}
    \qquad
    \infer[{\mbox{\small sum}}]{{\Upsilon;} \Gamma \vdash t \plus u:A}{{\Upsilon;} \Gamma \vdash t:A & {\Upsilon;} \Gamma \vdash u:A}
    \qquad
    \infer[{\mbox{\small prod}(a)}]{{\Upsilon;} \Gamma \vdash a \dotprod t:A}{{\Upsilon;} \Gamma \vdash t:A}
  \]
  \[
    \infer[{\mbox{\small $\one_i$}(a)}] {{\Upsilon;\varnothing} \vdash a.\star:\one}{}
    \qquad
    \infer[{\mbox{\small $\one_e$}}]{{\Upsilon;}\Gamma, \Delta \vdash \elimone(t,u):A}{{\Upsilon;}\Gamma \vdash t:\one & {\Upsilon;}\Delta \vdash u:A}
  \]
  \[
    \infer[{\mbox{\small $\multimap_i$}}]{{\Upsilon;}\Gamma \vdash \lambda x^A.t:A \multimap B}{{\Upsilon;}\Gamma, x^A \vdash t:B}
    \qquad
    \infer[{\mbox{\small $\multimap_e$}}]{{\Upsilon;}\Gamma, \Delta \vdash t~u:B}{{\Upsilon;}\Gamma \vdash t:A\multimap B & {\Upsilon;}\Delta \vdash u:A}
  \]
  \[
    \infer[{\mbox{\small $\otimes_i$}}]{{\Upsilon;}\Gamma, \Delta \vdash t \otimes u:A \otimes B}{{\Upsilon;}\Gamma \vdash t:A & {\Upsilon;}\Delta \vdash u:B}
    \qquad
    \infer[{\mbox{\small $\otimes_e$}}]{{\Upsilon;}\Gamma, \Delta \vdash \elimtens(t, x^A y^B.u):C}{{\Upsilon;}\Gamma \vdash t:A \otimes B & {\Upsilon;}\Delta, x:A, y:B \vdash u:C}
  \]
  \[
    \infer[{\mbox{\small $\top_i$}}]{{\Upsilon;} \Gamma \vdash \langle \rangle:\top}{}
    \qquad
    \infer[{\mbox{\small $\zero_e$}}]{{\Upsilon;}\Gamma, \Delta \vdash \elimzero(t):C}{{\Upsilon;}\Gamma \vdash t:\zero}
  \]
  \[
    \infer[{\mbox{\small $\with_i$}}]{{\Upsilon;}\Gamma \vdash \pair{t}{u}:A \with B}{{\Upsilon;}\Gamma \vdash t:A & {\Upsilon;}\Gamma \vdash u:B}
    \qquad
    \infer[{\mbox{\small $\with_{e1}$}}]{{\Upsilon;}\Gamma, \Delta \vdash \elimwith^1(t,x^A.u):C}{{\Upsilon;}\Gamma \vdash t:A \with B & {\Upsilon;}\Delta, x^A \vdash u:C}
    \qquad
    \infer[{\mbox{\small $\with_{e2}$}}]{{\Upsilon;}\Gamma, \Delta \vdash \elimwith^2(t,x^B.u):C}{{\Upsilon;}\Gamma \vdash t:A \with B & {\Upsilon;}\Delta, x^B \vdash u:C}
  \]
  \[
    \infer[{\mbox{\small $\oplus_{i1}$}}]{{\Upsilon;}\Gamma \vdash \inl(t):A \oplus B}{{\Upsilon;}\Gamma \vdash t:A}
    \qquad
    \infer[{\mbox{\small $\oplus_{i2}$}}]{{\Upsilon;}\Gamma \vdash \inr(t):A \oplus B}{{\Upsilon;}\Gamma \vdash t:B}
    \qquad
    \infer[{\mbox{\small $\oplus_e$}}]{{\Upsilon;}\Gamma, \Delta \vdash \elimplus(t,x^A.u,y^B.v):C}{{\Upsilon;}\Gamma \vdash t:A \oplus B & {\Upsilon;}\Delta, x^A \vdash u:C & {\Upsilon;}\Delta, y^B \vdash v:C}
  \]
  \[
    \infer[{\mbox{\small $\oc_i$}}]{\Upsilon; \varnothing \vdash \oc t: \oc A}{\Upsilon; \varnothing \vdash t:A}
    \qquad
    \infer[{\mbox{\small $\oc_e$}}]{\Upsilon; \Gamma, \Delta \vdash \elimbang(t, x^A.u):B}{\Upsilon; \Gamma \vdash t:\oc A & \Upsilon, x^A; \Delta \vdash u:B}
  \]
  \caption{The deduction rules of the $\OC$-calculus.\label{fig:figuretypingrules}}
\end{figure*}

\begin{figure}[t]
  \centering
  \begin{align*}
    \elimone(a.\star,t) & \lra  a \dotprod t \\
    (\lambda x^A.t)~u & \lra  (u/x)t \\
    \elimtens(u \otimes v,x^A y^B.w) & \lra  (u/x,v/y)w  \\
    \elimwith^1(\pair{t_1}{t_2}, x^A.v) & \lra  (t_1/x)v \\
    \elimwith^2(\pair{t_1}{t_2}, x^A.v) & \lra  (t_2/x)v \\
    \elimplus(\inl(t),x^A.v,y^B.w) & \lra  (t/x)v\\
    \elimplus(\inr(u),x^A.v,y^B.w) & \lra  (u/y)w\\
    \elimbang(\oc t, x^A.u) & \lra (t/x)u
  \end{align*}
  \[
    \begin{array}{r@{\,}l@{\qquad}r@{\,}l}
      {a.\star} \plus b.\star&\lra  (a+b).\star
      	&a \dotprod b.\star&\lra  (a \times b).\star\\
      (\lambda x^A.t) \plus (\lambda x^A.u) & \lra  \lambda x^A.(t \plus u)
      	&a \dotprod \lambda x^A. t &\lra  \lambda x^A. a \dotprod t\\
      \elimtens(t \plus u,x^A y^B.v) & \lra \elimtens(t,x^A y^B.v) \plus \elimtens(u,x^A y^B.v)
      	&\elimtens(a \dotprod t,x^A y^B.v) & \lra a \dotprod \elimtens(t,x^A y^B.v)\\
      \langle \rangle \plus \langle \rangle & \lra  \langle \rangle
      	&a \dotprod \langle\rangle &\lra  \langle\rangle\\
      \pair{t}{u} \plus \pair{v}{w} & \lra  \pair{t \plus v}{u \plus w}
      	&a \dotprod \pair{t}{u} &\lra  \pair{a \dotprod t}{a \dotprod u}\\
      \elimplus(t \plus u,x^A.v,y^B.w) & \lra \elimplus(t,x^A.v,y^B.w) \plus \elimplus(u,x^A.v,y^B.w)
      	&\elimplus(a \dotprod t,x^A.v,y^B.w) & \lra a \dotprod \elimplus(t,x^A.v,y^B.w)\\
	\elimbang(t \plus u, x^A.v) & \lra \elimbang(t, x^A.v) \plus \elimbang(u, x^A.v)
	&\elimbang(a \dotprod t, x^A.v) & \lra a \dotprod \elimbang(t, x^A.v)
    \end{array}
  \]
  \caption{The reduction rules of the $\OC$-calculus.\label{fig:figureductionrules}}
\end{figure}

The reduction rules are those of Figure~\ref{fig:figureductionrules}.  As
usual, the reduction relation is written $\lra$, its inverse $\lla$,
its reflexive-transitive closure $\lras$, the reflexive-transitive
closure of its inverse $\llas$, and its reflexive-symmetric-transitive
closure $\equiv$.  The first group of rules correspond to the reduction of
cuts on the connectives $\one$, $\multimap$, $\otimes$, $\with$, $\oplus$, and $\oc$.
The next two groups enable us to commute the interstitial rules sum and prod($a$)
with the introduction rules of the connectives $\one$, $\multimap$,
$\top$, and $\with$, and with the elimination rule of the connectives
$\otimes$, $\oplus$, and $\oc$.
For instance, the rule
\(
  \pair{t}{u} \plus \pair{v}{w} \lra  \pair{t \plus v}{u \plus w}
\)
pushes the symbol $\plus$ inside the pair.  The zero-ary commutation rules add
and multiply the scalars:
\(
  {a.\star} \plus b.\star \lra  (a+b).\star
\),
\(
  a \dotprod b.\star \lra  (a \times b).\star
\).

In the preliminary presentation of the calculus
\cite{DiazcaroDowekIvniskyMalherbeWoLLIC2024}, the interstitial rules commuted
with the introduction of the $\oc$ connective, instead of the elimination
rules. This is not appropriate, since it would require its model to use a
semiadditive comonad to interpret $\oc$, which is not the case in the
considered models.

\subsection{Correctness}\label{sec:correctness}
We now prove the subject reduction, confluence, strong normalisation, and introduction
properties of the $\OC$-calculus.

\subsubsection{Subject reduction}
\label{sec:SR}
The subject reduction property is not completely trivial.  As noted in the
introduction, we commute the sum rule with the introductions \emph{as much as
possible}. It is not possible in the case of the connectives $\otimes$ and
$\oplus$, since it would break subject reduction.  For example, the rule $(t\otimes u)\plus(v\otimes w)\lra (t\plus v)\otimes(u\plus w)$ would not be valid.
Intuitively, if we consider the $\otimes$ to be a mathematical tensor product
and $+$ a vector sums, it is clear that the rule is not valid.
Formally, it suffices to check that we have
\(
  \varnothing;x^A,y^A\vdash (x\otimes y)\plus (y\otimes x):A\otimes A
\)
but
\(
  \varnothing;x^A,y^A \nvdash (x \plus y)\otimes (y \plus x):A\otimes A
\),
since the rule $\otimes_i$ is multiplicative. 

\begin{restatable}[Subject reduction]{theorem}{SR}
  \label{thm:SR}
  If $\Upsilon; \Gamma \vdash t:A$ and $t \lra u$, then $\Upsilon; \Gamma \vdash u:A$.
\end{restatable}
\begin{proof}
  By induction on the relation $\lra$. The proof is given in Appendix~\ref{proof:SR}.
\end{proof}

\subsubsection{Confluence}

\begin{theorem}[Confluence]
\label{thm:Confluence}
The $\OC$-calculus is confluent.
\end{theorem}

\begin{proof}
The reduction system of Figure~\ref{fig:figureductionrules} applied to
well-formed proof-terms is left linear and has no critical pairs.  By
\cite[Theorem 6.8]{Nipkow}, it is confluent.
\end{proof}

\subsubsection{Strong normalisation}\label{sec:ST}
We now prove that all reduction sequences are finite.  To handle the symbols
$\plus$ and $\dotprod$ and the associated reduction rules, we prove the strong
normalisation of an extended reduction system, in the spirit of Girard's
ultra-reduction\footnote{Ultra-reduction is used, in particular, in the
adequacy of the connectives $\plus$ and $\dotprod$.}~\cite{GirardPhDThesis},
whose strong normalisation obviously implies that of the rules of
Figure~\ref{fig:figureductionrules}.

\begin{definition}[Ultra-reduction]
  Ultra-reduction is defined with the rules of Figure~\ref{fig:figureductionrules},
  plus the rules
  \[
    t \plus u  \lra  t \qquad\qquad
    t \plus u  \lra  u \qquad\qquad
    a \dotprod t  \lra  t
  \]
\end{definition}

Our proof is a direct extension from the proof of the $\mathcal L^{\mathcal
S}$-calculus~\cite{DiazcaroDowekMSCS24}.  We use a presentation ready to be
extended to polymorphism, by following the method introduced by
Tait~\cite{TaitJSL67} for Gödel's System T and generalised to System F by
Girard~\cite{GirardPhDThesis}.

\begin{definition}
  $\SN$ is the set of strongly normalising proof-terms and $Red(t)$ is the set of one-step reducts of $t$. That is,
  $\SN = \{ t \mid t\textrm{ strongly normalises}\}$ and $Red(t) = \{u\mid t\lra u\}$.
\end{definition}

\begin{definition}[Reducibility candidates]
  A set of proof-terms $E$ is a reducibility candidate if and only if the following conditions are satisfied.

    (CR1) $E \subseteq \SN$.

    (CR2) If $t \in E$ and $t \lra t'$, then $t' \in E$.

    (CR3) If $t$ is not an introduction and $Red(t) \subseteq E$, then $t \in E$.

The set of all reducibility candidates is called $\mathcal R$. 
\end{definition}

\begin{definition}
  Let $E,F$ be sets of proof-terms. We define the following sets.
  \begin{align*}
    E \hatmultimap F &= \{t \in \SN \mid \textrm{if }t \lras \lambda x^A.u\textrm{, then for every }v \in E, (v/x)u \in F\}\\
    E \hatotimes F &= \{t \in \SN \mid \textrm{if }t \lras u \otimes v\textrm{, then }u \in E\textrm{ and }v \in F\}\\
    E \hatand F &= \{t \in \SN \mid \textrm{if }t \lras \langle u,v \rangle\textrm{, then }u \in E\textrm{ and }v \in F\}\\
    E \hatoplus F &= \{t \in \SN \mid \textrm{if }t \lras \inl(u)\textrm{,\,then\,}u \in E\textrm{\,and if }t \lras \inr(v)\textrm{,\,then }v \in F\}\\
    \hatbang E &= \{t \in \SN \mid \textrm{if }t \lras \oc u\textrm{, then }u \in E\}
  \end{align*}
\end{definition}

\begin{remark}
  Notice that ultra-reduction is used to define the sets $E \hatotimes F$, $E
  \hatoplus F$ and $\hatbang$.
  
  Indeed, we only need to require, for example, that the term reduces to a
  tensor, and not to a linear combination of tensors, since any linear
  combination will itself reduce by the ultra-reduction rules.
\end{remark}

\begin{definition}
  For any proposition $A$, the set of proof-terms $\llbracket A \rrbracket$ is defined as follows:
  \[
    \begin{array}{r@{\,}l@{\qquad}r@{\,}l@{\qquad}r@{\,}l@{\qquad}r@{\,}l}
      \llbracket\one\rrbracket &= \SN
      &\llbracket A \otimes B \rrbracket &= \llbracket A \rrbracket \hatotimes \llbracket B \rrbracket
      &\llbracket \zero \rrbracket &= \SN
      &\llbracket A \oplus B \rrbracket &= \llbracket A \rrbracket \hatoplus \llbracket B \rrbracket\\
      \llbracket A \multimap B \rrbracket &= \llbracket A \rrbracket \hatmultimap \llbracket B \rrbracket 
      &\llbracket \top \rrbracket &= \SN
      &\llbracket A \with B \rrbracket &= \llbracket A \rrbracket \hatand \llbracket B \rrbracket
      &\llbracket \oc A \rrbracket &= \hatbang\llbracket A\rrbracket
    \end{array}
  \]
\end{definition}

\begin{restatable}{lemma}{typeinterpretationsarerc}
  \label{lem:typeinterpretationsarerc}
  For any proposition $A$, $\llbracket A \rrbracket \in \mathcal R$.
\end{restatable}
\begin{proof}
  By induction on $A$. The proof is given in Appendix~\ref{proof:ST}.
\end{proof}

\begin{lemma}[Variables]
  \label{lem:Var}
  For any proposition $A$, the set $\llbracket A \rrbracket$ contains all the proof-term variables.
\end{lemma}
\begin{proof}
  By Lemma~\ref{lem:typeinterpretationsarerc}, $\llbracket A \rrbracket\in \mathcal R$. Since $Red(x)$ is empty, by CR3, $x \in \llbracket A \rrbracket$.
\end{proof}

\begin{restatable}[Adequacy for strong normalisation]{lemma}{adequacy}
  \label{lem:adequacy}
  If $\Upsilon;\Gamma\vdash t:A$, then for any substitution $\sigma$ such that for each $x^B\in\Upsilon\cup\Gamma$, $\sigma(x)\in\llbracket B\rrbracket$, we have
  $\sigma t \in \llbracket A \rrbracket$.
\end{restatable}
  \begin{proof}
  By induction on $t$.
  If $t$ is a variable, then, by hypothesis, $\sigma t \in \llbracket A
  \rrbracket$. For the other proof-term constructors, we use the adequacy
  lemmas provided in Appendix~\ref{proof:Adequacy}.
\end{proof}

\begin{theorem}[Strong normalisation]\label{thm:ST}
  If $\Upsilon;\Gamma\vdash t:A$, then,
  $t\in\SN$.
\end{theorem}

\begin{proof}
  By Lemma~\ref{lem:Var}, for each $x^B\in\Upsilon\cup\Gamma$, $x^{B}\in\llbracket B \rrbracket$.  Then, by Lemma~\ref{lem:adequacy}, $t = \mathsf{id}(t) \in \llbracket A \rrbracket$. Hence, by Lemma~\ref{lem:typeinterpretationsarerc}, $t\in\SN$. 
\end{proof}

\subsubsection{Introduction property}\label{sec:introductionthm}

\begin{restatable}[Introduction]{theorem}{introductions}
  \label{thm:introductions}
  Let $t$ be a closed irreducible proof-term of $A$.
  \begin{itemize}
    \item If $A$ is $\one$, then $t$ has the form $a.\star$.

    \item If $A$ has the form $B \multimap C$, then $t$ has the form
      $\lambda x^B.u$.

    \item If $A$ has the form $B \otimes C$, then $t$ has the form $u \otimes v$,
      $u \plus v$, or $a \dotprod u$.

    \item If $A$ is $\top$, then $t$ is $\langle \rangle$.

    \item The proposition $A$ is not $\zero$.

    \item If $A$ has the form $B \with C$, then $t$ has the form
      $\pair{u}{v}$.

    \item If $A$ has the form $B \oplus C$, then $t$ has the form $\inl(u)$, $\inr(u)$, $u \plus v$, or $a \dotprod u$.
    \item If $A$ has the form $\oc B$, then $t$ has the form $\oc u$, $u \plus v$, or $a \dotprod u$.
  \end{itemize}
\end{restatable}
\begin{proof}
  By induction on $t$. The proof is given in
Appendix~\ref{proof:introductionthm}.
\end{proof}

\subsection{Encodings}\label{sec:secvectorsmatrices}
In this section, we present the encodings for vectors
(Section~\ref{sec:secvectors}) and matrices
(Section~\ref{sec:secmatrices}), which were
initially introduced in~\cite{DiazcaroDowekMSCS24} and are replicated here for
self-containment. 

Since $\mathcal S$ is a commutative semiring, we work with semimodules.
However, to help intuition, the reader may think of $\mathcal S$ as a field,
obtaining a vector space instead.

\subsubsection{Vectors}
\label{sec:secvectors}

As there is one rule $\one_i$ for each scalar $a$, there is one closed
irreducible proof-term $a.\star$ for each scalar $a$.  Thus, the closed irreducible
proof-terms $a.\star$ of $\one$ are in one-to-one correspondence with the elements
of ${\mathcal S}$.  Therefore, the proof-terms $\pair{a.\star}{b.\star}$ of $\one
\with \one$ are in one-to-one correspondence with the elements of ${\mathcal S}^2$, the
proof-terms $\pair{\pair{a.\star}{b.\star}}{c.\star}$ of $(\one \with \one) \with
\one$, and also the proof-terms $\pair{a.\star}{\pair{b.\star}{c.\star}}$ of $\one
\with (\one \with \one)$, are in one-to-one correspondence with the elements
of ${\mathcal S}^3$, etc.

\begin{definition}[The set ${\mathcal V}$]
  The set ${\mathcal V}$ is inductively defined as follows: $\one \in
  {\mathcal V}$, and if $A$ and $B$ are in ${\mathcal V}$, then so is $A
  \with B$.
\end{definition}

We now show that if $A \in {\mathcal V}$, then the set of closed
irreducible proof-terms of $A$ has a structure of $\mathcal S$-semimodule.

\begin{definition}[Zero vector]
  If $A \in {\mathcal V}$, we define the proof-term $0_A$ of $A$ by induction
on $A$.  If $A = \one$, then $0_A = 0.\star$.  If $A = A_1 \with A_2$,
then $0_A = \pair{0_{A_1}}{0_{A_2}}$.
\end{definition}

\begin{lemma}[$\mathcal S$-semimodule structure~\citeintitle{Lemma 3.4}{DiazcaroDowekMSCS24}] \label{lem:vecstructure}
  If $A \in {\mathcal V}$ and $t$, $t_1$, $t_2$, and $t_3$ are closed proof-terms of
  $A$, then
  \begin{multicols}{2}
    \begin{enumerate}
      \item $(t_1 \plus t_2) \plus t_3 \equiv t_1 \plus (t_2 \plus t_3)$
      \item $t_1 \plus t_2 \equiv t_2 \plus t_1$
      \item $t \plus 0_A \equiv t$
      \item $a \dotprod b \dotprod t \equiv (a \times b) \dotprod t$
      \item $1 \dotprod t \equiv t$
      \item $a \dotprod (t_1 \plus t_2) \equiv a \dotprod t_1 \plus a \dotprod t_2$
      \item $(a + b) \dotprod t \equiv a \dotprod t \plus b \dotprod t$
	\qed
    \end{enumerate}
  \end{multicols}
\end{lemma}

\begin{definition}[Dimension of a proposition in ${\mathcal V}$]
To each proposition $A \in {\mathcal V}$, we associate a positive
natural number $d(A)$, which is the number of occurrences of the
symbol $\one$ in $A$: $d(\one) = 1$ and $d(B
\with C) = d(B) + d(C)$.
\end{definition}

If $A \in {\mathcal V}$ and $d(A) = n$, then the closed irreducible proof-terms
of $A$ and the vectors of ${\mathcal S}^n$ are in one-to-one
correspondence.

\begin{definition}[One-to-one correspondence]
\label{def:onetoone}
Let $A \in {\mathcal V}$ with $d(A) = n$.  To each closed irreducible
proof-term $t$ of $A$, we associate a vector $\underline{t}$ of ${\mathcal
  S}^n$ as follows:

  If $A = \one$, then $t = a.\star$. We let $\underline{t} =
  \left(\begin{smallmatrix} a \end{smallmatrix}\right)$.

  If $A = A_1 \with A_2$, then $t = \pair{u}{v}$.  We let
  $\underline{t}$ be the vector with two blocks $\underline{u}$ and
  $\underline{v}$: $\underline{t} = \left(\begin{smallmatrix}
    \underline{u}\\\underline{v} \end{smallmatrix}\right)$.

To each vector ${\bf u}$ of ${\mathcal S}^n$, we associate a closed irreducible proof-term $\overline{\bf u}^A$ of $A$ as follows:

If $n = 1$, then ${\bf u} = \left(\begin{smallmatrix}
  a \end{smallmatrix}\right)$. We let $\overline{\bf u}^A = a.\star$.

If $n > 1$, then $A = A_1 \with A_2$, let $n_1$ and $n_2$ be
  the dimensions of $A_1$ and $A_2$.  Let ${\bf u}_1$ and ${\bf u}_2$
  be the two blocks of ${\bf u}$ of $n_1$ and $n_2$ lines, so ${\bf u}
  = \left(\begin{smallmatrix} {\bf u}_1\\ {\bf
      u}_2\end{smallmatrix}\right)$.  We let $\overline{\bf u}^A =
    \pair{\overline{{\bf u}_1}^{A_1}}{\overline{{\bf u}_2}^{A_2}}$.
\end{definition}

We extend the definition of $\underline{t}$ to any closed proof-term of
$A$, $\underline{t}$ is by definition $\underline{t'}$ where $t'$ is
the irreducible form of $t$.

The next theorem shows that the symbol $\plus$ expresses the sum of
vectors and the symbol $\dotprod$, the product of a vector by a scalar.

\begin{theorem}[Sum and scalar product of vectors~\citeintitle{Lemmas 3.7 and 3.8}{DiazcaroDowekMSCS24}]
\label{thm:parallelsum}
Let $A \in {\mathcal V}$, $u$ and $v$ two closed proof-terms of $A$, and $a\in\mathcal S$ a scalar.
Then, $\underline{u \plus v} = \underline{u} + \underline{v}$ and
$\underline{a \dotprod u} = a \underline{u}$.
\end{theorem}

\subsubsection{Matrices}
\label{sec:secmatrices}

We now want to prove that if $A, B \in {\mathcal V}$ with $d(A) = m$
and $d(B) = n$, and $F$ is a linear function from ${\mathcal S}^m$ to
${\mathcal S}^n$, then there exists a closed proof-term $f$ of $A
\multimap B$ such that, for all vectors ${\bf u} \in {\mathcal S}^m$,
$\underline{f~\overline{\bf u}^A} = F({\bf u})$.  This can
equivalently be formulated as the fact that if $M$ is a matrix with
$m$ columns and $n$ lines, then there exists a closed proof-term $f$ of $A
\multimap B$ such that for all vectors ${\bf u} \in {\mathcal S}^m$,
$\underline{f~\overline{\bf u}^A} = M {\bf u}$.

A similar theorem has been proved also in~\cite{odot} for a non-linear
calculus. 

\begin{theorem}[Matrices~\citeintitle{Theorem 3.10}{DiazcaroDowekMSCS24}]
  \label{thm:matrices}
  Let $A, B \in {\mathcal V}$ with $d(A) = m$ and $d(B) = n$ and let $M$
  be a matrix with $m$ columns and $n$ lines, then there exists a closed
  proof-term $t$ of $A \multimap B$ such that, for all vectors ${\bf u} \in
  {\mathcal S}^m$, $\underline{t~\overline{\bf u}^A} = M {\bf u}$.
  \qed
\end{theorem}

\begin{example}[Matrices with two columns and two lines]
  The matrix
  $\left(\begin{smallmatrix} a & c\\b & d \end{smallmatrix}\right)$
  is expressed as the proof-term
    $t = \lambda x^{\one \with \one}. \elimwith^1(x,y^\one.
    \elimone(y,\pair{a.\star}{b.\star})) \plus \elimwith^2(x,z^\one.
    \elimone(z,\pair{c.\star}{d.\star}))$.
  Then,
    $t~\pair{e.\star}{f.\star} \lras
    \pair{(a \times e + c \times f).\star}{(b \times e + d \times f).\star}$.
\end{example}

\subsection{Linearity}
\label{sec:seclinearity}

This section is an extension of the linearity results from \cite{DiazcaroDowekMSCS24}.

In this section, we state the converse to Theorem~\ref{thm:matrices}, that is,
that if $A,B \in {\mathcal V}$, then each closed proof-term $t$ of $A \multimap B$
expresses a linear function.

We want to state that for any closed proof-term $t$ of $A \multimap B$, if $u_1$ and
$u_2$ are closed proof-terms of $A$, then 
\begin{equation*}
  t~(u_1 \plus u_2) \equiv t~u_1 \plus t~u_2 
  \qquad\qquad\textrm{and}\qquad\qquad
  t~(a \dotprod u_1) \equiv a \dotprod t~u_1
\end{equation*}

The property, however, is not true in general. Consider, for example,
\(
  t = \lambda x^\one.\lambda y^{\one\multimap\one}. y~x
\)
and we have 
\[
  t~(1.\star \plus\,2.\star)
  \lras \lambda y^{\one\multimap\one}. y~3.{\star}
  \not\equiv
  \lambda y^{\one\multimap\one}. {(y~1.\star)} \plus {(y~2.\star)}
  \llas
  (t~{1.\star}) \plus (t~{2.\star})
\]

Nevertheless, although the proof-terms $\lambda y^{\one\multimap\one}.
y~3.\star$ and $\lambda y^{\one\multimap\one}. {(y~1.\star)} \plus
{(y~2.\star)}$ are not equivalent, if we apply them to the identity $\lambda
z^\one. z$, then both proof-terms
$(\lambda y^{\one\multimap\one}.
y~3.\star)~\lambda z^\one. z$ and $(\lambda y^{\one\multimap\one}.
{(y~1.\star)} \plus {(y~2.\star)})~\lambda z^\one. z$ reduce to
$3.\star$.  This leads to the introduction of a notion of observational equivalence.

For convenience, if $\fv(t) \subseteq \{x\}$, we will use the notation
$t[u]$ for $(u/x)t$.

\begin{definition}[Observational equivalence]
  Two proof-terms $t_1$ and $t_2$ of a proposition $A$ are
  \textit{observationally equivalent}, $t_1 \sim t_2$, if for all propositions
  $B$ in $\mathcal V$ and for all proof-terms $c$ such that $\varnothing;x^A
  \vdash c:B$, we have $c[t_1] \equiv c[t_2]$.
\end{definition}

\begin{remark}\label{rmk:intuitionisticvslinearobsctx}
  Testing the proof-terms in an observational context $c$ where
  $x^A;\varnothing \vdash c:B$ is equivalent to testing the proof-terms in the
  observational context $\varnothing; y^{\oc A} \vdash \elimbang(y, x^A.c):B$.
\end{remark}	

We first define the measure of a proof-term (Definition~\ref{def:measure}) and elimination contexts (Definition~\ref{def:eliminationContext}), and we prove some technical results (Lemmas~\ref{lem:decomp}, \ref{horrible} and \ref{lem:measureofelimcontext}). These results are used in the proof of Lemma~\ref{lem:linearity}. As a corollary we obtain the syntactic linearity result (Theorem~\ref{thm:linearity}). Then, Corollary~\ref{cor:corollary2}
generalises the proof for any $B$ stating the observational equivalence.

\begin{definition}[Measure of a proof-term]\label{def:measure}
	We define the measure $\mu$ as follows:
  \[
    \begin{array}{rl@{\qquad }rl}
      \mu(x) &= 0 &
      \mu(t \plus u) &= 1 + \max(\mu(t), \mu(u)) \\
      \mu(a \dotprod t) &= 1 + \mu(t) &
      \mu(a.\star) &= 1 \\
      \mu(\elimone(t,u)) &= 1 + \mu(t) + \mu(u) &
      \mu(\lambda x^A.t) &= 1 + \mu(t) \\
      \mu(t~u) &= 1 + \mu(t) + \mu(u) &
      \mu(t \otimes u) &= 1 + \mu(t) + \mu(u) \\
      \mu(\elimtens(t,x^A y^B.u)) &= 1 + \mu(t) + \mu(u) &
      \mu(\topintro) &= 1\\
      \mu(\elimzero(t)) &= 1 + \mu(t) &
      \mu(\pair{t}{u}) &= 1 + \max(\mu(t), \mu(u))\\
      \mu(\elimwith^1(t,y^A.u)) &= 1 + \mu(t) + \mu(u) &
      \mu(\elimwith^2(t,y^A.u)) &= 1 + \mu(t) + \mu(u)\\
      \mu(\inl(t)) &= 1 + \mu(t) &
      \mu(\inr(t)) &= 1 + \mu(t)\\
      \multicolumn{4}{c}{\mu(\elimplus(t,y^A.u,z^B.v)) = 1 + \mu(t) + \max(\mu(u), \mu(v))}\\
	  \mu(\oc t) &= 1 + \mu(t) &
	  \mu(\elimbang(t, x^A.u)) &= 1 + \mu(t) + \mu(u)
    \end{array}
  \]
\end{definition}

\begin{definition}[Elimination context]
  An elimination context is a proof-term with a free variable, written
  $\_$, that is a proof-term in the language
  \begin{align*}
    K =& \_
    \mid \elimone(K,u)
    \mid K~u
    \mid \elimtens(K,x^A y^B.u)
    \mid \elimzero(K)\\
    \mid& \elimwith^1(K,x^A.u)
    \mid \elimwith^2(K,x^B.u)
    \mid \elimplus(K,x^A.u,y^B.v)
	\mid \elimbang(K, x^A.u)
  \end{align*}
\end{definition}

It is possible to have elimination contexts where the free variable $\_$ is in the intuitionistic or in the linear context, since it is used exactly once by definition.

\begin{restatable}[Decomposition of a proof-term]{lemma}{decomp}
\label{lem:decomp}
	If $t$ is an irreducible proof-term such that $\varnothing;x^C \vdash t:A$, then there exist an elimination context $K$, a proof-term $u$ and a proposition $B$ such that $\_^B;\varnothing \vdash K:A$ or $\varnothing; \_^B \vdash K:A$, $\varnothing;x^C \vdash u:B$, $u$ is either the variable $x$, an introduction, a sum, or a product, and $t = K[u]$.
\end{restatable}
\begin{proof}
	By induction on $t$. The proof is given in Appendix~\ref{proof:decomp}.
\end{proof}

\begin{lemma}[Decomposition of an elimination context]
  \label{horrible}
  If $K$ is an elimination context such that $\_^A; \varnothing \vdash K:B$ or $\varnothing; \_^A \vdash K:B$
  and $K \neq \_$, then $K$ has the form $K_1[K_2]$
  where $K_1$ is an elimination context and
  $K_2$ is an elimination context formed with a single
  elimination rule, that is the elimination rule of the top symbol of $A$.
\end{lemma}

\begin{proof}
  As $K$ is not $\_$, then it has the form $K = L_1[L_2]$. This is because the free variable $\_$ appears exactly once as a subterm of every elimination context.
  If $L_2 = \_$, we take $K_1 = \_$, $K_2 = L_1$ and, as the proof-term is
  well-formed, $K_2$ must be an elimination of the top symbol of $A$.
  Otherwise, by induction hypothesis, $L_2$ has the form $L_2 = K'_1[K'_2]$,
  and $K'_2$ is an elimination of the top symbol of $A$.
  Hence, $K = L_1[K'_1[K'_2]]$. We take $K_1 = L_1[K'_1]$, $K_2 = K'_2$.
\end{proof}

\begin{restatable}{lemma}{measureofelimcontext}
  \label{lem:measureofelimcontext}
  $\mu(K[t]) = \mu(K) + \mu(t)$
\end{restatable}

\begin{proof}
  By induction on $K$. The proof is given in Appendix~\ref{proof:measureofelimcontext}.
\end{proof}

\begin{restatable}{lemma}{linearity}
	\label{lem:linearity}
	If $A$ is a proposition, $B$ is a proposition of ${\mathcal V}$, $t$ is
  	a proof-term such that $\varnothing;x^A \vdash t:B$ and $u_1$ and $u_2$ are two closed proof-terms of $A$, then
	\(
    t[u_1 \plus u_2] \equiv t[u_1] \plus t[u_2]
  \)
  and
  \(
    t[a \dotprod u_1] \equiv a \dotprod t[u_1]
  \).
\end{restatable}
\begin{proof}
	The proof is given in Appendix~\ref{proof:linearity}.
\end{proof}

\begin{theorem}[Linearity]\label{thm:linearity}
	If $A$ is a proposition, $B$ is a proposition of ${\mathcal V}$, $f$ is
  	a closed proof-term of $A \multimap B$ and $u_1$ and $u_2$ are two closed proof-terms of $A$, then \(
	f~(u_1 \plus u_2) \equiv f~u_1 \plus f~u_2
  \)
  and
  \(
	f~(a \dotprod u_1) \equiv a \dotprod (f~u_1)
  \)
\end{theorem}
\begin{proof}
	By strong normalisation and the introduction property, we have that $f \lras \lambda x^A.t$, where $\varnothing;x^A \vdash t:B$. Thus,
	\begin{align*}
		f~(u_1 \plus u_2) 
		&\lras t[u_1 \plus u_2]\\
		f~u_1 \plus f~u_2
		&\lras t[u_1] \plus t[u_2]\\
		f~(a \dotprod u_1)
		&\lras t[a \dotprod u_1]\\
		a \dotprod (f~u_1)
		&\lras a \dotprod (t[u_1])
	\end{align*}

	We conclude by Lemma~\ref{lem:linearity}.
\end{proof}

\begin{restatable}{corollary}{linearitygeneral}
\label{coro:linearitygeneral}
  If $A$ and $B$ are any propositions,
  $f$ a proof-term of $A \multimap B$ and
  $u_1$ and $u_2$ two closed proof-terms of $A$, then
  \[
	f~(u_1 \plus u_2) \sim f~u_1 \plus f~u_2
  	\qquad\qquad\textrm{and}\qquad\qquad
  	f~(a \dotprod u_1) \sim a \dotprod (f~u_1)
  \]
\end{restatable}
\begin{proof}
	The proof is given in Appendix~\ref{proof:linearitygeneral}.
\end{proof}

\begin{remark}
	Notice that if in Theorem~\ref{thm:linearity} the free variable $x$ was in the intuitionistic context, that is $x^A; \varnothing \vdash t:A$, the term $\lambda x^A.t$ would not be a proof of $A \multimap B$. In this case, the term $\lambda y^{\oc A}.\elimbang(y, x^A.t)$ is a proof of the proposition $\bang A \multimap B$, and this is covered by the linearity theorem, cf. Remark~\ref{rmk:intuitionisticvslinearobsctx}.
\end{remark}

Finally, the next corollary is the converse of Theorem~\ref{thm:matrices}.

\begin{corollary}
  \label{cor:corollary2}
  Let $A, B \in {\mathcal V}$, such that $d(A) = m$ and $d(B) = n$, 
  and  $f$ be a closed proof-term of $A \multimap B$.
  Then the function $F$ from ${\mathcal S}^m$ to ${\mathcal S}^n$,
  defined as
  $F({\bf u}) = \underline{f~\overline{\bf u}^A}$ is linear.
  \qed
\end{corollary}
\begin{proof}
  Using Lemma~\ref{lem:linearity} and Theorem~\ref{thm:parallelsum}, we have
\begin{align*}
      F({\bf u} + {\bf v}) 
      = \underline{t~\overline{\bf u + \bf v}^A}
      = \underline{t~(\overline{\bf u}^A \plus \overline{\bf v}^A)} 
      &= \underline{t~\overline{\bf u}^A \plus t~\overline{\bf v}^A}
      = \underline{t~\overline{\bf u}^A} + \underline{t~\overline{\bf v}^A}
      = F({\bf u}) + F({\bf v})
      \\
      F(a {\bf u}) 
      = \underline{t~\overline{a \bf u}^A}
      = \underline{t~(a \dotprod \overline{\bf u}^A)} 
      &= \underline{a\dotprod t~\overline{\bf u}^A}
      = a \underline{t~\overline{\bf u}^A}
      = a F({\bf u})
      \tag*{\qedhere}
\end{align*}
\end{proof}

\section{Denotational semantics}\label{sec:denotationalsemantics}
\subsection{The categorical construction and its properties}
In this section, we introduce the fundamental categorical constructions that serve as the basis for defining the denotational semantics of the $\OC$-calculus. These constructions are inspired by the work of Bierman~\cite{bierman} on intuitionistic linear logic. Many of the additive properties required for our setting have been established in a recent draft~\cite{DiazcaroMalherbe24}, which develops a categorical semantics for the $\mathcal L\odot^{\mathcal S}$-calculus~\cite{DiazcaroDowekMSCS24}, an extension of the $\mathcal{L^S}$-calculus incorporating the non-deterministic \emph{sup} connective $\odot$. We adapt these constructions to incorporate the exponential connective $\oc$.

\subsubsection{Linear categories}
The first key definition is that of a linear category (Definition~\ref{def:linearcategory}). To introduce it properly, we first recall some basic notions and properties (Definitions~\ref{def:symmetricmonoidalcomonad}, \ref{def:coalgebra}, and \ref{def:commutativecomonoid}).

\begin{definition}[Symmetric monoidal comonad \citeintitle{Definition  28}{bierman}]\label{def:symmetricmonoidalcomonad}
	A comonad on a category $\mathcal C$ is a triple $(\oc,\varepsilon,\delta)$ where $\oc:\mathcal C \to \mathcal C$, and $\varepsilon:\oc \Rightarrow Id$, $\delta:\oc \Rightarrow \oc\oc$ are natural transformations which make the following diagrams commute:
\[\begin{tikzcd}[ampersand replacement=\&,cramped]
	{\oc A} \&\& {\oc \oc A} \&\&\&\& {\oc A} \\
	\\
	{\oc \oc A} \&\& {\oc \oc \oc A} \&\& {\oc A} \&\& {\oc \oc A} \&\& {\oc A}
	\arrow["{\delta_A}", from=1-1, to=1-3]
	\arrow["{\delta_A}"', from=1-1, to=3-1]
	\arrow["{\oc \delta_A}", from=1-3, to=3-3]
	\arrow[equals, from=1-7, to=3-5]
	\arrow["{\delta_A}", from=1-7, to=3-7]
	\arrow[equals, from=1-7, to=3-9]
	\arrow["{\delta_{\oc A}}"', from=3-1, to=3-3]
	\arrow["{\varepsilon_{\oc A}}", from=3-7, to=3-5]
	\arrow["{\oc \varepsilon_A}"', from=3-7, to=3-9]
\end{tikzcd}\]

If $\mathcal C$ is a symmetric monoidal category, a comonad is symmetric monoidal when $\oc$ is a symmetric monoidal functor and $\varepsilon$, $\delta$ are monoidal natural transformations.
\end{definition}

\begin{notation}
	In a symmetric monoidal category $(\mathcal C, \otimes, I)$, a symmetric monoidal comonad is denoted by\\ $(\oc, \varepsilon, \delta, m_{A,B}, m_I)$, where $m_{A,B}: \oc A \otimes \oc B \to \oc(A \otimes B)$ and $m_I: I \to \oc I$ are the structure morphisms that make $\oc$ a monoidal functor.
\end{notation}

\begin{definition}[Coalgebra \citeintitle{Definition 29}{bierman}]\label{def:coalgebra}
	Given a comonad $(\oc, \varepsilon, \delta)$ on a category $\mathcal C$, a coalgebra is a pair $(A, h_A: A \to \oc A)$ where $A$ is an object of $\mathcal C$ and $h_A$ is a morphism in $\mathcal C$ which makes the following diagrams commute:
\[\begin{tikzcd}[ampersand replacement=\&,cramped]
	A \&\& {\oc A} \&\& A \&\& {\oc A} \\
	\\
	{\oc A} \&\& {\oc \oc A} \&\&\&\& A
	\arrow["{h_A}", from=1-1, to=1-3]
	\arrow["{h_A}"', from=1-1, to=3-1]
	\arrow["{\oc h_A}", from=1-3, to=3-3]
	\arrow["{h_A}", from=1-5, to=1-7]
	\arrow[equals, from=1-5, to=3-7]
	\arrow["{\varepsilon_A}", from=1-7, to=3-7]
	\arrow["{\delta_A}"', from=3-1, to=3-3]
\end{tikzcd}\]

	A coalgebra morphism between two coalgebras $(A,h)$ and $(A',h')$ is a morphism $f:A \to A'$ in $\mathcal C$ which makes the following diagram commute:
\[\begin{tikzcd}[ampersand replacement=\&,cramped]
	A \&\& {\oc A} \\
	\\
	{A'} \&\& {\oc A'}
	\arrow["h", from=1-1, to=1-3]
	\arrow["f"', from=1-1, to=3-1]
	\arrow["{\oc f}", from=1-3, to=3-3]
	\arrow["{h'}"', from=3-1, to=3-3]
\end{tikzcd}\]
\end{definition}

\begin{definition}[Commutative comonoid \citeintitle{Definitions 31 and 32}{bierman}]\label{def:commutativecomonoid}
	Given a symmetric monoidal category\\ $(\mathcal C, \otimes, I, \alpha, \rho, \lambda, \sigma)$, a commutative comonoid in $\mathcal C$ is a triple $(C, d, e)$ where $C$ is an object in $\mathcal C$, and $d:C \to C \otimes C$ and $e:C \to I$ are morphisms in $\mathcal C$ such that the following diagrams commute:

\[\begin{tikzcd}[ampersand replacement=\&,cramped,column sep=small]
	{C \otimes C} \&\& C \&\& {C \otimes C} \&\&\& C \&\&\& C \&\& {C \otimes C} \\
	\\
	{C \otimes (C \otimes C)} \&\&\&\& {(C \otimes C) \otimes C} \& {I \otimes C} \&\& {C \otimes C} \&\& {C \otimes I} \& C \&\& {C \otimes C}
	\arrow["{id_C \otimes d}"', from=1-1, to=3-1]
	\arrow["d"', from=1-3, to=1-1]
	\arrow["d", from=1-3, to=1-5]
	\arrow["{d \otimes id_C}", from=1-5, to=3-5]
	\arrow["{\lambda_C^{-1}}"', from=1-8, to=3-6]
	\arrow["d", from=1-8, to=3-8]
	\arrow["{\rho_C^{-1}}", from=1-8, to=3-10]
	\arrow["d", from=1-11, to=1-13]
	\arrow[equals, from=1-11, to=3-11]
	\arrow["{\sigma_{C,C}}", from=1-13, to=3-13]
	\arrow["{\alpha_{C,C,C}}"', from=3-1, to=3-5]
	\arrow["{e \otimes id_C}", from=3-8, to=3-6]
	\arrow["{id_C \otimes e}"', from=3-8, to=3-10]
	\arrow["d"', from=3-11, to=3-13]
\end{tikzcd}\]

A comonoid morphism between two comonoids $(C, d, e)$ and $(C', d', e')$ is a morphism $f:C \to C'$ such that the following diagrams commute:
\[\begin{tikzcd}[ampersand replacement=\&,cramped]
	C \&\& {C'} \&\& C \&\& {C'} \\
	\\
	I \&\& I \&\& {C \otimes C} \&\& {C' \otimes C'}
	\arrow["f", from=1-1, to=1-3]
	\arrow["e"', from=1-1, to=3-1]
	\arrow["{e'}", from=1-3, to=3-3]
	\arrow["f", from=1-5, to=1-7]
	\arrow["d"', from=1-5, to=3-5]
	\arrow["{d'}", from=1-7, to=3-7]
	\arrow[equals, from=3-1, to=3-3]
	\arrow["{f \otimes f}"', from=3-5, to=3-7]
\end{tikzcd}\]
\end{definition}

\begin{definition}[Linear category \citeintitle{Definition 35}{bierman}]\label{def:linearcategory}
A linear category consists of a symmetric monoidal closed category $(\mathcal{C},\otimes,I)$ with finite products and coproducts, equipped with a symmetric monoidal comonad $(\oc, \varepsilon, \delta, m_{A,B}, m_I)$, such that:
  \begin{enumerate}
    \item For every free $\oc$-coalgebra $(\oc A, \delta_A)$ there are two distinguished monoidal natural transformations with components $e_A:\oc A \to I$ and $d_A:\oc A \to \oc A \otimes \oc A$ which form a commutative comonoid and are coalgebra morphisms.
    \item Whenever $f:(\oc A,\delta_A) \to (\oc B,\delta_B)$ is a coalgebra morphism between free coalgebras, then it is also a comonoid morphism.
  \end{enumerate}
\end{definition}

The second condition can be generalised to any coalgebra morphism:

\begin{lemma}[\hspace{-0.1mm}\citeintitle{Corollary 3}{bierman}]\label{lem:biermanprop}
	Whenever $f: (A, h) \to (A', h')$ is a coalgebra morphism between any two coalgebras, then it is a comonoid morphism.
	\qed
\end{lemma}

In the next section we use the fact that the following constructions are coalgebras:
\begin{lemma}[\hspace{-0.1mm}\citeintitle{Propositions 7 and 8}{bierman}]\label{lem:biermancoalgebras}
	Given a monoidal comonad $(\oc, \varepsilon, \delta, m_{A,B}, m_I)$ on a symmetric monoidal category $\mathcal C$, then $(\oc A, \delta_A:\oc A \to \oc \oc A)$ and $(I, m_I: I \to \oc I)$ are coalgebras.
	\qed
\end{lemma}

\begin{restatable}{lemma}{generalizacionpropnueve}
	\label{lem:generalizacionprop9}
	Given a monoidal comonad $(\oc, \varepsilon, \delta, m_{A,B}, m_I)$ on a symmetric monoidal category $\mathcal C$, and a coalgebra $(A,h: A \to \oc A)$, then $(A \otimes A, m \circ (h \otimes h))$ is a coalgebra.
\end{restatable}
\begin{proof}
	See Appendix~\ref{proof:gen9}.
\end{proof}

\subsubsection{Generalised properties}\label{subsec:genprop}
This section provides generalised versions of $\varepsilon$ and $\delta$ (from Definition~\ref{def:symmetricmonoidalcomonad}), and of $e$ and $d$ (from Definition~\ref{def:linearcategory}). These generalisations are useful for handling the typing contexts of the calculus, and follow the definitions given in \cite[Appendix~A]{Maneggia}. We give their definitions (Definition~\ref{def:generalizededeepsilon}) and several properties.

\begin{definition}[Generalised versions of $\varepsilon$, $\delta$, $e$, and $d$]\label{def:generalizededeepsilon}~
  \begin{itemize}
    \item $\varepsilon_{A_1, \dots, A_n}:\bigotimes_{i=1}^n \oc A_i \to \bigotimes_{i=1}^n A_i$
      \[
	\varepsilon_{A_1, \dots, A_n} = \left\{
	  \begin{array}{ll}
	    I \xrightarrow{id_I} I & n=0\\
	    \oc A_1 \xrightarrow{\varepsilon_{A_1}} A_1 & n=1\\
	    (\bigotimes_{i=1}^{n-1}\oc A_i) \otimes \oc A_n \xrightarrow{\varepsilon_{A_1, \dots, A_{n-1}} \otimes \varepsilon_{A_n}} \bigotimes_{i=1}^n A_i & n > 1
	  \end{array}
	\right.
      \]
     \item $\delta_{A_1, \dots, A_n}:\bigotimes_{i=1}^n \oc A_i \to \oc (\bigotimes_{i=1}^n \oc A_i)$
     \[
      \delta_{A_1, \dots, A_n} = \left\{
        \begin{array}{ll}
          I \xrightarrow{m_I} \oc I & n=0\\
          \oc A_1 \xrightarrow{\delta_{A_1}} \oc \oc A_1 & n=1\\
          (\bigotimes_{i=1}^{n-1}\oc A_i) \otimes \oc A_n \xrightarrow{\delta_{A_1, \dots, A_{n-1}} \otimes \delta_{A_n}} \oc (\bigotimes_{i=1}^{n-1}\oc A_i) \otimes \oc \oc A_n \xrightarrow{m} \oc (\bigotimes_{i=1}^n \oc A_i) & n>1
        \end{array}
      \right.
     \]
     
    \item $e_{A_1, \dots, A_n}: \bigotimes_{i=1}^n \oc A_i \to I$
      \[
	e_{A_1, \dots, A_n} = \left\{
	  \begin{array}{ll}
	    I \xrightarrow{id_I} I & n=0\\
	    \oc A_1 \xrightarrow{e_{A_1}} I & n=1\\
	    (\bigotimes_{i=1}^{n-1}\oc A_i) \otimes \oc A_n \xrightarrow{e_{A_1, \dots, A_{n-1}} \otimes e_{A_n}} I \otimes I \xrightarrow{\lambda_I} I & n > 1
	  \end{array}
	\right. 
      \]
     \item $d_{A_1, \dots, A_n}:\bigotimes_{i=1}^n \oc A_i \to (\bigotimes_{i=1}^n \oc A_i) \otimes (\bigotimes_{i=1}^n \oc A_i)$
     \[
      d_{A_1, \dots, A_n} = \left\{
        \begin{array}{ll}
          I \xrightarrow{\lambda_I^{-1}} I \otimes I & n=0\\
          \oc A_1 \xrightarrow{d_{A_1}} \oc A_1 \otimes \oc A_1 & n=1\\
          (\bigotimes_{i=1}^{n-1} \oc A_i) \otimes \oc A_n \xrightarrow{d_{A_1, \dots, A_{n-1}} \otimes d_{A_n}} ((\bigotimes_{i=1}^{n-1} \oc A_i) \otimes (\bigotimes_{i=1}^{n-1} \oc A_i)) \otimes (\oc A_n \otimes \oc A_n)\\ 
          \phantom{(\bigotimes_{i=1}^{n-1} \oc A_i) \otimes \oc A_n} \xrightarrow{id \otimes \sigma \otimes id} (\bigotimes_{i=1}^n \oc A_i) \otimes (\bigotimes_{i=1}^n \oc A_i) & n>1
        \end{array}
      \right.
     \]
  \end{itemize}
\end{definition}

\begin{restatable}{lemma}{dnat}
	  \label{lem:dnat}
		For every $n$, $d_{A_1, \dots, A_n}$ is a natural transformation.
\end{restatable}
\begin{proof}
	See Appendix~\ref{proof:genprop}.
\end{proof}

Lemmas~\ref{lem:comonoidgen} to~\ref{lem:generalizedaxiom} show that these generalised versions satisfy the properties needed in the proof of Soundness (Theorem~\ref{thm:soundness}).

\begin{restatable}{lemma}{comonoidgen}
	\label{lem:comonoidgen}
	For every $n$, $(\bigotimes_{i=1}^n \oc A_i, d_{A_1, \dots, A_n}, e_{A_1, \dots, A_n})$ is a commutative comonoid.
\end{restatable}
\begin{proof}
	See Appendix~\ref{proof:genprop}.
\end{proof}

\begin{restatable}{lemma}{coalgebragen}
	\label{lem:coalgebragen}
	For every $n$, $(\bigotimes_{i=1}^n \oc A_i, \delta_{A_1, \dots, A_n})$ is a coalgebra.
\end{restatable}
\begin{proof}
	See Appendix~\ref{proof:genprop}.
\end{proof}
	
\begin{corollary}\label{cor:deltacoalgebramorph}
	For every $n$, $\delta_{A_1, \dots, A_n}:(\bigotimes_{i=1}^n \oc A_i, \delta_{A_1, \dots, A_n}) \to (\oc (\bigotimes_{i=1}^n \oc A_i), \delta_{\bigotimes_{i=1}^n \oc A_i})$ is a coalgebra morphism.
\end{corollary}
\begin{proof}
	$(\bigotimes_{i=1}^n \oc A_i, \delta_{A_1, \dots, A_n})$ is a coalgebra by Lemma~\ref{lem:coalgebragen}. $(\oc (\bigotimes_{i=1}^n \oc A_i), \delta_{\bigotimes_{i=1}^n \oc A_i})$ is a coalgebra by Lemma~\ref{lem:biermancoalgebras}. This Corollary is a direct consequence of Lemma~\ref{lem:coalgebragen}.
\end{proof}
		
\begin{restatable}{lemma}{dcoalgebramorph}
	\label{lem:dcoalgebramorph}
	For every $n$, $d_{A_1, \dots, A_n}:(\bigotimes_{i=1}^n \oc A_i, \delta_{A_1, \dots, A_n}) \to ((\bigotimes_{i=1}^n \oc A_i) \otimes (\bigotimes_{i=1}^n \oc A_i), m \circ (\delta_{A_1, \dots, A_n} \otimes \delta_{A_1, \dots, A_n}))$ is a coalgebra morphism.
\end{restatable}
\begin{proof}
	$(\bigotimes_{i=1}^n \oc A_i, \delta_{A_1, \dots, A_n})$ is a coalgebra by Lemma~\ref{lem:coalgebragen}. $((\bigotimes_{i=1}^n \oc A_i) \otimes (\bigotimes_{i=1}^n \oc A_i), m \circ (\delta_{A_1, \dots, A_n} \otimes \delta_{A_1, \dots, A_n}))$ is a coalgebra by Lemmas~\ref{lem:coalgebragen} and \ref{lem:generalizacionprop9}. See proof details in Appendix~\ref{proof:genprop}.
\end{proof}

\begin{restatable}{lemma}{generalizedaxiom}
		\label{lem:generalizedaxiom}
		For every $n$, $e_{A_1, \dots, A_n}: (\bigotimes_{i=1}^n \oc A_i, \delta_{A_1, \dots, A_n}) \to (I,m_I)$ is a coalgebra morphism.
\end{restatable}
\begin{proof}
	$(\bigotimes_{i=1}^n \oc A_i, \delta_{A_1, \dots, A_n})$ is a coalgebra by Lemma~\ref{lem:coalgebragen}. $(I,m_I)$ is a coalgebra by Lemma~\ref{lem:biermancoalgebras}. See proof details in Appendix~\ref{proof:genprop}.
\end{proof}

\subsubsection{Semiadditive categories}
The next key ingredient is the notion of semiadditive category (Definition~\ref{def:semiadditivecategory}) and semiadditive functor (Definition~\ref{def:semiadditivefunctor}). To introduce these concepts, we need to recall the notions of zero object and zero morphism (Definition~\ref{def:zeroobject}), as well as the sum of morphisms (Definition~\ref{def:homsetsum}). We then state that any category with biproducts admits a semiadditive structure (Theorem~\ref{thm:semiadditivestructure}) and that its hom-sets carry a natural semiring structure (Corollary~\ref{cor:semiringstructure}).
We also state the preservation of the biproduct structure by semiadditive functors (Lemma~\ref{lem:semiadditiveiso}), a property we will need in what follows.

\begin{definition}[Zero object and zero morphism]\label{def:zeroobject}
	A zero object $\mathbf 0$ in a category $\mathcal C$ is an object which is both initial and terminal.
	In a category $\mathcal C$ with zero object $\mathbf 0$, a morphism $f: A \to B$ is a zero morphism when it factors through $\mathbf 0$. It is unique for every pair of objects, and we denote it by $0_{A,B}$.
\end{definition}

\begin{definition}[Sum of morphisms]\label{def:homsetsum}
  Let $\mathcal C$ be a category with biproduct, let $\Delta = \langle id, id \rangle$, $\nabla = [id, id]$ and let $f,g:A \to B$. Then, $f+g:A \to B$ is defined as the morphism that makes the following diagram commute:
\[\begin{tikzcd}[ampersand replacement=\&,cramped]
	A \&\&\& B \\
	\\
	{A \oplus A} \&\&\& {B \oplus B}
	\arrow["{f + g}", from=1-1, to=1-4]
	\arrow["\Delta"', from=1-1, to=3-1]
	\arrow["{f \oplus g}"', from=3-1, to=3-4]
	\arrow["\nabla"', from=3-4, to=1-4]
\end{tikzcd}\]
\end{definition}

\begin{definition}[Semiadditive category \citeintitle{Chapter I, Section 18}{mitchell}]\label{def:semiadditivecategory}
	A semiadditive category is a category $\mathcal C$ together with an abelian monoid structure on each of its morphism sets, subject to the following conditions:
	\begin{enumerate}
		\item The composition functions $[B \to C] \times [A \to B] \to [A \to C]$ are bilinear.
		\item The zero elements of the monoids behave as zero morphisms.
	\end{enumerate}
\end{definition}

\begin{theorem}[Semiadditive structure \citeintitle{Proposition 18.4}{mitchell}]\label{thm:semiadditivestructure}
  Suppose that $\mathcal C$ is a category with biproducts. Then $\mathcal C$ has a unique semiadditive structure, where the abelian monoid structure operation is given by Definition~\ref{def:homsetsum}, and the identity is given by the zero morphisms.
	\qed
\end{theorem}

\begin{corollary}[Semiring structure \citeintitle{Corollary 3.4}{DiazcaroMalherbe24}]\label{cor:semiringstructure}
	In a category with biproduct every $Hom(A,A)$ is a semiring, where addition is given by Definition~\ref{def:homsetsum}, with unit $0_{A,A}$ and product is given by composition, with unit $id_A$.
	\qed
\end{corollary}

\begin{definition}[Semiadditive functor]\label{def:semiadditivefunctor}
	Let $\mathcal C, \mathcal D$ be semiadditive categories. A functor $F:\mathcal C \to \mathcal D$ is semiadditive when the mapping of morphisms preserves the monoid structure.
\end{definition}

\begin{lemma}[\hspace{-.1mm}\citeintitle{Lemma 3.9}{DiazcaroMalherbe24}]\label{lem:semiadditiveiso}
  Let $F: \abscat \to \abscat$ be a semiadditive functor. Then, $F(A) \oplus F(B)$ and $F(A \oplus B)$ are naturally isomorphic, with the
  isomorphisms given by $f = \langle F(\pi_1), F(\pi_2)\rangle$ and $f^{-1} = [F(i_1), F(i_2)]$.
  \qed
\end{lemma}
\subsubsection{The category \texorpdfstring{$\abscat$}{C-!-S}}
Finally, we define the category $\abscat$ (Definition~\ref{def:abscat}), which will serve as the basis for the denotational semantics of the $\OC$-calculus.

\begin{definition}[The category $\abscat$]\label{def:abscat}
	Let $\scalars$ be a fixed semiring. We define $\abscat$ to be a linear
	category with biproduct $\oplus$ where there is a monomorphism from
	$\scalars$ to $Hom(I,I)$.
\end{definition}

\begin{notation}
  In the category $\abscat$, the symmetric monoidal structure is given by $(\abscat, \otimes, I)$ where the coherence isomorphisms are denoted by:
  \begin{align*}
    \alpha_{A,B,C} : A \otimes (B \otimes C) &\to (A \otimes B) \otimes C &
    \lambda_A : I \otimes A & \to A &
    \rho_A : A \otimes I & \to A &
    \sigma_{A,B} : A \otimes B &\to B \otimes A
  \end{align*}
  the biproduct maps are denoted by:
  \begin{align*}
    \pi_1 : A \oplus B &\to A & 
    \pi_2 : A \oplus A &\to B & 
    i_1 : A &\to A \oplus B &
    i_2 : B &\to A \oplus B
  \end{align*}
  and the semiring monomorphism is denoted by $\mono{\cdot}: \scalars \to Hom(I,I)$.
\end{notation}

The following is an example of a concrete category with the properties stated in Definition~\ref{def:abscat}.

\begin{example}\label{ex:sm}
  The category $(\sm,\otimes,\scalars,\oplus)$, where objects are semimodules over the semiring $\scalars$, arrows are semimodule homomorphisms, the tensor product is the tensor product over semimodules and the biproduct is the Cartesian product. This category is symmetric monoidal closed, and the map $\mono{a}:\scalars \to \scalars$ is defined as $b \mapsto a \cdot_\scalars b$ for every $a \in \scalars$. The first model for the $\mathcal L^\scalars$-calculus has been defined in this category \cite{DiazcaroMalherbe24v2}.

  \begin{itemize}
    \item We can define a monoidal adjunction between $\sm$ and $\set$, which induces a monoidal comonad.
  
    Let $F:\set \to \sm$ be the functor defined by its action on $X$ and $f:X \to Y$ in $\set$:
    \begin{align*}
      FX &= \{\sum_{x \in X} s_x \cdot x \mid s_x \in \scalars\} &
      Ff(\sum_{x \in X} s_x \cdot x) &= \sum_{x \in X} s_x \cdot (f x)
      \end{align*}
  
    where $\sum$ are formal finite sums.

    Let $\phi_{XY}:FX \otimes FY \to F(X \times Y)$ and $\phi_\scalars:\scalars \to F \singleton$ be the following natural transformations:
    \begin{align*}
      \phi_{XY} \left(\sum_k(\sum_i s_{ik} \cdot x_{ik}) \otimes (\sum_j s'_{jk} \cdot y_{jk})\right) &= \sum_{ijk} (s_{ik} s'_{jk}) \cdot (x_{ik}, y_{jk}) &
      \phi_\scalars(s) &= s \cdot \ast
    \end{align*}

    With these definitions, $(F,\phi)$ is a strong monoidal functor.

    Let $G:\sm \to \set$ be the forgetful functor, and let $\theta_{MN}:GM \times GN \to G(M \otimes N)$ and $\theta_\singleton:\singleton \to G\scalars$ be the following natural transformations:
    \begin{align*}
      \theta_{MN}(x, y) &= x \otimes y &
      \theta_\singleton(\ast) &= 1_\scalars
    \end{align*}

    With these definitions, $(G,\theta)$ is a monoidal functor. $(F,\phi) \dashv (G,\theta)$ is a monoidal adjunction, where the counit of adjuction $\varepsilon:FG \Rightarrow Id$ and the unit of adjunction $\eta: Id \Rightarrow GF$ are given by:
    \begin{align*}
      \varepsilon_M \left(\sum_{x \in GM} s_x \cdot x\right) &= \sum_{x \in GM} s_x \bullet_M x &
      \eta_X(x) &= 1_\scalars \cdot x
    \end{align*}

    Let $m_{MN}:FGM \otimes FGN \to FG(M \otimes N)$ and $m_\scalars:\scalars \to FG \scalars$ be the following natural transformations:
    \begin{align*}
      m_{MN} &= F \theta_{MN} \circ \phi_{GM\ GN} &
      m_\scalars &= F \theta_\singleton \circ \phi_\scalars
    \end{align*}

    With these definitions, $(FG,m)$ is a monoidal functor, and the monoidal adjunction $((F,\phi),(G,\theta),\eta,\varepsilon)$ induces the monoidal comonad $(FG,\varepsilon:FG \Rightarrow Id,\delta:FG \Rightarrow FGFG,m)$, where $\delta = F \eta G$.

    \item  For every $M$ in $\sm$, $(FGM, d_M:FGM \to FGM \otimes FGM, e_M:FGM \to \scalars)$ is a commutative comonoid, where $d_M$ and $e_M$ are monoidal natural transformations defined as:
    \begin{align*}
      d_M &= \phi^{-1} \circ F \Delta &
      e_M &= \phi^{-1} \circ F \oc
    \end{align*}

    $d_M:(FGM,\delta_M) \to (FGM \otimes FGM, m_{FGM\ FGM} \circ (\delta_M \otimes \delta_M))$ and $e_M:(FGM, \delta_M) \to (\scalars,m_\scalars)$ are coalgebra morphisms.
    
    \item In $\sm$, when $f:(FGA,\delta_A) \to (FGB,\delta_B)$ is a coalgebra morphism, it is also a comonoid morphism.
  \end{itemize}

  Therefore, $\sm$ is a linear category.
\end{example}

\subsubsection{Distribution and scalar maps}
In this section, we define the distribution maps $dist$ and $\gamma$ in
$\abscat$ (Definition~\ref{def:distmaps})
and prove that they are natural transformations
(Lemma~\ref{lem:distmaps}). We also introduce the scalar map
$\hat{a}_A$
(Lemma~\ref{lem:hatanat}),
and establish several useful properties it satisfies with respect to the unit
object, the tensor product, the biproduct, and the hom-functor
(Lemma~\ref{lem:hataprop}). 

Corollaries~\ref{cor:nablasemiadditive} and~\ref{cor:deltasemiadditive} establish properties of the distribution maps with respect to $\nabla$ and $\Delta$, respectively. These results follow from Lemmas~\ref{lem:nablasemiadditive} and~\ref{lem:deltasemiadditive}.
Finally, Lemma~\ref{lem:sigmanabladelta} states the compatibility of the monoidal symmetry induced by the biproduct with $\nabla$ and $\Delta$.

\begin{definition}[Distribution maps]\label{def:distmaps} 
	In the category $\abscat$ we define the isomorphisms $dist: (A \oplus
	B) \otimes C \to (A \otimes C) \oplus (B \otimes C)$ and $\gamma :[A
	\to B \oplus C] \to [A \to B] \oplus [A \to C]$ as follows:\footnote{The isomorphism $\gamma$ may seem unusual because it does not correspond to a theorem of Linear Logic. Although these maps are not required to be theorems of Linear Logic, note that in Section~\ref{sec:interpretation} we define the interpretation of both connectives $\oplus$ and $\with$ as the categorical biproduct. Under this interpretation, the sequents $(A \multimap B \with C) \vdash (A \multimap B) \oplus (A \multimap C)$ and $(A \multimap B \with C) \vdash (A \multimap B) \with (A \multimap C)$ are theorems of Linear Logic, which are interpreted as morphisms of the same type as $\gamma$.}
	\begin{align*}
	  dist &= \langle \pi_1 \otimes id, \pi_2 \otimes id\rangle  &\qquad\qquad  \gamma &= \langle[id \to \pi_1], [id \to \pi_2]\rangle\\
	  dist^{-1} &= [i_1 \otimes id, i_2 \otimes id]  &  \gamma^{-1} &= [[id \to i_1], [id \to i_2]]
	\end{align*}
\end{definition}

\begin{lemma}\label{lem:distmaps}
	$dist$ and $\gamma$ are natural transformations.
\end{lemma}
\begin{proof}
	$- \otimes A$ and $[A \to -]$ are semiadditive functors, and $dist$ and $\gamma$ are the natural isomorphisms corresponding to Lemma~\ref{lem:semiadditiveiso}.
\end{proof}

\begin{lemma}[Scalar map \citeintitle{Lemma 3.11}{DiazcaroMalherbe24}]\label{lem:hatanat}
	Let $a: I \to I$ in $\abscat$. The map
    $\hat{a}_A: A \to A$ where $\hat a_A = \rho_A \circ (id \otimes a) \circ \rho_A^{-1}$ defines a natural transformation between the identity functors.
    \qed
\end{lemma}

\begin{lemma}[Properties of the scalar map \citeintitle{Lemmas 3.12 and 3.14}{DiazcaroMalherbe24}]~ \label{lem:hataprop}
  \label{lem:ahathom}
	\begin{enumerate}
		\item $\hat a_I = a$
		\item $\hat a_{A \otimes B} = \hat a_A \otimes id_B = id_A \otimes \hat a_B$
		\item $\hat a_{A \oplus B} = \hat a_A \oplus \hat a_B$
		\item $\hat{a}_{[A \to B]} = [A \to \hat a_B]$
		  \qed
	\end{enumerate}
\end{lemma}

\begin{lemma}\label{lem:nablasemiadditive}
	Let $F:\abscat \to \abscat$ be a semiadditive functor. Then,
	\(
	  \nabla \circ \langle F(\pi_1), F(\pi_2) \rangle = F(\nabla)
	\).
\end{lemma}
\begin{proof}
	We show that $F(\nabla) \circ [F(i_1),F(i_2)] = \nabla$:
  \begin{align*}
    F(\nabla) \circ [F(i_1),F(i_2)] 
    &= F([id,id]) \circ [F(i_1),F(i_2)]\\ 
    &= [F([id,id]) \circ F(i_1), F([id,id]) \circ F(i_2)]\\
    &= [F([id,id] \circ i_1), F([id,id] \circ i_2)]\\
    &= [F(id), F(id)]\\
    &= [id, id]\\
    &= \nabla
  \end{align*}

  Therefore, $F(\nabla) \circ [F(i_1),F(i_2)] \circ \langle F(\pi_1), F(\pi_2) \rangle = \nabla \circ \langle F(\pi_1), F(\pi_2) \rangle$, and by the isomorphisms from Lemma~\ref{lem:semiadditiveiso}, $[F(i_1),F(i_2)] \circ \langle F(\pi_1), F(\pi_2) \rangle = id$.
\end{proof}

\begin{corollary}\label{cor:nablasemiadditive}~ 
	\begin{enumerate}
		\item $\nabla \circ dist = \nabla \otimes id$
		\item $\nabla \circ \gamma = [A \to \nabla]$
	\end{enumerate}
\end{corollary}
\begin{proof}
	$- \otimes A$ and $[A \to -]$ are semiadditive functors, we conclude by Lemma~\ref{lem:nablasemiadditive}.
\end{proof}

\begin{lemma}[\hspace{-0.1mm}\citeintitle{Lemma 3.19}{DiazcaroMalherbe24}]\label{lem:deltasemiadditive}
	Let $F: \abscat \to \abscat$ be a semiadditive functor. Then,
	\(
	  [F(i_1), F(i_2)] \circ \Delta = F(\Delta)
	\).
	\qed
\end{lemma}

\begin{corollary}\label{cor:deltasemiadditive}~ 
	\begin{enumerate}
		\item $dist^{-1} \circ \Delta = \Delta \otimes id$
		\item $\gamma^{-1} \circ \Delta = [A \to \Delta]$
	\end{enumerate}
\end{corollary}
\begin{proof}
	$- \otimes A$ and $[A \to -]$ are semiadditive functors, we conclude by Lemma~\ref{lem:deltasemiadditive}.
\end{proof}

\begin{lemma}[\hspace{-0.1mm}\citeintitle{Lemma 3.21}{DiazcaroMalherbe24}]\label{lem:sigmanabladelta}~ 
	\begin{enumerate}
		\item $(\nabla \oplus \nabla) \circ (id \oplus \sigma' \oplus id) = \nabla$
		\item $(id \oplus \sigma' \oplus id) \circ (\Delta \oplus \Delta) = \Delta$
	\end{enumerate}
	Here $\sigma'$ is the monoidal symmetry induced by the biproduct.
	\qed
\end{lemma}

\subsection{Interpretation of the calculus}\label{sec:interpretation}
In this section, we define the interpretation of the $\OC$-calculus in $\abscat$. We begin by introducing the interpretation of types (Definition~\ref{def:interpretation}) and typing contexts (Definition~\ref{def:interpretationcontext}). We then define the interpretation of the typing rules (Definition~\ref{def:ruleinterpretation}), for which we first introduce a convenient notation for the duplication of banged contexts.

\begin{definition}[Interpretation of types]\label{def:interpretation}
	\begin{align*}
		\interpretation{\one} &= I
			& \interpretation{\zero} &= \mathbf 0\\
		\interpretation{A \multimap B} &= [\interpretation{A} \to \interpretation{B}]
			& \interpretation{A \with B} &= \interpretation{A} \oplus \interpretation{B}\\
		\interpretation{A \otimes B} &= \interpretation{A} \otimes \interpretation{B}
			& \interpretation{A \oplus B} &= \interpretation{A} \oplus \interpretation{B}\\
		\interpretation{\top} &= \mathbf 0
			& \interpretation{\bang A} &= \oc \interpretation{A}
	  \end{align*}
\end{definition}

\begin{notation}
  \begin{align*}
    \bang \varnothing &= \varnothing&
    \bang (x^A, \Gamma) &= x^{\bang A}, \bang \Gamma&
    \varnothing,\Gamma &= \Gamma, \varnothing = \Gamma&
  \end{align*}
\end{notation}

\begin{definition}[Interpretation of typing contexts]\label{def:interpretationcontext}
  \begin{align*}
		\interpretation{\varnothing} &= I&
		\interpretation{x^A,\Gamma} &= \interpretation{A} \otimes \interpretation{\Gamma}
  \end{align*}
\end{definition}

\begin{notation}[Duplication of banged contexts]
  $\overline{d}_{\Upsilon, \Gamma, \Delta} : \interpretation{\bang \Upsilon} \otimes \interpretation{\Gamma} \otimes \interpretation{\Delta} \to \interpretation{\bang \Upsilon} \otimes \interpretation{\Gamma} \otimes \interpretation{\bang \Upsilon} \otimes \interpretation{\Delta}$
  \begin{align*}
    &	\overline{d}_{\Upsilon, \Gamma, \Delta} =\\
    & \left\{
      \begin{array}{ll}
	\interpretation{\bang \Upsilon} \otimes \interpretation{\Gamma} \otimes \interpretation{\Delta} \xrightarrow{d_\Upsilon \otimes id} \interpretation{\bang \Upsilon} \otimes \interpretation{\bang \Upsilon} \otimes \interpretation{\Gamma} \otimes \interpretation{\Delta} \xrightarrow{id \otimes \sigma \otimes id} \interpretation{\bang \Upsilon} \otimes \interpretation{\Gamma} \otimes \interpretation{\bang \Upsilon} \otimes \interpretation{\Delta} & \Gamma \neq \varnothing, \Delta \neq \varnothing\\
	\interpretation{\bang \Upsilon} \otimes \interpretation{\Delta} \xrightarrow{d_\Upsilon \otimes id} \interpretation{\bang \Upsilon} \otimes \interpretation{\bang \Upsilon} \otimes \interpretation{\Delta} \xrightarrow{\rho^{-1} \otimes id} \interpretation{\bang \Upsilon} \otimes I \otimes \interpretation{\bang \Upsilon} \otimes \interpretation{\Delta} & \Gamma = \varnothing, \Delta \neq \varnothing\\
	\interpretation{\bang \Upsilon} \otimes \interpretation{\Gamma} \xrightarrow{d_\Upsilon \otimes id} \interpretation{\bang \Upsilon} \otimes \interpretation{\bang \Upsilon} \otimes \interpretation{\Gamma} \xrightarrow{id \otimes \sigma} \interpretation{\bang \Upsilon} \otimes \interpretation{\Gamma} \otimes \interpretation{\bang \Upsilon} \xrightarrow{id \otimes \rho^{-1}} \interpretation{\bang \Upsilon} \otimes \interpretation{\Gamma} \otimes \interpretation{\bang \Upsilon} \otimes I & \Gamma \neq \varnothing, \Delta = \varnothing
      \end{array}
    \right.
  \end{align*}
\end{notation}

\begin{definition}[Interpretation of typing rules]\label{def:ruleinterpretation}	~

    \noindent $\ruleinterpretation{\infer[{\mbox{\small lin-ax}}]{\Upsilon; x^A \vdash x:A}{}} = \interpretation{\bang \Upsilon} \otimes \interpretation{A} \xrightarrow{e_\Upsilon \otimes id} I \otimes \interpretation{A} \xrightarrow{\lambda} \interpretation{A}$
    \smallskip

    \noindent $\ruleinterpretation{\infer[{\mbox{\small ax}}]{\Upsilon, x^A;\varnothing \vdash x:A}{}} = \interpretation{\bang \Upsilon} \otimes \oc \interpretation{A} \otimes I \xrightarrow{\rho} \interpretation{\bang \Upsilon} \otimes \oc \interpretation{A} \xrightarrow{e_\Upsilon \otimes \varepsilon_A} I \otimes \interpretation{A} \xrightarrow{\lambda} \interpretation{A}$
    \smallskip

    \noindent $\ruleinterpretation{\infer[{\mbox{\small sum}}]{{\Upsilon;} \Gamma \vdash t \plus u:A}{{\Upsilon;} \Gamma \vdash t:A & \Upsilon; \Gamma \vdash u:A}} = \interpretation{\bang \Upsilon} \otimes \interpretation{\Gamma} \xrightarrow{\Delta} (\interpretation{\bang \Upsilon} \otimes \interpretation{\Gamma}) \oplus (\interpretation{\bang \Upsilon} \otimes \interpretation{\Gamma}) \xrightarrow{\interpretation{t} \oplus \interpretation{u}} \interpretation{A} \oplus \interpretation{A} \xrightarrow{\nabla} \interpretation{A}$
    \smallskip

    \noindent $\ruleinterpretation{\infer[{\mbox{\small prod}(a)}]{\Upsilon;\Gamma \vdash a \dotprod t:A}{{\Upsilon;} \Gamma \vdash t:A}} = \interpretation{\bang \Upsilon} \otimes \interpretation{\Gamma} \xrightarrow{\interpretation{t}} \interpretation{A} \xrightarrow{\widehat{\mono{a}}} \interpretation{A}$
    \smallskip

    \noindent $\ruleinterpretation{\infer[{\mbox{\small $\one_i$}(a)}] {{\Upsilon;\varnothing} \vdash a.\star:\one}{}} = \interpretation{\bang \Upsilon} \otimes I \xrightarrow{\rho} \interpretation{\bang \Upsilon} \xrightarrow{e_\Upsilon} I \xrightarrow{\mono{a}} I$
    \smallskip

    \noindent	$\ruleinterpretation{\infer[{\mbox{\small $\one_e$}}]{{\Upsilon;}\Gamma, \Delta \vdash \elimone(t,u):A}{{\Upsilon;}\Gamma \vdash t:\one &  {\Upsilon;}\Delta \vdash u:A}} = \interpretation{\bang \Upsilon} \otimes \interpretation{\Gamma} \otimes \interpretation{\Delta} \xrightarrow{\overline{d}_{\Upsilon,\Gamma,\Delta}} \interpretation{\bang \Upsilon} \otimes \interpretation{\Gamma} \otimes \interpretation{\bang \Upsilon} \otimes \interpretation{\Delta} \xrightarrow{\interpretation{t} \otimes \interpretation{u}} I \otimes \interpretation{A} \xrightarrow{\lambda} \interpretation{A}$
    \smallskip

    \noindent	$\ruleinterpretation{\infer[{\mbox{\small $\multimap_i$}}]{\Upsilon;\Gamma \vdash \lambda x^A.t:A \multimap B}{{\Upsilon;}\Gamma, x^A \vdash t:B}} = \interpretation{\bang \Upsilon} \otimes \interpretation{\Gamma} \xrightarrow{\eta_{\interpretation{A}}} [\interpretation{A} \to \interpretation{\bang \Upsilon} \otimes \interpretation{\Gamma} \otimes \interpretation{A}] \xrightarrow{[\interpretation{A} \to \interpretation{t}]} [\interpretation{A} \to \interpretation{B}]$
    \smallskip

    \noindent	$\ruleinterpretation{\infer[{\mbox{\small $\multimap_e$}}]{{\Upsilon;}\Gamma, \Delta \vdash t~u:B}{\Upsilon;\Gamma \vdash u:A & \Upsilon;\Delta \vdash t:A\multimap B}} =
      \begin{aligned}[t]
	&\interpretation{\bang \Upsilon} \otimes \interpretation{\Gamma} \otimes \interpretation{\Delta} \xrightarrow{\overline{d}_{\Upsilon,\Gamma,\Delta}} \interpretation{\bang \Upsilon} \otimes \interpretation{\Gamma} \otimes \interpretation{\bang \Upsilon} \otimes \interpretation{\Delta} \\
	&\xrightarrow{\interpretation{u} \otimes \interpretation{t}} \interpretation{A} \otimes [\interpretation{A} \to \interpretation{B}] \xrightarrow{\sigma} [\interpretation{A} \to \interpretation{B}] \otimes \interpretation{A} \xrightarrow{\varepsilon_{\interpretation{A}}} \interpretation{B}
      \end{aligned}$
    \smallskip

    \noindent $\ruleinterpretation{\infer[{\mbox{\small $\otimes_i$}}]{{\Upsilon;}\Gamma, \Delta \vdash t \otimes u:A \otimes B}{{\Upsilon;}\Gamma \vdash t:A & {\Upsilon;}\Delta \vdash u:B}} = \interpretation{\bang \Upsilon} \otimes \interpretation{\Gamma} \otimes \interpretation{\Delta} \xrightarrow{\overline{d}_{\Upsilon,\Gamma,\Delta}} \interpretation{\bang \Upsilon} \otimes \interpretation{\Gamma} \otimes \interpretation{\bang \Upsilon} \otimes \interpretation{\Delta} \xrightarrow{\interpretation{t} \otimes \interpretation{u}} \interpretation{A} \otimes \interpretation{B}$
    \smallskip

    \noindent $\ruleinterpretation{\infer[{\mbox{\small $\otimes_e$}}]{{\Upsilon;}\Gamma, \Delta \vdash \elimtens(t, x^A y^B.u):C}{{\Upsilon;}\Gamma \vdash t:A \otimes B & {\Upsilon;}\Delta, x^A, y^B \vdash u:C}} =
      \begin{aligned}[t]
	& \interpretation{\bang \Upsilon} \otimes \interpretation{\Gamma} \otimes \interpretation{\Delta} \xrightarrow{\overline{d}_{\Upsilon,\Gamma,\Delta}} \interpretation{\bang \Upsilon} \otimes \interpretation{\Gamma} \otimes \interpretation{\bang \Upsilon} \otimes \interpretation{\Delta}\\
	&\xrightarrow{\interpretation{t} \otimes id} \interpretation{A} \otimes \interpretation{B} \otimes \interpretation{\bang \Upsilon} \otimes \interpretation{\Delta} \xrightarrow{\sigma} \interpretation{\bang \Upsilon} \otimes \interpretation{\Delta} \otimes \interpretation{A} \otimes \interpretation{B} \\
	&\xrightarrow{\interpretation{u}} \interpretation{C}
      \end{aligned} $
    \smallskip

      \noindent $\ruleinterpretation{\infer[{\mbox{\small $\top_i$}}]{{\Upsilon;} \Gamma \vdash \topintro:\top}{}} = \interpretation{\bang \Upsilon} \otimes \interpretation{\Gamma} \xrightarrow{!} \mathbf 0 $
    \smallskip

      \noindent $ \ruleinterpretation{\infer[{\mbox{\small $\zero_e$}}]{{\Upsilon;}\Gamma, \Delta \vdash \elimzero(t):C}{{\Upsilon;}\Gamma \vdash t:\zero}} = \interpretation{\bang \Upsilon} \otimes \interpretation{\Gamma} \otimes \interpretation{\Delta} \xrightarrow{\interpretation{t} \otimes id} \mathbf 0 \otimes \interpretation{\Delta} \xrightarrow{0} \interpretation{C}$
    \smallskip

      \noindent $\ruleinterpretation{\infer[{\mbox{\small $\with_i$}}]{{\Upsilon;}\Gamma \vdash \pair{t}{u}:A \with B}{{\Upsilon;}\Gamma \vdash t:A & {\Upsilon;}\Gamma \vdash u:B}} = \interpretation{\bang \Upsilon} \otimes \interpretation{\Gamma} \xrightarrow{\Delta}(\interpretation{\bang \Upsilon} \otimes \interpretation{\Gamma}) \oplus (\interpretation{\bang \Upsilon} \otimes \interpretation{\Gamma}) \xrightarrow{\interpretation{t} \oplus \interpretation{u}} \interpretation{A} \oplus \interpretation{B}$
    \smallskip

      \noindent $\ruleinterpretation{\infer[{\mbox{\small $\with_{e1}$}}]{{\Upsilon;}\Gamma, \Delta \vdash \elimwith^1(t,x^A.u):C}{{\Upsilon;}\Gamma \vdash t:A \with B & {\Upsilon;}\Delta, x^A \vdash u:C}} =
	\begin{aligned}[t]
	  & \interpretation{\bang \Upsilon} \otimes \interpretation{\Gamma} \otimes \interpretation{\Delta} \xrightarrow{\overline{d}_{\Upsilon,\Gamma,\Delta}} \interpretation{\bang \Upsilon} \otimes \interpretation{\Gamma} \otimes \interpretation{\bang \Upsilon} \otimes \interpretation{\Delta}\\
	  &\xrightarrow{\interpretation{t} \otimes id} (\interpretation{A} \oplus \interpretation{B}) \otimes \interpretation{\bang \Upsilon} \otimes \interpretation{\Delta} \xrightarrow{\pi_1 \otimes id} \interpretation{A} \otimes \interpretation{\bang \Upsilon} \otimes \interpretation{\Delta}\\ 
	  &\xrightarrow{\sigma} \interpretation{\bang \Upsilon} \otimes \interpretation{\Delta} \otimes \interpretation{A} \xrightarrow{\interpretation{u}} \interpretation{C}
	\end{aligned}$
    \smallskip

      \noindent $\ruleinterpretation{\infer[{\mbox{\small $\with_{e2}$}}]{{\Upsilon;}\Gamma, \Delta \vdash \elimwith^2(t,x^B.u):C}{{\Upsilon;}\Gamma \vdash t:A \with B & {\Upsilon;}\Delta, x^B \vdash u:C}} =
	\begin{aligned}[t]
	  & \interpretation{\bang \Upsilon} \otimes \interpretation{\Gamma} \otimes \interpretation{\Delta} \xrightarrow{\overline{d}_{\Upsilon,\Gamma,\Delta}} \interpretation{\bang \Upsilon} \otimes \interpretation{\Gamma} \otimes \interpretation{\bang \Upsilon} \otimes \interpretation{\Delta}\\
	  &\xrightarrow{\interpretation{t} \otimes id} (\interpretation{A} \oplus \interpretation{B}) \otimes \interpretation{\bang \Upsilon} \otimes \interpretation{\Delta} \xrightarrow{\pi_2 \otimes id} \interpretation{B} \otimes \interpretation{\bang \Upsilon} \otimes \interpretation{\Delta}\\
	  &\xrightarrow{\sigma} \interpretation{\bang \Upsilon} \otimes \interpretation{\Delta} \otimes \interpretation{B} 
	  \xrightarrow{\interpretation{u}} \interpretation{C}
	\end{aligned}$
    \smallskip

      \noindent $\ruleinterpretation{\infer[{\mbox{\small $\oplus_{i1}$}}]{{\Upsilon;}\Gamma \vdash \inl(t):A \oplus B}{{\Upsilon;}\Gamma \vdash t:A}} = \interpretation{\bang \Upsilon} \otimes \interpretation{\Gamma} \xrightarrow{\interpretation{t}} \interpretation{A} \xrightarrow{i_1} \interpretation{A} \oplus \interpretation{B}$
    \smallskip

      \noindent $\ruleinterpretation{\infer[{\mbox{\small $\oplus_{i2}$}}]{{\Upsilon;}\Gamma \vdash \inr(t):A \oplus B}{{\Upsilon;}\Gamma \vdash t:B}} = \interpretation{\bang \Upsilon} \otimes \interpretation{\Gamma} \xrightarrow{\interpretation{t}} \interpretation{B} \xrightarrow{i_2} \interpretation{A} \oplus \interpretation{B}$
    \smallskip

      \noindent $\ruleinterpretation{\infer[{\mbox{\small $\oplus_e$}}]{{\Upsilon;}\Gamma, \Delta \vdash \elimplus(t,x^A.u,y^B.v):C}{{\Upsilon;}\Gamma \vdash t:A \oplus B & {\Upsilon;}\Delta, x^A \vdash u:C & {\Upsilon;}\Delta, y^B \vdash v:C}} =
	\begin{aligned}[t]
	& \interpretation{\bang \Upsilon} \otimes \interpretation{\Gamma} \otimes \interpretation{\Delta} \xrightarrow{\overline{d}_{\Upsilon,\Gamma,\Delta}} \interpretation{\bang \Upsilon} \otimes \interpretation{\Gamma} \otimes \interpretation{\bang \Upsilon} \otimes \interpretation{\Delta}\\
	&\xrightarrow{\interpretation{t} \otimes id} (\interpretation{A} \oplus \interpretation{B}) \otimes \interpretation{\bang \Upsilon} \otimes \interpretation{\Delta}\\
	&\xrightarrow{dist} (\interpretation{A} \otimes \interpretation{\bang \Upsilon} \otimes \interpretation{\Delta}) \oplus (\interpretation{B} \otimes \interpretation{\bang \Upsilon} \otimes \interpretation{\Delta})\\
	&\xrightarrow{\sigma \oplus \sigma} (\interpretation{\bang \Upsilon} \otimes \interpretation{\Delta} \otimes \interpretation{A}) \oplus (\interpretation{\bang \Upsilon} \otimes \interpretation{\Delta} \otimes \interpretation{B}) \\
	&\xrightarrow{[\interpretation{u}, \interpretation{v}]} \interpretation{C}
      \end{aligned}$
    \smallskip

    \noindent $\ruleinterpretation{\infer[{\mbox{\small $\oc_i$}}]{\Upsilon; \varnothing \vdash \oc t: \oc A}{\Upsilon; \varnothing \vdash t:A}} = \interpretation{\bang \Upsilon} \otimes I \xrightarrow{\rho} \interpretation{\bang \Upsilon} \xrightarrow{\delta_\Upsilon} \oc  \interpretation{\bang \Upsilon} \xrightarrow{\oc (\rho^{-1})} \oc  (\interpretation{\bang \Upsilon} \otimes I) \xrightarrow{\oc  \interpretation{t}} \oc  \interpretation{A}$
    \smallskip

    \noindent $\ruleinterpretation{\infer[{\mbox{\small $\oc_e$}}]{\Upsilon; \Gamma, \Delta \vdash \elimbang(t, x^A.u):B}{\Upsilon; \Gamma \vdash t:\oc A & \Upsilon, x^A; \Delta \vdash u:B}} =
      \begin{aligned}[t]
	& \interpretation{\bang \Upsilon} \otimes \interpretation{\Gamma} \otimes \interpretation{\Delta} \xrightarrow{\overline{d}_{\Upsilon,\Gamma,\Delta}} \interpretation{\bang \Upsilon} \otimes \interpretation{\Gamma} \otimes \interpretation{\bang \Upsilon} \otimes \interpretation{\Delta}\\
	&\xrightarrow{\interpretation{t} \otimes id} \oc  \interpretation{A} \otimes \interpretation{\bang \Upsilon} \otimes \interpretation{\Delta}
	\xrightarrow{\sigma \otimes id} \interpretation{\bang \Upsilon} \otimes \oc  \interpretation{A} \otimes \interpretation{\Delta} \xrightarrow{\interpretation{u}} \interpretation{B}
      \end{aligned}$
\end{definition}

\begin{example}
	Let the semiring $\scalars = (\naturalnumbers, \cdot, +)$ be the natural numbers with standard multiplication and addition, and let $a, b \in \mathbb N$. The term $v = \langle a.\star, b.\star\rangle$ has type $\one \with \one$ in the empty context, corresponding to a vector type as defined in Section~\ref{sec:secvectors}. In the category $\mathsf{SM}_\naturalnumbers$ (Example~\ref{ex:sm}), the term $v$ is interpreted by the following morphism:
	\[
		\interpretation{\langle a.\star, b.\star\rangle} =
		\naturalnumbers \xrightarrow{\Delta} \naturalnumbers \times \naturalnumbers \xrightarrow{\mono{a} \times \mono{b}} \naturalnumbers \times \naturalnumbers
	\]
	\[
		n \mapsto (an, bn)
	\]

	The scalar multiplication of $v$ by $2 \in \naturalnumbers$ is interpreted as follows:
	\[
		\interpretation{2 \dotprod \langle a.\star, b.\star\rangle} =
		\naturalnumbers
		\xrightarrow{\Delta} \naturalnumbers \times \naturalnumbers
		\xrightarrow{\mono{a} \times \mono{b}} \naturalnumbers \times \naturalnumbers
		\xrightarrow{\rho^{-1}_{\naturalnumbers \times \naturalnumbers}} \naturalnumbers \times \naturalnumbers \times \naturalnumbers
		\xrightarrow{id_{\naturalnumbers \times \naturalnumbers} \times \mono{2}} \naturalnumbers \times \naturalnumbers \times \naturalnumbers
		\xrightarrow{\rho_{\naturalnumbers \times \naturalnumbers}} \naturalnumbers \times \naturalnumbers
	\]
	\[
		n \mapsto (2an, 2bn)
	\]

	The interpretation of $v$ added to itself coincides with $\interpretation{2 \bullet v}$, as expected:
	\[
		\interpretation{\langle a.\star, b.\star\rangle \plus \langle a.\star, b.\star\rangle} =
		\naturalnumbers 
		\xrightarrow{\Delta} \naturalnumbers \times \naturalnumbers 
		\xrightarrow{\Delta \times \Delta} \naturalnumbers \times \naturalnumbers \times \naturalnumbers \times \naturalnumbers 
		\xrightarrow{\mono{a} \times \mono{b} \times \mono{a} \times \mono{b}} \naturalnumbers \times \naturalnumbers \times \naturalnumbers \times \naturalnumbers 
		\xrightarrow{\nabla} \naturalnumbers \times \naturalnumbers
	\]
	\[
		n \mapsto (2an, 2bn)
	\]

	The interpretation of $v$ as a duplicable vector differs from that where each coordinate can be duplicated individually:
	\[
		\interpretation{\oc \langle a.\star, b.\star\rangle} = 
		\naturalnumbers
		\xrightarrow{m_{\naturalnumbers}} FG \naturalnumbers
		\xrightarrow{FG \Delta} FG(\naturalnumbers \times \naturalnumbers)
		\xrightarrow{FG(\mono{a} \times \mono{b})} FG(\naturalnumbers \times \naturalnumbers)
	\]
	\[
		n \mapsto n \cdot (a,b)
	\]

	\[
		\interpretation{\langle \oc a.\star, \oc b.\star \rangle} =
		\naturalnumbers
		\xrightarrow{\Delta} \naturalnumbers \times \naturalnumbers
		\xrightarrow{m_{\naturalnumbers} \times m_{\naturalnumbers}} FG \naturalnumbers \times FG \naturalnumbers
		\xrightarrow{FG \mono{a} \times FG \mono{b}} FG \naturalnumbers \times FG \naturalnumbers
	\]
	\[
		n \mapsto (n \cdot a, n \cdot b)
	\]
\end{example}

\begin{example}
	Let $\scalars = \mathbb C$, and let $H$ be the term representing the Hadamard operator $\left(\begin{smallmatrix} 1 & 1\\1 & -1 \end{smallmatrix}\right)$, encoded as in Section~\ref{sec:secmatrices}:
	\[
		H = \lambda x^{\one \with \one}. \elimwith^1(x,y^\one.
    \elimone(y,\pair{1.\star}{1.\star})) \plus \elimwith^2(x,z^\one.
    \elimone(z,\pair{1.\star}{(-1).\star}))
	\]

	This term is interpreted in the category $\mathsf{SM}_{\mathbb C}$ (Example~\ref{ex:sm}) by the following morphism:
	\begin{align*}
		\interpretation{H} =&
		\complexnumbers
		\xrightarrow{\eta_{\complexnumbers \times \complexnumbers}} \mathsf{hom}(\complexnumbers \times \complexnumbers, \complexnumbers \otimes (\complexnumbers \times \complexnumbers))
		\xrightarrow{\mathsf{hom}(\complexnumbers \times \complexnumbers, \Delta)} \mathsf{hom}(\complexnumbers \times \complexnumbers, (\complexnumbers \otimes (\complexnumbers \times \complexnumbers)) \times (\complexnumbers \otimes (\complexnumbers \times \complexnumbers)))\\
		&\xrightarrow{\mathsf{hom}(\complexnumbers \times \complexnumbers, \lambda \times \lambda)} \mathsf{hom}(\complexnumbers \times \complexnumbers, (\complexnumbers \times \complexnumbers) \times (\complexnumbers \times \complexnumbers))
		\xrightarrow{\mathsf{hom}(\complexnumbers \times \complexnumbers, \pi_1 \times \pi_2)} \mathsf{hom}(\complexnumbers \times \complexnumbers, \complexnumbers \times \complexnumbers)\\
		&\xrightarrow{\mathsf{hom}(\complexnumbers \times \complexnumbers, \Delta \times \Delta)} \mathsf{hom}(\complexnumbers \times \complexnumbers, (\complexnumbers \times \complexnumbers) \times (\complexnumbers \times \complexnumbers))
		\xrightarrow{\mathsf{hom}(\complexnumbers \times \complexnumbers, (\mono{1} \times \mono{1}) \times (\mono{1} \times \mono{-1}))} \mathsf{hom}(\complexnumbers \times \complexnumbers, (\complexnumbers \times \complexnumbers) \times (\complexnumbers \times \complexnumbers))\\
		&\xrightarrow{\mathsf{hom}(\complexnumbers \times \complexnumbers, \nabla)} \mathsf{hom}(\complexnumbers \times \complexnumbers, \complexnumbers \times \complexnumbers)
	\end{align*}
	\[
		x \mapsto ((y, z) \mapsto x(y + z, y - z))
	\]
	This is the linear map that represents the action of the Hadamard operator on a vector in $\complexnumbers^2$.
\end{example}

\subsection{Soundness}
In this section, we prove the soundness of the $\OC$-calculus with respect to
the denotational semantics
(Theorem~\ref{thm:soundness}). We begin by
introducing the notion of substitution, which plays a key role in the proof.
This is done through two lemmas: one for the substitution of an intuitionistic
variable (Lemma~\ref{lem:nonlinearsubstitution})
and another for the substitution of a linear variable
(Lemma~\ref{lem:linearsubstitution}).

The proofs of these lemmas rely not only on the properties established in previous sections, but also on several technical lemmas, which are stated and proved in Appendix~\ref{proof:technical}.

\begin{restatable}[Intuitionistic substitution]{lemma}{nonlinearsubstitution}
  \label{lem:nonlinearsubstitution}
  Let $\Upsilon, x^B; \Gamma \vdash t:A$ and $\Upsilon; \varnothing \vdash u:B$. 
Then,
\[\begin{tikzcd}[cramped,ampersand replacement=\&]
	{\interpretation{\oc \Upsilon} \otimes \interpretation{\Gamma}} \& {\interpretation{\oc \Upsilon} \otimes \interpretation{\oc \Upsilon} \otimes \interpretation{\Gamma}} \\
	\& {\interpretation{\oc \Upsilon} \otimes I \otimes \interpretation{\oc \Upsilon} \otimes \interpretation{\Gamma}} \\
	\& {\oc\interpretation{B} \otimes \interpretation{\oc \Upsilon} \otimes \interpretation{\Gamma}} \\
	{\interpretation A} \& {\interpretation{\oc \Upsilon} \otimes \oc\interpretation{B} \otimes \interpretation{\Gamma}}
	\arrow["{d_\Upsilon \otimes id}", from=1-1, to=1-2]
	\arrow["{\interpretation{(v/x)u}}"', from=1-1, to=4-1]
	\arrow["{\rho^{-1}\otimes id\otimes id}", from=1-2, to=2-2]
	\arrow["{\interpretation{\oc u}\otimes id\otimes id}", from=2-2, to=3-2]
	\arrow["\sigma", from=3-2, to=4-2]
	\arrow["{\interpretation t}", from=4-2, to=4-1]
\end{tikzcd}\]
\end{restatable}
\begin{proof}
  By induction on $t$. See Appendix~\ref{proof:soundness}.
\end{proof}

\begin{restatable}[Linear substitution]{lemma}{linearsubstitution}
  \label{lem:linearsubstitution}
	Let $\Upsilon;\Gamma \vdash v:A$ and $\Upsilon; \Delta, x^A \vdash u: B$, then the following diagram commutes:
\[\begin{tikzcd}[ampersand replacement=\&,cramped]
	{\interpretation{\oc \Upsilon} \otimes \interpretation{\Gamma} \otimes \interpretation{\Delta}} \&\&\& {\interpretation{\oc \Upsilon} \otimes \interpretation{\oc \Upsilon} \otimes \interpretation{\Gamma} \otimes \interpretation{\Delta}} \\
	\\
	\&\&\& {\interpretation{\oc \Upsilon} \otimes \interpretation{\Delta} \otimes\interpretation{\oc \Upsilon} \otimes \interpretation{\Gamma}} \\
	\\
	{\interpretation B} \&\&\& {\interpretation{\oc \Upsilon} \otimes \interpretation{\Delta} \otimes\interpretation A}
	\arrow["{d_\Upsilon \otimes id}", from=1-1, to=1-4]
	\arrow["{\interpretation{(v/x)u}}"', from=1-1, to=5-1]
	\arrow["{id \otimes \sigma}", from=1-4, to=3-4]
	\arrow["{id \otimes \interpretation v}", from=3-4, to=5-4]
	\arrow["{\interpretation u}", from=5-4, to=5-1]
\end{tikzcd}\]
\end{restatable}
\begin{proof}
  By induction on $u$. See Appendix~\ref{proof:soundness}.
\end{proof}

\begin{restatable}[Soundness]{theorem}{soundness}
  \label{thm:soundness}
  Let $\Upsilon;\Gamma \vdash t:A$. If $t \to r$, then $\interpretation{t} = \interpretation{r}$.
\end{restatable}
\begin{proof}
  By induction on the relation $\to$. See Appendix~\ref{proof:soundness}.
\end{proof}

\subsection{Adequacy}
The adequacy theorem (Theorem~\ref{thm:semanticsadequacy}) states that if two proofs are interpreted by the same morphism in our model, then they are observationally equivalent. By observationally equivalent, we mean that the proofs \emph{produce} identical results in any context.

To formally define this notion of equivalence (Definition~\ref{def:equivalence}), we first introduce the notion of \emph{elimination context} (Definition~\ref{def:eliminationContext}). Unlike the standard notion of context, an elimination context is specifically designed to eliminate connectives: a proof of an implication is applied, a proof of a pair is projected, and so on. We define elimination contexts under the invariant that if a proof of a proposition $A$ can be placed within a context yielding a proof of a proposition $B$, then $B$ is \emph{smaller} than $A$.

To this end, we begin by defining the notion of the size of a proposition (Definition~\ref{def:size}).

\begin{definition}[Size of a proposition]
  \label{def:size}
  The size of a proposition (written $|A|$) is given by
  \begin{align*}
    |\one| &= 1 &
    |A \otimes B| &= |A|+|B| &
    |\zero| &= 1 &
    |A \oplus B|  &= |A|+|B| &
    \\
    |A \multimap B| &= |A|+|B| &
    |\top| &= 1 &
    |A \with B| &= |A|+|B| &
    |\bang A|  &= |A|+1
  \end{align*}
\end{definition}

\begin{definition}[Elimination context]
  \label{def:eliminationContext}
  An elimination context is a proof produced by the following grammar, where $\_$ denotes a distinguished variable.
  \[
    K = \_
    \mid K~u
    \mid \elimtens(K,x^A y^B.v)
    \mid \elimwith^1(K,x^A.r)
    \mid \elimwith^2(K,x^B.r)
    \mid \elimplus(K,x^A.r,y^B.s)
    \mid \elimbang(K,x^A.r)
  \]
  
  where:
  \begin{itemize}
    \item In the proof $\elimtens(K,x^A y^B.v)$, $\varnothing;x^A,y^B \vdash v:C$ with $|C| < |A \otimes B|$.
    \item In the proof $\elimplus(K,x^A.r,y^B.s)$, $\varnothing;x^A \vdash r:C$ and $\varnothing;y^B \vdash s:C$ with $|C| < |A \oplus B|$.
    \item In the proof $\elimbang(K,x^A.r)$, $x^A;\varnothing \vdash r:B$ with $|B| < |\bang A|$.
  \end{itemize}
  We write $K[t]$ for $(t/\_)K$.
\end{definition}

\begin{notation}
We write ${\_}^A \vdash K : \one$ to mean either ${\_}^A ; \varnothing \vdash K : \one$ or $\varnothing ; {\_}^A \vdash K : \one$.
\end{notation}

\begin{definition}[Observational equivalence]
  \label{def:equivalence}
  A proof $\vdash t : A$ is observationally equivalent to a proof $\vdash u : A$ (written $t \equiv u$) if, for every elimination context ${\_}^A \vdash K : \one$, we have
  \[
  K[t] \lra^* a.\star \quad \text{if and only if} \quad K[u] \lra^* a.\star.
  \]
\end{definition}

\begin{restatable}[Adequacy]{theorem}{semanticsadequacy}
  \label{thm:semanticsadequacy}
  Let $\vdash t:A$ and $\vdash r:A$. If $\interpretation{t} = \interpretation{r}$, then $t \equiv r$.
\end{restatable}
\begin{proof}
  By induction on the size of $A$.
  See Appendix~\ref{proof:adequacy}.
\end{proof}

\section{Conclusion}
In this paper, we have presented the $\OC$-calculus, an extension of the
$\mathcal{L^S}$-calculus with the
exponential connective, making it a more
expressive language.  We have proved all its correctness properties, and stated its
algebraic linearity for the linear fragment, which is naturally kept.

The $\mathcal{L^S}$-calculus was originally introduced as a core
language for quantum computing. Its ability to represent matrices and vectors
makes it suitable for expressing quantum programs when taking $\mathcal
S=\mathbb C$.  Moreover, by taking $\mathcal S=\mathbb R^+$, one can consider a
probabilistic language, and by taking $\mathcal S=\{\star\}$, a linear
extension of the parallel lambda calculus~\cite{BoudolIC94}.

To consider this calculus as a proper quantum language, we would need not only
to ensure algebraic linearity but also to ensure unitarity, using techniques
such as those in~\cite{DiazcaroGuillermoMiquelValironLICS19}.
Also, the language $\mathcal{L^S}$ can be extended with a non-deterministic connective $\odot$~\cite{odot}, from which a quantum measurement operator can be encoded. We did not add such a connective to our presentation, to stay in a pure linear logic setting, however, the extension is straightforward.
Another future work is to extend the categorical model
of the $\mathcal L\odot^{\mathcal S}$-calculus given in~\cite{DiazcaroMalherbe24}.
To accommodate the polymorphic version of the $\OC$-calculus~\cite{DiazcaroDowekIvniskyMalherbeWoLLIC2024}, we would need to use hyperdoctrines~\cite{Crole},
following the approach of~\cite{Maneggia}, a direction we are willing to pursue.

\section{Acknowledgments}
This work is supported by the European Union through the MSCA SE project QCOMICAL (Grant Agreement ID: 101182520), by the Plan France 2030 through the PEPR integrated project EPiQ (ANR-22-PETQ-0007), and by the Uruguayan CSIC grant 22520220100073UD.

\bibliographystyle{abbrv}
\bibliography{biblio-paper}

\appendix
\section{Proofs of the correctness properties (Section~\ref{sec:correctness})}
  \subsection{Theorem~\ref{thm:SR} (Subject Reduction)}\label{proof:SR}
  As usual, we need a substitution lemma to prove subject reduction (Theorem~\ref{thm:SR}).

\begin{lemma}[Substitution]
  \label{lem:polysubstitution}
  ~
  \begin{enumerate}
    \item If $\Upsilon;\Gamma, x^B \vdash t:A$ and $\Upsilon;\Delta \vdash u:B$, then $\Upsilon;\Gamma, \Delta \vdash (u/x)t:A$.
    \item If $\Upsilon, x^B;\Gamma \vdash t:A$ and $\Upsilon;\varnothing\vdash u:B$, then $\Upsilon;\Gamma\vdash (u/x)t:A$.
  \end{enumerate}
\end{lemma}
\begin{proof}
  ~
  \begin{enumerate}
    \item By induction on $t$.
      \begin{itemize}
	\item If $t = x$, then $\Gamma$ is empty and $A = B$. Thus, $\Upsilon;\Gamma, \Delta \vdash (u/x)t:A$ is the same as $\Upsilon;\Delta \vdash u:B$ and this is valid by hypothesis.

	\item If $t = v \plus w$, then $\Upsilon;\Gamma, x^B \vdash v:A$ and $\Upsilon;\Gamma, x^B \vdash w:A$. By the induction hypothesis, $\Upsilon;\Gamma, \Delta \vdash (u/x)v:A$ and $\Upsilon;\Gamma, \Delta \vdash (u/x)w:A$. Therefore, by rule sum, $\Upsilon;\Gamma, \Delta \vdash (u/x)(v \plus w):A$.

	\item If $t = a \dotprod v$, then $\Upsilon;\Gamma, x^B \vdash v:A$. By the induction hypothesis, $\Upsilon;\Gamma, \Delta \vdash (u/x)v:A$. Therefore, by rule prod($a$), $\Upsilon;\Gamma, \Delta \vdash (u/x)(a \dotprod v):A$.

	\item If $t = \elimone(v, w)$, then there are $\Gamma_1, \Gamma_2$ such that $\Gamma = \Gamma_1, \Gamma_2$ and there are two cases.
	  \begin{itemize}
	    \item If $\Upsilon;\Gamma_1, x^B \vdash v:\one$ and $\Upsilon;\Gamma_2 \vdash w:A$, by the induction hypothesis $\Upsilon;\Gamma_1, \Delta \vdash (u/x)v:\one$. By rule $\one_e$, $\Upsilon;\Gamma, \Delta \vdash \elimone((u/x)v, w):A$.
	    \item If $\Upsilon;\Gamma_1 \vdash v:\one$ and $\Upsilon;\Gamma_2, x^B \vdash w:A$, by the induction hypothesis $\Upsilon;\Gamma_2, \Delta \vdash (u/x)w:A$. By rule $\one_e$, $\Upsilon;\Gamma, \Delta \vdash \elimone(v,(u/x)w):A$.
	  \end{itemize}
	  Therefore, $\Upsilon;\Gamma, \Delta \vdash (u/x)\elimone(v,w):A$.

	\item If $t = \lambda y^C.v$, then $A = C \multimap D$ and $\Upsilon;\Gamma, x^B, y^C \vdash v:D$. By the induction hypothesis, $\Upsilon;\Gamma, y^C, \Delta \vdash (u/x)v:D$. Therefore, by rule $\multimap_i$, $\Upsilon;\Gamma, \Delta \vdash (u/x)\lambda y^C.v:C \multimap D$.

	\item If $t = v~w$, then there are $\Gamma_1$ and $\Gamma_2$ such that $\Gamma = \Gamma_1, \Gamma_2$ and there are two cases.
	  \begin{itemize}
	    \item If $\Upsilon;\Gamma_1, x^B \vdash v:C \multimap A$ and $\Upsilon;\Gamma_2 \vdash w:C$, by the induction hypothesis $\Upsilon;\Gamma_1, \Delta \vdash (u/x)v:C \multimap A$. By rule $\multimap_e$, $\Upsilon;\Gamma, \Delta \vdash (u/x)v~w:A$.
	    \item If $\Upsilon;\Gamma_1 \vdash v:C \multimap A$ and $\Upsilon;\Gamma_2, x^B \vdash w:C$, by the induction hypothesis $\Upsilon;\Gamma_2, \Delta \vdash (u/x)w:C$. By rule $\multimap_e$, $\Upsilon;\Gamma, \Delta \vdash v~(u/x)w:A$.
	  \end{itemize}
	  Therefore, $\Upsilon;\Gamma, \Delta \vdash (u/x)(v~w):A$.

	\item If $t = v \otimes w$, then $A = C \otimes D$ and there are $\Gamma_1$ and $\Gamma_2$ such that $\Gamma = \Gamma_1, \Gamma_2$ and there are two cases.
	  \begin{itemize}
	    \item If $\Upsilon;\Gamma_1, x^B \vdash v:C$ and $\Upsilon;\Gamma_2 \vdash w:D$, by the induction hypothesis $\Upsilon;\Gamma_1, \Delta \vdash (u/x)v:C$. By rule $\otimes_i$, $\Upsilon;\Gamma, \Delta \vdash ((u/x)v) \otimes w:C \otimes D$.
	    \item If $\Upsilon;\Gamma_1 \vdash v:C$ and $\Upsilon;\Gamma_2, x^B \vdash w:D$, by the induction hypothesis $\Upsilon;\Gamma_2, \Delta \vdash (u/x)w:D$. By rule $\otimes_i$, $\Upsilon;\Gamma, \Delta \vdash v \otimes ((u/x)w):C \otimes D$.
	  \end{itemize}
	  Therefore, $\Upsilon;\Gamma, \Delta \vdash (u/x)(v \otimes w):C \otimes D$.

	\item If $t = \elimtens(v, y^C z^D.w)$, then there are $\Gamma_1$ and $\Gamma_2$ such that $\Gamma = \Gamma_1, \Gamma_2$ and there are two cases.
	  \begin{itemize}
	    \item If $\Upsilon;\Gamma_1, x^B \vdash v:C \otimes D$ and $\Upsilon;\Gamma_2, y^C, z^D \vdash w:A$, by the induction hypothesis $\Upsilon;\Gamma_1, \Delta \vdash (u/x)v:C \otimes D$. By rule $\otimes_e$, $\Upsilon;\Gamma, \Delta \vdash \elimtens((u/x)v, y^C z^D.w):A$.
	    \item If $\Upsilon;\Gamma_1 \vdash v:C \otimes D$ and $\Upsilon;\Gamma_2, x^B, y^C, z^D \vdash w:A$, by the induction hypothesis $\Upsilon;\Gamma_2, y^C, z^D, \Delta \vdash (u/x)w:A$. By rule $\otimes_e$, $\Upsilon;\Gamma, \Delta \vdash \elimtens(v, y^C z^D.(u/x)w):A$.
	  \end{itemize}
	  Therefore, $\Upsilon;\Gamma, \Delta \vdash (u/x)\elimtens(v, y^C z^D.w):A$.

	\item If $t = \langle\rangle$, then $A = \top$. By rule $\top_i$ $\Upsilon;\Gamma, \Delta \vdash \langle\rangle:\top$.

	\item If $t = \elimzero(v)$, then there are $\Gamma_1$ and $\Gamma_2$ such that $\Gamma = \Gamma_1, \Gamma_2$ and there are two cases.
	  \begin{itemize}
	    \item If $\Upsilon;\Gamma_1, x^B \vdash v:\zero$, by the induction hypothesis $\Upsilon;\Gamma_1, \Delta \vdash (u/x)v:\zero$. By rule $\zero_e$, $\Upsilon;\Gamma, \Delta \vdash \elimzero((u/x)v):A$.
	    \item If $\Upsilon;\Gamma_1 \vdash v:\zero$ and $x \notin \fv(v)$, by rule $\zero_e$ $\Upsilon;\Gamma, \Delta \vdash \elimzero(v):A$.
	  \end{itemize}
	  Therefore, $\Upsilon;\Gamma, \Delta \vdash (u/x)\elimzero(v):A$.

	\item If $t = \langle v, w \rangle$, then $A = C \with D$, $\Upsilon;\Gamma, x^B \vdash v:C$ and $\Upsilon;\Gamma, x^B \vdash w:D$. By the induction hypothesis, $\Upsilon;\Gamma, \Delta \vdash (u/x)v:C$ and $\Upsilon;\Gamma, \Delta \vdash (u/x)w:D$. By rule $\with_i$, $\Upsilon;\Gamma, \Delta \vdash (u/x)\langle v,w \rangle:C \with D$.

	\item If $t = \elimwith^1(v, y^C.w)$, then there are $\Gamma_1$ and $\Gamma_2$ such that $\Gamma = \Gamma_1, \Gamma_2$ and there are two cases.
	  \begin{itemize}
	    \item If $\Upsilon;\Gamma_1, x^B \vdash v:C \with D$ and $\Upsilon;\Gamma_2, y^C \vdash w:A$, by the induction hypothesis $\Upsilon;\Gamma_1, \Delta \vdash (u/x)v:C \with D$. By rule $\with_{e1}$, $\Upsilon;\Gamma, \Delta \vdash \elimwith^1((u/x)v, y^C.w):A$.
	    \item If $\Upsilon;\Gamma_1 \vdash v:C \with D$ and $\Upsilon;\Gamma_2, x^B, y^C \vdash w:A$, by the induction hypothesis $\Upsilon;\Gamma_2, y^C, \Delta \vdash (u/x)w:A$. By rule $\with_{e1}$, $\Upsilon;\Gamma, \Delta \vdash \elimwith^1(v, y^C.(u/x)w):A$.
	  \end{itemize}
	  Therefore, $\Upsilon;\Gamma, \Delta \vdash (u/x)\elimwith^1(v, y^C.w):A$.

	\item If $t = \elimwith^2(v, y^D.w)$, then there are $\Gamma_1$ and $\Gamma_2$ such that $\Gamma = \Gamma_1, \Gamma_2$ and there are two cases.
	  \begin{itemize}
	    \item If $\Upsilon;\Gamma_1, x^B \vdash v:C \with D$ and $\Upsilon;\Gamma_2, y^D \vdash w:A$, by the induction hypothesis $\Upsilon;\Gamma_1, \Delta \vdash (u/x)v:C \with D$. By rule $\with_{e2}$, $\Upsilon;\Gamma, \Delta \vdash \elimwith^2((u/x)v, y^D.w):A$.
	    \item If $\Upsilon;\Gamma_1 \vdash v:C \with D$ and $\Upsilon;\Gamma_2, x^B, y^D \vdash w:A$, by the induction hypothesis $\Upsilon;\Gamma_2, y^D, \Delta \vdash (u/x)w:A$. By rule $\with_{e2}$, $\Upsilon;\Gamma, \Delta \vdash \elimwith^2(v, y^D.(u/x)w):A$.
	  \end{itemize}
	  Therefore, $\Upsilon;\Gamma, \Delta \vdash (u/x)\elimwith^2(v, y^D.w):A$.

	\item If $t = \inl(v)$, then $A = C \oplus D$ and $\Upsilon;\Gamma, x^B \vdash v:C$. By the induction hypothesis, $\Upsilon;\Gamma, \Delta \vdash (u/x)v:C$. Therefore, by rule $\oplus_{i1}$, $\Upsilon;\Gamma, \Delta \vdash (u/x)\inl(v):C \oplus D$.

	\item If $t = \inr(v)$, then $A = C \oplus D$ and $\Upsilon;\Gamma, x^B \vdash v:D$. By the induction hypothesis, $\Upsilon;\Gamma, \Delta \vdash (u/x)v:D$. Therefore, by rule $\oplus_{i2}$, $\Upsilon;\Gamma, \Delta \vdash (u/x)\inr(v):C \oplus D$.

	\item If $t = \elimplus(v, y^C.w, z^D.s)$, then there are $\Gamma_1$ and $\Gamma_2$ such that $\Gamma = \Gamma_1, \Gamma_2$ and there are two cases.
	  \begin{itemize}
	    \item If $\Upsilon;\Gamma_1, x^B \vdash v:C \oplus D$, $\Upsilon;\Gamma_2, y^C \vdash w:A$ and $\Upsilon;\Gamma_2, z^D \vdash s:A$, by the induction hypothesis $\Upsilon;\Gamma_1, \Delta \vdash (u/x)v:C \oplus D$. By rule $\oplus_e$, $\Upsilon;\Gamma, \Delta \vdash \elimplus((u/x)v, y^C.w, z^D.s):A$.
	    \item If $\Upsilon;\Gamma_1 \vdash v:C \oplus D$, $\Upsilon;\Gamma_2, x^B, y^C \vdash w:A$ and $\Upsilon;\Gamma_2, x^B, z^D \vdash s:A$, by the induction hypothesis $\Upsilon;\Gamma_2, y:C \vdash (u/x)w:A$ and $\Upsilon;\Gamma_2, z:D \vdash (u/x)s:A$. By rule $\oplus_e$,\\ $\Upsilon;\Gamma, \Delta \vdash \elimplus(v, y^C.(u/x)w, z^D.(u/x)s):A$.
	  \end{itemize}
	  Therefore, $\Upsilon;\Gamma, \Delta \vdash (u/x)\elimplus(v, y^C.w, z^D.s):A$.

	\item If $t = \oc v$, this is not possible since the linear context should be empty.

	\item If $t = \elimbang(v, y^C.w)$, then there are $\Gamma_1$ and $\Gamma_2$ such that $\Gamma = \Gamma_1, \Gamma_2$ and there are two cases.
	  \begin{itemize}
	    \item If $\Upsilon; \Gamma_1, x^B \vdash v:\oc C$ and $\Upsilon, y^C; \Gamma_2 \vdash w:A$, by the induction hypothesis $\Upsilon; \Gamma_1, \Delta \vdash (u/x) v:\oc C$. By rule $\oc_e$, $\Upsilon; \Gamma, \Delta \vdash \elimbang((u/x)v, y^C.w):A$.
	    \item If $\Upsilon; \Gamma_1 \vdash v:\oc C$ and $\Upsilon, y^C; \Gamma_2, x^B \vdash w:A$, by the induction hypothesis $\Upsilon, y^C; \Gamma_2, \Delta \vdash (u/x)w:A$. By rule $\oc_e$,\\ $\Upsilon; \Gamma, \Delta \vdash \elimbang(v, y^C.(u/x)w):A$.
	  \end{itemize}
	  Therefore, $\Upsilon; \Gamma, \Delta \vdash (u/x)\elimbang(v, y^C.w):A$.

      \end{itemize}
    \item By induction on $t$.
      \begin{itemize}
	\item If $t = x$, then $\Gamma$ is empty and $A = B$. Thus, $\Upsilon;\Gamma \vdash (u/x)t:A$ is the same as $\Upsilon;\varnothing \vdash u:B$ and this is valid by hypothesis.

	\item If $t = y \neq x$, then either $\Gamma = \{y^A\}$ or $\Gamma$ is empty and $y^A \in \Upsilon$.
	  \begin{itemize}
	    \item In the first case, $\Upsilon; y^A \vdash y:A$ by rule lin-ax.
	    \item In the second case, $\Upsilon; \varnothing \vdash y:A$ by rule ax.
	  \end{itemize}
	  Therefore, $\Upsilon; \Gamma \vdash y:A$.

	\item If $t = v \plus w$, then $\Upsilon, x^B;\Gamma \vdash v:A$ and $\Upsilon, x^B;\Gamma \vdash w:A$. By the induction hypothesis, $\Upsilon;\Gamma \vdash (u/x)v:A$ and $\Upsilon;\Gamma \vdash (u/x)w:A$. Therefore, by rule sum, $\Upsilon;\Gamma \vdash (u/x)(v \plus w):A$.

	\item If $t = a \dotprod v$, then $\Upsilon, x^B;\Gamma \vdash v:A$. By the induction hypothesis, $\Upsilon;\Gamma \vdash (u/x)v:A$. Therefore, by rule prod($a$), $\Upsilon;\Gamma \vdash (u/x)(a \dotprod v):A$.

	\item If $t = \elimone(v, w)$, then there are $\Gamma_1, \Gamma_2$ such that $\Gamma = \Gamma_1, \Gamma_2$, $\Upsilon, x^B; \Gamma_1 \vdash v:\one$ and $\Upsilon, x^B; \Gamma_2 \vdash w:A$. By the induction hypothesis, $\Upsilon; \Gamma_1 \vdash (u/x)v:\one$ and $\Upsilon; \Gamma_2 \vdash (u/x)w:A$. By rule $\one_e$, $\Upsilon; \Gamma \vdash (u/x)\elimone(v, w):A$.

	\item If $t = \lambda y^C.v$, then $A = C \multimap D$ and $\Upsilon, x^B;\Gamma, y^C \vdash v:D$. By the induction hypothesis, $\Upsilon;\Gamma, y^C \vdash (u/x)v:D$. Therefore, by rule $\multimap_i$, $\Upsilon;\Gamma \vdash (u/x)\lambda y^C.v:C \multimap D$.

	\item If $t = v~w$, then there are $\Gamma_1$ and $\Gamma_2$ such that $\Gamma = \Gamma_1, \Gamma_2$, $\Upsilon, x^B; \Gamma_1 \vdash v:C \multimap A$ and $\Upsilon, x^B; \Gamma_2 \vdash w:C$. By the induction hypothesis, $\Upsilon; \Gamma_1 \vdash (u/x)v:C \multimap A$ and $\Upsilon; \Gamma_2 \vdash (u/x)w:C$. By rule $\multimap_e$, $\Upsilon;\Gamma \vdash (u/x)(v~w):A$.

	\item If $t = v \otimes w$, then $A = C \otimes D$ and there are $\Gamma_1$ and $\Gamma_2$ such that $\Gamma = \Gamma_1, \Gamma_2$, $\Upsilon, x^B; \Gamma_1 \vdash v:C$ and $\Upsilon, x^B; \Gamma_2 \vdash w:D$. By the induction hypothesis, $\Upsilon; \Gamma_1 \vdash (u/x)v:C$ and $\Upsilon; \Gamma_2 \vdash (u/x)w:D$. By rule $\otimes_i$, $\Upsilon;\Gamma \vdash (u/x)(v \otimes w):C \otimes D$.

	\item If $t = \elimtens(v, y^C z^D.w)$, then there are $\Gamma_1$ and $\Gamma_2$ such that $\Gamma = \Gamma_1, \Gamma_2$, $\Upsilon, x^B; \Gamma_1 \vdash v:C \otimes D$ and $\Upsilon, x^B; \Gamma_2, y^C, z^D \vdash w:A$. By the induction hypothesis, $\Upsilon; \Gamma_1 \vdash (u/x)v:C \otimes D$ and $\Upsilon; \Gamma_2, y^C, z^D \vdash (u/x)w:A$. By rule $\otimes_e$, $\Upsilon;\Gamma \vdash (u/x)\elimtens(v, y^C z^D.w):A$.

	\item If $t = \langle\rangle$, then $A = \top$. By rule $\top_i$ $\Upsilon;\Gamma \vdash \langle\rangle:\top$.

	\item If $t = \elimzero(v)$, then there are $\Gamma_1$ and $\Gamma_2$ such that $\Gamma = \Gamma_1, \Gamma_2$ and $\Upsilon, x^B; \Gamma_1 \vdash v:\zero$. By the induction hypothesis $\Upsilon; \Gamma_1 \vdash (u/x)v:\zero$. By rule $\zero_e$, $\Upsilon;\Gamma \vdash (u/x)\elimzero(v):A$.

	\item If $t = \langle v, w \rangle$, then $A = C \with D$, $\Upsilon, x^B;\Gamma \vdash v:C$ and $\Upsilon, x^B;\Gamma \vdash w:D$. By the induction hypothesis, $\Upsilon;\Gamma, \vdash (u/x)v:C$ and $\Upsilon;\Gamma, \vdash (u/x)w:D$. By rule $\with_i$, $\Upsilon;\Gamma \vdash (u/x)\langle v,w \rangle:C \with D$.

	\item If $t = \elimwith^1(v, y^C.w)$, then there are $\Gamma_1$ and $\Gamma_2$ such that $\Gamma = \Gamma_1, \Gamma_2$, $\Upsilon, x^B; \Gamma_1 \vdash v:C \with D$ and $\Upsilon, x^B; \Gamma_2, y^C \vdash w:A$. By the induction hypothesis, $\Upsilon; \Gamma_1 \vdash (u/x)v:C \with D$ and $\Upsilon; \Gamma_2, y^C \vdash (u/x)w:A$. By rule $\with_{e1}$, $\Upsilon;\Gamma \vdash (u/x)\elimwith^1(v, y^C.w):A$.

	\item If $t = \elimwith^2(v, y^D.w)$, then there are $\Gamma_1$ and $\Gamma_2$ such that $\Gamma = \Gamma_1, \Gamma_2$, $\Upsilon, x^B; \Gamma_1 \vdash v:C \with D$ and $\Upsilon, x^B; \Gamma_2, y^C \vdash w:A$. By the induction hypothesis, $\Upsilon; \Gamma_1 \vdash (u/x)v:C \with D$ and $\Upsilon; \Gamma_2, y^C \vdash (u/x)w:A$. By rule $\with_{e2}$, $\Upsilon;\Gamma \vdash (u/x)\elimwith^2(v, y^D.w):A$.

	\item If $t = \inl(v)$, then $A = C \oplus D$ and $\Upsilon, x^B;\Gamma \vdash v:C$. By the induction hypothesis, $\Upsilon;\Gamma \vdash (u/x)v:C$. Therefore, by rule $\oplus_{i1}$, $\Upsilon;\Gamma \vdash (u/x)\inl(v):C \oplus D$.

	\item If $t = \inr(v)$, then $A = C \oplus D$ and $\Upsilon, x^B;\Gamma \vdash v:D$. By the induction hypothesis, $\Upsilon;\Gamma \vdash (u/x)v:D$. Therefore, by rule $\oplus_{i2}$, $\Upsilon;\Gamma \vdash (u/x)\inr(v):C \oplus D$.

	\item If $t = \elimplus(v, y^C.w, z^D.s)$, then there are $\Gamma_1$ and $\Gamma_2$ such that $\Gamma = \Gamma_1, \Gamma_2$, $\Upsilon, x^B; \Gamma_1 \vdash v:C \oplus D$, $\Upsilon, x^B; \Gamma_2, y^C \vdash w:A$ and $\Upsilon, x^B; \Gamma_2, z^D \vdash s:A$. By the induction hypothesis, $\Upsilon; \Gamma_1 \vdash (u/x)v:C \oplus D$, $\Upsilon; \Gamma_2, y^C \vdash (u/x)w:A$ and $\Upsilon; \Gamma_2, z^D \vdash (u/x)s:A$. By rule $\oplus_e$,\\ $\Upsilon;\Gamma \vdash (u/x)\elimplus(v, y^C.w, z^D.s):A$.

	\item If $t = \oc v$, then $A = \oc C$, $\Gamma$ is empty and $\Upsilon, x^B; \varnothing \vdash v:C$. By the induction hypothesis $\Upsilon; \varnothing \vdash (u/x)v:C$, and by rule $\oc_i$ $\Upsilon; \varnothing \vdash (u/x)\oc v: \oc C$.

	\item If $t = \elimbang(v, y^C.w)$, then there are $\Gamma_1$ and $\Gamma_2$ such that $\Gamma = \Gamma_1, \Gamma_2$, $\Upsilon, x^B; \Gamma_1 \vdash v:\oc C$ and $\Upsilon, x^B; \Gamma_2, y^C \vdash w:A$. By the induction hypothesis, $\Upsilon; \Gamma_1 \vdash (u/x)v:\oc C$ and $\Upsilon; \Gamma_2, y^C \vdash (u/x)w:A$. By rule $\oc_e$, $\Upsilon; \Gamma \vdash (u/x)\elimbang(v, y^C.w):A$.\qedhere
      \end{itemize}
  \end{enumerate}
\end{proof}

\SR*
\begin{proof}
  By induction on the relation $\lra$. 
  \begin{itemize}
    \item If $t = \elimone(a.\star, v)$ and $u = a \dotprod v$, then $\Upsilon;\Gamma \vdash v:A$. Therefore, by rule prod($a$), $\Upsilon;\Gamma \vdash a \dotprod v:A$.

    \item If $t = (\lambda x^B.v_1)~v_2$ and $u = (v_2/x)v_1$, then there are $\Gamma_1$ and $\Gamma_2$ such that $\Gamma = \Gamma_1, \Gamma_2$, $\Upsilon;\Gamma_1, x^B \vdash v_1:A$ and $\Upsilon;\Gamma_2 \vdash v_2:B$. By Lemma~\ref{lem:polysubstitution}, $\Upsilon;\Gamma \vdash (v_2/x)v_1:A$.

    \item If $t = \elimtens(v_1 \otimes v_2, x^B y^C.v_3)$ and $u = (v_1/x, v_2/y)v_3$, then there are $\Gamma_1$, $\Gamma_2$ and $\Gamma_3$ such that $\Gamma = \Gamma_1, \Gamma_2, \Gamma_3$, $\Upsilon;\Gamma_1 \vdash v_1:B$, $\Upsilon;\Gamma_2 \vdash v_2:C$ and $\Upsilon;\Gamma_3, x^B, y^C \vdash v_3:A$. By Lemma~\ref{lem:polysubstitution} twice, $\Upsilon;\Gamma \vdash (v_1/x, v_2/y)v_3:A$.

    \item If $t = \elimwith^1(\langle v_1, v_2 \rangle, x^B.v_3)$ and $u = (v_1/x)v_3$, then there are $\Gamma_1$ and $\Gamma_2$ such that $\Gamma = \Gamma_1, \Gamma_2$, $\Upsilon;\Gamma_1 \vdash v_1:B$, $\Upsilon;\Gamma_1 \vdash v_2:C$ and $\Upsilon;\Gamma_2, x^B \vdash v_3:A$. By Lemma~\ref{lem:polysubstitution}, $\Upsilon;\Gamma \vdash (v_1/x)v_3:A$.

    \item If $t = \elimwith^2(\langle v_1, v_2 \rangle, x^B.v_3)$ and $u = (v_2/x)v_3$, then there are $\Gamma_1$ and $\Gamma_2$ such that $\Gamma = \Gamma_1, \Gamma_2$, $\Upsilon;\Gamma_1 \vdash v_1:C$, $\Upsilon;\Gamma_1 \vdash v_2:B$ and $\Upsilon;\Gamma_2, x^B \vdash v_3:A$. By Lemma~\ref{lem:polysubstitution}, $\Upsilon;\Gamma \vdash (v_2/x)v_3:A$.

    \item If $t = \elimplus(\inl(v_1), x^B.v_2, y^C.v_3)$ and $u = (v_1/x)v_2$, then there are $\Gamma_1$ and $\Gamma_2$ such that $\Gamma = \Gamma_1, \Gamma_2$, $\Upsilon;\Gamma_1 \vdash v_1:B$, $\Upsilon;\Gamma_2, x^B \vdash v_2:A$ and $\Upsilon;\Gamma_2, y^C \vdash v_3:A$. By Lemma~\ref{lem:polysubstitution}, $\Upsilon;\Gamma \vdash (v_1/x)v_2:A$.

    \item If $t = \elimplus(\inr(v_1), x^B.v_2, y^C.v_3)$ and $u = (v_1/y)v_3$, then there are $\Gamma_1$ and $\Gamma_2$ such that $\Gamma = \Gamma_1, \Gamma_2$, $\Upsilon;\Gamma_1 \vdash v_1:C$, $\Upsilon;\Gamma_2, x^B \vdash v_2:A$ and $\Upsilon;\Gamma_2, y^C \vdash v_3:A$. By Lemma~\ref{lem:polysubstitution}, $\Upsilon;\Gamma \vdash (v_1/y)v_3:A$.

    \item If $t = \elimbang(\oc v_1, x^B.v_2)$ and $u = (v_1/x)v_2$, then $\Upsilon; \varnothing \vdash v_1:B$ and $\Upsilon, x^B; \Gamma \vdash v_2:A$. By Lemma~\ref{lem:polysubstitution}, $\Upsilon; \Gamma \vdash (v_1/x)v_2:A$.

    \item If $t = {a.\star} \plus b.\star$ and $u = (a+b).\star$, then $A = \one$ and $\Gamma$ is empty. By rule $\one_i$($a+b$), $\Upsilon; \varnothing \vdash (a+b).\star:\one$.

    \item If $t = (\lambda x^B.v_1) \plus (\lambda x^B.v_2)$ and $u = \lambda x^B.(v_1 \plus v_2)$, then $A = B \multimap C$, $\Upsilon; \Gamma, x^B \vdash v_1:C$ and $\Upsilon; \Gamma, x^B \vdash v_2:C$. By rule sum, $\Upsilon; \Gamma, x^B \vdash v_1 \plus v_2:C$. Therefore, by rule $\multimap_i$, $\Upsilon; \Gamma \vdash \lambda x^B.(v_1 \plus v_2): B \multimap C$.

    \item If $t = \elimtens(v_1 \plus v_2, x^B y^C.v_3)$ and $u = \elimtens(v_1, x^B y^C.v_3) \plus \elimtens(v_2, x^B y^C.v_3)$, then there are $\Gamma_1$ and $\Gamma_2$ such that $\Gamma = \Gamma_1, \Gamma_2$, $\Upsilon; \Gamma_1 \vdash v_1:B \otimes C$, $\Upsilon; \Gamma_1 \vdash v_2:B \otimes C$ and $\Upsilon; \Gamma_2, x^B, y^C \vdash v_3:A$. By rule $\otimes_e$, $\Upsilon; \Gamma \vdash \elimtens(v_1, x^B y^C.v_3):A$ and $\Upsilon; \Gamma \vdash \elimtens(v_2, x^B y^C.v_3):A$. Therefore, by rule sum, $\Upsilon; \Gamma \vdash \elimtens(v_1, x^B y^C.v_3) \plus \elimtens(v_2, x^B y^C.v_3):A$.

    \item If $t = \langle\rangle \plus \langle\rangle$ and $u = \langle\rangle$, then $A = \top$. By rule $\top_i$, $\Upsilon;\Gamma \vdash \langle\rangle:\top$.

    \item If $t = \langle v_1, v_2 \rangle \plus \langle v_3, v_4 \rangle$ and $u = \langle v_1 \plus v_3, v_2 \plus v_4 \rangle$, then $A = B \with C$, $\Upsilon; \Gamma \vdash v_1:B$, $\Upsilon; \Gamma \vdash v_2:C$, $\Upsilon; \Gamma \vdash v_3:B$ and $\Upsilon; \Gamma \vdash v_4:C$. By rule sum, $\Upsilon; \Gamma \vdash v_1 \plus v_3:B$ and $\Upsilon; \Gamma \vdash v_2 \plus v_4:C$. Therefore, by rule $\with_i$, $\Upsilon; \Gamma \vdash \langle v_1 \plus v_3, v_2 \plus v_4 \rangle:B \with C$.

    \item If $t = \elimplus(v_1 \plus v_2, x^B.v_3, y^C.v_4)$ and $u = \elimplus(v_1, x^B.v_3, y^C.v_4) \plus \elimplus(v_2, x^B.v_3, y^C.v_4)$, then there are $\Gamma_1$ and $\Gamma_2$ such that $\Gamma = \Gamma_1, \Gamma_2$, $\Upsilon; \Gamma_1 \vdash v_1:B \oplus C$, $\Upsilon; \Gamma_1 \vdash v_2:B \oplus C$, $\Upsilon; \Gamma_2, x^B \vdash v_3:A$ and $\Upsilon; \Gamma_2, y^C \vdash v_4:A$. By rule $\oplus_e$, $\Upsilon; \Gamma \vdash \elimplus(v_1, x^B.v_3, y^C.v_4):A$ and $\Upsilon; \Gamma \vdash \elimplus(v_2, x^B.v_3, y^C.v_4):A$. Therefore, by rule sum, $\Upsilon; \Gamma \vdash \elimplus(v_1, x^B.v_3, y^C.v_4) \plus \elimplus(v_2, x^B.v_3, y^C.v_4):A$.

    \item If $t = \elimbang(t \plus u, x^B.v)$ and $u = \elimbang(t, x^B.v)
      \plus \elimbang(u, x^B.v)$, the $\Gamma = \Gamma_1, \Gamma_2$,
      $\Upsilon;\Gamma_1 \vdash t:\oc B$, $\Upsilon;\Gamma_1 \vdash u:\oc B$
      and $\Upsilon, x^B; \Gamma_2 \vdash v:A$. By rule $\oc_e$,
      $\Upsilon;\Gamma \vdash \elimbang(t, x^B.v):A$ and $\Upsilon;\Gamma
      \vdash \elimbang(u, x^B.v):A$. By rule sum, $\Upsilon;\Gamma \vdash
    \elimbang(t, x^B.v) \plus \elimbang(u, x^B.v):A$.

    \item If $t = a \dotprod b.\star$ and $u = (a \times b).\star$, then $A = \one$ and $\Gamma$ is empty. By rule $\one_i$($a \times b$), $\Upsilon; \varnothing \vdash (a \times b).\star:\one$.

    \item If $t = a \dotprod \lambda x^B.v$ and $u = \lambda x^B.a \dotprod v$, then $A = B \multimap C$ and $\Upsilon; \Gamma, x^B \vdash v:C$. By rule prod($a$), $\Upsilon; \Gamma, x^B \vdash a \dotprod v:C$. Therefore, by rule $\multimap_i$, $\Upsilon; \Gamma \vdash \lambda x^B.a \dotprod v:B \multimap C$.

    \item If $t = \elimtens(a \dotprod v_1, x^B y^C.v_2)$ and $u = a \dotprod \elimtens(v_1, x^B y^C.v_2)$, then there are $\Gamma_1$ and $\Gamma_2$ such that $\Gamma = \Gamma_1, \Gamma_2$, $\Upsilon; \Gamma_1 \vdash v_1:B \otimes C$ and $\Upsilon; \Gamma_2, x^B, y^C \vdash v_2:A$. By rule $\otimes_e$, $\Upsilon; \Gamma \vdash \elimtens(v_1, x^B y^C.v_2):A$. Therefore, by rule prod($a$), $\Upsilon; \Gamma \vdash a \dotprod \elimtens(v_1, x^B y^C.v_2):A$.

    \item If $t = a \dotprod \langle\rangle$ and $u = \langle\rangle$, then $A = \top$. Therefore, by rule $\top_i$, $\Upsilon; \Gamma \vdash \langle\rangle:\top$.

    \item If $t = a \dotprod \langle v_1, v_2 \rangle$ and $u = \langle a \dotprod v_1, a \dotprod v_2 \rangle$, then $A = B \with C$, $\Upsilon; \Gamma \vdash v_1:B$ and $\Upsilon; \Gamma \vdash v_2:C$. By rule prod($a$), $\Upsilon; \Gamma \vdash a \dotprod v_1:B$ and $\Upsilon; \Gamma \vdash a \dotprod v_2:C$. Therefore, by rule $\with_i$, $\Upsilon; \Gamma \vdash \langle a \dotprod v_1, a \dotprod v_2 \rangle:B \with C$.

    \item If $t = \elimplus(a \dotprod v_1, x^B.v_2, y^C.v_3)$ and $u = a \dotprod \elimplus(v_1, x^B.v_2, y^C.v_3)$, then there are $\Gamma_1$ and $\Gamma_2$ such that $\Gamma = \Gamma_1, \Gamma_2$, $\Upsilon; \Gamma_1 \vdash v_1:B \oplus C$, $\Upsilon; \Gamma_2, x^B \vdash v_2:A$ and $\Upsilon; \Gamma_2, y^C \vdash v_3:A$. By rule $\oplus_e$, $\Upsilon; \Gamma \vdash \elimplus(v_1, x^B.v_2, y^C.v_3):A$. Therefore, by rule prod($a$), $\Upsilon; \Gamma \vdash a \dotprod \elimplus(v_1, x^B.v_2, y^C.v_3):A$.

    \item If $t = \elimbang(a \dotprod t, x^B.v)$ and $u = a \dotprod \elimbang(t, x^B.v)$, then $\Gamma = \Gamma_1, \Gamma_2$, $\Upsilon,\Gamma_1 \vdash t:\oc B$ and $\Upsilon,x^B;\Gamma_2 \vdash v:A$. By rule $\oc_e$, $\Upsilon;\Gamma \vdash \elimbang(t, x^B.v):A$. By rule prod($a$), $\Upsilon;\Gamma \vdash a \dotprod \elimbang(t, x^B.v):A$.
      \qedhere
  \end{itemize}
\end{proof}

\subsection{Lemma~\ref{lem:typeinterpretationsarerc}}\label{proof:ST}
\typeinterpretationsarerc*
\begin{proof}
  By induction on $A$.
  \begin{itemize}
    \item If $A = \one$, $\SN$ has the properties CR1, CR2, and CR3.
    \item If $A = B \multimap C$:
      \begin{itemize}
        \item Let $t \in \llbracket B \rrbracket \hatmultimap \llbracket C \rrbracket$, then $t \in \SN$.
        \item Let $t \in \llbracket B \rrbracket \hatmultimap \llbracket C \rrbracket$ such that $t \lra t'$. Then $t' \in \SN$, and if $t' \lras \lambda x^D.u$, $t \lras \lambda x^D.u$. Therefore, for all $v \in \llbracket B \rrbracket$, $(v/x)u \in \llbracket C \rrbracket$.
        \item Let $t$ be a proof-term that is not an introduction such that $Red(t) \subseteq \llbracket B \rrbracket \hatmultimap \llbracket C \rrbracket$. Since $Red(t) \subseteq \SN$, $t \in \SN$. If $t \lras \lambda x^D.u$, the rewrite sequence has at least one step because $t$ is not an introduction. Then, there is a proof-term $t' \in Red(t)$ such that $t' \lras \lambda x^D.u$. Therefore, for all $v \in \llbracket B \rrbracket$, $(v/x)u \in \llbracket C \rrbracket$.
      \end{itemize}
    \item If $A = B \otimes C$:
      \begin{itemize}
        \item Let $t \in \llbracket B \rrbracket \hatotimes \llbracket C \rrbracket$, then $t \in \SN$.
        \item Let $t \in \llbracket B \rrbracket \hatotimes \llbracket C \rrbracket$ such that $t \lra t'$. Then $t' \in \SN$, and if $t' \lras u \otimes v$, $t \lras u \otimes v$. Therefore, $u \in \llbracket B \rrbracket$ and $v \in \llbracket C \rrbracket$.
        \item Let $t$ be a proof-term that is not an introduction such that $Red(t) \subseteq \llbracket B \rrbracket \hatotimes \llbracket C \rrbracket$. Since $Red(t) \subseteq \SN$, $t \in \SN$. If $t \lras u \otimes v$, the rewrite sequence has at least one step because $t$ is not an introduction. Then, there is a proof-term $t' \in Red(t)$ such that $t' \lras u \otimes v$. Therefore, $u \in \llbracket B \rrbracket$ and $v \in \llbracket C \rrbracket$.
      \end{itemize}
    \item If $A = \top$, $\SN$ has the properties CR1, CR2, and CR3.
    \item If $A = \zero$, $\SN$ has the properties CR1, CR2, and CR3.
    \item If $A = B \with C$:
    \begin{itemize}
      \item Let $t \in \llbracket B \rrbracket \hatand \llbracket C \rrbracket$, then $t \in \SN$.
      \item Let $t \in \llbracket B \rrbracket \hatand \llbracket C \rrbracket$ such that $t \lra t'$. Then $t' \in \SN$, and if $t' \lras \langle u,v \rangle$, $t \lras \langle u,v \rangle$. Therefore, $u \in \llbracket B \rrbracket$ and $v \in \llbracket C \rrbracket$.
      \item Let $t$ be a proof-term that is not an introduction such that $Red(t) \subseteq \llbracket B \rrbracket \hatand \llbracket C \rrbracket$. Since $Red(t) \subseteq \SN$, $t \in \SN$. If $t \lras \langle u,v \rangle$, the rewrite sequence has at least one step because $t$ is not an introduction. Then, there is a proof-term $t' \in Red(t)$ such that $t' \lras \langle u,v \rangle$. Therefore, $u \in \llbracket B \rrbracket$ and $v \in \llbracket C \rrbracket$.
    \end{itemize}
    \item If $A = B \oplus C$:
      \begin{itemize}
        \item Let $t \in \llbracket B \rrbracket \hatoplus \llbracket C \rrbracket$, then $t \in \SN$.
        \item Let $t \in \llbracket B \rrbracket \hatoplus \llbracket C \rrbracket$ such that $t \lra t'$. Then $t' \in \SN$, and if $t' \lras \inl(u)$, $t \lras \inl(u)$. Therefore, $u \in \llbracket B \rrbracket$. If $t' \lras \inr(v)$, $t \lras \inr(v)$. Therefore, $v \in \llbracket C \rrbracket$.
        \item Let $t$ be a proof-term that is not an introduction such that $Red(t) \subseteq \llbracket B \rrbracket \hatoplus \llbracket C \rrbracket$. Since $Red(t) \subseteq \SN$, $t \in \SN$. If $t \lras \inl(u)$, the rewrite sequence has at least one step because $t$ is not an introduction. Then, there is a proof-term $t' \in Red(t)$ such that $t' \lras \inl(u)$. Therefore, $u \in \llbracket B \rrbracket$. If $t \lras \inr(v)$, the rewrite sequence has at least one step because $t$ is not an introduction. Then, there is a proof-term $t' \in Red(t)$ such that $t' \lras \inr(v)$. Therefore, $v \in \llbracket C \rrbracket$.
      \end{itemize}
    \item If $A = \oc B$:
      \begin{itemize}
	\item Let $t \in \hatbang\llbracket B\rrbracket$, then $t\in \SN$.
	\item Let $t \in \hatbang\llbracket B \rrbracket$ such that $t \lra t'$. Then $t' \in \SN$, and if $t' \lras \oc u$, $t \lras\oc u$. Therefore, $u \in \llbracket B \rrbracket$. 
	\item Let $t$ be a proof-term that is not an introduction such that $Red(t) \subseteq \hatbang\llbracket B \rrbracket$. Since $Red(t) \subseteq \SN$, $t \in \SN$. If $t \lras \oc u$, the rewrite sequence has at least one step because $t$ is not an introduction. Then, there is a proof-term $t' \in Red(t)$ such that $t' \lras \oc u$. Therefore, $u \in \llbracket B \rrbracket$.\qedhere 
      \end{itemize}
  \end{itemize}
\end{proof}

\subsection{Adequacy lemmas for the proof of Lemma~\ref{lem:adequacy}}
\label{proof:Adequacy}

In Lemmas~\ref{lem:sum}
to~{\ref{lem:elimbang}},
we prove the adequacy of each proof-term constructor. If $t$ is a strongly
normalising proof-term, we write $|t|$ for the maximum length of a reduction
sequence issued from $t$.

\begin{lemma}[Normalisation of a sum]
    \label{lem:terminationsum}
    If $t$ and $u$ strongly normalise, then so does $t \plus u$. 
  \end{lemma}
  \begin{proof}
    We prove that all the one-step reducts of 
    $t \plus u$ strongly normalise, by induction first on 
    $|t| + |u|$ and then on the size of $t$. 
  
    If the reduction takes place in $t$ or in $u$ we apply the induction
    hypothesis.
    Otherwise, the reduction occurs at the root and the rule used is one of 
    \[
      \begin{array}{r@{\,}l@{\qquad}r@{\,}l}
	{a.\star} \plus {b.\star} &\lra  (a + b).\star
	&	\pair{t'_1}{t'_2} \plus \pair{u'_1}{u'_2} &\lra  \pair{t'_1 \plus u'_1}{t'_2 \plus u'_2}\\
	(\lambda x^A.t') \plus (\lambda x^A.u') &\lra  \lambda x^A.(t' \plus u')
	&	t \plus u &\lra t\\
	\langle \rangle \plus \langle \rangle &\lra \langle \rangle
	&	t \plus u &\lra u
      \end{array}
    \]
  
    In the first case, the proof-term $(a + b).\star$ is irreducible, hence it
    strongly normalises. In the second case,
    by induction hypothesis the proof-term $t' \plus u'$
    strongly normalises, thus so does the proof-term
    $\lambda x^A.(t' \plus u')$.
    In the third case, the proof-term $\langle \rangle$ is irreducible, hence it
    strongly normalises. 
    In the fourth case,
    by induction hypothesis the proof-terms
    $t'_1 \plus u'_1$ and $t'_2 \plus u'_2$
    strongly normalise, hence so does the proof-term
    $\pair{t'_1 \plus u'_1}{t'_2 \plus u'_2}$.
    In the fifth and sixth cases, the proof-terms $t$ and $u$ strongly normalise. 
  \end{proof}
  
  \begin{lemma}[Normalisation of a product]
  \label{lem:terminationprod}
  If $t$ strongly normalises, then so does $a \dotprod t$. 
  \end{lemma}
  \begin{proof}
  We prove that all the one-step reducts of 
  $a \dotprod t$ strongly normalise, by induction first on 
  $|t|$ and then on the size of $t$. 
  
  If the reduction takes place in $t$, we apply the induction
  hypothesis.
  Otherwise, the reduction occurs at the root, and the rule used is one of
  \[
    \begin{array}{r@{\,}l@{\qquad}r@{\,}l}
      a \dotprod b.\star &\lra  (a \times b).\star
      &		a \dotprod \pair{t'_1}{t'_2} &\lra  \pair{a \dotprod t'_1}{a \dotprod t'_2}\\
      a \dotprod (\lambda x^A.t') &\lra  \lambda x^A. a \dotprod t'
      &		a \dotprod t &\lra t\\
      a \dotprod \langle \rangle &\lra  \langle \rangle
    \end{array}
  \]
  In the first case, the proof-term $(a \times b).\star$ is irreducible,
  hence it strongly normalises. In the second case, by the induction hypothesis
  the proof-term $a \dotprod t'$ strongly normalises, thus so does the proof-term
  $\lambda x^A.a \dotprod t'$. In the third case, the proof-term $\langle
  \rangle$ is irreducible, hence it strongly normalises.  In the
  fourth case, by the induction hypothesis the proof-terms $a \dotprod t'_1$ and
  $a \dotprod t'_2$ strongly normalise, hence so does the proof-term
  $\pair{a \dotprod t'_1}{a \dotprod t'_2}$. In the fifth case, the proof-term $t$
  strongly normalises. 
  \end{proof}

\begin{lemma}[Adequacy of $\plus$]
\label{lem:sum}
If $\Upsilon;\Gamma \vdash t_1:A$, $\Upsilon;\Gamma \vdash t_2:A$, $t_1 \in \llbracket A \rrbracket$, and $t_2 \in \llbracket A
\rrbracket$, then $t_1 \plus t_2 \in \llbracket A \rrbracket$.
\end{lemma}
\begin{proof}
  By induction on $A$.
  The proof-terms $t_1$ and $t_2$ strongly normalise.  Thus, by
  Lemma~\ref{lem:terminationsum}, the proof-term $t_1 \plus t_2$ strongly
  normalises.
  Furthermore:
  \begin{itemize}
    \item If the proposition $A$ has the form $\one$, $\top$, or $\zero$, then $t_1\plus t_2\in\SN=\llbracket A\rrbracket$.
    \item If the proposition $A$ has the form $B \multimap C$, and $t_1
      \plus t_2 \lra^* \lambda x^B. v$ then either $t_1 \lra^*
      \lambda x^B. u_1$, $t_2 \lra^* \lambda x^B. u_2$, and $u_1
      \plus u_2 \lra^* v$, or $t_1 \lra^* \lambda x^B. v$, or $t_2
      \lra^* \lambda x^B. v$.

      In the first case, as $t_1$ and $t_2$ are in $\llbracket A
      \rrbracket$, for every $w$ in $\llbracket B \rrbracket$, $(w/x)u_1
      \in \llbracket C \rrbracket$ and $(w/x)u_2 \in \llbracket C
      \rrbracket$.  By induction hypothesis, $(w/x)(u_1 \plus u_2) =
      (w/x)u_1 \plus (w/x)u_2 \in \llbracket C \rrbracket$ and by
      {CR2}, $(w/x)v \in \llbracket C \rrbracket$.

      In the second and the third, as $t_1$ and $t_2$ are in $\llbracket A
      \rrbracket$, for every $w$ in $\llbracket B \rrbracket$, $(w/x)v \in
      \llbracket C \rrbracket$.

    \item If the proposition $A$ has the form $B \otimes C$, and $t_1
      \plus t_2 \lra^* v \otimes v'$ then $t_1 \lra^* v \otimes v'$, or
      $t_2 \lra^* v \otimes v'$.  As $t_1$ and $t_2$ are in $\llbracket
      A \rrbracket$, $v \in \llbracket B \rrbracket$ and $v' \in
      \llbracket C \rrbracket$.

    \item If the proposition $A$ has the form $B \with C$, and $t_1 \plus t_2
      \lra^* \pair{v}{v'}$ then $t_1 \lra^* \pair{u_1}{u'_1}$, $t_2 \lra^*
      \pair{u_2}{u'_2}$, $u_1 \plus u_2 \lra^* v$, and $u'_1 \plus u'_2
      \lra^* v'$, or $t_1 \lra^* \pair{v}{v'}$, or $t_2 \lra^*
      \pair{v}{v'}$.

      In the first case, as $t_1$ and $t_2$ are in $\llbracket A
      \rrbracket$, $u_1$ and $u_2$ are in $\llbracket B \rrbracket$ and
      $u'_1$ and $u'_2$ are in $\llbracket C \rrbracket$.  By induction
      hypothesis, $u_1 \plus u_2 \in \llbracket B \rrbracket$ and $u'_1
      \plus u'_2 \in \llbracket C \rrbracket$ and by {CR2}, $v \in \llbracket B \rrbracket$ and $v'
      \in \llbracket C \rrbracket$.

      In the second and the third, as $t_1$ and $t_2$ are in $\llbracket A
      \rrbracket$, $v \in \llbracket B \rrbracket$ and $v' \in \llbracket
      C \rrbracket$.

    \item If the proposition $A$ has the form $B \oplus C$, and $t_1 \plus
      t_2 \lra^* \inl(v)$ then $t_1 \lra^* \inl(v)$ or $t_2 \lra^*
      \inl(v)$.  As $t_1$ and $t_2$ are in $\llbracket A \rrbracket$, $v
      \in \llbracket B \rrbracket$.

      The proof is similar if $t_1 \plus t_2 \lra^* \inr(v)$.
    \item
	If the proposition $A$ has the form $\oc B$, and $t_1 \plus t_2
	\lra^* \oc v$ then $t_1 \lra^* \oc v$, or $t_2 \lra^* \oc v$.
	In both cases, as $t_1$ and $t_2$ are in $\llbracket A
	\rrbracket$, $v \in \llbracket B \rrbracket$.\qedhere
  \end{itemize}
\end{proof}

\begin{lemma}[Adequacy of $\dotprod$]
\label{lem:prod}
If $\Upsilon;\Gamma \vdash t:A$ and $t \in \llbracket A \rrbracket$, then $a \dotprod t \in \llbracket A
\rrbracket$.
\end{lemma}

\begin{proof}
  By induction on $A$.  The proof-term $t$ strongly normalises.  Thus, by
  Lemma~\ref{lem:terminationprod}, the proof-term $a \dotprod t$ strongly
  normalises.  Furthermore:

  \begin{itemize}
    \item If the proposition $A$ has the form $\one$, $\top$, or $\zero$, then $a\dotprod t\in\SN=\llbracket A\rrbracket$.
    \item If the proposition $A$ has the form $B \multimap C$, and $a
      \dotprod t \lra^* \lambda x^B. v$ then either $t \lra^* \lambda
      x^B. u$ and $a \dotprod u \lra^* v$, or $t \lra^* \lambda
      x^B. v$.

      In the first case, as $t$ is in $\llbracket A \rrbracket$, for
      every $w$ in $\llbracket B \rrbracket$, $(w/x)u \in \llbracket C
      \rrbracket$.  By induction hypothesis, $(w/x) (a \dotprod u) = a
      \dotprod (w/x)u \in \llbracket C \rrbracket$ and by
      {CR2}, $(w/x)v \in \llbracket C \rrbracket$.

      In the second, as $t$ is in $\llbracket A \rrbracket$, for every $w$
      in $\llbracket B \rrbracket$, $(w/x)v \in \llbracket C \rrbracket$.

    \item If the proposition $A$ has the form $B \otimes C$, and $a
      \dotprod t \lra^* v \otimes v'$ then $t \lra^* v \otimes v'$.  As $t$
      is in $\llbracket A \rrbracket$, $v \in \llbracket B \rrbracket$ and
      $v' \in \llbracket C \rrbracket$.

    \item If the proposition $A$ has the form $B \with C$, and $a \dotprod t
      \lra^* \pair{v}{v'}$ then $t \lra^* \pair{u}{u'}$, $a \dotprod u
      \lra^* v$, and $a \dotprod u' \lra^* v'$, or $t \lra^* \pair{v}{v'}$.

      In the first case, as $t$ is in $\llbracket A \rrbracket$, $u$ is in
      $\llbracket B \rrbracket$ and $u'$ is in $\llbracket C \rrbracket$.
      By induction hypothesis, $a \dotprod u \in \llbracket B \rrbracket$
      and $a \dotprod u' \in \llbracket C \rrbracket$ and by
      {CR2}, $v \in \llbracket B \rrbracket$ and $v'
      \in \llbracket C \rrbracket$.

      In the second, as $t$ is in $\llbracket A \rrbracket$, $v \in
      \llbracket B \rrbracket$ and $v' \in \llbracket C \rrbracket$.

    \item If the proposition $A$ has the form $B \oplus C$, and $a \dotprod t
      \lra^* \inl(v)$ then $t \lra^* \inl(v)$.
      Then, by {CR2},
      $\inl(v) \in \llbracket A \rrbracket$ hence, $v \in \llbracket
      B \rrbracket$.

      The proof is similar if $a \dotprod t \lra^* \inr(v)$.
    \item 
	If the proposition $A$ has the form $\oc B$, and $a \dotprod t \lra^*
	\oc v$ then $t \lra^* \oc v$. As $t$ is in $\llbracket A \rrbracket$,
	$v \in
	\llbracket B \rrbracket$.\qedhere
  \end{itemize}
\end{proof}

\begin{lemma}[Adequacy of $a.\star$]
  \label{lem:star}
   $a.\star \in \llbracket \one \rrbracket$.
  \end{lemma}
  
  \begin{proof}
  As $a.\star$ is irreducible, it strongly normalises, hence
  $a.\star \in \llbracket \one \rrbracket$.
  \end{proof}
  
  \begin{lemma}[Adequacy of $\lambda$]
  \label{lem:abstraction}
  If, for all $u \in \llbracket A \rrbracket$, $(u/x)t \in \llbracket B
  \rrbracket$, then $\lambda x^C.t \in \llbracket A \multimap B
  \rrbracket$.
  \end{lemma}
  
  \begin{proof}
  By Lemma~\ref{lem:Var}, $x \in \llbracket A \rrbracket$, thus
  $t = (x/x)t \in \llbracket B \rrbracket$. Hence, $t$ strongly
  normalises.  Consider a reduction sequence issued from $\lambda
  x^C.t$.  This sequence can only reduce $t$ hence it is finite. Thus,
  $\lambda x^C.t$ strongly normalises.
  
  Furthermore, if $\lambda x^C.t \lras \lambda x^C.t'$, then
  $t \lra^* t'$.  Let $u \in \llbracket A \rrbracket$,
  $(u/x)t \lra^* (u/x)t'$.
  As $(u/x)t \in \llbracket B
  \rrbracket$, by {CR2}, $(u/x)t' \in
  \llbracket B \rrbracket$.
  \end{proof}
  
  \begin{lemma}[Adequacy of $\otimes$]
  \label{lem:tensor}
  If $t_1 \in \llbracket A \rrbracket$ and $t_2 \in \llbracket B
  \rrbracket$, then $t_1 \otimes t_2 \in \llbracket A \otimes B
  \rrbracket$.
  \end{lemma}
  
  \begin{proof}
  The proof-terms $t_1$ and $t_2$ strongly normalise. Consider a reduction
  sequence issued from $t_1 \otimes t_2$.  This sequence can only reduce
  $t_1$ and $t_2$, hence it is finite.  Thus, $t_1 \otimes t_2$ strongly
  normalises.
  
  Furthermore, if $t_1 \otimes t_2 \lras t'_1 \otimes t'_2$,
  then $t_1 \lra^* t'_1$ and $t_2 \lra^* t'_2$.  By {CR2},
  $t'_1 \in \llbracket A \rrbracket$ and $t'_2 \in
  \llbracket B \rrbracket$. 
  \end{proof}

  \begin{lemma}[Adequacy of $\langle \rangle$]
  \label{lem:unit}
  $\langle \rangle \in \llbracket \top \rrbracket$.
  \end{lemma}
  
  \begin{proof}
  As $\langle \rangle$ is irreducible, it strongly normalises, hence
  $\langle \rangle \in \llbracket \top \rrbracket$. 
  \end{proof}
  
  \begin{lemma}[Adequacy of $\pair{}{}$]
  \label{lem:pair}
  If $t_1 \in \llbracket A \rrbracket$ and $t_2 \in \llbracket B
  \rrbracket$, then $\pair{t_1}{t_2} \in \llbracket A \with B
  \rrbracket$.
  \end{lemma}
  
  \begin{proof}
  The proof-terms $t_1$ and $t_2$ strongly normalise. Consider a reduction
  sequence issued from $\pair{t_1}{t_2}$.  This sequence can only
  reduce $t_1$
  and $t_2$, hence it is finite.  Thus, $\pair{t_1}{t_2}$
  strongly normalises.
  
  Furthermore, if $\pair{t_1}{t_2} \lras \pair
  {t'_1}{t'_2}$, then $t_1 \lra^* t'_1$ and $t_2 \lra^* t'_2$.  By
  {CR2}, $t'_1 \in \llbracket A \rrbracket$ and
  $t'_2 \in \llbracket B \rrbracket$. 
  \end{proof}
  
  \begin{lemma}[Adequacy of $\inl$]
  \label{lem:inl}
  If $t \in \llbracket A \rrbracket$, then $\inl(t) \in \llbracket A
  \oplus B \rrbracket$.
  \end{lemma}
  
  \begin{proof}
  The proof-term $t$ strongly normalises. Consider a reduction sequence
  issued from $\inl(t)$.  This sequence can only reduce $t$, hence it is
  finite.  Thus, $\inl(t)$ strongly normalises.
  
  Furthermore, if $\inl(t) \lras \inl(t')$, then $t \lra^* t'$.  By {CR2}, $t' \in \llbracket A
  \rrbracket$. And $\inl(t)$ never reduces to $\inr(t')$. 
  \end{proof}
  
  \begin{lemma}[Adequacy of $\inr$]
  \label{lem:inr}
  If $t \in \llbracket B \rrbracket$, then $\inr(t) \in \llbracket A
  \oplus B \rrbracket$.
  \end{lemma}
  
  \begin{proof}
    Similar to the proof of Lemma~\ref{lem:inl}. 
  \end{proof}

  \begin{lemma}[Adequacy of $\oc$]
    \label{lem:bang}
    If $t \in \llbracket A\rrbracket$, then $\oc t \in \llbracket\oc A\rrbracket$.
  \end{lemma}
  \begin{proof}
    The proof-term $t$ strongly normalises. Consider a reduction sequence
    issued from $\oc t$.  This sequence can only reduce $t$, hence it is
    finite.  Thus, $\oc t$ strongly normalises.
    Furthermore, if $\oc t \lras \oc t'$, then $t \lra^* t'$.  By CR2, $t' \in \llbracket A
    \rrbracket$. 
  \end{proof}
  
  \begin{lemma}[Adequacy of $\elimone$]
  \label{lem:elimone}
  If $t_1 \in \llbracket \one \rrbracket$ and $t_2 \in \llbracket C \rrbracket$, 
  then $\elimone(t_1,t_2) \in \llbracket C \rrbracket$.
  \end{lemma}
  
  \begin{proof}
  The proof-terms $t_1$ and $t_2$ strongly normalise.  We prove, by
  induction on $|t_1| + |t_2|$, that $\elimone(t_1,t_2)
  \in \llbracket C \rrbracket$.  Using {CR3}, we only
  need to prove that every of its one step reducts is in $\llbracket C
  \rrbracket$.  If the reduction takes place in $t_1$ or $t_2$, then we
  apply {CR2} and the induction hypothesis.
  
  Otherwise, the proof-term $t_1$ is $a.\star$ and the
  reduct is $a \dotprod t_2$. We conclude with Lemma~\ref{lem:prod}. 
  \end{proof}
  
  \begin{lemma}[Adequacy of application]
  \label{lem:application}
  If $t_1 \in \llbracket A \multimap B \rrbracket$ and $t_2 \in
  \llbracket A \rrbracket$, then $t_1~t_2 \in \llbracket B
  \rrbracket$.
  \end{lemma}
  
  \begin{proof}
  The proof-terms $t_1$ and $t_2$ strongly normalise. We prove, by induction
  on $|t_1| + |t_2|$, that $t_1~t_2 \in \llbracket B \rrbracket$. Using
  {CR3}, we only need to prove that every of its one
  step reducts is in $\llbracket B \rrbracket$.  If the reduction takes
  place in $t_1$ or in $t_2$, then we apply {CR2}
  and the induction hypothesis.
  
  Otherwise, the proof-term $t_1$ has the form $\lambda x^C.u$ and the reduct
  is $(t_2/x)u$.  As $\lambda x^C.u \in \llbracket A \multimap B
  \rrbracket$, we have $(t_2/x)u \in \llbracket B \rrbracket$. 
  \end{proof}
  
  \begin{lemma}[Adequacy of $\elimtens$]
  \label{lem:elimtens}
  If $t_1 \in \llbracket A \otimes B \rrbracket$,
  for all $u$ in $\llbracket A \rrbracket$,
  for all $v$ in $\llbracket B \rrbracket$,
  $(u/x,v/y)t_2 \in \llbracket C \rrbracket$, 
  then $\elimtens(t_1, x^D y^E.t_2) \in \llbracket C \rrbracket$.
  \end{lemma}
  
  \begin{proof}
  By Lemma~\ref{lem:Var}, $x \in \llbracket A \rrbracket$ and $y \in
  \llbracket B \rrbracket$, thus $t_2 = (x/x,y/y)t_2 \in \llbracket C
  \rrbracket$.  Hence, $t_1$ and $t_2$ strongly normalise.  We prove, by
  induction on $|t_1| + |t_2|$, that $\elimtens(t_1,
  x^D y^E.t_2) \in \llbracket C \rrbracket$.  Using
  {CR3}, we only need to prove that every of its one
  step reducts is in $\llbracket C \rrbracket$.  If the reduction takes
  place in $t_1$ or $t_2$, then we apply {CR2} and the
  induction hypothesis. Otherwise, either:
  \begin{itemize}
  \item The proof-term $t_1$ has the form $w_2 \otimes w_3$ and the reduct is
    $(w_2/x,w_3/y)t_2$. As
    $w_2 \otimes w_3 \in \llbracket A \otimes B \rrbracket$, we
    have $w_2 \in \llbracket A \rrbracket$
  and $w_3 \in \llbracket B \rrbracket$. 
    Hence, $(w_2/x,w_3/y)t_2 \in
    \llbracket C \rrbracket$.
  
  \item The proof-term $t_1$ has the form $t_1' \plus t''_1$ and the
    reduct is\\ $\elimtens(t'_1, x^D y^E.t_2) \plus
    \elimtens(t''_1, x^D y^E.t_2)$. As $t_1
    \lra t'_1$ with an ultra-reduction rule, we have by
    {CR2}, $t'_1 \in \llbracket A \otimes B
    \rrbracket$.  Similarly, $t''_1 \in \llbracket A \otimes B
    \rrbracket$.  Thus, by induction hypothesis, $\elimtens(t'_1,
    x^D y^E.t_2) \in \llbracket A \otimes B \rrbracket$
    and $\elimtens(t''_1, x^D y^E.t_2) \in \llbracket A
    \otimes B \rrbracket$.  We conclude with Lemma~\ref{lem:sum}.
  
  \item The proof-term $t_1$ has the form $a \dotprod t_1'$ and the
    reduct is $a \dotprod \elimtens(t'_1, x^D y^E.t_2)$. As $t_1
    \lra t'_1$ with an ultra-reduction rule, we have by
    {CR2}, $t'_1 \in \llbracket A \oplus B
    \rrbracket$.  
    Thus, by induction hypothesis, $\elimtens(t'_1,
    x^D y^E.t_2) \in \llbracket A \otimes B \rrbracket$.
    We conclude with Lemma~\ref{lem:prod}. \qedhere
  \end{itemize}
  \end{proof}
  
  \begin{lemma}[Adequacy of $\elimzero$]
  \label{lem:elimzero}
  If $t \in \llbracket \zero \rrbracket$, 
  then $\elimzero(t) \in \llbracket C \rrbracket$.
  \end{lemma}
  
  \begin{proof}
  The proof-term $t$ strongly normalises.  Consider a reduction sequence
  issued from $\elimzero(t)$.  This sequence can only reduce $t$, hence it
  is finite.  Thus, $\elimzero(t)$ strongly normalises.  Moreover, it
  never reduces to an introduction. 
  \end{proof}
  
  \begin{lemma}[Adequacy of $\elimwith^1$]
  \label{lem:elimwith1}
  If $t_1 \in \llbracket A \with B \rrbracket$
  and, for all $u$ in $\llbracket A \rrbracket$,
  $(u/x)t_2 \in \llbracket C \rrbracket$, 
  then $\elimwith^1(t_1, x^D.t_2) \in \llbracket C \rrbracket$.
  \end{lemma}
  
  \begin{proof}
  By Lemma~\ref{lem:Var}, $x \in \llbracket A \rrbracket$
  thus $t_2 = (x/x)t_2 \in \llbracket C
  \rrbracket$.  Hence, $t_1$ and $t_2$ strongly normalise.  We prove, by
  induction on $|t_1| + |t_2|$, that $\elimwith^1(t_1, x^D.t_2)
  \in \llbracket C \rrbracket$.  Using {CR3}, we only
  need to prove that every of its one step reducts is in $\llbracket C
  \rrbracket$.  If the reduction takes place in $t_1$ or $t_2$, then we
  apply {CR2} and the induction hypothesis.
  
  Otherwise, the proof-term $t_1$ has the form $\pair{u}{v}$ and the
  reduct is $(u/x)t_2$.  As $\pair{u}{v} \in \llbracket A
  \with B \rrbracket$, we have $u \in \llbracket A \rrbracket$.
  Hence, $(u/x)t_2 \in \llbracket C \rrbracket$. 
  \end{proof}
  
  \begin{lemma}[Adequacy of $\elimwith^2$]
  \label{lem:elimwith2}
  If $t_1 \in \llbracket A \with B \rrbracket$ and,
  for all $u$ in $\llbracket B \rrbracket$,
  $(u/x)t_2 \in \llbracket C \rrbracket$, 
  then $\elimwith^2(t_1, x^D.t_2) \in \llbracket C \rrbracket$.
  \end{lemma}
  
  \begin{proof}
  Similar to the proof of Lemma~\ref{lem:elimwith1}. 
  \end{proof}
  
  \begin{lemma}[Adequacy of $\elimplus$]
  \label{lem:elimplus}
  If $t_1 \in \llbracket A \oplus B \rrbracket$, for all $u$ in $\llbracket A
  \rrbracket$, $(u/x)t_2 \in \llbracket C \rrbracket$, and, for all $v$
  in $\llbracket B \rrbracket$, $(v/y)t_3 \in \llbracket C \rrbracket$,
  then $\elimplus(t_1, x^D.t_2, y^E.t_3) \in \llbracket C \rrbracket$.
  \end{lemma}
  
  \begin{proof}
  By Lemma~\ref{lem:Var}, $x \in \llbracket A \rrbracket$, thus $t_2 =
  (x/x)t_2 \in \llbracket C \rrbracket$. In the same way, $t_3 \in
  \llbracket C \rrbracket$.  Hence, $t_1$, $t_2$, and $t_3$ strongly
  normalises.  We prove, by induction on $|t_1| + |t_2| + |t_3|$,
  that $\elimplus(t_1, x^D.t_2,
  y^E.t_3) \in \llbracket C \rrbracket$.  Using {CR3}, we
  only need to prove that every of its one step reducts
  is in $\llbracket C \rrbracket$.  If the reduction takes place in
  $t_1$, $t_2$, or $t_3$, then we apply {CR2} and
  the induction hypothesis. Otherwise, either:
  \begin{itemize}
  \item The proof-term $t_1$ has the form $\inl(w_2)$ and the reduct is
    $(w_2/x)t_2$. As $\inl(w_2) \in \llbracket A \oplus B \rrbracket$, we
    have $w_2 \in \llbracket A \rrbracket$.  Hence, $(w_2/x)t_2 \in
    \llbracket C \rrbracket$.
  
  \item The proof-term $t_1$ has the form $\inr(w_3)$ and the reduct is
    $(w_3/x)t_3$. As $\inr(w_3) \in \llbracket A \oplus B \rrbracket$, we
    have $w_3 \in \llbracket B \rrbracket$.  Hence, $(w_3/x)t_3 \in
    \llbracket C \rrbracket$.
  
  \item The proof-term $t_1$ has the form $t_1' \plus t''_1$ and the
    reduct is $\elimplus(t'_1, x^D.t_2, y^E.t_3) \plus
    \elimplus(t''_1, x^D.t_2, y^E.t_3)$. As $t_1
    \lra t'_1$ with an ultra-reduction rule, we have by
    {CR2}, $t'_1 \in \llbracket A \oplus B
    \rrbracket$.  Similarly, $t''_1 \in \llbracket A \oplus B
    \rrbracket$.  Thus, by induction hypothesis, $\elimplus(t'_1,
    x^D.t_2, y^E.t_3) \in \llbracket A \oplus B \rrbracket$
    and $\elimplus(t''_1, x^D.t_2, y^E.t_3) \in \llbracket A
    \oplus B \rrbracket$.  We conclude with Lemma~\ref{lem:sum}.
  
  \item The proof-term $t_1$ has the form $a \dotprod t_1'$ and the
    reduct is $a \dotprod \elimplus(t'_1, x^D.t_2, y^E.t_3)$. As $t_1
    \lra t'_1$ with an ultra-reduction rule, we have by
    {CR2}, $t'_1 \in \llbracket A \oplus B
    \rrbracket$.  
    Thus, by induction hypothesis, $\elimplus(t'_1,
    x^D.t_2, y^E.t_3) \in \llbracket A \oplus B \rrbracket$.
    We conclude with Lemma~\ref{lem:prod}. \qedhere
  \end{itemize}
  \end{proof}

\begin{lemma}[Adequacy of $\elimbang$]\label{lem:elimbang}
  If $t_1 \in \llbracket \oc A \rrbracket$
  and, for all $u$ in $\llbracket A \rrbracket$,
  $(u/x)t_2 \in \llbracket B \rrbracket$, 
  then $\elimbang(t_1, x^C.t_2) \in \llbracket B \rrbracket$.
\end{lemma}

\begin{proof}
  By Lemma~\ref{lem:Var}, $x \in \llbracket A \rrbracket$
  thus $t_2 = (x/x)t_2 \in \llbracket B
  \rrbracket$.  Hence, $t_1$ and $t_2$ strongly normalise.  We prove, by
  induction on $|t_1| + |t_2|$, that $\elimbang(t_1, x^C.t_2)
  \in \llbracket B \rrbracket$.  Using CR3, we only
  need to prove that every of its one step reducts is in $\llbracket B
  \rrbracket$.  If the reduction takes place in $t_1$ or $t_2$, then we
  apply CR2 and the induction hypothesis. Otherwise, either:
  \begin{itemize}
	\item The proof-term $t_1$ has the form $\oc u$ and the reduct is $(u/x)t_2$.  As $\oc u \in \llbracket \oc A \rrbracket$, we have $u \in \llbracket A \rrbracket$.
  Hence, $(u/x)t_2 \in \llbracket B \rrbracket$.

\item The proof-term $t_1$ has the form $t_1' \plus t_1''$ and the reduct is $\elimbang(t_1', x^C.t_2) \plus \elimbang(t_1'', x^C.t_2)$. As $t_1 \lra t_1'$ with an ultra-reduction rule, we have that $t_1' \in \llbracket \oc A \rrbracket$ by CR2. Similarly, $t_1'' \in \llbracket \oc A \rrbracket$. Thus, by the induction hypothesis, $\elimbang(t_1', x^C.t_2) \in \llbracket B \rrbracket$ and $\elimbang(t_1'', x^C.t_2) \in \llbracket B \rrbracket$. We conclude with Lemma~\ref{lem:sum}.
  
\item The proof-term $t_1$ has the form $a \dotprod t_1'$ and the reduct is $a \dotprod \elimbang(t_1', x^C.t_2)$. As $t_1 \lra t_1'$ with an ultra-reduction rule, we have that $t_1' \in \llbracket \oc A \rrbracket$ by CR2. Thus, by the induction hypothesis, $\elimbang(t_1', x^C.t_2) \in \llbracket B \rrbracket$. We conclude with Lemma~\ref{lem:prod}.\qedhere
  \end{itemize}
\end{proof}

\subsection{Theorem~\ref{thm:introductions} (Introduction)}\label{proof:introductionthm}
\introductions*
\begin{proof}
  By induction on $t$.

  We first remark that, as the proof-term $t$ is closed, it is not a
  variable. Then, we prove that it cannot be an elimination.
  \begin{itemize}
    \item If $t = \elimone(u,v)$, then $u$ is a closed
      irreducible proof-term of $\one$, hence, by induction
      hypothesis, it has the form $a.\star$
      and the proof-term $t$ is reducible.

    \item If $t = u~v$, then $u$ is a closed  irreducible proof-term of $B
      \multimap A$, hence, by induction hypothesis, it has the form
      $\lambda \abstr{x:B}u_1$ and the proof-term $t$ is reducible.

    \item If $t = \elimtens(u,x^B y^C.v)$, then $u$ is a closed
      irreducible proof-term of $B \otimes C$, hence, by induction hypothesis, it
      has the form $u_1 \otimes u_2$, $u_1 \plus u_2$, or $a \dotprod
      u_1$ and the proof-term $t$ is reducible.

    \item If $t = \elimzero(u)$, then $u$ is a closed irreducible proof-term of
      $\zero$ and, by induction hypothesis, no such proof-term exists.

    \item If $t = \elimwith^1(u,x^B.v)$, then $u$ is a closed
      irreducible proof-term of $B \with C$, hence, by induction
      hypothesis, it has the form $\pair{u_1}{u_2}$
      and the proof-term $t$ is reducible.

    \item If $t = \elimwith^2(u,x^C.v)$, then $u$ is a closed
      irreducible proof-term of $B \with C$, hence, by induction
      hypothesis, it has the form $\pair{u_1}{u_2}$
      and the proof-term $t$ is reducible.

    \item If $t = \elimplus(u,x^B.v,y^C.w)$, then $u$ is a closed
      irreducible proof-term of $B \oplus C$, hence, by induction hypothesis, it
      has the form $\inl(u_1)$, $\inr(u_1)$, $u_1 \plus u_2$, or $a \dotprod
      u_1$ and the proof-term $t$ is reducible.

    \item If $t = \elimbang(u,x^B.v)$, then $u$ is a closed
      irreducible proof-term of $\oc B$, hence, by induction
      hypothesis, it has the form $\oc u_1$, $u_1 \plus u_2$, or $a \dotprod u_1$
      and the proof-term $t$ is reducible.

  \end{itemize}

  Hence, $t$ is an introduction, a sum, or a product.

  \begin{itemize}
    \item It $t$ has the form $a.\star$, then $A$ is $\one$. 
    \item If it has the form $\lambda x^B.u$, then $A$ has the form $B \multimap C$.
    \item If it has the form $u \otimes v$, then $A$ has the form $B \otimes C$.
    \item If it is $\langle \rangle$, then $A$ is $\top$.
    \item If it has the form $\pair{u}{v}$, then $A$ has the form $B \with C$.
    \item If it has the form $\inl(u)$ or $\inr(u)$, then $A$ has the form $B \oplus C$.
    \item If it has the form $\oc u$, then $A$ has the form $\oc B$.
  \end{itemize}
  We prove that, if it has the form $u \plus v$ or $a \dotprod u$, $A$ has
  the form $B \otimes C$, $B \oplus C$ or $\oc B$.
  \begin{itemize}
    \item 
      If $t = u \plus v$, then the proof-terms $u$ and $v$ are two closed and irreducible proof-terms of $A$.
      \begin{itemize}
	\item     If $A = \one$ then, by induction hypothesis, they both have the form $a.\star$ and the proof-term $t$ is reducible.
	\item If $A$ has the form $B \multimap C$ then, by induction hypothesis, they are both abstractions and the proof-term $t$ is reducible.
	\item If $A = \top$ then, by induction hypothesis, they both are $\langle \rangle$ and the proof-term $t$ is reducible.
	\item If $A = \zero$ then, they are irreducible proof-terms of $\zero$ and, by induction hypothesis, no such proof-terms exist.
	\item If $A$ has the form $B \with C$, then, by induction hypothesis,
	  they are both pairs and the proof-term $t$ is reducible.
      \end{itemize}
	Hence, $A$ has the form $B \otimes C$, $B \oplus C$ or $\oc B$.

    \item
      If $t = a \dotprod u$, then the proof-term $u$ is a closed and irreducible proof-term of $A$.
      \begin{itemize}
	\item If $A = \one$ then, by induction hypothesis, $u$ has the form $a.\star$ and the proof-term $t$ is reducible.
	\item If $A$ has the form $B \multimap C$ then, by induction hypothesis, it is an abstraction and the proof-term $t$ is reducible.
	\item If $A = \top$ then, by induction hypothesis, it is $\langle \rangle$ and the proof-term $t$ is reducible.
	\item If $A = \zero$ then, it is an irreducible proof-term of $\zero$ and, by induction hypothesis, no such proof-term exists.
	\item If $A$ has the form $B \with C$, then, by induction hypothesis, it is a pair and the proof-term $t$ is reducible. 
      \end{itemize}
	Hence, $A$ has the form $B \otimes C$, $B \oplus C$ or $\oc B$.\qedhere
  \end{itemize}
\end{proof}

\section{Proofs of the linearity properties (Section~\ref{sec:seclinearity})}\label{proof:seclinearity}

\subsection{Lemma~\ref{lem:decomp}}\label{proof:decomp}
\decomp*
\begin{proof}
	By induction on $t$.
	\begin{itemize}
    \item If $t$ is the variable $x$, an introduction, a sum, or a
     product, we take $K = \_$, $u = t$, and $B = A$.

    \item If $t = \elimone(t_1,t_2)$, then $t_1$ is not a closed proof-term, as
      otherwise it would be a closed irreducible proof of $\one$, hence,
      by Theorem \ref{thm:introductions}, it would be an introduction and $t$
      would not be irreducible. Then, $\varnothing;x^C \vdash t_1:\one$ and $\varnothing;\varnothing \vdash t_2:A$. By the induction hypothesis, there exist $K_1$, $u_1$ and $B_1$ such that $\Gamma_1;\Delta_1 \vdash K_1:\one$, $\varnothing;x^C \vdash u_1:B_1$ and $t_1 = K_1[u_1]$, where $\Gamma_1;\Delta_1 = \_^{B_1};\varnothing$ or $\Gamma_1;\Delta_1 = \varnothing; \_^{B_1}$. We take $K = \elimone(K_1, t_2)$, $u=u_1$ and $B=B_1$. We have $\varnothing;x^C \vdash u:B$, $K[u] = \elimone(K_1[u_1], t_2)=t$ and $\Gamma_1;\Delta_1 \vdash K:A$.
      
    \item If $t = t_1~t_2$,
          we apply the same method as for the case $t = \elimone(t_1,t_2)$.

    \item If $t = \elimtens(t_1,y^{D_1} z^{D_2}.t_2)$, then $t_1$ is not a
      closed proof-term, as otherwise it would be a closed irreducible proof of $D_1 \otimes D_2$, hence,
      by Theorem \ref{thm:introductions}, it would be an introduction, a sum, or a
      product, and $t$ would not be irreducible. Then, $\varnothing;x^C \vdash t_1:D_1 \otimes D_2$ and $\varnothing; y^{D_1}, z^{D_2} \vdash t_2:A$. By the induction hypothesis, there exist $K_1$, $u_1$ and $B_1$ such that $\Gamma_1;\Delta_1 \vdash K_1:D_1 \otimes D_2$, $\varnothing;x^C \vdash u_1:B_1$ and $t_1 = K_1[u_1]$, where $\Gamma_1;\Delta_1 = \_^{B_1};\varnothing$ or $\Gamma_1;\Delta_1 = \varnothing; \_^{B_1}$. We take $K = \elimtens(K_1, y^{D_1} z^{D_2}.t_2)$, $u=u_1$ and $B=B_1$. We have $\varnothing;x^C \vdash u:B$, $K[u] = \elimtens(K_1[u_1], t_2) = t$ and $\Gamma_1;\Delta_1 \vdash K:A$.

    \item If $t = \elimzero(t_1)$, then,
      by Theorem \ref{thm:introductions},
      $t_1$ is not a closed proof-term as there is
      no closed irreducible proof of $\zero$. Then, $\varnothing;x^C \vdash t_1:\zero$. By the induction hypothesis, there exist $K_1$, $u_1$ and $B_1$ such that $\Gamma_1;\Delta_1 \vdash K_1:\zero$, $\varnothing;x^C \vdash u_1:B_1$ and $t_1 = K_1[u_1]$, where $\Gamma_1;\Delta_1 = \_^{B_1};\varnothing$ or $\Gamma_1;\Delta_1 = \varnothing; \_^{B_1}$. We take $K = \elimzero(K_1)$, $u=u_1$ and $B=B_1$. We have $\varnothing;x^C \vdash u:B$, $K[u] = \elimzero(K_1[u_1]) = t$ and $\Gamma_1, \Delta_1 \vdash K:A$.

    \item If $t = \elimwith^1(t_1,y^D.t_2)$ or
      $t = \elimwith^2(t_1,y^D.t_2)$,
      we apply the same method as for the case $t = \elimone(t_1,t_2)$.

    \item If $t =  \elimplus(t_1,y^{D_1}.t_2,z^{D_2}.t_3)$, we apply the same
      method as for the case $t = \elimtens(t_1,y^{D_1} z^{D_2}.t_2)$.
	
	\item If $t = \elimbang(t_1, y^D.t_2)$, then $t_1$ is not a closed proof-term, as otherwise it would be a closed irreducible proof of $\oc D$, hence, by Theorem~\ref{thm:introductions}, it would be an introduction, a sum, or a product, and $t$ would not be irreducible. Then, $\varnothing;x^C \vdash t_1:\oc D$ and $y^D;\varnothing \vdash t_2:A$. By the induction hypothesis, there exist $K_1$, $u_1$ and $B_1$ such that $\Gamma_1;\Delta_1 \vdash K_1:\oc D$, $\varnothing;x^C \vdash u_1:B_1$ and $K_1[u_1] = t_1$, where $\Gamma_1;\Delta_1 = \_^{B_1};\varnothing$ or $\Gamma_1;\Delta_1 = \varnothing; \_^{B_1}$. We take $K = \elimbang(K_1, y^D.t_2)$, $u = u_1$ and $B = B_1$. We have $\varnothing;x^C \vdash u:B$, $K[u] = \elimbang(K_1[u_1], y^D.t_2) = t$ and $\Gamma_1;\Delta_1 \vdash K:A$.
	  \qedhere
\end{itemize}
\end{proof}

\subsection{Lemma~\ref{lem:measureofelimcontext}}\label{proof:measureofelimcontext}
\measureofelimcontext*
\begin{proof}
  By induction on $K$.

  \begin{itemize}
    \item If $K = \_$, then $\mu(K) = 0$ and $K[t] = t$.  We have
      $\mu(K[t]) = \mu(t) = \mu(K) + \mu(t)$.

    \item If $K = \elimone(K_1,u)$ then $K[t] = \elimone(K_1[t],u)$.
      We have, by
      induction hypothesis, $\mu(K[t]) = 1 + \mu(K_1[t]) + \mu(u)
      = 1 + \mu(K_1) + \mu(t) + \mu(u) = \mu(K) + \mu(t)$.

    \item If $K = K_1~u$ then $K[t] = K_1[t]~u$.  We have, by
      induction hypothesis, $\mu(K[t]) = 1 + \mu(K_1[t]) + \mu(u) = 1 +
      \mu(K_1) +   \mu(t) + \mu(u) = \mu(K) + \mu(t)$.

    \item 
      If $K = \elimtens(K_1,x^A y^B.v)$, then $K[t] =
      \elimtens(K_1[t],x^A y^B.v)$.  We have, by induction
      hypothesis, $\mu(K[t]) = 1 + \mu(K_1[t]) + \mu(v) = 1 +
      \mu(K_1) + \mu(t) + \mu(v) = \mu(K) + \mu(t)$.

    \item If $K = \elimzero(K_1)$, then $K[t] = \elimzero(K_1[t])$. We
      have, by induction hypothesis, $\mu(K[t]) = 1 + \mu(K_1[t])= 1 +
      \mu(K_1) + \mu (t) = \mu(K) + \mu(t)$.

    \item If $K = \elimwith^1(K_1,x^A.r)$, then $K[t] =
      \elimwith^1(K_1[t],x^A.r)$. We have, by induction hypothesis,
      $\mu(K[t]) = 1 + \mu(K_1[t]) + \mu(r) = 1 + \mu(K_1) + \mu (t) +
      \mu(r) = \mu(K) + \mu(t)$.

      The same holds if $K = \elimwith^2(K_1,x^A.s)$.

    \item 
      If $K = \elimplus(K_1,x^A.r,y^B.s)$, then $K[t] =
      \elimplus(K_1[t],x^A.r,y^B.s)$.  We have, by induction
      hypothesis, $\mu(K[t]) = 1 + \mu(K_1[t]) + \max(\mu(r), \mu(s)) = 1 +
      \mu(K_1) + \mu(t) + \max(\mu(r), \mu(s)) = \mu(K) + \mu(t)$.

	\item  If $K = \elimbang(K_1,x^A.r)$, then $K[t] = \elimbang(K_1[t],x^A.r)$. We have, by the induction hypothesis, $\mu(K[t]) = 1 + \mu(K_1[t]) + \mu(r) = 1 + \mu(K_1) + \mu(t) + \mu(r) = \mu(K) + \mu(t)$.
      \qedhere
  \end{itemize}
\end{proof}

\subsection{Lemma~\ref{lem:linearity}}\label{proof:linearity}
\linearity*
\begin{proof}
  Without loss of generality, we can assume that $t$ is irreducible.  We proceed by induction on $\mu(t)$. 

  Using Lemma~\ref{lem:decomp}, the proof-term $t$ can be decomposed as $K[t']$
  where $t'$ is either the variable $x$, an introduction, a sum, or a product, $\varnothing;x^A \vdash t':C$ and $\_^C;\varnothing \vdash K:B$ or $\varnothing; \_^C \vdash K:B$.
 \begin{itemize}
	\item If $t'$ is an introduction, as $t$ is irreducible, $K = \_$, and $t'$ is a proof-term of $B$ in $\mathcal V$, then $t'$ is either $a.\star$ or $\pair{t_1}{t_2}$. Since $\varnothing;x^A \nvdash a.\star:\one$, we have that $t' = \pair{t_1}{t_2}$, with $\mu(t_1) < \mu(t') = \mu(t)$ and $\mu(t_2) < \mu(t') = \mu(t)$. Then, by the induction hypothesis we have that:
	\begin{align*}
		t[u_1 \plus u_2]
		&= \pair{t_1[u_1 \plus u_2]}{t_2[u_1 \plus u_2]}\\
		&\equiv \pair{t_1[u_1 \plus u_2]}{t_2[u_1 \plus u_2]}\\
		&\longleftarrow \pair{t_1[u_1]}{t_2[u_1]} \plus \pair{t_1[u_2]}{t_2[u_2]}\\
		&= t[u_1] \plus t[u_2]
	\end{align*}
	\begin{align*}
		t[a \dotprod u_1]
		&= \pair{t_1[a \dotprod u_1]}{t_2[a \dotprod u_1]}\\
		&\equiv \pair{a \dotprod t_1[u_1]}{a \dotprod t_2[u_1]}\\
		&\longleftarrow a \dotprod \pair{t_1[u_1]}{t_2[u_1]}\\
		&= a \dotprod t[u_1]
	\end{align*}

	\item If $t'=t_1 \plus t_2$, since $\mu(K) < \mu(t)$, $\mu(t_1) < \mu(t)$ and $\mu(t_2) < \mu(t)$, by the induction hypothesis and by Lemma~\ref{lem:vecstructure} we have that:
	\begin{align*}
		t[u_1 \plus u_2] 
		&= K[t_1 \plus t_2][u_1 \plus u_2]\\
		&= K[t_1[u_1 \plus u_2] \plus t_2[u_1 \plus u_2]]\\
		&\equiv K[(t_1[u_1] \plus t_1[u_2]) \plus (t_2[u_1] \plus t_2[u_2])]\\
		&\equiv K[(t_1[u_1] \plus t_2[u_1]) \plus (t_1[u_2] \plus t_2[u_2])]\\
		&\equiv K[(t_1[u_1] \plus t_2[u_1])] \plus K[t_1[u_2] \plus t_2[u_2]]\\
		&= K[t_1 \plus t_2][u_1] \plus K[t_1 \plus t_2][u_2]\\
		&= t[u_1] \plus t[u_2]
	\end{align*}
	\begin{align*}
		t[a \dotprod u_1]
		&= K[t_1 \plus t_2][a \dotprod u_1]\\
		&= K[t_1[a \dotprod u_1] \plus t_2[a \dotprod u_1]]\\
		&\equiv K[(a \dotprod t_1[u_1]) \plus (a \dotprod t_2[u_1])]\\
		&\equiv K[a \dotprod (t_1[u_1] \plus t_2[u_1])]\\
		&\equiv a \dotprod K[t_1[u_1] \plus t_2[u_1]]\\
		&= a \dotprod K[t_1 \plus t_2][u_1]\\
		&= a \dotprod t[u_1]
	\end{align*}

	\item If $t' = b \dotprod t_1$, since $\mu(K) < \mu(t)$ and $\mu(t_1) < \mu(t)$, by the induction hypothesis and by Lemma~\ref{lem:vecstructure} we have that:
	\begin{align*}
		t[u_1 \plus u_2]
		&= K[b \dotprod t_1][u_1 \plus u_2]\\
		&= K[b \dotprod t_1[u_1 \plus u_2]]\\
		&\equiv K [b \dotprod (t_1[u_1] \plus t_1[u_2])]\\
		&\equiv K [b \dotprod t_1[u_1] \plus b \dotprod t_1[u_2]]\\
		&\equiv K[b \dotprod t_1[u_1]] \plus K[b \dotprod t_1[u_2]]\\
		&= K[b \dotprod t_1][u_1] \plus K[b \dotprod t_1][u_2]\\
		&= t[u_1] \plus t[u_2]
	\end{align*}
	\begin{align*}
		t[a \dotprod u_1]
		&= K[b \dotprod t_1][a \dotprod u_1]\\
		&= K[b \dotprod t_1[a \dotprod u_1]]\\
		&\equiv K[b \dotprod (a \dotprod t_1[u_1])]\\
		&\equiv K[(b \times a) \dotprod t_1[u_1]]\\
		&= K[(a \times b) \dotprod t_1[u_1]]\\
		&\equiv K[a \dotprod (b \dotprod t_1[u_1])]\\
		&\equiv a \dotprod K[b \dotprod t_1[u_1]]\\
		&= a \dotprod K[b \dotprod t_1][u_1]\\
		&= a \dotprod t[u_1]
	\end{align*}

	\item If $t'$ is the variable $x$, by Lemma~\ref{horrible}, $K$ has the form $K_1[K_2]$ and $K_2$ is an
  elimination of the top symbol of $A$. We consider the various cases for $K_2$. 

  \begin{itemize}
	\item If $K = K_1[\elimone(\_,r)]$, then $u_1$ and $u_2$ are closed proof-terms of $\one$, thus $u_1 \lras b.\star$ and $u_2 \lras c.\star$. Since $\mu(K_1) < \mu(K) = \mu(t)$, by the induction hypothesis and Lemma~\ref{lem:vecstructure} we have that:
	\begin{align*}
		K[u_1 \plus u_2] 
		&= K_1[\elimone(u_1 \plus u_2, r)]\\
		&\lras K_1[\elimone(b.\star \plus c.\star, r)]\\
		&\lra K_1[(b + c) \dotprod r]\\
		&\equiv (b + c) \dotprod K_1[r]\\
		&\equiv b \dotprod K_1[r] \plus c \dotprod K_1[r]\\
		&\equiv K_1[b \dotprod r] \plus K_1[c \dotprod r]\\
		&\llas K_1[\elimone(b.\star, r)] \plus K_1[\elimone(c.\star, r)]\\
		&\llas K_1[\elimone(u_1, r)] \plus K_1[\elimone(u_2, r)]\\
		&= K[u_1] \plus K[u_2]
	\end{align*}
	\begin{align*}
		K[a \dotprod u_1]
		&= K_1[\elimone(a \dotprod u_1, r)]\\
		&\lras K_1[\elimone(a \dotprod b.\star, r)]\\
		&\lra K_1[\elimone((a \times b) \dotprod r)]\\
		&\lra K_1[(a \times b) \dotprod r]\\
		&\equiv K_1[a \dotprod (b \dotprod r)]\\
		&\equiv a \dotprod K_1[b \dotprod r]\\
		&\lla a \dotprod K_1[\elimone(b.\star, r)]\\
		&\llas a \dotprod K_1[\elimone(u_1, r)]\\
		&= a \dotprod K[u_1]
	\end{align*}

	\item If $K = K_1[\_~r]$, then $u_1$ and $u_2$ are closed proof-terms of a linear implication, thus $u_1 \lras \lambda y^D.u_1'$ and $u_2 \lras \lambda y^D.u_2'$. Since $\mu(K_1) < \mu(K) = \mu(t)$, by the induction hypothesis we have that:
	\begin{align*}
		K[u_1 \plus u_2]
		&= K_1[(u_1 \plus u_2)~r]\\
		&\lras K_1[((\lambda y^D.u_1') \plus (\lambda y^D.u_2'))~r]\\
		&\lra K_1[(\lambda y^D.u_1' \plus u_2')~r]\\
		&\lra K_1[(r/y)(u_1' \plus u_2')]\\
		&= K_1[(r/y)u_1' \plus (r/y)u_2']\\
		&\equiv K_1[(r/y)u_1'] \plus K_1[(r/y)u_2']\\
		&\llas K_1[(\lambda y^D.u_1')~r] \plus K_1[(\lambda y^D.u_2')~r]\\
		&\llas K_1[u_1~r] \plus K_1[u_2~r]\\
		&= K[u_1] \plus K[u_2]
	\end{align*}
	\begin{align*}
		K[a \dotprod u_1]
		&= K_1[(a \dotprod u_1)~r]\\
		&\lras K_1[(a \dotprod (\lambda y^D.u_1'))~r]\\
		&\lra K_1[(\lambda y^D.a \dotprod u_1')~r]\\
		&\lra K_1[(r/y)(a \dotprod u_1')]\\
		&= K_1[a \dotprod (r/y)u_1']\\
		&\equiv a \dotprod K_1[(r/y)u_1']\\
		&\lla a \dotprod K_1[(\lambda y^D.u_1')~r]\\
		&\llas a \dotprod K_1[u_1~r]\\
		&= a \dotprod K[u_1]
	\end{align*}
	\item If $K = K_1[\elimtens(\_, y^{D_1} z^{D_2}.r)]$, since $\mu(K_1) < \mu(K) = \mu(t)$, by the induction hypothesis we have that:
	\begin{align*}
		K[u_1 \plus u_2]
		&= K_1[\elimtens(u_1 \plus u_2, y^{D_1} z^{D_2}.r)]\\
		&\lra K_1[\elimtens(u_1, y^{D_1} z^{D_2}.r) \plus \elimtens(u_2, y^{D_1} z^{D_2}.r)]\\
		&\equiv K_1[\elimtens(u_1, y^{D_1} z^{D_2}.r)] \plus K_1[\elimtens(u_2, y^{D_1} z^{D_2}.r)]\\
		&= K[u_1] \plus K[u_2]
	\end{align*}
	\begin{align*}
		K[a \dotprod u_1]
		&= K_1[\elimtens(a \dotprod u_1, y^{D_1} z^{D_2}.r)]\\
		&\lra K_1[a \dotprod \elimtens(u_1, y^{D_1} z^{D_2}.r)]\\
		&\equiv a \dotprod K_1[\elimtens(u_1, y^{D_1} z^{D_2}.r)]\\
		&= a \dotprod K[u_1]
	\end{align*}

	\item If $K = K_1[\elimzero(\_)]$, then $u_1$ and $u_2$ are closed proof-terms of $\zero$ and this is not possible.
	\item If $K = K_1[\elimwith^1(\_,y^D.r)]$, then $u_1$ and $u_2$ are closed proof-terms of the additive conjunction $\with$, thus $u_1 \lras \pair{u_{11}}{u_{12}}$ and $u_2 \lras \pair{u_{21}}{u_{22}}$. 
	
	By Lemma~\ref{lem:measureofelimcontext}, we have that $\mu(K_1[r]) = \mu(K_1) + \mu(r) < 1 + \mu(K_1) + \mu(r) = \mu(K) = \mu(t)$. We have that $\varnothing;y^D \vdash K_1[r]:B$. Then, by the induction hypothesis on $K_1[r]$, and since $y \notin \fv(K_1)$, we have that:
	\begin{align*}
		K[u_1 \plus u_2]
		&= K_1[\elimwith^1(u_1 \plus u_2,y^D.r)]\\
		&\lras K_1[\elimwith^1(\pair{u_{11}}{u_{12}} \plus \pair{u_{21}}{u_{22}}, y^D.r)]\\
		&\lra K_1[\elimwith^1(\pair{u_{11} \plus u_{21}}{u_{12} \plus u_{22}}, y^D.r)]\\
		&\lra K_1[r[u_{11} \plus u_{21}]]\\
		&= (K_1[r])[u_{11} \plus u_{21}]\\
		&\equiv (K_1[r])[u_{11}] \plus (K_1[r])[u_{21}]\\
		&= K_1[r[u_{11}]] \plus K_1[r[u_{21}]]\\
		&\llas K_1[\elimwith^1(\pair{u_{11}}{u_{12}}, y^D.r)] \plus K_1[\elimwith^1(\pair{u_{21}}{u_{22}}, y^D.r)]\\
		&\llas K_1[\elimwith^1(u_1, y^D.r)] \plus K_1[\elimwith^1(u_2, y^D.r)]\\
		&= K[u_1] \plus K[u_2]
	\end{align*}

	\begin{align*}
		K[a \dotprod u_1]
		&= K_1[\elimwith^1(a \dotprod u_1, y^D.r)]\\
		&\lras K_1[\elimwith^1(a \dotprod \pair{u_{11}}{u_{12}}, y^D.r)]\\
		&\lra K_1[\elimwith^1(\pair{a \dotprod u_{11}}{a \dotprod u_{12}}, y^D.r)]\\
		&\lra K_1[r[a \dotprod u_{11}]]\\
		&= (K_1[r])[a \dotprod u_{11}]\\
		&\equiv a \dotprod (K_1[r])[u_{11}]\\
		&= a \dotprod K_1[r[u_{11}]]\\
		&\lla a \dotprod K_1[\elimwith^1(\pair{u_{11}}{u_{12}}, y^D.r)]\\
		&\llas a \dotprod K_1[\elimwith^1(u_1, y^D.r)]\\
		&= a \dotprod K[u_1]
	\end{align*}
	\item If $K = K_1[\elimwith^2(\_,y^D.r)]$ the proof is similar.
	\item If $K = K_1[\elimplus(\_,y^{D_1}.r, z^{D_2}.s)]$, since $\mu(K_1) < \mu(K) = \mu(t)$, by the induction hypothesis we have that:
	\begin{align*}
		K[u_1 \plus u_2]
		&= K_1[\elimplus(u_1 \plus u_2, y^{D_1}.r, z^{D_2}.s)]\\
		&\lra K_1[\elimplus(u_1, y^{D_1}.r, z^{D_2}.s) \plus \elimplus(u_2, y^{D_1}.r, z^{D_2}.s)]\\
		&\equiv K_1[\elimplus(u_1, y^{D_1}.r, z^{D_2}.s)] \plus K_1[\elimplus(u_2, y^{D_1}.r, z^{D_2}.s)]\\
		&= K[u_1] \plus K[u_2]
	\end{align*}
	\begin{align*}
		K[a \dotprod u_1]
		&= K_1[\elimplus(a \dotprod u_1, y^{D_1}.r, z^{D_2}.s)]\\
		&\lra K_1[a \dotprod \elimplus(u_1, y^{D_1}.r, z^{D_2}.s)]\\
		&\equiv a \dotprod K_1[\elimplus(u_1, y^{D_1}.r, z^{D_2}.s)]\\
		&= a \dotprod K[u_1]
	\end{align*}

	\item If $K = K_1[\elimbang(\_,y^D.r)]$, since $\mu(K_1) < \mu(K) = \mu(t)$, by the induction hypothesis we have that:
	\begin{align*}
		K[u_1 \plus u_2]
		&= K_1[\elimbang(u_1 \plus u_2, y^D.r)]\\
		&\lra K_1[\elimbang(u_1, y^D.r) \plus \elimbang(u_2, y^D.r)]\\
		&\equiv K_1[\elimbang(u_1, y^D.r)] \plus K_1[\elimbang(u_2, y^D.r)]\\
		&= K[u_1] \plus K[u_2]
	\end{align*}
	\begin{align*}
		K(a \dotprod u_1)
		&= K_1[\elimbang(a \dotprod u_1, y^D.r)]\\
		&\lra K_1[a \dotprod \elimbang(u_1, y^D.r)]\\
		&\equiv a \dotprod K_1[\elimbang(u_1, y^D.r)]\\
		&= a \dotprod K[u_1]
		\tag*{\qedhere}
	\end{align*}
  \end{itemize}
\end{itemize}
\end{proof}

\subsection{Corollary~\ref{coro:linearitygeneral}}\label{proof:linearitygeneral}
\linearitygeneral*
\begin{proof}
  By strong normalisation and the introduction property, we have that $f \lras \lambda x^A.t$, where $\varnothing;x^A \vdash t:B$. Let $C \in {\mathcal V}$, and $c$ be a proof-term such that $\varnothing;\_^B \vdash c:C$. Thus,
	\begin{align*}
		c[f~(u_1 \plus u_2)] 
		&\lras c[t[u_1 \plus u_2]]\\
		c[f~u_1 \plus f~u_2]
		&\lras c[t[u_1] \plus t[u_2]]\\
		c[f~(a \dotprod u_1)]
		&\lras c[t[a \dotprod u_1]]\\
		c[a \dotprod (f~u_1)]
		&\lras c[a \dotprod (t[u_1])]
	\end{align*}

   Then, $\varnothing;x^A \vdash c[t]:C$, and applying Lemma~\ref{lem:linearity} to the proof-term $c[t]$ we get 
	\[
		c[t[u_1 \plus u_2]] = c[t][u_1 \plus u_2] \equiv c[t][u_1] \plus c[t][u_2] = c[t[u_1]] \plus c[t[u_2]]
	\]
	and
	\[
		c[t][a \dotprod u_1] = c[t[a \dotprod u_1]] \equiv a \dotprod c[t[u_1]] = a \dotprod c[t][u_1]
	\]

  and applying it again to the proof-term $c$ we get
  \[
	c[t[u_1] \plus t[u_2]] \equiv c[t[u_1]] \plus c[t[u_2]]
	\qquad\qquad\textrm{and}\qquad\qquad
	c[a \dotprod t[u_1]] \equiv a \dotprod c[t[u_1]]
  \]
  
  Thus
  \[
	c[t[u_1 \plus u_2]] \equiv c[t[u_1] \plus t[u_2]]
	\qquad\qquad\textrm{and}\qquad\qquad
	c[t[a \dotprod u_1]] \equiv c[a \dotprod t[u_1]]
  \]
  that is 
  \begin{align*}
    f~(u_1 \plus u_2) \sim f~u_1 \plus f~u_2
  	&\qquad\qquad\textrm{and}\qquad\qquad
  	f~(a \dotprod u_1) \sim a \dotprod (f~u_1)
    \tag*{\qedhere}
  \end{align*}
\end{proof}

\section{Proofs of the denotational semantics (Section~\ref{sec:denotationalsemantics})}\label{proof:secdenotationalsemantics}

\subsection{Lemma~\ref{lem:generalizacionprop9}}\label{proof:gen9}
\generalizacionpropnueve*
\begin{proof}
	This proof is a generalisation of the proof of \cite[Proposition 9]{bierman}. The following diagrams commute:
\[\begin{tikzcd}[ampersand replacement=\&,cramped]
	{A \otimes A} \&\& {\oc A \otimes \oc A} \&\& {\oc(A \otimes A)} \\
	\&\&\&\&\&\& {A \otimes A} \&\& {\oc A \otimes \oc A} \&\& {\oc(A \otimes A)} \\
	{\oc A \otimes \oc A} \&\& {\oc\oc A \otimes \oc\oc A} \&\& {\oc(\oc A \otimes \oc A)} \\
	\&\&\&\&\&\&\&\& {A \otimes A} \&\& {A \otimes A} \\
	{\oc(A \otimes A)} \&\&\&\& {\oc\oc(A \otimes A)}
	\arrow["{h \otimes h}", from=1-1, to=1-3]
	\arrow["{h \otimes h}"', from=1-1, to=3-1]
	\arrow["{(1)}"{description}, draw=none, from=1-1, to=3-3]
	\arrow["m", from=1-3, to=1-5]
	\arrow["{\oc h \otimes \oc h}"{description}, dashed, from=1-3, to=3-3]
	\arrow["{(2)}"{description}, draw=none, from=1-3, to=3-5]
	\arrow["{\oc(h \otimes h)}", from=1-5, to=3-5]
	\arrow["{h \otimes h}", from=2-7, to=2-9]
	\arrow[""{name=0, anchor=center, inner sep=0}, equals, from=2-7, to=4-9]
	\arrow["m", from=2-9, to=2-11]
	\arrow["{\varepsilon_A \otimes \varepsilon_A}"{description}, dashed, from=2-9, to=4-9]
	\arrow["{(5)}"{description}, draw=none, from=2-9, to=4-11]
	\arrow["{\varepsilon_{A \otimes A}}", from=2-11, to=4-11]
	\arrow["{\delta_A \otimes \delta_A}"{description}, dashed, from=3-1, to=3-3]
	\arrow["m"', from=3-1, to=5-1]
	\arrow["{(3)}"{description}, draw=none, from=3-1, to=5-5]
	\arrow["m"{description}, dashed, from=3-3, to=3-5]
	\arrow["{\oc m}", from=3-5, to=5-5]
	\arrow[equals, from=4-9, to=4-11]
	\arrow["{\delta_{A \otimes A}}"', from=5-1, to=5-5]
	\arrow["{(4)}"{description}, draw=none, from=2-9, to=0]
\end{tikzcd}\]
(1) commutes because $(A,h)$ is a coalgebra, (2) commutes by naturality of $m$, (3) commutes because $\delta$ is a monoidal natural transformation, (4) commutes because $(A,h)$ is a coalgebra and (5) commutes because $\varepsilon$ is a monoidal natural transformation.
\end{proof}

\subsection{The generalised properties}\label{proof:genprop}

\dnat*
\begin{proof}
	By induction on $n$ we show that the following diagram commutes, for every $f_i : A_i \to B_i$ where $1 \leq i \leq n$.
	\[\begin{tikzcd}[cramped, ampersand replacement=\&]
	{\bigotimes_{i=1}^n \oc A_i} \&\& {\bigotimes_{i=1}^n \oc A_i \otimes \bigotimes_{i=1}^n \oc A_i } \\
	\\
	{\bigotimes_{i=1}^n \oc B_i} \&\& {\bigotimes_{i=1}^n \oc B_i \otimes \bigotimes_{i=1}^n \oc B_i}
	\arrow["{d_{A_1, \dots, A_n}}", from=1-1, to=1-3]
	\arrow["{\bigotimes_{i=1}^n \oc f_i}"', from=1-1, to=3-1]
	\arrow["{\bigotimes_{i=1}^n \oc f_i \otimes \bigotimes_{i=1}^n \oc f_i}", from=1-3, to=3-3]
	\arrow["{d_{B_1, \dots, B_n}}"', from=3-1, to=3-3]
\end{tikzcd}\]

	\begin{itemize}
		\item If $n=0$, the diagram commutes trivially.
\[\begin{tikzcd}[cramped]
	I && {I \otimes I} \\
	\\
	I && {I \otimes I}
	\arrow["{\lambda_I^{-1}}", from=1-1, to=1-3]
	\arrow["id"', from=1-1, to=3-1]
	\arrow["{id \otimes id}", from=1-3, to=3-3]
	\arrow["{\lambda_I^{-1}}"', from=3-1, to=3-3]
\end{tikzcd}\]
		\item If $n=1$, the diagram commutes by naturality of $d$.
\[\begin{tikzcd}[cramped]
	{\oc A} && {\oc A \otimes \oc A} \\
	\\
	{\oc B} && {\oc B \otimes \oc B}
	\arrow["{d_A}", from=1-1, to=1-3]
	\arrow["{\oc f}"', from=1-1, to=3-1]
	\arrow["{\oc f \otimes \oc f}", from=1-3, to=3-3]
	\arrow["{d_B}"', from=3-1, to=3-3]
\end{tikzcd}\]
		\item If $n>1$:
\[\begin{tikzcd}[cramped,column sep=0pt]
	{\bigotimes_{i=1}^n \oc A_i} &&&&&&& \begin{array}{c} \bigotimes_{i=1}^n \oc A_i \\ \otimes \bigotimes_{i=1}^n \oc A_i \end{array} \\
	&& \begin{array}{c} \bigotimes_{i=1}^{n-1} \oc A_i \otimes \\ \bigotimes_{i=1}^{n-1} \oc A_i \otimes \oc A_n \end{array} && \begin{array}{c} \bigotimes_{i=1}^{n-1} \oc A_i \otimes \\ \bigotimes_{i=1}^{n-1} \oc A_i \\ \otimes \oc A_n \otimes \oc A_n \end{array} \\
	\\
	& \begin{array}{c} \bigotimes_{i=1}^{n-1} \oc B_i \\ \otimes \oc A_n \end{array} & \begin{array}{c} \bigotimes_{i=1}^{n-1} \oc B_i \otimes \\ \bigotimes_{i=1}^{n-1} \oc B_i \otimes \\ \oc A_n \end{array} && \begin{array}{c} \bigotimes_{i=1}^{n-1} \oc B_i \otimes \\ \bigotimes_{i=1}^{n-1} \oc B_i \otimes\\ \oc A_n \otimes \oc A_n \end{array} && \begin{array}{c} \bigotimes_{i=1}^{n-1} \oc B_i \otimes \oc A_n \otimes \\ \bigotimes_{i=1}^{n-1} \oc B_i \otimes \oc A_n \end{array} \\
	\\
	&& \begin{array}{c} \bigotimes_{i=1}^{n-1} \oc B_i \otimes \\ \bigotimes_{i=1}^{n-1} \oc B_i \otimes \\ \oc B_n \end{array} && \begin{array}{c} \bigotimes_{i=1}^{n-1} \oc B_i \otimes \\ \bigotimes_{i=1}^{n-1} \oc B_i \otimes\\ \oc B_n \otimes \oc B_n \end{array} \\
	{\bigotimes_{i=1}^n \oc B_i} &&&&&&& \begin{array}{c} \bigotimes_{i=1}^n \oc B_i \otimes \\ \bigotimes_{i=1}^n \oc B_i \end{array}
	\arrow[""{name=0, anchor=center, inner sep=0}, "{d_{A_1, \dots, A_n}}", from=1-1, to=1-8]
	\arrow["{d_{A_1, \dots, A_{n-1}} \otimes id}"{description}, dashed, from=1-1, to=2-3]
	\arrow["{\bigotimes_{i=1}^{n-1} \oc f_i \otimes id}"{description, pos=0.4}, dashed, from=1-1, to=4-2]
	\arrow["{(2)}"{description}, draw=none, from=1-1, to=4-3]
	\arrow["{\bigotimes_{i=1}^n \oc f_i}"{description}, from=1-1, to=7-1]
	\arrow["{(4)}"{description, pos=0.6}, draw=none, from=1-8, to=4-5]
	\arrow["\begin{array}{c} \bigotimes_{i=1}^{n-1} \oc f_i \otimes id \otimes \\ \bigotimes_{i=1}^{n-1} \oc f_i \otimes id \end{array}"{description}, dashed, from=1-8, to=4-7]
	\arrow[""{name=1, anchor=center, inner sep=0}, "\begin{array}{c} \bigotimes_{i=1}^n \oc f_i \otimes \\ \bigotimes_{i=1}^n \oc f_i \end{array}"{description}, from=1-8, to=7-8]
	\arrow["{id \otimes d_{A_n}}"', dashed, from=2-3, to=2-5]
	\arrow["\begin{array}{c} \bigotimes_{i=1}^{n-1} \oc f_i \otimes \\ \bigotimes_{i=1}^{n-1} \oc f_i \otimes id \end{array}"{description}, dashed, from=2-3, to=4-3]
	\arrow["{id \otimes \sigma \otimes id}"{description}, dashed, from=2-5, to=1-8]
	\arrow["\begin{array}{c} \bigotimes_{i=1}^{n-1} \oc f_i \otimes \\ \bigotimes_{i=1}^{n-1} \oc f_i \otimes id \end{array}"{description}, dashed, from=2-5, to=4-5]
	\arrow["{d_{B_1, \dots, B_{n-1}} \otimes id}"', curve={height=24pt}, dashed, from=4-2, to=4-3]
	\arrow["{(3)}"{description}, draw=none, from=4-3, to=2-5]
	\arrow["{id \otimes d_{A_n}}"', dashed, from=4-3, to=4-5]
	\arrow["{id \otimes \oc f_n}"{description}, dashed, from=4-3, to=6-3]
	\arrow["{id \otimes \sigma \otimes id}"', curve={height=18pt}, dashed, from=4-5, to=4-7]
	\arrow["{id \otimes \oc f_n \otimes \oc f_n}"{description}, dashed, from=4-5, to=6-5]
	\arrow["{id \otimes \oc f_n \otimes id \otimes \oc f_n}"{description}, dashed, from=4-7, to=7-8]
	\arrow["{(6)}"{description}, draw=none, from=6-3, to=4-5]
	\arrow["{id \otimes d_{B_n}}"', dashed, from=6-3, to=6-5]
	\arrow["{(7)}"{description}, curve={height=24pt}, draw=none, from=6-5, to=4-7]
	\arrow["{id \otimes \sigma \otimes id}"{description}, dashed, from=6-5, to=7-8]
	\arrow["{(5)}"{description, pos=0.6}, draw=none, from=7-1, to=4-2]
	\arrow["{d_{B_1, \dots, B_{n-1}} \otimes id}"{description}, dashed, from=7-1, to=6-3]
	\arrow[""{name=2, anchor=center, inner sep=0}, "{d_{B_1, \dots, B_n}}"', from=7-1, to=7-8]
	\arrow["{(1)}"{description}, draw=none, from=2-5, to=0]
	\arrow["{(9)}"{description}, curve={height=24pt}, draw=none, from=4-7, to=1]
	\arrow["{(8)}"{description}, draw=none, from=6-3, to=2]
\end{tikzcd}\]
(1) commutes by definition of $d_{A_1, \dots, A_n}$, (2) commutes by the induction hypothesis, (3) commutes by functoriality of $\otimes$, (4) commutes by naturality of $\sigma$ and functoriality of $\otimes$, (5) commutes by functoriality of $\otimes$, (6) commutes by naturality of $d$, (7) commutes by naturality of $\sigma$ and functoriality of $\otimes$, (8) commutes by definition of $d_{B_1, \dots, B_n}$ and (9) commutes by functoriality of $\otimes$.\qedhere
	\end{itemize}
\end{proof}

Lemmas \ref{lem:dcomonoidaxiom}, \ref{lem:comonoidn} and \ref{lem:dsymm} prove the properties needed to show that $(\bigotimes_{i=1}^n \oc A_i, d_{A_1, \dots, A_n}, e_{A_1, \dots, A_n})$ is a comonoid (Lemma \ref{lem:comonoidgen}).

\begin{lemma}
	\label{lem:dcomonoidaxiom}
	For every $n$ the following diagram commutes:
\[\begin{tikzcd}[ampersand replacement=\&,cramped]
	{\bigotimes_{i=1}^n \oc A_i \otimes \bigotimes_{i=1}^n \oc A_i} \&\& {\bigotimes_{i=1}^n \oc A_i} \&\& {\bigotimes_{i=1}^n \oc A_i \otimes \bigotimes_{i=1}^n \oc A_i} \\
	\\
	\&\& {\bigotimes_{i=1}^n \oc A_i \otimes \bigotimes_{i=1}^n \oc A_i \otimes \bigotimes_{i=1}^n \oc A_i}
	\arrow["{id \otimes d_{A_1, \dots, A_n}}"{description}, from=1-1, to=3-3]
	\arrow["{d_{A_1, \dots, A_n}}"', from=1-3, to=1-1]
	\arrow["{d_{A_1, \dots, A_n}}", from=1-3, to=1-5]
	\arrow["{d_{A_1, \dots, A_n} \otimes id}"{description}, from=1-5, to=3-3]
\end{tikzcd}\]
\end{lemma}
\begin{proof}
	By induction on $n$:
	\begin{itemize}
		\item If $n=0$, the diagram commutes by coherence.
\[\begin{tikzcd}[ampersand replacement=\&,cramped]
	{I \otimes I} \&\& I \&\& {I \otimes I} \\
	\\
	\&\& {I \otimes I \otimes I}
	\arrow["{id \otimes \lambda_I^{-1}}"{description}, from=1-1, to=3-3]
	\arrow["{\lambda_I^{-1}}"', from=1-3, to=1-1]
	\arrow["{\lambda_I^{-1}}", from=1-3, to=1-5]
	\arrow["{\lambda_I^{-1} \otimes id}"{description}, from=1-5, to=3-3]
\end{tikzcd}\]
		\item If $n=1$, the diagram commutes because $(\oc A, d_A, e_A)$ is a comonoid.
\[\begin{tikzcd}[ampersand replacement=\&,cramped]
	{\oc A \otimes \oc A} \&\& {\oc A} \&\& {\oc A \otimes \oc A} \\
	\\
	\&\& {\oc A \otimes \oc A \otimes \oc A}
	\arrow["{id \otimes d_A}"{description}, from=1-1, to=3-3]
	\arrow["{d_A}"', from=1-3, to=1-1]
	\arrow["{d_A}", from=1-3, to=1-5]
	\arrow["{d_A \otimes id}"{description}, from=1-5, to=3-3]
\end{tikzcd}\]
		\item If $n > 1$, the commuting diagram is shown in Figure~\ref{fig:lem:dcomonoidaxiom} (Appendix~\ref{app:diagramasgrandes}).
		
		(1) and (2) commute by definition of $d_{A_1,\dots,A_n}$, (3) commutes by the induction hypothesis, (4) and (5) commute by functoriality of $\otimes$, (6) and (7) commute by naturality of $\sigma$, (8) and (9) commute by definition of $d_{A_1,\dots,A_n}$, (10) and (11) commute by functoriality of $\otimes$, (12) and (13) commute trivially and (14) commutes by coherence.\qedhere
	\end{itemize}
\end{proof}

\begin{lemma}
	\label{lem:comonoidn}
	For every $n$ the following diagram commutes:
\[\begin{tikzcd}[ampersand replacement=\&,cramped, column sep=large]
	\&\& {\bigotimes_{i=1}^n \oc A_i} \\
	\\
	{I \otimes \bigotimes_{i=1}^n \oc A_i} \&\& {\bigotimes_{i=1}^n \oc A_i \otimes \bigotimes_{i=1}^n \oc A_i} \&\& {\bigotimes_{i=1}^n \oc A_i \otimes I}
	\arrow["{\lambda^{-1}}"', from=1-3, to=3-1]
	\arrow["{d_{A_1, \dots, A_n}}", from=1-3, to=3-3]
	\arrow["{\rho^{-1}}", from=1-3, to=3-5]
	\arrow["{e_{A_1, \dots, A_n} \otimes id}", from=3-3, to=3-1]
	\arrow["{id \otimes e_{A_1, \dots, A_n}}"', from=3-3, to=3-5]
\end{tikzcd}\]
  \end{lemma}
\begin{proof}
We prove the commutation of the diagram on the left, the proof of diagram on the right is similar.

  By induction on $n$.
  \begin{itemize}
    \item If $n = 0$, the diagram commutes trivially.
\[\begin{tikzcd}[cramped]
	I && {I \otimes I} \\
	\\
	&& {I \otimes I}
	\arrow["{\lambda_I^{-1}}", from=1-1, to=1-3]
	\arrow["{\lambda_I^{-1}}"', from=1-1, to=3-3]
	\arrow["id", from=1-3, to=3-3]
\end{tikzcd}\]
    \item If $n = 1$, the diagram commutes because $(\oc A, d_A, e_A)$ is a comonoid.
\[\begin{tikzcd}[cramped]
	{\oc A} && {\oc A \otimes \oc A} \\
	\\
	&& {I \otimes \oc A}
	\arrow["{d_A}", from=1-1, to=1-3]
	\arrow["{\lambda^{-1}}"', from=1-1, to=3-3]
	\arrow["{e_A \otimes id}", from=1-3, to=3-3]
\end{tikzcd}\]
    \item If $n > 1$:
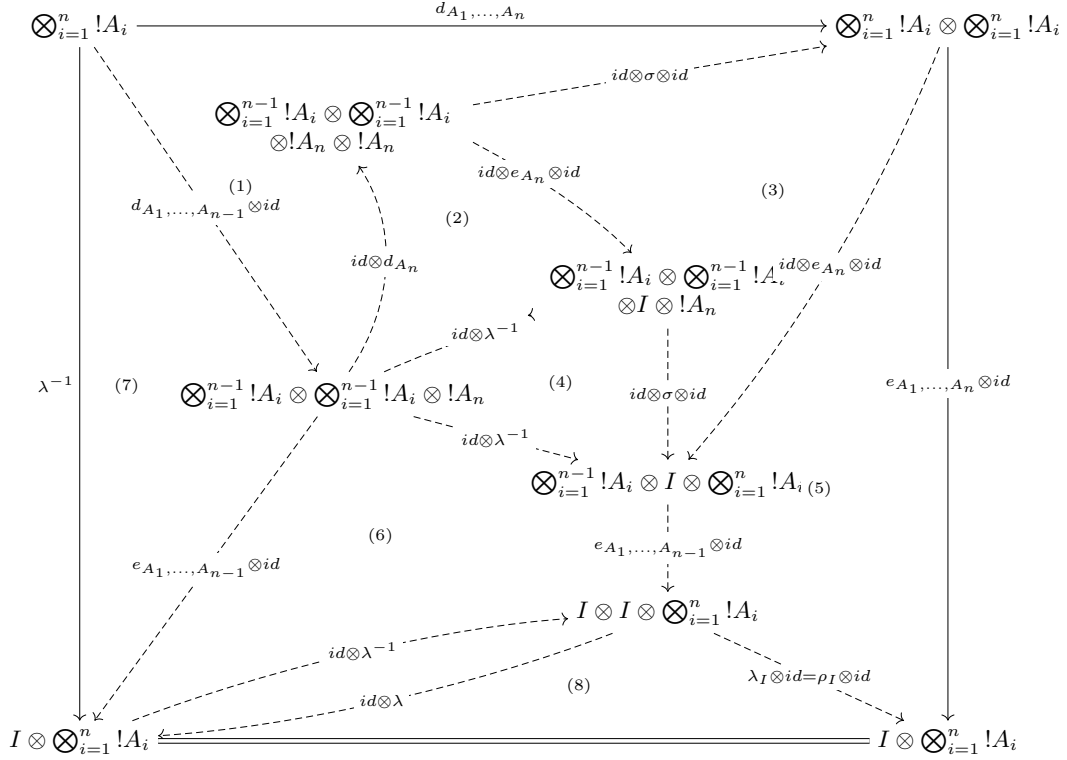
\[\begin{tikzcd}[ampersand replacement=\&,cramped,column sep=tiny]
	{\bigotimes_{i=1}^n \oc A_i} \&\&\&\& {\bigotimes_{i=1}^n \oc A_i \otimes \bigotimes_{i=1}^n \oc A_i} \\
	\& \begin{array}{c} \bigotimes_{i=1}^{n-1} \oc A_i \otimes \bigotimes_{i=1}^{n-1} \oc A_i \\\otimes \oc A_n \otimes \oc A_n \end{array} \\
	\\
	\&\&\& \begin{array}{c} \bigotimes_{i=1}^{n-1} \oc A_i \otimes \bigotimes_{i=1}^{n-1} \oc A_i \\\otimes I \otimes \oc A_n \end{array} \\
	\& {\bigotimes_{i=1}^{n-1} \oc A_i \otimes \bigotimes_{i=1}^{n-1} \oc A_i \otimes \oc A_n} \\
	\&\&\& {\bigotimes_{i=1}^{n-1} \oc A_i \otimes I \otimes \bigotimes_{i=1}^n \oc A_i} \\
	\\
	\&\&\& {I \otimes I \otimes \bigotimes_{i=1}^n \oc A_i} \\
	\\
	{I \otimes \bigotimes_{i=1}^n \oc A_i} \&\&\&\& {I\otimes \bigotimes_{i=1}^n \oc A_i}
	\arrow["{d_{A_1, \dots, A_n}}", from=1-1, to=1-5]
	\arrow[""{name=0, anchor=center, inner sep=0}, "{d_{A_1, \dots, A_{n-1}} \otimes id}"{description}, dashed, from=1-1, to=5-2]
	\arrow[""{name=1, anchor=center, inner sep=0}, "{\lambda^{-1}}"', from=1-1, to=10-1]
	\arrow["{(3)}"{description, pos=0.7}, draw=none, from=1-5, to=4-4]
	\arrow["{id \otimes e_{A_n} \otimes id}"{description}, curve={height=-12pt}, dashed, from=1-5, to=6-4]
	\arrow[""{name=2, anchor=center, inner sep=0}, "{e_{A_1, \dots, A_n} \otimes id}"{description}, from=1-5, to=10-5]
	\arrow["{id \otimes \sigma \otimes id}"{description}, dashed, from=2-2, to=1-5]
	\arrow[""{name=3, anchor=center, inner sep=0}, "{id \otimes e_{A_n} \otimes id}"{description}, curve={height=-18pt}, dashed, from=2-2, to=4-4]
	\arrow["{id \otimes \sigma \otimes id}"{description}, dashed, from=4-4, to=6-4]
	\arrow[""{name=4, anchor=center, inner sep=0}, "{id \otimes d_{A_n}}"{description}, curve={height=24pt}, dashed, from=5-2, to=2-2]
	\arrow["{id \otimes \lambda^{-1}}"{description}, curve={height=-6pt}, dashed, from=5-2, to=4-4]
	\arrow[""{name=5, anchor=center, inner sep=0}, "{id \otimes \lambda^{-1}}"{description}, dashed, from=5-2, to=6-4]
	\arrow[""{name=6, anchor=center, inner sep=0}, "{e_{A_1, \dots, A_{n-1}} \otimes id}"{description}, dashed, from=5-2, to=10-1]
	\arrow["{e_{A_1, \dots, A_{n-1}} \otimes id}"{description}, dashed, from=6-4, to=8-4]
	\arrow["{id \otimes \lambda}"{description}, curve={height=-12pt}, dashed, from=8-4, to=10-1]
	\arrow["{\lambda_I \otimes id = \rho_I \otimes id}"{description}, dashed, from=8-4, to=10-5]
	\arrow["{id \otimes \lambda^{-1}}"{description}, curve={height=-12pt}, dashed, from=10-1, to=8-4]
	\arrow[""{name=7, anchor=center, inner sep=0}, equals, from=10-1, to=10-5]
	\arrow["{(1)}"{description}, draw=none, from=2-2, to=0]
	\arrow["{(4)}"{description}, draw=none, from=4-4, to=5]
	\arrow["{(2)}"{description}, draw=none, from=4, to=3]
	\arrow["{(7)}"{description}, draw=none, from=5-2, to=1]
	\arrow["{(6)}"{description}, draw=none, from=6-4, to=6]
	\arrow["{(5)}"{description}, draw=none, from=8-4, to=2]
	\arrow["{(8)}"{description}, draw=none, from=8-4, to=7]
\end{tikzcd}\]

(1) commutes by definition of $d_{A_1, \dots, A_n}$, (2) commutes because $(\oc A_n, d_{A_n}, e_{A_n})$ is a comonoid, (3) commutes by naturality of $\sigma$, (4) commutes by coherence, (5) commutes by definition of $e_{A_1, \dots, A_n}$, (6) commutes by functoriality of $\otimes$, (7) commutes by the induction hypothesis and (8) commutes by coherence.\qedhere
  \end{itemize}
\end{proof}

\begin{lemma}
	\label{lem:dsymm}
	  For every $n$ the following diagram commutes:
  
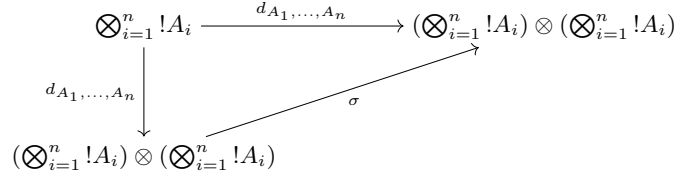
\[\begin{tikzcd}[ampersand replacement=\&,cramped]
	  {\bigotimes_{i=1}^n \oc A_i} \&\& {(\bigotimes_{i=1}^n \oc A_i) \otimes (\bigotimes_{i=1}^n \oc A_i)} \\
	  \\
	  {(\bigotimes_{i=1}^n \oc A_i) \otimes (\bigotimes_{i=1}^n \oc A_i)}
	  \arrow["{d_{A_1,\dots,A_n}}", from=1-1, to=1-3]
	  \arrow["{d_{A_1,\dots,A_n}}"', from=1-1, to=3-1]
	  \arrow["\sigma"', from=3-1, to=1-3]
  \end{tikzcd}\]
  \end{lemma}
\begin{proof}
	By induction on $n$:
	\begin{itemize}
		\item If $n=0$, the diagram commutes by coherence.
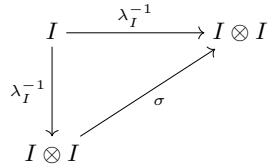
\[\begin{tikzcd}[ampersand replacement=\&,cramped]
	I \&\& {I \otimes I} \\
	\\
	{I \otimes I}
	\arrow["{\lambda_I^{-1}}", from=1-1, to=1-3]
	\arrow["{\lambda_I^{-1}}"', from=1-1, to=3-1]
	\arrow["\sigma"', from=3-1, to=1-3]
\end{tikzcd}\]

		\item If $n=1$, the diagram commutes because $(\oc A, d_A, e_A)$ is a commutative comonoid.
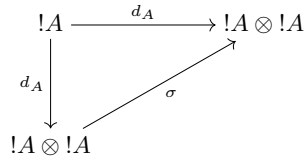
\[\begin{tikzcd}[ampersand replacement=\&,cramped]
	{\oc A} \&\& {\oc A \otimes \oc A} \\
	\\
	{\oc A \otimes \oc A}
	\arrow["{d_A}", from=1-1, to=1-3]
	\arrow["{d_A}"', from=1-1, to=3-1]
	\arrow["\sigma"', from=3-1, to=1-3]
\end{tikzcd}\]

\item If $n>1$:
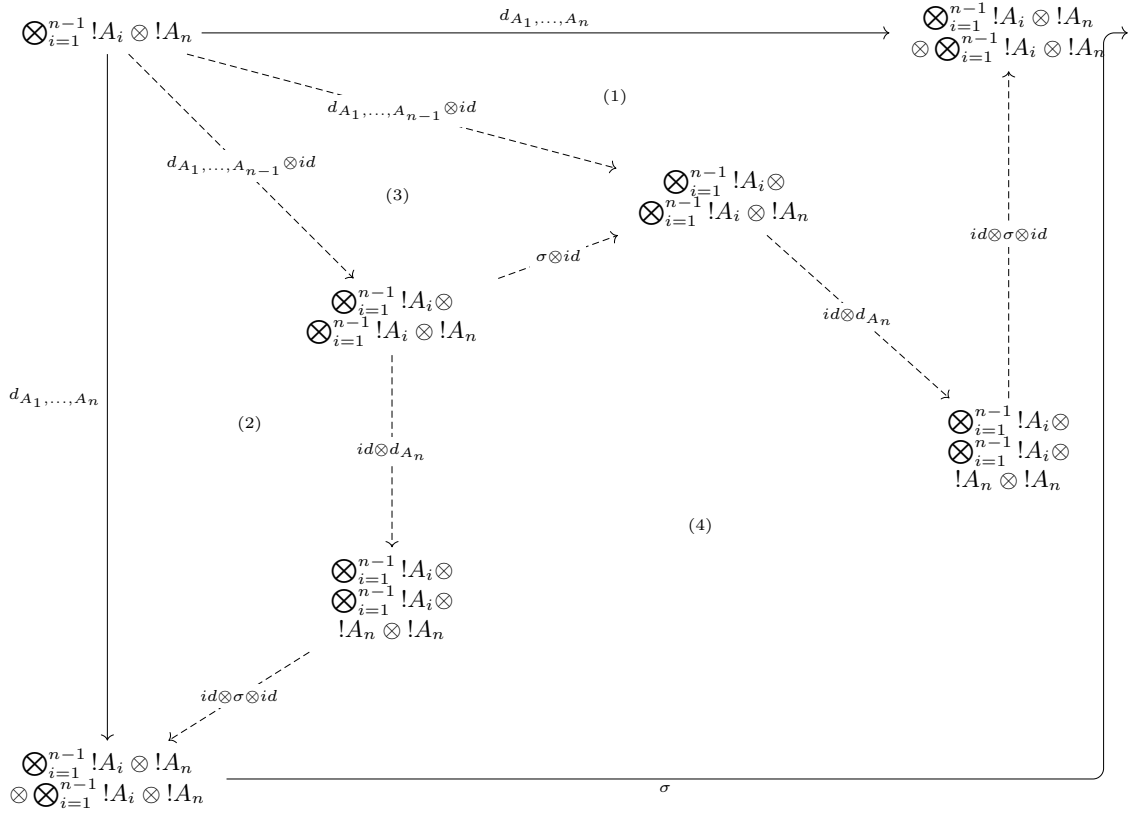
\[\begin{tikzcd}[ampersand replacement=\&,cramped]
	{\bigotimes_{i=1}^{n-1} \oc A_i \otimes \oc A_n} \&\&\&\& \begin{array}{c} \bigotimes_{i=1}^{n-1} \oc A_i \otimes \oc A_n\\ \otimes \bigotimes_{i=1}^{n-1} \oc A_i \otimes \oc A_n \end{array} \\
	\\
	\&\&\& \begin{array}{c} \bigotimes_{i=1}^{n-1} \oc A_i \otimes \\\bigotimes_{i=1}^{n-1} \oc A_i \otimes \oc A_n \end{array} \\
	\& \begin{array}{c} \bigotimes_{i=1}^{n-1} \oc A_i \otimes \\\bigotimes_{i=1}^{n-1} \oc A_i \otimes \oc A_n \end{array} \\
	\&\&\&\& \begin{array}{c} \bigotimes_{i=1}^{n-1} \oc A_i \otimes \\\bigotimes_{i=1}^{n-1} \oc A_i \otimes \\ \oc A_n \otimes \oc A_n \end{array} \\
	\& \begin{array}{c} \bigotimes_{i=1}^{n-1} \oc A_i \otimes \\\bigotimes_{i=1}^{n-1} \oc A_i \otimes \\ \oc A_n \otimes \oc A_n \end{array} \\
	\\
	\begin{array}{c} \bigotimes_{i=1}^{n-1} \oc A_i \otimes \oc A_n \\ \otimes \bigotimes_{i=1}^{n-1} \oc A_i \otimes \oc A_n \end{array}
	\arrow[""{name=0, anchor=center, inner sep=0}, "{d_{A_1,\dots,A_n}}", from=1-1, to=1-5]
	\arrow[""{name=1, anchor=center, inner sep=0}, "{d_{A_1,\dots,A_{n-1}} \otimes id}"{description}, dashed, from=1-1, to=3-4]
	\arrow["{d_{A_1,\dots,A_{n-1}} \otimes id}"{description}, dashed, from=1-1, to=4-2]
	\arrow[""{name=2, anchor=center, inner sep=0}, "{d_{A_1,\dots,A_n}}"', from=1-1, to=8-1]
	\arrow["{id \otimes d_{A_n}}"{description}, dashed, from=3-4, to=5-5]
	\arrow["{\sigma \otimes id}"{description}, dashed, from=4-2, to=3-4]
	\arrow[""{name=3, anchor=center, inner sep=0}, "{id \otimes d_{A_n}}"{description}, dashed, from=4-2, to=6-2]
	\arrow["{id \otimes \sigma \otimes id}"{description}, dashed, from=5-5, to=1-5]
	\arrow["{(4)}"{description}, draw=none, from=6-2, to=5-5]
	\arrow["{id \otimes \sigma \otimes id}"{description}, dashed, from=6-2, to=8-1]
	\arrow["\sigma"', from=8-1, to=1-5, rounded corners, to path={-- ([xshift=330pt]\tikztostart.east) \tikztonodes |- (\tikztotarget.east)}]
	\arrow["{(1)}"{description}, draw=none, from=3-4, to=0]
	\arrow["{(3)}"{description}, draw=none, from=4-2, to=1]
	\arrow["{(2)}"{description}, draw=none, from=3, to=2]
\end{tikzcd}\]
(1) and (2) commute by definition of $d_{A_1, \dots, A_n}$, (3) commutes by the induction hypothesis and functoriality of $\otimes$ and (4) commutes because $(\oc A_n, d_{A_n}, e_{A_n})$ is a commutative comonoid.\qedhere
	\end{itemize}
\end{proof}

\comonoidgen*
\begin{proof}
	By Lemmas \ref{lem:dcomonoidaxiom}, \ref{lem:comonoidn} and \ref{lem:dsymm}.
\end{proof}

Lemmas \ref{lem:deltacoalgebraprop1} and \ref{lem:deltacoalgebraprop2} prove the properties needed to show that $(\bigotimes_{i=1}^n \oc A_i, \delta_{A_1, \dots, A_n})$ is a coalgebra (Lemma \ref{lem:coalgebragen}).

\begin{lemma}
	\label{lem:deltacoalgebraprop1}
	  For every $n$ the following diagram commutes:
  \[\begin{tikzcd}[ampersand replacement=\&,cramped]
	  {\bigotimes_{i=1}^n \oc A_i} \&\& {\oc (\bigotimes_{i=1}^n \oc A_i)} \\
	  \\
	  {\oc (\bigotimes_{i=1}^n \oc A_i)} \&\& {\oc \oc (\bigotimes_{i=1}^n \oc A_i)}
	  \arrow["{\delta_{A_1, \dots, A_n} }", from=1-1, to=1-3]
	  \arrow["{\delta_{A_1, \dots, A_n} }"', from=1-1, to=3-1]
	  \arrow["{\delta_{\bigotimes_{i=1}^n \oc A_i}}", from=1-3, to=3-3]
	  \arrow["{\oc (\delta_{A_1, \dots, A_n})}"', from=3-1, to=3-3]
  \end{tikzcd}\]
\end{lemma}
\begin{proof}
	By induction on $n$:
\begin{itemize}
	\item If $n=0$, the diagram commutes because $(\oc,\varepsilon, \delta, m_{A,B}, m_I)$ is a monoidal comonad.
\[\begin{tikzcd}[ampersand replacement=\&,cramped]
	I \&\& {\oc I} \\
	\\
	{\oc I} \&\& {\oc\oc I}
	\arrow["{m_I}", from=1-1, to=1-3]
	\arrow["{m_I}"', from=1-1, to=3-1]
	\arrow["{\delta_I}", from=1-3, to=3-3]
	\arrow["{\oc m_I}"', from=3-1, to=3-3]
\end{tikzcd}\]

	\item If $n=1$, the diagram commutes because $(\oc,\varepsilon, \delta, m_{A,B}, m_I)$ is a comonad.
\[\begin{tikzcd}[ampersand replacement=\&,cramped]
	{\oc A} \&\& {\oc \oc A} \\
	\\
	{\oc \oc A} \&\& {\oc \oc \oc A}
	\arrow["{\delta_{A}}", from=1-1, to=1-3]
	\arrow["{\delta_A}"', from=1-1, to=3-1]
	\arrow["{\delta_{\oc A}}", from=1-3, to=3-3]
	\arrow["{\oc \delta_A}"', from=3-1, to=3-3]
\end{tikzcd}\]

\item If $n > 1$:
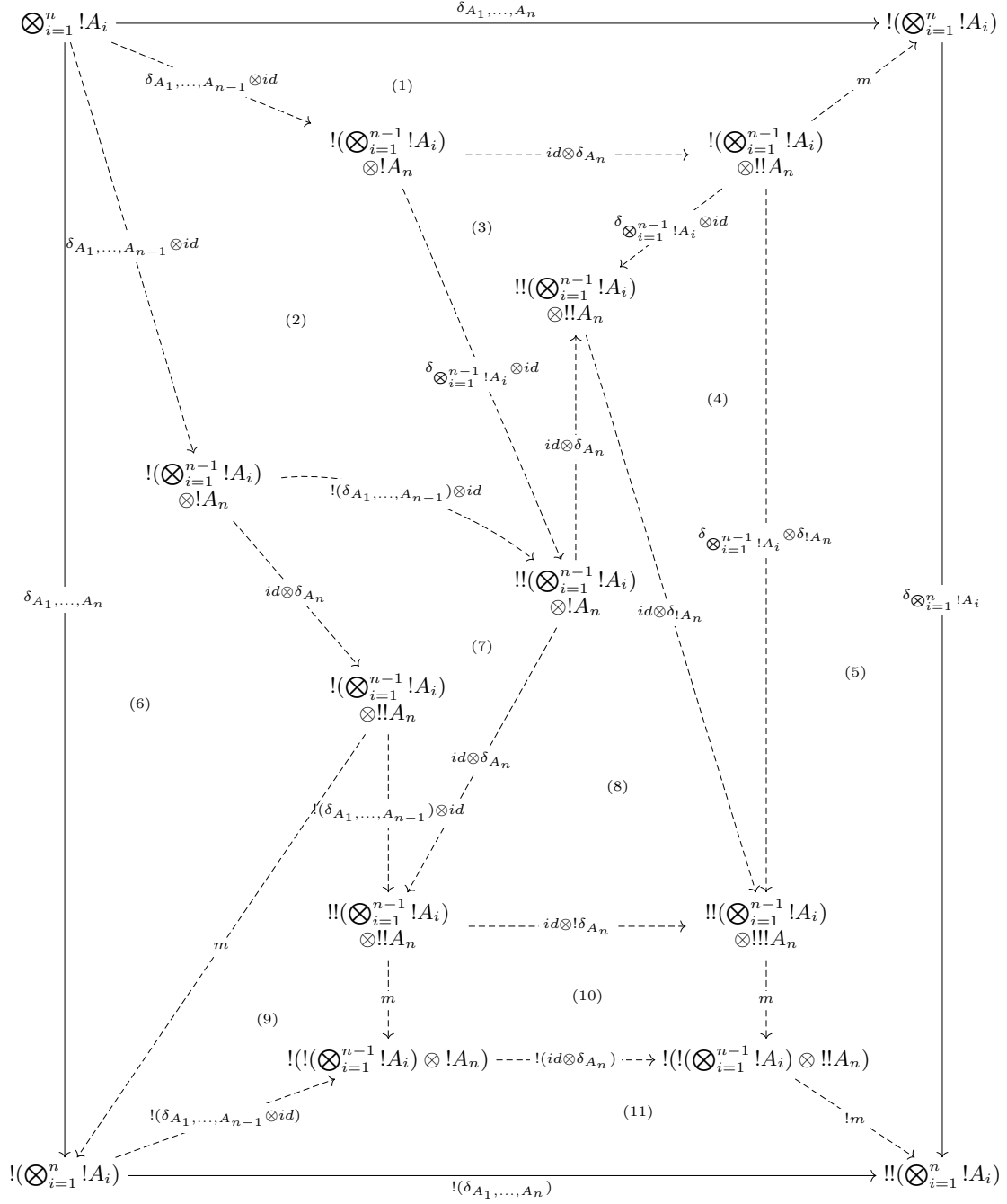
\[\begin{tikzcd}[ampersand replacement=\&,cramped,column sep=0pt]
	{\bigotimes_{i=1}^n \oc A_i} \&\&\&\&\& {\oc (\bigotimes_{i=1}^n \oc A_i)} \\
	\\
	\&\& \begin{array}{c} \oc (\bigotimes_{i=1}^{n-1} \oc A_i)\\ \otimes \oc A_n \end{array} \&\& \begin{array}{c} \oc (\bigotimes_{i=1}^{n-1} \oc A_i)\\ \otimes \oc \oc A_n \end{array} \\
	\\
	\&\&\& \begin{array}{c} \oc \oc (\bigotimes_{i=1}^{n-1} \oc A_i)\\ \otimes \oc \oc A_n \end{array} \\
	\\
	\\
	\& \begin{array}{c} \oc (\bigotimes_{i=1}^{n-1} \oc A_i)\\ \otimes \oc A_n \end{array} \\
	\&\&\& \begin{array}{c} \oc \oc (\bigotimes_{i=1}^{n-1} \oc A_i)\\ \otimes \oc A_n \end{array} \\
	\&\& \begin{array}{c} \oc (\bigotimes_{i=1}^{n-1} \oc A_i)\\ \otimes \oc \oc A_n \end{array} \\
	\\
	\\
	\\
	\&\& \begin{array}{c} \oc \oc (\bigotimes_{i=1}^{n-1} \oc A_i)\\ \otimes \oc \oc A_n \end{array} \&\& \begin{array}{c} \oc \oc (\bigotimes_{i=1}^{n-1} \oc A_i)\\ \otimes \oc \oc \oc A_n \end{array} \\
	\\
	\&\& {\oc (\oc (\bigotimes_{i=1}^{n-1} \oc A_i) \otimes \oc A_n)} \&\& {\oc (\oc (\bigotimes_{i=1}^{n-1} \oc A_i) \otimes \oc \oc A_n)} \\
	\\
	{\oc (\bigotimes_{i=1}^n \oc A_i)} \&\&\&\&\& {\oc \oc (\bigotimes_{i=1}^n \oc A_i)}
	\arrow["{\delta_{A_1, \dots, A_n} }", from=1-1, to=1-6]
	\arrow["{\delta_{A_1, \dots, A_{n-1}} \otimes id }"{description}, dashed, from=1-1, to=3-3]
	\arrow["{(1)}"{description}, draw=none, from=1-1, to=3-5]
	\arrow["{\delta_{A_1, \dots, A_{n-1}} \otimes id }"{description}, dashed, from=1-1, to=8-2]
	\arrow["{\delta_{A_1, \dots, A_n} }"{description}, from=1-1, to=18-1]
	\arrow["{\delta_{\bigotimes_{i=1}^n \oc A_i}}"{description}, from=1-6, to=18-6]
	\arrow["{id \otimes \delta_{A_n}}"{description}, dashed, from=3-3, to=3-5]
	\arrow["{(3)}"{description}, draw=none, from=3-3, to=5-4]
	\arrow["{(2)}"{description}, draw=none, from=3-3, to=8-2]
	\arrow["{\delta_{\bigotimes_{i=1}^{n-1} \oc A_i} \otimes id}"{description}, dashed, from=3-3, to=9-4]
	\arrow["m"{description}, dashed, from=3-5, to=1-6]
	\arrow["{\delta_{\bigotimes_{i=1}^{n-1} \oc A_i} \otimes id}"{description}, dashed, from=3-5, to=5-4]
	\arrow["{\delta_{\bigotimes_{i=1}^{n-1} \oc A_i} \otimes \delta_{\oc A_n}}"{description}, dashed, from=3-5, to=14-5]
	\arrow["{(4)}"{description, pos=0.3}, curve={height=30pt}, draw=none, from=3-5, to=14-5]
	\arrow["{(5)}"{description}, draw=none, from=3-5, to=18-6]
	\arrow["{id \otimes \delta_{\oc A_n}}"{description}, dashed, from=5-4, to=14-5]
	\arrow["{\oc (\delta_{A_1, \dots, A_{n-1}}) \otimes id}"{description}, curve={height=-24pt}, dashed, from=8-2, to=9-4]
	\arrow["{id \otimes \delta_{A_n}}"{description}, dashed, from=8-2, to=10-3]
	\arrow["{(6)}"{description, pos=0.3}, curve={height=12pt}, draw=none, from=8-2, to=18-1]
	\arrow["{id \otimes \delta_{A_n}}"{description}, dashed, from=9-4, to=5-4]
	\arrow["{(7)}"{description}, draw=none, from=9-4, to=10-3]
	\arrow["{id \otimes \delta_{A_n}}"{description}, dashed, from=9-4, to=14-3]
	\arrow["{(8)}"{description}, curve={height=30pt}, draw=none, from=9-4, to=14-5]
	\arrow["{\oc (\delta_{A_1, \dots, A_{n-1}}) \otimes id}"{description}, dashed, from=10-3, to=14-3]
	\arrow["m"{description}, dashed, from=10-3, to=18-1]
	\arrow["{id \otimes \oc \delta_{A_n}}"{description}, dashed, from=14-3, to=14-5]
	\arrow["m"{description}, dashed, from=14-3, to=16-3]
	\arrow["{(10)}"{description}, draw=none, from=14-3, to=16-5]
	\arrow["{(9)}"{description, pos=0.3}, draw=none, from=14-3, to=18-1]
	\arrow["m"{description}, dashed, from=14-5, to=16-5]
	\arrow["{\oc (id \otimes \delta_{A_n})}"{description}, dashed, from=16-3, to=16-5]
	\arrow["{(11)}"{description, pos=0.4}, draw=none, from=16-3, to=18-6]
	\arrow["{\oc m}"{description}, dashed, from=16-5, to=18-6]
	\arrow["{\oc (\delta_{A_1, \dots, A_{n-1}} \otimes id)}"{description}, dashed, from=18-1, to=16-3]
	\arrow["{\oc (\delta_{A_1, \dots, A_n})}"', from=18-1, to=18-6]
\end{tikzcd}\]
(1) commutes by definition of $\delta_{A_1, \dots, A_n}$, (2) commutes by the induction hypothesis, (3) and (4) commute by functoriality of $\otimes$, (5) commutes because $\delta$ is a monoidal natural transformation between $\oc$ and $\oc\oc$, (6) commutes by definition of $\delta_{A_1, \dots, A_n}$, (7) commutes by functoriality of $\otimes$, (8) commutes because $(\oc, \varepsilon, \delta)$ is a comonad and by functoriality of $\otimes$, (9) and (10) commute by naturality of $m$, and (11) commutes by definition of $\delta_{A_1, \dots, A_n}$ and functoriality of $\otimes$.\qedhere
\end{itemize}
\end{proof}

\begin{lemma}
	\label{lem:deltacoalgebraprop2}
	  For every $n$ the following diagram commutes:
  \[\begin{tikzcd}[ampersand replacement=\&,cramped]
	  {\oc (\bigotimes_{i=1}^n \oc A_i)} \\
	  \\
	  {\bigotimes_{i=1}^n \oc A_i} \& {\bigotimes_{i=1}^n \oc A_i}
	  \arrow["{\varepsilon_{\bigotimes_{i=1}^n \oc A_i}}", from=1-1, to=3-2]
	  \arrow["{\delta_{A_1,\dots,A_n}}", from=3-1, to=1-1]
	  \arrow[equals, from=3-1, to=3-2]
  \end{tikzcd}\]
  \end{lemma}
\begin{proof}
  By induction on $n$:
  \begin{itemize}
    \item If $n = 0$, the diagram commutes because $\varepsilon$ is a monoidal natural transformation.
  \[\begin{tikzcd}[ampersand replacement=\&,cramped]
    {\oc I} \\
    \\
    I \& I
    \arrow["{\varepsilon_I}", from=1-1, to=3-2]
    \arrow["{\delta_\varnothing = m_I}", from=3-1, to=1-1]
    \arrow[equals, from=3-1, to=3-2]
  \end{tikzcd}\]
    \item If $n = 1$, the diagram commutes because $(\oc,\varepsilon,\delta)$ is a comonad.
  \[\begin{tikzcd}[ampersand replacement=\&,cramped]
    {\oc \oc A} \\
    \\
    {\oc A} \& {\oc A}
    \arrow["{\varepsilon_{\oc A}}", from=1-1, to=3-2]
    \arrow["{\delta_A}", from=3-1, to=1-1]
    \arrow[equals, from=3-1, to=3-2]
  \end{tikzcd}\]
    \item If $n > 1$:
  \[\begin{tikzcd}[ampersand replacement=\&,cramped]
    {\oc (\bigotimes_{i=1}^n \oc A_i)} \\
    \\
    \&\& {\oc (\bigotimes_{i=1}^{n-1} \oc A_i) \otimes \oc \oc A_n} \& {(4)} \\
    \& {(1)} \\
    \&\& {\oc (\bigotimes_{i=1}^{n-1} \oc A_i) \otimes \oc A_n} \& {(3)} \\
    \&\& {(2)} \\
    {\bigotimes_{i=1}^{n-1} \oc A_i \otimes \oc A_n} \&\&\&\& {\bigotimes_{i=1}^{n-1} \oc A_i \otimes \oc A_n}
    \arrow["{\varepsilon_{\bigotimes_{i=1}^n \oc A_i}}", from=1-1, to=7-5, rounded corners, to path={ (\tikztostart.east) -- ([xshift=200pt]\tikztostart.east) \tikztonodes -| (\tikztotarget.north)}]
    \arrow["m"{description}, dashed, from=3-3, to=1-1]
    \arrow["{\varepsilon_{\bigotimes_{i=1}^{n-1}\oc A_i} \otimes \varepsilon_{\oc A_n}}"{description}, curve={height=-30pt}, dashed, from=3-3, to=7-5]
    \arrow["{id \otimes \delta_{A_n}}"{description}, dashed, from=5-3, to=3-3]
    \arrow["{\varepsilon_{\bigotimes_{i=1}^{n-1}\oc A_i} \otimes id_{\oc A_n}}"{description}, dashed, from=5-3, to=7-5]
    \arrow["{\delta_{A_1,\dots,A_n}}", from=7-1, to=1-1]
    \arrow["{\delta_{A_1, \dots, A_{n-1}} \otimes id_{\oc A_n}}"{description}, dashed, from=7-1, to=5-3]
    \arrow[equals, from=7-1, to=7-5]
  \end{tikzcd}\]
  (1) commutes by definition of $\delta_{A_1, \dots, A_n}$, (2) commutes by the induction hypothesis and functoriality of $\otimes$, (3) commutes by comonoidality of $(\oc,\varepsilon,\delta)$ and functoriality of $\otimes$, and (4) commutes because $\varepsilon$ is a monoidal natural transformation.\qedhere
  \end{itemize}
  \end{proof}

\coalgebragen*
\begin{proof}
	By Lemmas~\ref{lem:deltacoalgebraprop1} and \ref{lem:deltacoalgebraprop2}.
\end{proof}

The following is an auxiliary lemma used to prove that $d_{A_1, \dots, A_n}$ is a coalgebra morphism (Lemma~\ref{lem:dcoalgebramorph}).

\begin{lemma}
	\label{lem:dcoalgebramorph:aux}
	The following diagram commutes:
\[\begin{tikzcd}[ampersand replacement=\&,cramped]
	{!A \otimes !B \otimes !C \otimes !D} \&\&\& {!(A \otimes B) \otimes !(C \otimes D)} \&\&\& {!(A \otimes B \otimes C \otimes D)} \\
	\\
	\\
	{!A \otimes !C \otimes !B \otimes !D} \&\&\& {!(A \otimes C) \otimes !(B \otimes D)} \&\&\& {!(A \otimes C \otimes B \otimes D)}
	\arrow["{m \otimes m}", from=1-1, to=1-4]
	\arrow["{id \otimes \sigma \otimes id}"{description}, from=1-1, to=4-1]
	\arrow["m", from=1-4, to=1-7]
	\arrow["{!(id \otimes \sigma \otimes id)}"{description}, from=1-7, to=4-7]
	\arrow["{m \otimes m}"', from=4-1, to=4-4]
	\arrow["m"', from=4-4, to=4-7]
\end{tikzcd}\]
\end{lemma}
\begin{proof}
\[\begin{tikzcd}[cramped,column sep=tiny]
	{!A \otimes !B \otimes !C \otimes !D} &&& {!(A \otimes B) \otimes !(C \otimes D)} &&& {!(A \otimes B \otimes C \otimes D)} \\
	&& {!A \otimes !(B \otimes C) \otimes !D} && {!(A \otimes B \otimes C) \otimes !D} \\
	\\
	&& {!A \otimes !(C \otimes B) \otimes !D} && {!(A \otimes C \otimes B) \otimes !D} \\
	{!A \otimes !C \otimes !B \otimes !D} &&& {!(A \otimes C) \otimes !(B \otimes D)} &&& {!(A \otimes C \otimes B \otimes D)}
	\arrow["{m \otimes m}", from=1-1, to=1-4]
	\arrow["{id \otimes m \otimes id}"{description}, dashed, from=1-1, to=2-3]
	\arrow[""{name=0, anchor=center, inner sep=0}, "{id \otimes \sigma \otimes id}"{description}, from=1-1, to=5-1]
	\arrow["m", from=1-4, to=1-7]
	\arrow[""{name=1, anchor=center, inner sep=0}, "{!(id \otimes \sigma \otimes id)}"{description}, from=1-7, to=5-7]
	\arrow[""{name=2, anchor=center, inner sep=0}, "{m \otimes id}"{description}, dashed, from=2-3, to=2-5]
	\arrow[""{name=3, anchor=center, inner sep=0}, "{id \otimes !\sigma \otimes id}"{description}, dashed, from=2-3, to=4-3]
	\arrow["{(3)}"{description}, draw=none, from=2-3, to=4-5]
	\arrow["m"{description}, dashed, from=2-5, to=1-7]
	\arrow[""{name=4, anchor=center, inner sep=0}, "{!(id \otimes \sigma) \otimes id}"{description}, dashed, from=2-5, to=4-5]
	\arrow[""{name=5, anchor=center, inner sep=0}, "{m \otimes id}"{description}, dashed, from=4-3, to=4-5]
	\arrow["m"{description}, dashed, from=4-5, to=5-7]
	\arrow["{id \otimes m \otimes id}"{description}, dashed, from=5-1, to=4-3]
	\arrow["{m \otimes m}"', from=5-1, to=5-4]
	\arrow["m"', from=5-4, to=5-7]
	\arrow["{(1)}"{description}, draw=none, from=1-4, to=2]
	\arrow["{(2)}"{description}, draw=none, from=3, to=0]
	\arrow["{(4)}"{description}, draw=none, from=4, to=1]
	\arrow["{(5)}"{description}, draw=none, from=5-4, to=5]
\end{tikzcd}\]
(1) commutes because $\oc$ is a monoidal functor, (2) commutes because $\oc$ is symmetric monoidal, (3) and (4) commute by naturality of $m$ and (5) commutes because $\oc$ is monoidal.
\end{proof}

\dcoalgebramorph*
\begin{proof}
	By induction on $n$ we show that the following diagram commutes:
\[\begin{tikzcd}[ampersand replacement=\&,cramped]
	{\bigotimes_{i=1}^n \oc A_i} \&\&\&\&\&\& {\oc (\bigotimes_{i=1}^n \oc A_i)} \\
	\\
	{\bigotimes_{i=1}^n \oc A_i \otimes \bigotimes_{i=1}^n \oc A_i} \&\&\& {\oc (\bigotimes_{i=1}^n \oc A_i) \otimes \oc (\bigotimes_{i=1}^n \oc A_i)} \&\&\& {\oc (\bigotimes_{i=1}^n \oc A_i \otimes \bigotimes_{i=1}^n \oc A_i)}
	\arrow["{\delta_{A_1, \dots, A_n}}", from=1-1, to=1-7]
	\arrow["{d_{A_1, \dots, A_n}}"', from=1-1, to=3-1]
	\arrow["{\oc (d_{A_1, \dots, A_n})}", from=1-7, to=3-7]
	\arrow["{\delta_{A_1, \dots, A_n} \otimes \delta_{A_1, \dots, A_n}}"', from=3-1, to=3-4]
	\arrow["{m_{\bigotimes_{i=1}^n \oc A_i,\bigotimes_{i=1}^n \oc A_i}}"', from=3-4, to=3-7]
\end{tikzcd}\]

	\begin{itemize}
		\item If $n=0$:
\[\begin{tikzcd}[cramped]
	I &&&&& {\oc I} \\
	&& {(1)} \\
	{I \otimes I} &&& {I \otimes !I} \\
	& {(2)} \\
	&&&& {(3)} \\
	{!I \otimes !I} &&&&& {!(I \otimes I)}
	\arrow["{m_I}", from=1-1, to=1-6]
	\arrow["{\lambda_I}"{description}, from=3-1, to=1-1]
	\arrow["{id \otimes m_I}"{description}, dashed, from=3-1, to=3-4]
	\arrow["{m_I \otimes m_I}"', from=3-1, to=6-1]
	\arrow["{\lambda_I}"{description}, dashed, from=3-4, to=1-6]
	\arrow["{m_I \otimes id}"{description}, dashed, from=3-4, to=6-1]
	\arrow["{m_{I,I}}"', from=6-1, to=6-6]
	\arrow["{!\lambda_I}"{description}, shift left=3, from=6-6, to=1-6]
\end{tikzcd}\]
(1) commutes by naturality of $\lambda$, (2) commutes by functoriality of $\otimes$ and (3) commutes because $\oc$ is a monoidal functor.
	\item If $n=1$, the diagram commutes because $d_A$ is a coalgebra morphism.
\[\begin{tikzcd}[cramped]
	{\oc A} && {\oc \oc A} \\
	\\
	{\oc A \otimes \oc A} \\
	\\
	{\oc \oc A \otimes \oc \oc A} && {\oc (\oc A \otimes \oc A)}
	\arrow["{\delta_A}", from=1-1, to=1-3]
	\arrow["{d_A}"', from=1-1, to=3-1]
	\arrow["{\oc d_A}", from=1-3, to=5-3]
	\arrow["{\delta_A \otimes \delta_A}"', from=3-1, to=5-1]
	\arrow["{m_{\oc A,\oc A}}"', from=5-1, to=5-3]
\end{tikzcd}\]

		\item If $n>1$, the commuting diagram is shown in Figure~\ref{fig:lem:dcoalgebramorph} (Appendix~\ref{app:diagramasgrandes}).

		(1) commutes by definition of $\delta_{A_1,\dots, A_n}$, (2) commutes by definition of $d_{A_1,\dots, A_n}$, (3) commutes by the induction hypothesis, (4) commutes by functoriality of $\otimes$, (5) commutes by naturality of $m$, (6) and (7) commute by functoriality of $\otimes$, (8) commutes because $d_{A_n}$ is a coalgebra morphism, (9) commutes by naturality of $\sigma$, (10) commutes by definition of $\delta_{A_1,\dots, A_n}$, (11) commutes by naturality of $\sigma$, (12) commutes by functoriality of $\otimes$, (13) commutes by naturality of $m$, (14) commutes by definition of $d_{A_1,\dots, A_n}$ and functoriality of $\oc$ and (15) commutes by Lemma~\ref{lem:dcoalgebramorph:aux}.\qedhere
	\end{itemize}
\end{proof}

\generalizedaxiom*
\begin{proof}
	By induction on $n$ we show that the following diagram commutes:
\[\begin{tikzcd}[ampersand replacement=\&,cramped]
	{\bigotimes_{i=1}^n \oc A_i} \&\& I \\
	\\
	{\oc (\bigotimes_{i=1}^n \oc A_i)} \&\& {\oc I}
	\arrow["{e_{A_1, \dots, A_n}}", from=1-1, to=1-3]
	\arrow["{\delta_{A_1, \dots, A_n}}"', from=1-1, to=3-1]
	\arrow["{m_I}", from=1-3, to=3-3]
	\arrow["{\oc (e_{A_1, \dots, A_n})}"', from=3-1, to=3-3]
\end{tikzcd}\]

	\begin{itemize}
		\item If $n = 0$, the diagram commutes trivially.
\[\begin{tikzcd}[ampersand replacement=\&,cramped]
	I \&\& I \\
	\\
	{\oc I} \&\& {\oc I}
	\arrow["{id_I}", from=1-1, to=1-3]
	\arrow["{m_I}"', from=1-1, to=3-1]
	\arrow["{m_I}", from=1-3, to=3-3]
	\arrow["{id_{\oc I}}"', from=3-1, to=3-3]
\end{tikzcd}\]

		\item If $n = 1$, the diagram commutes because $e_A$ is a coalgebra morphism.
\[\begin{tikzcd}[ampersand replacement=\&,cramped]
	{\oc A} \&\& I \\
	\\
	{\oc \oc A} \&\& {\oc I}
	\arrow["{e_A}", from=1-1, to=1-3]
	\arrow["{\delta_A}"', from=1-1, to=3-1]
	\arrow["{m_I}", from=1-3, to=3-3]
	\arrow["{\oc e_A}"', from=3-1, to=3-3]
\end{tikzcd}\]
		\item If $n > 1$:
\[\begin{tikzcd}[ampersand replacement=\&,cramped,sep=small]
	{\bigotimes_{i=1}^n \oc A_i} \&\&\&\&\& I \\
	\&\& {(1)} \\
	\& {I \otimes \oc A_n} \&\& {I \otimes I} \\
	\&\& {(2)} \\
	\& {(3)} \& {\oc I \otimes \oc A_n} \& {\oc I \otimes I} \\
	\&\& {(6)} \&\& {(4)} \\
	\& {\oc (\bigotimes_{i=1}^{n-1} \oc A_i) \otimes \oc A_n} \\
	\&\&\& {\oc (\bigotimes_{i=1}^{n-1} \oc A_i) \otimes I} \\
	\& {(5)} \& {(8)} \\
	\&\& {\oc (\bigotimes_{i=1}^{n-1} \oc A_i) \otimes \oc I} \& {(9)} \\
	\&\&\& {\oc (\bigotimes_{i=1}^{n-1} \oc A_i)} \\
	\& {\oc (\bigotimes_{i=1}^{n-1} \oc A_i) \otimes \oc \oc A_n} \& {\oc (\bigotimes_{i=1}^{n-1} \oc A_i \otimes I)} \&\& {(7)} \\
	\& {(10)} \\
	\&\& {(12)} \\
	\&\& {\oc (I \otimes I)} \& {(11)} \\
	\&\& {(13)} \\
	{\oc (\bigotimes_{i=1}^n \oc A_i)} \&\&\&\&\& {\oc I}
	\arrow["{e_{A_1, \dots, A_n}}", from=1-1, to=1-6]
	\arrow["{e_{A_1, \dots, A_{n-1}} \otimes id}"{description}, dashed, from=1-1, to=3-2]
	\arrow["{\delta_{A_1, \dots, A_{n-1}} \otimes id}"{description}, dashed, from=1-1, to=7-2]
	\arrow["{\delta_{A_1, \dots, A_n}}"', from=1-1, to=17-1]
	\arrow["{m_I}", from=1-6, to=17-6]
	\arrow["{id \otimes e_{A_n}}"{description}, dashed, from=3-2, to=3-4]
	\arrow["{m_I \otimes id}"{description}, dashed, from=3-2, to=5-3]
	\arrow["{\lambda_I = \rho_I}"{description}, dashed, from=3-4, to=1-6]
	\arrow["{m_I \otimes id}"{description}, dashed, from=3-4, to=5-4]
	\arrow["{id \otimes e_{A_n}}"{description}, dashed, from=5-3, to=5-4]
	\arrow["\rho"{description}, curve={height=-24pt}, dashed, from=5-4, to=17-6]
	\arrow["{\oc (e_{A_1, \dots, A_{n-1}}) \otimes id}"{description}, dashed, from=7-2, to=5-3]
	\arrow["{id \otimes e_{A_n}}"{description}, dashed, from=7-2, to=8-4]
	\arrow["{id \otimes \delta_{A_n}}"{description}, curve={height=-30pt}, dashed, from=7-2, to=12-2]
	\arrow["{\oc (e_{A_1, \dots, A_{n-1}}) \otimes id}"{description}, dashed, from=8-4, to=5-4]
	\arrow["{id \otimes m_I}"{description}, dashed, from=8-4, to=10-3]
	\arrow["\rho"{description}, curve={height=-24pt}, dashed, from=8-4, to=11-4]
	\arrow["m"{description}, dashed, from=10-3, to=12-3]
	\arrow["{\oc (e_{A_1, \dots, A_{n-1}})}"{description}, dashed, from=11-4, to=17-6]
	\arrow["{id \otimes \oc e_{A_n}}"{description}, dashed, from=12-2, to=10-3]
	\arrow["m"{description}, dashed, from=12-2, to=17-1]
	\arrow["{\oc \rho}"{description}, dashed, from=12-3, to=11-4]
	\arrow["{\oc (e_{A_1, \dots, A_{n-1}} \otimes id)}"{description}, curve={height=-30pt}, dashed, from=12-3, to=15-3]
	\arrow["{\oc \rho_I = \oc \lambda_I}"{description}, dashed, from=15-3, to=17-6]
	\arrow["{\oc (id \otimes e_{A_n})}"{description}, curve={height=-24pt}, dashed, from=17-1, to=12-3]
	\arrow["{\oc (e_{A_1, \dots, A_{n-1}} \otimes e_{A_n})}"{description}, dashed, from=17-1, to=15-3]
	\arrow["{\oc (e_{A_1, \dots, A_n})}"', from=17-1, to=17-6]
\end{tikzcd}\]

(1) commutes by definition of $e_{A_1, \dots, A_n}$, (2) commutes by functoriality of $\otimes$, (3) commutes by the induction hypothesis and functoriality of $\otimes$, (4) commutes by naturality of $\rho$, (5) commutes by definition of $\delta_{A_1, \dots, A_n}$, (6) commutes by functoriality of $\otimes$, (7) commutes by naturality of $\rho$, (8) commutes because $e_A$ is a coalgebra morphism and by functoriality of $\otimes$, (9) commutes because $\oc$ is a monoidal functor, (10) commutes by naturality of $m$, (11) commutes by naturality of $\rho$ and functoriality of $\oc$, (12) commutes by functoriality of $\otimes$ and $\oc$ and (13) commutes by definition of $e_{A_1, \dots, A_n}$ and functoriality of $\oc$.\qedhere
	\end{itemize}
\end{proof}

\subsection{Technical lemmas to prove Soundness}\label{proof:technical}

Lemmas \ref{lem:substitutionaux} to \ref{lem:weakeningcat}
are auxiliary results used in the proofs of
substitution and soundness (Lemmas \ref{lem:nonlinearsubstitution}
and
\ref{lem:linearsubstitution}, and Theorem
\ref{thm:soundness}).

\begin{lemma}
	\label{lem:substitutionaux}
	  Let $g: \bigotimes_{i=1}^n \oc A_i \otimes I \to D$ and let $f$ such that:
	  \begin{align*}
		  f =& \bigotimes_{i=1}^n \oc A_i
		  \xrightarrow{d_{A_1, \dots, A_n}}
		  \bigotimes_{i=1}^n \oc A_i \otimes \bigotimes_{i=1}^n \oc A_i
		  \xrightarrow{\delta_{A_1, \dots, A_n} \otimes id}
		  \oc (\bigotimes_{i=1}^n \oc A_i) \otimes 
		  \bigotimes_{i=1}^n \oc A_i\\
		  &\xrightarrow{\oc (\rho^{-1}) \otimes id}
		  \oc (\bigotimes_{i=1}^n \oc A_i \otimes I) \otimes 
		  \bigotimes_{i=1}^n \oc A_i
		  \xrightarrow{\oc g \otimes id}
		  \oc D \otimes
		  \bigotimes_{i=1}^n \oc A_i
		  \xrightarrow{\sigma}
		  \bigotimes_{i=1}^n \oc A_i \otimes \oc D
	  \end{align*}
	  Then the following diagram commutes:
  \[\begin{tikzcd}[ampersand replacement=\&,cramped]
	  {\bigotimes_{i=1}^n \oc A_i} \&\& {\bigotimes_{i=1}^n \oc A_i \otimes \bigotimes_{i=1}^n \oc A_i} \\
	  \\
	  {\bigotimes_{i=1}^n \oc A_i \otimes \oc B} \&\& {\bigotimes_{i=1}^n \oc A_i \otimes \oc B \otimes \bigotimes_{i=1}^n \oc A_i \otimes \oc B}
	  \arrow["{d_{A_1, \dots, A_n}}", from=1-1, to=1-3]
	  \arrow["f"', from=1-1, to=3-1]
	  \arrow["{f \otimes f}", from=1-3, to=3-3]
	  \arrow["{d_{A_1, \dots, A_n, B}}"', from=3-1, to=3-3]
  \end{tikzcd}\]
  \end{lemma}
\begin{proof}
	Replacing $f$ by its definition, we want to show that the following diagram commutes:
\[\begin{tikzcd}[ampersand replacement=\&,cramped,column sep=small]
	{\bigotimes_{i=1}^n \oc A_i} \&\&\&\& {(\bigotimes_{i=1}^n \oc A_i)^{\otimes 2}} \\
	\&\& {(1)} \\
	{(\bigotimes_{i=1}^n \oc A_i)^{\otimes 2}} \& {(\bigotimes_{i=1}^n \oc A_i)^{\otimes 3}} \&\& {(\bigotimes_{i=1}^n \oc A_i)^{\otimes 4}} \& {(\bigotimes_{i=1}^n \oc A_i)^{\otimes 4}} \\
	\&\& {(2)} \\
	{\oc (\bigotimes_{i=1}^n \oc A_i) \otimes \bigotimes_{i=1}^n \oc A_i} \&\&\&\& {(\oc (\bigotimes_{i=1}^n \oc A_i) \otimes \bigotimes_{i=1}^n \oc A_i)^{\otimes 2}} \\
	\\
	{\oc (\bigotimes_{i=1}^n \oc A_i \otimes I) \otimes \bigotimes_{i=1}^n \oc A_i} \&\& {(3)} \&\& {(\oc (\bigotimes_{i=1}^n \oc A_i \otimes I) \otimes \bigotimes_{i=1}^n \oc A_i)^{\otimes 2}} \\
	\\
	{\oc B \otimes \bigotimes_{i=1}^n \oc A_i} \&\&\&\& {(\oc B \otimes \bigotimes_{i=1}^n \oc A_i)^{\otimes 2}} \\
	\&\& {(4)} \\
	{\bigotimes_{i=1}^n \oc A_i \otimes \oc B} \&\&\&\& {(\bigotimes_{i=1}^n \oc A_i \otimes \oc B)^{\otimes 2}}
	\arrow["{d_{A_1, \dots, A_n}}", from=1-1, to=1-5]
	\arrow["{d_{A_1, \dots, A_n}}"{description}, from=1-1, to=3-1]
	\arrow["{d_{A_1, \dots, A_n} \otimes d_{A_1, \dots, A_n}}"{description}, from=1-5, to=3-5]
	\arrow["{d_{A_1, \dots, A_n} \otimes id}"', dashed, from=3-1, to=3-2]
	\arrow["{\delta_{A_1, \dots, A_n} \otimes id}"{description}, from=3-1, to=5-1]
	\arrow["{id \otimes d_{A_1, \dots, A_n}}"', dashed, from=3-2, to=3-4]
	\arrow["{id \otimes \sigma \otimes id}"', dashed, from=3-4, to=3-5]
	\arrow["{\delta_{A_1, \dots, A_n} \otimes id \otimes \delta_{A_1, \dots, A_n} \otimes id}"{description}, from=3-5, to=5-5]
	\arrow["{d_{\bigotimes_{i=1}^n \oc A_i, A_1, \dots, A_n}}"{description}, dashed, from=5-1, to=5-5]
	\arrow["{\oc(\rho^{-1}) \otimes id}"{description}, from=5-1, to=7-1]
	\arrow["{\oc (\rho^{-1}) \otimes id \otimes \oc (\rho^{-1}) \otimes id}"{description}, from=5-5, to=7-5]
	\arrow["{\oc g \otimes id}"{description}, from=7-1, to=9-1]
	\arrow["{\oc g \otimes id \otimes \oc g \otimes id}"{description}, from=7-5, to=9-5]
	\arrow["{d_{B,A_1, \dots, A_n}}"{description}, dashed, from=9-1, to=9-5]
	\arrow["\sigma"{description}, from=9-1, to=11-1]
	\arrow["{\sigma \otimes \sigma}"{description}, from=9-5, to=11-5]
	\arrow["{d_{A_1, \dots, A_n, B}}"', from=11-1, to=11-5]
\end{tikzcd}\]
\begin{itemize}
	\item Commutation of (1): the diagram commutes because by Lemma~\ref{lem:comonoidgen} $(\bigotimes_{i=1}^n \oc A_i, d_{A_1,\dots,A_n}, e_{A_1,\dots,A_n})$ is a commutative comonoid.

	\item Commutation of (2):
\[\resizebox{.85\textwidth}{!}{\begin{tikzcd}[ampersand replacement=\&,cramped,column sep=small]
	{(\bigotimes_{i=1}^n \oc A_i)^{\otimes 2}} \&\& {(\bigotimes_{i=1}^n \oc A_i)^{\otimes 3}} \&\& {(\bigotimes_{i=1}^n \oc A_i)^{\otimes 4}} \& {(\bigotimes_{i=1}^n \oc A_i)^{\otimes 4}} \\
	\\
	\\
	\&\& {(\oc (\bigotimes_{i=1}^n \oc A_i))^{\otimes 2} \otimes \bigotimes_{i=1}^n \oc A_i} \&\& {(\oc (\bigotimes_{i=1}^n \oc A_i))^{\otimes 2} \otimes (\bigotimes_{i=1}^n \oc A_i)^{\otimes 2}} \\
	\\
	\\
	{\oc (\bigotimes_{i=1}^n \oc A_i) \otimes \bigotimes_{i=1}^n \oc A_i} \&\&\&\&\& {(\oc (\bigotimes_{i=1}^n \oc A_i) \otimes \bigotimes_{i=1}^n \oc A_i)^{\otimes 2}}
	\arrow["{d_{A_1, \dots, A_n} \otimes id}", from=1-1, to=1-3]
	\arrow["{(5)}"{description}, draw=none, from=1-1, to=4-3]
	\arrow["{\delta_{A_1, \dots, A_n} \otimes id}"{description}, from=1-1, to=7-1]
	\arrow["{id \otimes d_{A_1, \dots, A_n}}", from=1-3, to=1-5]
	\arrow["{\delta_{A_1, \dots, A_n} \otimes \delta_{A_1, \dots, A_n} \otimes id \otimes id}"{description}, dashed, from=1-3, to=4-3]
	\arrow["{id \otimes \sigma \otimes id}", from=1-5, to=1-6]
	\arrow["{\delta_{A_1, \dots, A_n} \otimes \delta_{A_1, \dots, A_n} \otimes id \otimes id}"{description}, dashed, from=1-5, to=4-5]
	\arrow["{(7)}"{description}, draw=none, from=1-6, to=4-5]
	\arrow["{\delta_{A_1, \dots, A_n} \otimes id \otimes \delta_{A_1, \dots, A_n} \otimes id}"{description}, from=1-6, to=7-6]
	\arrow["{(6)}"{description}, draw=none, from=4-3, to=1-5]
	\arrow[""{name=0, anchor=center, inner sep=0}, "{id \otimes d_{A_1,\dots,A_n}}"', curve={height=12pt}, dashed, from=4-3, to=4-5]
	\arrow["{id \otimes \sigma \otimes id}"{description}, dashed, from=4-5, to=7-6]
	\arrow["{d_{\bigotimes_{i=1}^n \oc A_i} \otimes id}"{description}, dashed, from=7-1, to=4-3]
	\arrow[""{name=1, anchor=center, inner sep=0}, "{d_{\bigotimes_{i=1}^n \oc A_i, A_1, \dots, A_n}}"', from=7-1, to=7-6]
	\arrow["{(8)}"{description}, draw=none, from=0, to=1]
\end{tikzcd}}\]
(5) commutes because by Corollary~\ref{cor:deltacoalgebramorph} $\delta_{A_1,\dots,A_n}$ is a coalgebra morphism, therefore by Lemma~\ref{lem:biermanprop} it is also a comonoid morphism between $(\bigotimes_{i=1}^n \oc A_i, d_{A_1,\dots,A_n},e_{A_1,\dots,A_n})$ and $(\oc(\bigotimes_{i=1}^n \oc A_i),d_{\bigotimes_{i=1}^n \oc A_i}, e_{\bigotimes_{i=1}^n \oc A_i})$, (6) commutes by functoriality of $\otimes$, (7) commutes by naturality of $\sigma$ and (8) commutes by definition of $d\_$.

	\item Commutation of (3): the diagram commutes by naturality of $d_{A_1,\dots,A_n}$ (Lemma~\ref{lem:dnat}).
	\item Commutation of (4):
\[\begin{tikzcd}[ampersand replacement=\&,cramped,column sep=small]
	{\oc B \otimes \bigotimes_{i=1}^n \oc A_i} \&\&\&\& \begin{array}{c} \oc B \otimes \bigotimes_{i=1}^n \oc A_i \otimes \\ \oc B \otimes \bigotimes_{i=1}^n \oc A_i \end{array} \\
	\\
	\&\& {\oc B \otimes \oc B \otimes \bigotimes_{i=1}^n \oc A_i \otimes \bigotimes_{i=1}^n \oc A_i} \\
	\\
	\&\& {\bigotimes_{i=1}^n \oc A_i \otimes \bigotimes_{i=1}^n \oc A_i \otimes \oc B \otimes \oc B} \\
	\\
	{\bigotimes_{i=1}^n \oc A_i \otimes \oc B} \&\&\&\& \begin{array}{c} \bigotimes_{i=1}^n \oc A_i \otimes \oc B \otimes \\ \bigotimes_{i=1}^n \oc A_i \otimes \oc B \end{array}
	\arrow[""{name=0, anchor=center, inner sep=0}, "{d_{B,A_1, \dots, A_n}}", from=1-1, to=1-5]
	\arrow["{d_B \otimes d_{A_1, \dots, A_n}}"{description}, dashed, from=1-1, to=3-3]
	\arrow[""{name=1, anchor=center, inner sep=0}, "\sigma"', from=1-1, to=7-1]
	\arrow[""{name=2, anchor=center, inner sep=0}, "{\sigma \otimes \sigma}", from=1-5, to=7-5]
	\arrow["{id \otimes \sigma \otimes id}"{description}, dashed, from=3-3, to=1-5]
	\arrow[""{name=3, anchor=center, inner sep=0}, "\sigma"{description}, dashed, from=3-3, to=5-3]
	\arrow["{id \otimes \sigma \otimes id}"{description}, dashed, from=5-3, to=7-5]
	\arrow["{d_{A_1, \dots, A_n} \otimes d_B}"{description}, dashed, from=7-1, to=5-3]
	\arrow[""{name=4, anchor=center, inner sep=0}, "{d_{A_1, \dots, A_n, B}}"', from=7-1, to=7-5]
	\arrow["{(10)}"{description}, draw=none, from=1, to=3]
	\arrow["{(9)}"{description}, draw=none, from=3-3, to=0]
	\arrow["{(12)}"{description}, draw=none, from=3, to=2]
	\arrow["{(11)}"{description}, draw=none, from=5-3, to=4]
\end{tikzcd}\]
(9) and (11) commute by definition of $d\_$, (10) commutes by naturality of $\sigma$, and (12) commutes by coherence.\qedhere
\end{itemize}
\end{proof}

\begin{corollary}
	\label{coro:substitutionaux}
	  Let $g$ and $f$ as in Lemma~\ref{lem:substitutionaux}, then the following diagram commutes:
  \[\begin{tikzcd}[ampersand replacement=\&,cramped]
	  {\bigotimes_{i=1}^n \oc A_i \otimes B \otimes C} \&\&\& \begin{array}{c} \bigotimes_{i=1}^n \oc A_i \otimes \\\bigotimes_{i=1}^n \oc A_i \otimes B \otimes C \end{array} \&\&\& \begin{array}{c} \bigotimes_{i=1}^n \oc A_i \otimes B \otimes \\ \bigotimes_{i=1}^n \oc A_i \otimes C \end{array} \\
	  \\
	  \begin{array}{c} \bigotimes_{i=1}^n \oc A_i \otimes \oc D\\ \otimes B \otimes C \end{array} \&\&\& \begin{array}{c} \bigotimes_{i=1}^n \oc A_i \otimes \oc D \otimes \\\bigotimes_{i=1}^n \oc A_i \otimes \oc D \otimes B \otimes C \end{array} \&\&\& \begin{array}{c} \bigotimes_{i=1}^n \oc A_i \otimes \oc D \otimes B \\\otimes \bigotimes_{i=1}^n \oc A_i \otimes \oc D \otimes C \end{array}
	  \arrow["{d_{A_1,\dots,A_n} \otimes id_B \otimes id_C}", from=1-1, to=1-4]
	  \arrow["{f \otimes id_B \otimes id_C}"{description}, from=1-1, to=3-1]
	  \arrow["{id_{\bigotimes_{i=1}^n \oc A_i} \otimes \sigma \otimes id_C}", from=1-4, to=1-7]
	  \arrow["{f \otimes id_B \otimes f \otimes id_C}"{description}, from=1-7, to=3-7]
	  \arrow["{d_{A_1,\dots,A_n,D} \otimes id_B \otimes id_C}", from=3-1, to=3-4]
	  \arrow["{id_{\bigotimes_{i=1}^n \oc A_i \otimes \oc D} \otimes \sigma \otimes id_C}", from=3-4, to=3-7]
  \end{tikzcd}\]
  \end{corollary}
\begin{proof}
\[\begin{tikzcd}[ampersand replacement=\&,cramped]
	{\bigotimes_{i=1}^n \oc A_i \otimes B \otimes C} \&\& \begin{array}{c} \bigotimes_{i=1}^n \oc A_i \otimes \\\bigotimes_{i=1}^n \oc A_i \otimes B \otimes C \end{array} \&\& \begin{array}{c} \bigotimes_{i=1}^n \oc A_i \otimes B \otimes \\ \bigotimes_{i=1}^n \oc A_i \otimes C \end{array} \\
	\& {(1)} \&\& {(2)} \\
	\begin{array}{c} \bigotimes_{i=1}^n \oc A_i \otimes \oc D\\ \otimes B \otimes C \end{array} \&\& \begin{array}{c} \bigotimes_{i=1}^n \oc A_i \otimes \oc D \otimes \\\bigotimes_{i=1}^n \oc A_i \otimes \oc D \otimes B \otimes C \end{array} \&\& \begin{array}{c} \bigotimes_{i=1}^n \oc A_i \otimes \oc D \otimes B \\\otimes \bigotimes_{i=1}^n \oc A_i \otimes \oc D \otimes C \end{array}
	\arrow["{d_{A_1,\dots,A_n} \otimes id_B \otimes id_C}", from=1-1, to=1-3]
	\arrow["{f \otimes id_B \otimes id_C}"{description}, from=1-1, to=3-1]
	\arrow["{id_{\bigotimes_{i=1}^n \oc A_i} \otimes \sigma \otimes id_C}", from=1-3, to=1-5]
	\arrow["{f \otimes f \otimes id_B \otimes id_C}"{description}, dashed, from=1-3, to=3-3]
	\arrow["{f \otimes id_B \otimes f \otimes id_C}"{description}, from=1-5, to=3-5]
	\arrow["{d_{A_1,\dots,A_n,D} \otimes id_B \otimes id_C}", from=3-1, to=3-3]
	\arrow["{id_{\bigotimes_{i=1}^n \oc A_i \otimes \oc D} \otimes \sigma \otimes id_C}", from=3-3, to=3-5]
\end{tikzcd}\]
(1) commutes by Lemma~\ref{lem:substitutionaux} and functoriality of $\otimes$, and (2) commutes by naturality of $\sigma$.
\end{proof}

\begin{lemma}
	\label{lem:deltaaltdef}
	  For every $n$ the following diagram commutes:
  \[\begin{tikzcd}[ampersand replacement=\&,cramped]
	  {\oc B \otimes \bigotimes_{i=1}^n \oc A_i} \&\&\&\& {\oc (\oc B \otimes \bigotimes_{i=1}^n \oc A_i)} \\
	  \\
	  \&\& {\oc \oc B \otimes \oc (\bigotimes_{i=1}^n \oc A_i)}
	  \arrow["{\delta_{B,A_1, \dots, A_n}}", from=1-1, to=1-5]
	  \arrow["{\delta_B \otimes \delta_{A_1, \dots, A_n}}"', from=1-1, to=3-3]
	  \arrow["m"', from=3-3, to=1-5]
  \end{tikzcd}\]
  \end{lemma}
\begin{proof}
	By induction on $n$:
	\begin{itemize}
		\item If $n=0$ the diagram commutes by definition of $\delta_{B,\varnothing}$:
\[\begin{tikzcd}[ampersand replacement=\&,cramped]
	{\oc B \otimes I} \&\&\&\& {\oc (\oc B \otimes I)} \\
	\\
	\&\& {\oc \oc B \otimes \oc I}
	\arrow["{\delta_{B,\varnothing}}", from=1-1, to=1-5]
	\arrow["{\delta_B \otimes m_I}"', from=1-1, to=3-3]
	\arrow["m"', from=3-3, to=1-5]
\end{tikzcd}\]
		\item If $n=1$ the diagram commutes by definition of $\delta_{B,A}$:
\[\begin{tikzcd}[ampersand replacement=\&,cramped]
	{\oc B \otimes \oc A} \&\&\&\& {\oc (\oc B \otimes \oc A)} \\
	\\
	\&\& {\oc \oc B \otimes \oc \oc A}
	\arrow["{\delta_{B,A}}", from=1-1, to=1-5]
	\arrow["{\delta_B \otimes \delta_A}"', from=1-1, to=3-3]
	\arrow["m"', from=3-3, to=1-5]
\end{tikzcd}\]
		\item If $n>1$:
\[\begin{tikzcd}[ampersand replacement=\&,cramped,column sep=small]
	{\oc B \otimes \bigotimes_{i=1}^n \oc A_i} \&\&\&\& {\oc (\oc B \otimes \bigotimes_{i=1}^n \oc A_i)} \\
	\&\& {(1)} \\
	\& {(2)} \& {\oc (\oc B \otimes \bigotimes_{i=1}^{n-1} \oc A_i) \otimes \oc \oc A_n} \\
	\\
	\& {(3)} \&\& {(4)} \\
	\&\& {\oc \oc B \otimes \oc (\bigotimes_{i=1}^{n-1} \oc A_i) \otimes \oc \oc A_n} \\
	{\oc \oc B \otimes \oc (\bigotimes_{i=1}^n \oc A_i)}
	\arrow["{\delta_{B,A_1, \dots, A_n}}", from=1-1, to=1-5]
	\arrow["{\delta_{B,A_1, \dots, A_{n-1}} \otimes \delta_{A_n}}"{description}, dashed, from=1-1, to=3-3]
	\arrow["{\delta_B \otimes \delta_{A_1, \dots, A_{n-1}} \otimes \delta_{A_n}}"{description}, curve={height=30pt}, dashed, from=1-1, to=6-3]
	\arrow["{\delta_B \otimes \delta_{A_1, \dots, A_n}}"', from=1-1, to=7-1]
	\arrow["m"{description}, dashed, from=3-3, to=1-5]
	\arrow["{m \otimes id}"{description}, dashed, from=6-3, to=3-3]
	\arrow["{id \otimes m}"{description}, dashed, from=6-3, to=7-1]
	\arrow["m"', from=7-1, to=1-5, rounded corners, to path={(\tikztostart.east) -| ([yshift=-130pt]\tikztotarget.south) -- (\tikztotarget.south) \tikztonodes}]
\end{tikzcd}\]
(1) commutes by definition of $\delta\_$, (2) commutes by the induction hypothesis and functoriality of $\otimes$, (3) commutes by definition of $\delta\_$ and functoriality of $\otimes$ and (4) commutes because $\oc$ is a monoidal functor.\qedhere
	\end{itemize}
\end{proof}

\begin{lemma}
	\label{lem:deltaprop}
	  Let $f:\bigotimes_{i=1}^n \oc A_i \otimes I \to B$, then the following diagram commutes:
  \[\begin{tikzcd}[ampersand replacement=\&,cramped,row sep=scriptsize]
	  {\bigotimes_{i=1}^n \oc A_i} \&\& {\oc (\bigotimes_{i=1}^n \oc A_i)} \\
	  \\
	  {\bigotimes_{i=1}^n \oc A_i \otimes \bigotimes_{i=1}^n \oc A_i} \&\& {\oc (\bigotimes_{i=1}^n \oc A_i \otimes \bigotimes_{i=1}^n \oc A_i)} \\
	  \\
	  {\oc (\bigotimes_{i=1}^n \oc A_i) \otimes \bigotimes_{i=1}^n \oc A_i} \&\& {\oc (\oc (\bigotimes_{i=1}^n \oc A_i) \otimes \bigotimes_{i=1}^n \oc A_i)} \\
	  \\
	  {\oc (\bigotimes_{i=1}^n \oc A_i \otimes I) \otimes \bigotimes_{i=1}^n \oc A_i} \&\& {\oc (\oc (\bigotimes_{i=1}^n \oc A_i \otimes I) \otimes \bigotimes_{i=1}^n \oc A_i)} \\
	  \\
	  {\oc B \otimes \bigotimes_{i=1}^n \oc A_i} \&\& {\oc (\oc B \otimes \bigotimes_{i=1}^n \oc A_i)} \\
	  \\
	  {\bigotimes_{i=1}^n \oc A_i \otimes \oc B} \&\& {\oc (\bigotimes_{i=1}^n \oc A_i \otimes \oc B)}
	  \arrow["{\delta_{A_1, \dots, A_n}}", from=1-1, to=1-3]
	  \arrow["{d_{A_1, \dots, A_n}}"', from=1-1, to=3-1]
	  \arrow["{\oc (d_{A_1, \dots, A_n})}", from=1-3, to=3-3]
	  \arrow["{\delta_{A_1, \dots, A_n} \otimes id}"', from=3-1, to=5-1]
	  \arrow["{\oc (\delta_{A_1, \dots, A_n} \otimes id)}", from=3-3, to=5-3]
	  \arrow["{\oc (\rho^{-1}) \otimes id}"', from=5-1, to=7-1]
	  \arrow["{\oc (\oc (\rho^{-1}) \otimes id)}", from=5-3, to=7-3]
	  \arrow["{\oc f \otimes id}"', from=7-1, to=9-1]
	  \arrow["{\oc (\oc f \otimes id)}", from=7-3, to=9-3]
	  \arrow["\sigma"', from=9-1, to=11-1]
	  \arrow["{\oc \sigma}", from=9-3, to=11-3]
	  \arrow["{\delta_{A_1, \dots, A_n, B}}"', from=11-1, to=11-3]
  \end{tikzcd}\]
  \end{lemma}
\begin{proof}
\[\begin{tikzcd}[ampersand replacement=\&,cramped,row sep=scriptsize]
	{\bigotimes_{i=1}^n \oc A_i} \&\&\&\& {\oc (\bigotimes_{i=1}^n \oc A_i)} \\
	\&\& {(1)} \\
	{\bigotimes_{i=1}^n \oc A_i \otimes \bigotimes_{i=1}^n \oc A_i} \&\& {\oc (\bigotimes_{i=1}^n \oc A_i) \otimes \oc (\bigotimes_{i=1}^n \oc A_i)} \&\& {\oc (\bigotimes_{i=1}^n \oc A_i \otimes \bigotimes_{i=1}^n \oc A_i)} \\
	\&\& {(2)} \\
	{\oc (\bigotimes_{i=1}^n \oc A_i) \otimes \bigotimes_{i=1}^n \oc A_i} \&\&\&\& {\oc (\oc (\bigotimes_{i=1}^n \oc A_i) \otimes \bigotimes_{i=1}^n \oc A_i)} \\
	\\
	{\oc (\bigotimes_{i=1}^n \oc A_i \otimes I) \otimes \bigotimes_{i=1}^n \oc A_i} \&\& {(3)} \&\& {\oc (\oc (\bigotimes_{i=1}^n \oc A_i \otimes I) \otimes \bigotimes_{i=1}^n \oc A_i)} \\
	\\
	{\oc B \otimes \bigotimes_{i=1}^n \oc A_i} \&\&\&\& {\oc (\oc B \otimes \bigotimes_{i=1}^n \oc A_i)} \\
	\&\& {(4)} \\
	{\bigotimes_{i=1}^n \oc A_i \otimes \oc B} \&\&\&\& {\oc (\bigotimes_{i=1}^n \oc A_i \otimes \oc B)}
	\arrow["{\delta_{A_1, \dots, A_n}}", from=1-1, to=1-5]
	\arrow["{d_{A_1, \dots, A_n}}"', from=1-1, to=3-1]
	\arrow["{\oc (d_{A_1, \dots, A_n})}", from=1-5, to=3-5]
	\arrow["{\delta_{A_1, \dots, A_n} \otimes \delta_{A_1, \dots, A_n}}"', curve={height=12pt}, dashed, from=3-1, to=3-3]
	\arrow["{\delta_{A_1, \dots, A_n} \otimes id}"', from=3-1, to=5-1]
	\arrow["m"', dashed, from=3-3, to=3-5]
	\arrow["{\oc (\delta_{A_1, \dots, A_n} \otimes id)}", from=3-5, to=5-5]
	\arrow["{\delta_{(\bigotimes_{i=1}^n \oc A_i), A_1, \dots, A_n}}"{description}, dashed, from=5-1, to=5-5]
	\arrow["{\oc (\rho^{-1}) \otimes id}"', from=5-1, to=7-1]
	\arrow["{\oc (\oc (\rho^{-1}) \otimes id)}", from=5-5, to=7-5]
	\arrow["{\oc f \otimes id}"', from=7-1, to=9-1]
	\arrow["{\oc (\oc f \otimes id)}", from=7-5, to=9-5]
	\arrow["{\delta_{B,A_1, \dots, A_n}}"{description}, dashed, from=9-1, to=9-5]
	\arrow["\sigma"', from=9-1, to=11-1]
	\arrow["{\oc \sigma}", from=9-5, to=11-5]
	\arrow["{\delta_{A_1, \dots, A_n, B}}"', from=11-1, to=11-5]
\end{tikzcd}\]
	\begin{itemize}
		\item Commutation of (1): the diagram commutes because $d_{A_1,\dots,A_n}$ is a coalgebra morphism (Lemma~\ref{lem:dcoalgebramorph}).
		\item Commutation of (2):
\[\resizebox{.85\textwidth}{!}{\begin{tikzcd}[ampersand replacement=\&,cramped,column sep=tiny]
	{\bigotimes_{i=1}^n \oc A_i \otimes \bigotimes_{i=1}^n \oc A_i} \&\&\&\& {\oc (\bigotimes_{i=1}^n \oc A_i) \otimes \oc (\bigotimes_{i=1}^n \oc A_i)} \& {\oc (\bigotimes_{i=1}^n \oc A_i \otimes \bigotimes_{i=1}^n \oc A_i)} \\
	\\
	\&\& {\oc (\bigotimes_{i=1}^n \oc A_i) \otimes \bigotimes_{i=1}^n \oc A_i} \\
	\\
	\\
	\&\& {\oc \oc (\bigotimes_{i=1}^n \oc A_i) \otimes \bigotimes_{i=1}^n \oc A_i} \&\& \begin{array}{c} \oc \oc (\bigotimes_{i=1}^n \oc A_i)\\ \otimes \oc (\bigotimes_{i=1}^n \oc A_i) \end{array} \\
	\\
	{\oc (\bigotimes_{i=1}^n \oc A_i) \otimes \bigotimes_{i=1}^n \oc A_i} \&\&\&\&\& {\oc (\oc (\bigotimes_{i=1}^n \oc A_i) \otimes \bigotimes_{i=1}^n \oc A_i)}
	\arrow["{\delta_{A_1, \dots, A_n} \otimes \delta_{A_1, \dots, A_n}}", from=1-1, to=1-5]
	\arrow["{\delta_{A_1, \dots, A_n} \otimes id}"{description}, curve={height=12pt}, dashed, from=1-1, to=3-3]
	\arrow["{(5)}"{description}, curve={height=-12pt}, draw=none, from=1-1, to=3-3]
	\arrow["{\delta_{A_1, \dots, A_n} \otimes id}"{description}, from=1-1, to=8-1]
	\arrow["m", from=1-5, to=1-6]
	\arrow["{\oc (\delta_{A_1, \dots, A_n}) \otimes id}"{description}, dashed, from=1-5, to=6-5]
	\arrow["{\oc (\delta_{A_1, \dots, A_n} \otimes id)}"{description}, from=1-6, to=8-6]
	\arrow["{id \otimes \delta_{A_1, \dots, A_n}}"{description}, dashed, from=3-3, to=1-5]
	\arrow["{\oc (\delta_{A_1, \dots, A_n}) \otimes id}"{description}, dashed, from=3-3, to=6-3]
	\arrow["{(7)}"{description, pos=0.4}, draw=none, from=3-3, to=6-5]
	\arrow["{(6)}"{description, pos=0.4}, draw=none, from=3-3, to=8-1]
	\arrow["{id \otimes \delta_{A_1, \dots, A_n}}", curve={height=-18pt}, dashed, from=6-3, to=6-5]
	\arrow["{(9)}"{description, pos=0.3}, draw=none, from=6-3, to=8-6]
	\arrow["{(8)}"{description}, draw=none, from=6-5, to=1-6]
	\arrow["m"{description}, dashed, from=6-5, to=8-6]
	\arrow["{\delta_{\bigotimes_{i=1}^n \oc A_i} \otimes id}"{description}, dashed, from=8-1, to=6-3]
	\arrow["{\delta_{(\bigotimes_{i=1}^n \oc A_i), A_1, \dots, A_n}}"', from=8-1, to=8-6]
\end{tikzcd}}\]
(5) commutes by functoriality of $\otimes$, (6) commutes because $\delta_{A_1, \dots, A_n}$ is a coalgebra morphism (Lemma~\ref{cor:deltacoalgebramorph}) and by functoriality of $\otimes$, (7) commutes by functoriality of $\otimes$, (8) commutes by naturality of $m$, and (9) commutes by Lemma~\ref{lem:deltaaltdef}.

		\item Commutation of (3):
\[\begin{tikzcd}[ampersand replacement=\&,cramped,row sep=small, column sep=tiny]
	{\oc (\bigotimes_{i=1}^n \oc A_i) \otimes \bigotimes_{i=1}^n \oc A_i} \&\&\&\& {\oc (\oc (\bigotimes_{i=1}^n \oc A_i) \otimes \bigotimes_{i=1}^n \oc A_i)} \\
	\&\& {(10)} \\
	\&\& {\oc \oc (\bigotimes_{i=1}^n \oc A_i) \otimes \oc (\bigotimes_{i=1}^n \oc A_i)} \\
	\& {(11)} \&\& {(12)} \\
	{\oc (\bigotimes_{i=1}^n \oc A_i \otimes I) \otimes \bigotimes_{i=1}^n \oc A_i} \&\&\&\& {\oc (\oc (\bigotimes_{i=1}^n \oc A_i \otimes I) \otimes \bigotimes_{i=1}^n \oc A_i)} \\
	\\
	\&\& {\oc \oc B \otimes \oc (\bigotimes_{i=1}^n \oc A_i)} \\
	\&\& {(13)} \\
	{\oc B \otimes \bigotimes_{i=1}^n \oc A_i} \&\&\&\& {\oc (\oc B \otimes \bigotimes_{i=1}^n \oc A_i)}
	\arrow["{\delta_{(\bigotimes_{i=1}^n \oc A_i), A_1, \dots, A_n}}", from=1-1, to=1-5]
	\arrow["{\delta_{(\bigotimes_{i=1}^n \oc A_i)} \otimes \delta_{A_1, \dots, A_n}}"{description}, dashed, from=1-1, to=3-3]
	\arrow["{\oc (\rho^{-1}) \otimes id}"', from=1-1, to=5-1]
	\arrow["{\oc (f \circ \rho^{-1}) \otimes id}"{description, pos=0.7}, curve={height=-90pt}, dashed, from=1-1, to=9-1]
	\arrow["{\oc (\oc (\rho^{-1}) \otimes id)}", from=1-5, to=5-5]
	\arrow["{\oc (\oc (f \circ \rho^{-1}) \otimes id)}"{description, pos=0.7}, curve={height=90pt}, dashed, from=1-5, to=9-5]
	\arrow["m"{description}, dashed, from=3-3, to=1-5]
	\arrow["{\oc \oc (f \circ \rho^{-1}) \otimes \oc id}"{description}, dashed, from=3-3, to=7-3]
	\arrow["{\oc f \otimes id}"', from=5-1, to=9-1]
	\arrow["{\oc (\oc f \otimes id)}", from=5-5, to=9-5]
	\arrow["m"{description}, dashed, from=7-3, to=9-5]
	\arrow["{\delta_B \otimes \delta_{A_1, \dots, A_n}}"{description}, dashed, from=9-1, to=7-3]
	\arrow["{\delta_{B,A_1, \dots, A_n}}"', from=9-1, to=9-5]
\end{tikzcd}\]
(10) and (13) commute by Lemma~\ref{lem:deltaaltdef}, (11) commutes by naturality of $\delta$ and functoriality of $\otimes$ and (12) commutes by naturality of $m$.
	\item Commutation of (4):
\[\begin{tikzcd}[ampersand replacement=\&,cramped,row sep=small]
	{\oc B \otimes \bigotimes_{i=1}^n \oc A_i} \&\&\&\& {\oc (\oc B \otimes \bigotimes_{i=1}^n \oc A_i)} \\
	\&\& {(14)} \\
	\&\& {\oc \oc B \otimes \oc (\bigotimes_{i=1}^n \oc A_i)} \\
	\& {(15)} \&\& {(16)} \\
	\&\& {\oc (\bigotimes_{i=1}^n \oc A_i) \otimes \oc \oc B} \\
	\&\& {(17)} \\
	{\bigotimes_{i=1}^n \oc A_i \otimes \oc B} \&\&\&\& {\oc (\bigotimes_{i=1}^n \oc A_i \otimes \oc B)}
	\arrow["{\delta_{B,A_1, \dots, A_n}}", from=1-1, to=1-5]
	\arrow["{\delta_B \otimes \delta_{A_1, \dots, A_n}}"{description}, dashed, from=1-1, to=3-3]
	\arrow["\sigma"', from=1-1, to=7-1]
	\arrow["{\oc \sigma}", from=1-5, to=7-5]
	\arrow["m"{description}, dashed, from=3-3, to=1-5]
	\arrow["\sigma"{description}, dashed, from=3-3, to=5-3]
	\arrow["m"{description}, dashed, from=5-3, to=7-5]
	\arrow["{\delta_{A_1, \dots, A_n} \otimes \delta_B}"{description}, dashed, from=7-1, to=5-3]
	\arrow["{\delta_{A_1, \dots, A_n, B}}"', from=7-1, to=7-5]
\end{tikzcd}\]
(14) commutes by Lemma~\ref{lem:deltaaltdef}, (15) commutes by naturality of $\sigma$, (16) commutes because $\oc$ is a symmetric monoidal functor and (17) commutes by definition of $\delta\_$.\qedhere
\end{itemize}
\end{proof}

\begin{lemma}
	\label{lem:eprop}
	  Let $f: (\bigotimes_{i=1}^n \oc A_i) \otimes I \to B$, then the following diagram commutes:
  \[\begin{tikzcd}[ampersand replacement=\&,cramped,row sep=scriptsize]
	  {\bigotimes_{i=1}^n \oc A_i} \&\& I \\
	  \\
	  {\bigotimes_{i=1}^n \oc A_i \otimes \bigotimes_{i=1}^n \oc A_i} \\
	  \&\& {\bigotimes_{i=1}^n \oc A_i \otimes \oc B} \\
	  {\bigotimes_{i=1}^n \oc A_i \otimes I\otimes \bigotimes_{i=1}^n \oc A_i} \\
	  \&\& {\oc B \otimes \bigotimes_{i=1}^n \oc A_i} \\
	  {\bigotimes_{i=1}^n \oc A_i \otimes \bigotimes_{i=1}^n \oc A_i} \\
	  \\
	  {\oc (\bigotimes_{i=1}^n \oc A_i) \otimes \bigotimes_{i=1}^n \oc A_i} \&\& {\oc (\bigotimes_{i=1}^n \oc A_i \otimes I) \otimes \bigotimes_{i=1}^n \oc A_i}
	  \arrow["{e_{A_1, \dots, A_n}}", from=1-1, to=1-3]
	  \arrow["{d_{A_1, \dots, A_n}}"', from=1-1, to=3-1]
	  \arrow["{\rho^{-1} \otimes id}"', from=3-1, to=5-1]
	  \arrow["{e_{A_1, \dots, A_n, B}}"', from=4-3, to=1-3]
	  \arrow["{\rho \otimes id}"', from=5-1, to=7-1]
	  \arrow["\sigma"', from=6-3, to=4-3]
	  \arrow["{\delta_{A_1, \dots, A_n} \otimes id}"', from=7-1, to=9-1]
	  \arrow["{\oc (\rho^{-1}) \otimes id}"', from=9-1, to=9-3]
	  \arrow["{\oc f \otimes id}"', from=9-3, to=6-3]
  \end{tikzcd}\]
  \end{lemma}
\begin{proof}
  \[\begin{tikzcd}[ampersand replacement=\&,cramped,column sep=tiny,row sep=scriptsize]
    {\bigotimes_{i=1}^n \oc A_i} \&\&\&\&\&\& I \\
    \&\& {(1)} \\
    {\bigotimes_{i=1}^n \oc A_i \otimes \bigotimes_{i=1}^n \oc A_i} \&\& {I \otimes I} \&\& {(3)} \\
    \&\&\&\& {I \otimes I} \& {(4)} \\
    \& {(2)} \& {I\otimes \bigotimes_{i=1}^n \oc A_i} \&\& {(5)} \&\& {\bigotimes_{i=1}^n \oc A_i \otimes \oc B} \\
    \&\&\&\& { I \otimes \bigotimes_{i=1}^n \oc A_i} \\
    {\bigotimes_{i=1}^n \oc A_i \otimes I \otimes \bigotimes_{i=1}^n \oc A_i} \&\&\& {(6)} \&\&\& {\oc B \otimes \bigotimes_{i=1}^n \oc A_i} \\
    \&\&\&\&\& {(9)} \\
    \&\& {\oc I \otimes \bigotimes_{i=1}^n \oc A_i} \&\& {(8)} \&\& {\oc (\bigotimes_{i=1}^n \oc A_i \otimes I) \otimes \bigotimes_{i=1}^n \oc A_i} \\
    \& {(7)} \\
    {\bigotimes_{i=1}^n \oc A_i \otimes \bigotimes_{i=1}^n \oc A_i} \&\&\&\&\&\& {\oc (\bigotimes_{i=1}^n \oc A_i) \otimes \bigotimes_{i=1}^n \oc A_i}
    \arrow["{e_{A_1, \dots, A_n}}", from=1-1, to=1-7]
    \arrow["{d_{A_1, \dots, A_n}}"', from=1-1, to=3-1]
    \arrow["\lambda"{description}, curve={height=-6pt}, dashed, from=1-1, to=5-3]
    \arrow["{\rho^{-1} \otimes id}"', from=3-1, to=7-1]
    \arrow[curve={height=-80pt}, equals, from=3-1, to=11-1]
    \arrow["\lambda"{description}, dashed, from=3-3, to=1-7]
    \arrow["\sigma"{description}, dashed, from=3-3, to=4-5]
    \arrow["\lambda"{description}, dashed, from=4-5, to=1-7]
    \arrow["{\lambda^{-1}}"{description}, curve={height=-12pt}, dashed, from=5-3, to=1-1]
    \arrow["{id_I\otimes e_{A_1,\dots,A_n}}"{description}, dashed, from=5-3, to=3-3]
    \arrow[equals, from=5-3, to=6-5]
    \arrow["{m_I \otimes id}"{description}, dashed, from=5-3, to=9-3]
    \arrow["{e_{A_1, \dots, A_n, B}}"', from=5-7, to=1-7]
    \arrow["{e_{A_1,\dots,A_n} \otimes e_B}"{description}, curve={height=-18pt}, dashed, from=5-7, to=4-5]
    \arrow["{\rho \otimes id}"', from=7-1, to=11-1]
    \arrow["\sigma"', from=7-7, to=5-7]
    \arrow["{e_B \otimes id}"{description}, dashed, from=7-7, to=6-5]
    \arrow["{e_I \otimes id}"{description}, curve={height=12pt}, dashed, from=9-3, to=6-5]
    \arrow["{\oc f \otimes id}"', from=9-7, to=7-7]
    \arrow["{e_{A_1,\dots,A_n} \otimes id}"{description}, curve={height=30pt}, dashed, from=11-1, to=5-3]
    \arrow["{\delta_{A_1, \dots, A_n} \otimes id}"', from=11-1, to=11-7]
    \arrow["{e_{\bigotimes_{i=1}^n \oc A_i} \otimes id}"{description}, curve={height=-18pt}, dashed, from=11-7, to=6-5]
    \arrow["{\oc (e_{A_1, \dots, A_n}) \otimes id}"{description}, dashed, from=11-7, to=9-3]
    \arrow["{\oc (\rho^{-1}) \otimes id}"', from=11-7, to=9-7]
  \end{tikzcd}\]
  (1) commutes by naturality of $\lambda$, (2) commutes by Lemma~\ref{lem:comonoidgen}, (3) commutes by coherence, (4) commutes by definition of $e\_$, (5) commutes by naturality of $\sigma$ and functoriality of $\otimes$, (6) commutes because $\varepsilon$ is a monoidal natural transformation, (7) commutes by Lemma~\ref{lem:generalizedaxiom} and functoriality of $\otimes$, and (8) and (9) commute by naturality of $e$ and functoriality of $\otimes$.
  \end{proof}

\begin{lemma}[Weakening]\label{lem:weakeningcat}
	Let $t$ be a proof-term such that $\Upsilon;\Gamma \vdash t:A$, and $x$ a variable such that $x \notin \Upsilon$ and $x \notin \Gamma$. Then, for every proposition $B$, $\Upsilon, x^B; \Gamma \vdash t:A$ and the following diagram commutes.
\[\begin{tikzcd}[ampersand replacement=\&,cramped]
	{\interpretation{\oc\Upsilon} \otimes \oc \interpretation B \otimes \interpretation \Gamma} \&\&\&\& {\interpretation A} \\
	\\
	{\interpretation{\oc\Upsilon} \otimes I \otimes \interpretation \Gamma} \&\&\&\& {\interpretation{\oc\Upsilon} \otimes \interpretation \Gamma}
	\arrow["{\interpretation{\Upsilon, x^B;\Gamma \vdash t:A}}", from=1-1, to=1-5]
	\arrow["{id \otimes e_B \otimes id}"', from=1-1, to=3-1]
	\arrow["{\rho \otimes id}"', from=3-1, to=3-5]
	\arrow["{\interpretation{\Upsilon;\Gamma \vdash t:A}}"', from=3-5, to=1-5]
\end{tikzcd}\]
\end{lemma}
\begin{proof}
	Straightforward by induction on $t$.
\end{proof}

\subsection{Soundness}\label{proof:soundness}

\nonlinearsubstitution*
\begin{proof}
	To avoid cumbersome notation, let $\alpha = \sigma\circ(\interpretation{\bang u} \otimes id)\circ(\rho^{-1} \otimes id)\circ d_\Upsilon$.

	By definition, 
	\[
		\interpretation{\bang u} = 
		\interpretation{\bang\Upsilon} \otimes I
		\xrightarrow{\rho}
		\interpretation{\bang\Upsilon}
		\xrightarrow{\delta_\Upsilon}
		\oc \interpretation{\bang\Upsilon}
		\xrightarrow{\oc (\rho^{-1})}
		\oc (\interpretation{\bang\Upsilon} \otimes I)
		\xrightarrow{\oc \interpretation{u}}
		\oc \interpretation{B}
	\]
  By induction on $t$:
  \begin{itemize}
    \item If $t = x$, then $B = A$ and $\Gamma = \varnothing$: the commuting diagram is shown in Figure~\ref{fig:lem:nonlinearsubstitutionx} (Appendix~\ref{app:diagramasgrandes}).
    
	(1) commutes by definition of $\alpha$, (2) commutes by definition of $\interpretation{\oc u}$, (3) commutes by naturality of $\sigma$, (4) commutes by naturality f $\rho^{-1}$, (5) commutes because $\oc$ is a monoidal functor, (6) commutes by naturality of $\sigma$, (7) commutes by Lemma~\ref{lem:coalgebragen}, (8) commutes because $\varepsilon$ is a monoidal natural transformation and by functoriality of $\otimes$, (9) commutes by naturality of $\sigma$, (10) commutes because $\varepsilon$ is a monoidal natural transformation, (11) commutes trivially, (12) commutes by naturality of $\sigma$, (13) commutes by naturality of $\rho$, (14) commutes by naturality of $\varepsilon$, (15) commutes by naturality of $\rho$, and (16) commutes by functoriality of $\otimes$.

	The following diagram shows the commutation of (17):
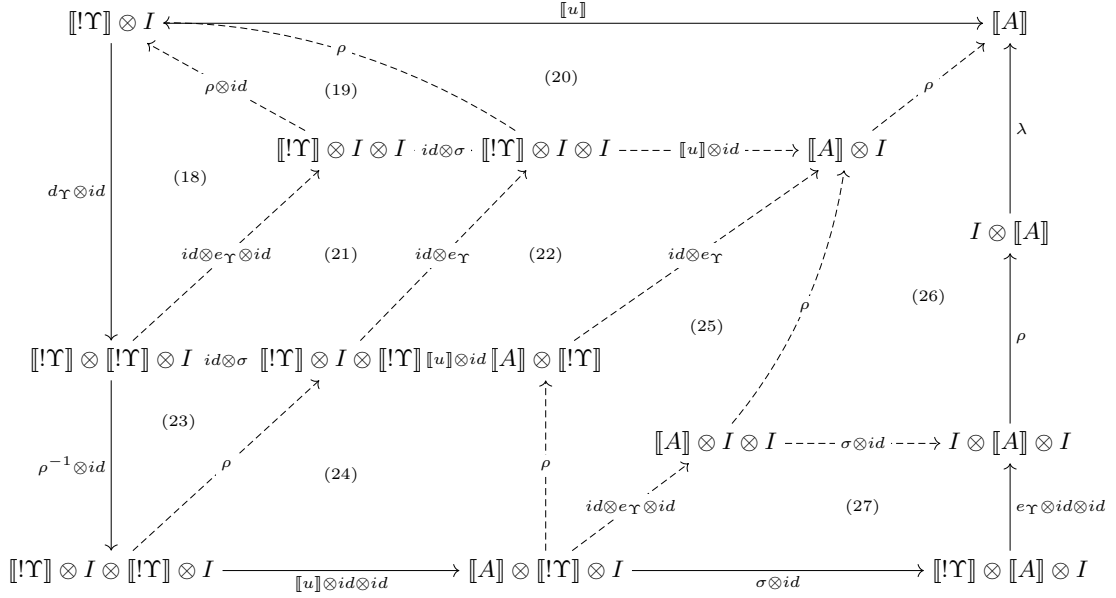
\[\begin{tikzcd}[ampersand replacement=\&,cramped,column sep=tiny]
	{\interpretation{\oc \Upsilon} \otimes I} \&\&\&\&\&\&\&\& {\interpretation{A}} \\
	\\
	\&\& {\interpretation{\oc \Upsilon} \otimes I \otimes I} \&\& {\interpretation{\oc \Upsilon} \otimes I\otimes I} \&\& {\interpretation{A} \otimes I} \\
	\&\&\&\&\&\&\&\& {I \otimes \interpretation{A}} \\
	\\
	{\interpretation{\oc \Upsilon} \otimes \interpretation{\oc \Upsilon} \otimes I} \&\& {\interpretation{\oc \Upsilon} \otimes I \otimes \interpretation{\oc \Upsilon}} \&\& {\interpretation{A} \otimes \interpretation{\oc \Upsilon}} \\
	\&\&\&\&\& {\interpretation{A} \otimes I \otimes I} \&\&\& {I \otimes \interpretation{A} \otimes I} \\
	\\
	{\interpretation{\oc \Upsilon} \otimes I \otimes \interpretation{\oc \Upsilon} \otimes I} \&\&\&\& {\interpretation{A} \otimes \interpretation{\oc \Upsilon} \otimes I} \&\&\&\& {\interpretation{\oc \Upsilon} \otimes \interpretation{A} \otimes I}
	\arrow[""{name=0, anchor=center, inner sep=0}, "{\interpretation{u}}", from=1-1, to=1-9]
	\arrow[""{name=1, anchor=center, inner sep=0}, "{d_\Upsilon \otimes id}"', from=1-1, to=6-1]
	\arrow["{\rho \otimes id}"{description}, dashed, from=3-3, to=1-1]
	\arrow["{id \otimes \sigma}"{description}, dashed, from=3-3, to=3-5]
	\arrow[""{name=2, anchor=center, inner sep=0}, "\rho"{description}, curve={height=18pt}, dashed, from=3-5, to=1-1]
	\arrow["{\interpretation{u} \otimes id}"{description}, dashed, from=3-5, to=3-7]
	\arrow["\rho"{description}, dashed, from=3-7, to=1-9]
	\arrow["{(26)}"{description}, draw=none, from=3-7, to=7-9]
	\arrow["\lambda"', from=4-9, to=1-9]
	\arrow["{id \otimes e_\Upsilon \otimes id}"{description}, dashed, from=6-1, to=3-3]
	\arrow["{id \otimes \sigma}"{description}, dashed, from=6-1, to=6-3]
	\arrow["{\rho^{-1} \otimes id}"', from=6-1, to=9-1]
	\arrow["{(21)}"{description}, draw=none, from=6-3, to=3-3]
	\arrow["{id \otimes e_\Upsilon}"{description}, dashed, from=6-3, to=3-5]
	\arrow["{\interpretation{u} \otimes id}"{description}, dashed, from=6-3, to=6-5]
	\arrow["{(22)}"{description}, draw=none, from=6-5, to=3-5]
	\arrow["{id \otimes e_\Upsilon}"{description}, dashed, from=6-5, to=3-7]
	\arrow[""{name=3, anchor=center, inner sep=0}, "\rho"{description}, curve={height=12pt}, dashed, from=7-6, to=3-7]
	\arrow["{\sigma \otimes id}"{description}, dashed, from=7-6, to=7-9]
	\arrow["{(27)}"{description}, draw=none, from=7-6, to=9-9]
	\arrow["\rho"', from=7-9, to=4-9]
	\arrow[""{name=4, anchor=center, inner sep=0}, "\rho"{description}, dashed, from=9-1, to=6-3]
	\arrow[""{name=5, anchor=center, inner sep=0}, "{\interpretation{u} \otimes id \otimes id}"', from=9-1, to=9-5]
	\arrow["\rho"{description}, dashed, from=9-5, to=6-5]
	\arrow["{id \otimes e_\Upsilon \otimes id}"{description}, dashed, from=9-5, to=7-6]
	\arrow["{\sigma \otimes id}"', from=9-5, to=9-9]
	\arrow["{e_\Upsilon \otimes id \otimes id}"', from=9-9, to=7-9]
	\arrow["{(19)}"{description}, draw=none, from=3-3, to=2]
	\arrow["{(18)}"{description}, draw=none, from=3-3, to=1]
	\arrow["{(20)}"{description}, draw=none, from=3-5, to=0]
	\arrow["{(23)}"{description}, draw=none, from=6-1, to=4]
	\arrow["{(24)}"{description}, draw=none, from=6-3, to=5]
	\arrow["{(25)}"{description}, draw=none, from=6-5, to=3]
\end{tikzcd}\]
(18) commutes by Lemma~\ref{lem:comonoidgen}, (19)  commutes by coherence, (20) commutes by naturality of $\rho$, (21) commutes by naturality of $\sigma$, (22) commutes by functoriality of $\otimes$, (23) commutes by coherence, (24) and (25) commute by naturality of $\rho$, (26) commutes by coherence and (27) commutes by naturality of $\sigma$.

    \item If $t = y \neq x$, then there are two cases:
    \begin{itemize}
      \item If $\Gamma = y^A$, the diagram commutes by Lemma~\ref{lem:eprop} and functoriality of $\otimes$.
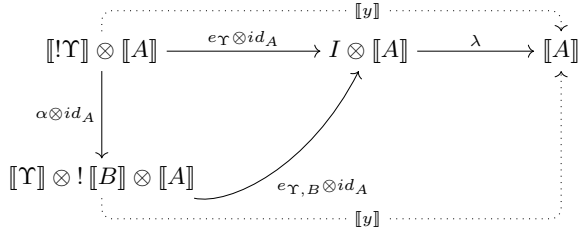
\[\begin{tikzcd}[ampersand replacement=\&,cramped]
	{\interpretation{\oc\Upsilon} \otimes \interpretation{A}} \&\& {I\otimes \interpretation{A}} \&\& {\interpretation{A}} \\
	\\
	{\interpretation{\Upsilon} \otimes \oc \interpretation{B} \otimes \interpretation{A}}
	\arrow["{e_\Upsilon \otimes id_A}", from=1-1, to=1-3]
	\arrow["{\interpretation{y}}"{description}, dotted, from=1-1, to=1-5, rounded corners, to path={ -- ([yshift=2ex]\tikztostart.north) -| ([yshift=2ex,xshift=100pt]\tikztostart.north) \tikztonodes -| (\tikztotarget.north)}]
	\arrow["{\alpha \otimes id_A}"', from=1-1, to=3-1]
	\arrow["\lambda", from=1-3, to=1-5]
	\arrow["{e_{\Upsilon, B} \otimes id_A}"', curve={height=30pt}, from=3-1, to=1-3]
	\arrow["{\interpretation{y}}"{description}, dotted, from=3-1, to=1-5, rounded corners, to path={ -- ([yshift=-2ex]\tikztostart.south) -| ([yshift=-2ex,xshift=100pt]\tikztostart.south) \tikztonodes -| (\tikztotarget.south)}]
\end{tikzcd}\]

      \item If $\Gamma = \varnothing$ and $\Upsilon = \Upsilon', y^A$: 
	  
	Let 
	\[
	\alpha' = \interpretation{\bang\Upsilon'} 
    \xrightarrow{d_{\Upsilon'}}
    \interpretation{\bang\Upsilon'} \otimes \interpretation{\bang\Upsilon'}
	\xrightarrow{\rho^{-1} \otimes id}
	\interpretation{\bang\Upsilon'} \otimes I \otimes \interpretation{\bang\Upsilon'}
	\xrightarrow{\interpretation{\bang u} \otimes id}
	\oc\interpretation{B} \otimes \interpretation{\bang\Upsilon'}
	\xrightarrow{\sigma}
	\interpretation{\bang\Upsilon'} \otimes \oc\interpretation{B}
  \]
  then it is easy to see that $(id \otimes \sigma_{\oc A,\oc B}) \circ \alpha = \alpha' \otimes id_{\oc A}$.
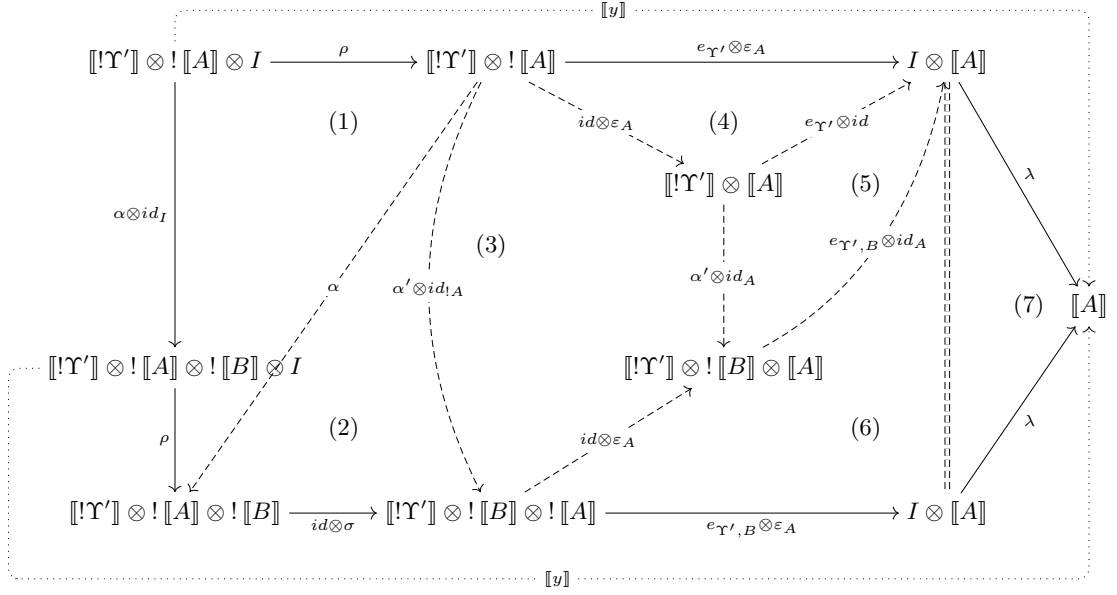
\[\begin{tikzcd}[ampersand replacement=\&,cramped,column sep=4pt,row sep=small]
	{\interpretation{\oc\Upsilon'} \otimes \oc \interpretation{A} \otimes I} \&\& {\interpretation{\oc\Upsilon'} \otimes \oc \interpretation{A}} \&\&\& {I \otimes \interpretation{A}} \\
	\& {(1)} \&\& {(4)} \\
	\&\&\& {\interpretation{\oc\Upsilon'} \otimes \interpretation{A}} \& {(5)} \\
	\&\& {(3)} \\
	\&\&\&\&\&\& {(7)} \& {\interpretation{A}} \\
	{\interpretation{\oc\Upsilon'} \otimes \oc \interpretation{A} \otimes \oc \interpretation{B} \otimes I} \&\&\& {\interpretation{\oc\Upsilon'} \otimes \oc \interpretation{B} \otimes \interpretation{A}} \\
	\& {(2)} \&\&\& {(6)} \\
	\\
	{\interpretation{\oc\Upsilon'} \otimes \oc \interpretation{A} \otimes \oc \interpretation{B}} \&\& {\interpretation{\oc\Upsilon'} \otimes \oc \interpretation{B} \otimes \oc \interpretation{A}} \&\&\& {I \otimes \interpretation{A}}
	\arrow["\rho", from=1-1, to=1-3]
	\arrow["{\interpretation{y}}"{description}, dotted, from=1-1, to=5-8, rounded corners, to path={-- ([yshift=3ex]\tikztostart.north) -- ([yshift=3ex,xshift=330pt]\tikztostart.north) \tikztonodes -| (\tikztotarget.north)}]
	\arrow["{\alpha \otimes id_I}"', from=1-1, to=6-1]
	\arrow["{e_{\Upsilon'} \otimes \varepsilon_A}", from=1-3, to=1-6]
	\arrow["{id \otimes \varepsilon_A}"{description}, dashed, from=1-3, to=3-4]
	\arrow["\alpha"{description}, dashed, from=1-3, to=9-1]
	\arrow["{\alpha' \otimes id_{\oc A}}"{description}, curve={height=30pt}, dashed, from=1-3, to=9-3]
	\arrow["\lambda", from=1-6, to=5-8]
	\arrow[equals, dashed, from=1-6, to=9-6]
	\arrow["{e_{\Upsilon'} \otimes id}"{description}, dashed, from=3-4, to=1-6]
	\arrow["{\alpha' \otimes id_A}"{description}, dashed, from=3-4, to=6-4]
	\arrow["{\interpretation{y}}"{description}, dotted, from=6-1, to=5-8, rounded corners, to path={-- ([xshift=-3ex]\tikztostart.west) |- ([yshift=-95pt,xshift=-400pt]\tikztotarget.south) -- ([yshift=-95pt]\tikztotarget.south) \tikztonodes -- (\tikztotarget.south)}]
	\arrow["\rho"', from=6-1, to=9-1]
	\arrow["{e_{\Upsilon',B} \otimes id_A}"{description}, curve={height=24pt}, dashed, from=6-4, to=1-6]
	\arrow["{id \otimes \sigma}"', from=9-1, to=9-3]
	\arrow["{id \otimes \varepsilon_A}"{description}, dashed, from=9-3, to=6-4]
	\arrow["{e_{\Upsilon',B} \otimes \varepsilon_A}"', from=9-3, to=9-6]
	\arrow["\lambda"', from=9-6, to=5-8]
\end{tikzcd}\]

(1) commutes by naturality of $\rho$, (2) commutes by defintion of $\alpha'$, (3) and (4) commute by functoriality of $\otimes$, (5) commutes by Lemma~\ref{lem:eprop}, (6) commutes by functoriality of $\otimes$ and (7) commutes trivially.

    \end{itemize}
    \item If $t = v \plus w$, then $\Upsilon,x^B;\Gamma \vdash v:A$ and $\Upsilon,x^B;\Gamma \vdash w:A$.
    
    By the induction hypothesis,
    \begin{align*}
      \interpretation{(u/x)v} &= \interpretation{!\Upsilon} \otimes \interpretation{\Gamma} \xrightarrow{\alpha \otimes id_\Gamma} \interpretation{\Upsilon} \otimes \oc\interpretation{B} \otimes \interpretation{\Gamma} \xrightarrow{\interpretation{v}} \interpretation{A} \\
      \interpretation{(u/x)v} &= \interpretation{!\Upsilon} \otimes \interpretation{\Gamma} \xrightarrow{\alpha \otimes id_\Gamma} \interpretation{\Upsilon} \otimes \oc\interpretation{B} \otimes \interpretation{\Gamma} \xrightarrow{\interpretation{w}} \interpretation{A}
    \end{align*}
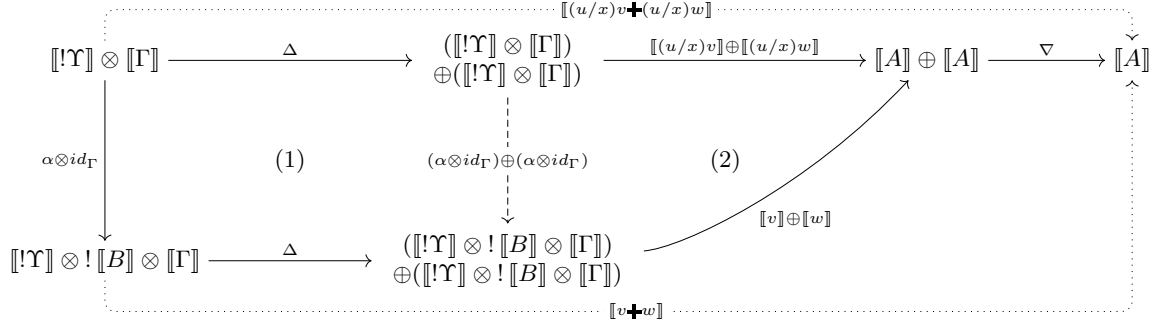
\[\begin{tikzcd}[ampersand replacement=\&,cramped]
	{\interpretation{\oc\Upsilon} \otimes \interpretation{\Gamma}} \&\& \begin{array}{c} (\interpretation{\oc\Upsilon} \otimes \interpretation{\Gamma})\\ \oplus (\interpretation{\oc\Upsilon} \otimes \interpretation{\Gamma}) \end{array} \&\&\& {\interpretation{A} \oplus \interpretation{A}} \&\& {\interpretation{A}} \\
	\& {(1)} \&\& {(2)} \\
	{\interpretation{\oc\Upsilon} \otimes \oc \interpretation{B} \otimes \interpretation{\Gamma}} \&\& \begin{array}{c} (\interpretation{\oc\Upsilon} \otimes \oc \interpretation{B} \otimes \interpretation{\Gamma})\\ \oplus (\interpretation{\oc\Upsilon} \otimes \oc \interpretation{B} \otimes \interpretation{\Gamma}) \end{array}
	\arrow["\Delta", from=1-1, to=1-3]
	\arrow["{\interpretation{(u/x)v \plus (u/x)w}}"{description}, dotted, from=1-1, to=1-8, rounded corners, to path={ -- ([yshift=3ex]\tikztostart.north) -| ([yshift=3ex,xshift=200pt]\tikztostart.north) \tikztonodes -| (\tikztotarget.north)}]
	\arrow["{\alpha \otimes id_\Gamma}"', from=1-1, to=3-1]
	\arrow["{\interpretation{(u/x)v} \oplus \interpretation{(u/x)w}}", from=1-3, to=1-6]
	\arrow["{(\alpha \otimes id_\Gamma) \oplus (\alpha \otimes id_\Gamma)}"{description}, dashed, from=1-3, to=3-3]
	\arrow["\nabla", from=1-6, to=1-8]
	\arrow["{\interpretation{v \plus w}}"{description}, dotted, from=3-1, to=1-8, rounded corners, to path={ -- ([yshift=-3ex]\tikztostart.south) -| ([yshift=-3ex,xshift=200pt]\tikztostart.south) \tikztonodes -| (\tikztotarget.south)}]
	\arrow["\Delta", from=3-1, to=3-3]
	\arrow["{\interpretation{v} \oplus \interpretation{w}}"', curve={height=24pt}, from=3-3, to=1-6]
\end{tikzcd}\]

(1) commutes by naturality of $\Delta$, and (2) commutes by the induction hypothesis and universal property of the biproduct.

\item If $t = a \dotprod v$, then $\Upsilon, x^B; \Gamma \vdash v:A$.

By the induction hypothesis,
\[
	\interpretation{(u/x)v} = \interpretation{!\Upsilon} \otimes \interpretation{\Gamma} \xrightarrow{\alpha \otimes id_\Gamma} \interpretation{\Upsilon} \otimes \oc\interpretation{B} \otimes \interpretation{\Gamma} \xrightarrow{\interpretation{v}} \interpretation{A}
\]
\[\begin{tikzcd}[ampersand replacement=\&,cramped]
	{\interpretation{\oc\Upsilon} \otimes \interpretation{\Gamma}} \&\& {\interpretation{A}} \&\& {\interpretation{A}} \\
	\\
	{\interpretation{\oc\Upsilon} \otimes \oc\interpretation{B} \otimes \interpretation{\Gamma}}
	\arrow["{\interpretation{(u/x)v}}", from=1-1, to=1-3]
	\arrow["{\interpretation{a \dotprod (u/x)v}}"{description}, dotted, from=1-1, to=1-5, rounded corners, to path={ -- ([yshift=2ex]\tikztostart.north) -| ([yshift=2ex,xshift=90pt]\tikztostart.north) \tikztonodes -| (\tikztotarget.north)}]
	\arrow["{\alpha \otimes id_\Gamma}"', from=1-1, to=3-1]
	\arrow["{\widehat{\mono a}}", from=1-3, to=1-5]
	\arrow["{\interpretation v}"', curve={height=30pt}, from=3-1, to=1-3]
	\arrow["{\interpretation{a \dotprod v}}"{description}, dotted, from=3-1, to=1-5, rounded corners, to path={ -- ([yshift=-2ex]\tikztostart.south) -| ([yshift=-2ex,xshift=90pt]\tikztostart.south) \tikztonodes -| (\tikztotarget.south)}]
\end{tikzcd}\]

The diagram commutes by the induction hypothesis.
    \item If $t = a.\star$, then $A = \one$ and $\Gamma = \varnothing$.
\[\begin{tikzcd}[ampersand replacement=\&,cramped]
	{\interpretation{\oc\Upsilon} \otimes I} \&\& {\interpretation{\oc\Upsilon}} \&\& I \&\& I \\
	\& {(1)} \&\& {(2)} \\
	{\interpretation{\oc\Upsilon} \otimes \oc \interpretation{B} \otimes I} \&\& {\interpretation{\oc\Upsilon} \otimes \oc \interpretation{B}}
	\arrow["\rho", from=1-1, to=1-3]
	\arrow["{\interpretation{a.\star}}"{description}, dotted, from=1-1, to=1-7, rounded corners, to path={-- ([yshift=2ex]\tikztostart.north) -- ([yshift=2ex]\tikztotarget.north) \tikztonodes -- (\tikztotarget)}]
	\arrow["{\alpha \otimes id_I}"', from=1-1, to=3-1]
	\arrow["{e_\Upsilon}", from=1-3, to=1-5]
	\arrow["\alpha"{description}, dashed, from=1-3, to=3-3]
	\arrow["{\mono{a}}", from=1-5, to=1-7]
	\arrow["{\interpretation{a.\star}}"{description}, dotted, from=3-1, to=1-7, rounded corners, to path={-- ([yshift=-2ex]\tikztostart.south) -- ([yshift=-2ex,xshift=260pt]\tikztostart.south) \tikztonodes -| (\tikztotarget)}]
	\arrow["\rho"', from=3-1, to=3-3]
	\arrow["{e_{\Upsilon,B}}"', curve={height=30pt}, from=3-3, to=1-5]
\end{tikzcd}\]

    (1) commutes by naturality of $\rho$ and (2) commutes by Lemma~\ref{lem:eprop}.

    \item If $t = \elimone(v, w)$, then $\Gamma = \Gamma_1, \Gamma_2$, $\Upsilon, x^B;\Gamma_1 \vdash v:\one$ and $\Upsilon, x^B;\Gamma_2 \vdash w:A$.
    
    By the induction hypothesis, 
    \begin{align*}
      \interpretation{(u/x)v} &= \interpretation{!\Upsilon} \otimes \interpretation{\Gamma_1} \xrightarrow{\alpha \otimes id_{\Gamma_1}} \interpretation{\Upsilon} \otimes \oc\interpretation{B} \otimes \interpretation{\Gamma_1} \xrightarrow{\interpretation{v}} I \\
      \interpretation{(u/x)w} &= \interpretation{!\Upsilon} \otimes \interpretation{\Gamma_2} \xrightarrow{\alpha \otimes id_{\Gamma_2}} \interpretation{\Upsilon} \otimes \oc\interpretation{B} \otimes \interpretation{\Gamma_2} \xrightarrow{\interpretation{w}} \interpretation{A}
    \end{align*}
\[\begin{tikzcd}[ampersand replacement=\&,cramped,column sep=small]
	{\interpretation{\oc\Upsilon} \otimes \interpretation{\Gamma_1} \otimes \interpretation{\Gamma_2}} \&\& \begin{array}{c} \interpretation{\oc\Upsilon} \otimes \interpretation{\Gamma_1}\\\otimes \interpretation{\oc\Upsilon} \otimes \interpretation{\Gamma_2} \end{array} \&\&\& {I \otimes \interpretation A} \&\& {\interpretation A} \\
	\& {(1)} \&\& {(2)} \\
	{\interpretation{\oc\Upsilon} \otimes \oc \interpretation{B} \otimes \interpretation{\Gamma_1} \otimes \interpretation{\Gamma_2}} \&\& \begin{array}{c} \interpretation{\oc\Upsilon} \otimes \oc \interpretation{B} \otimes \interpretation{\Gamma_1}\\\otimes \interpretation{\oc\Upsilon} \otimes \oc \interpretation{B} \otimes \interpretation{\Gamma_2} \end{array}
	\arrow["{\overline d_{\Upsilon,\Gamma_1,\Gamma_2}}", from=1-1, to=1-3]
	\arrow["{\interpretation{\elimone((u/x)v, (u/x)w)}}"{description}, dotted, from=1-1, to=1-8, rounded corners, to path={ -- ([yshift=3ex]\tikztostart.north) -| ([yshift=3ex,xshift=170pt]\tikztostart.north) \tikztonodes -| (\tikztotarget.north)}]
	\arrow["{\alpha \otimes id_{\Gamma_1} \otimes id_{\Gamma_2}}"', from=1-1, to=3-1]
	\arrow["{\interpretation{(u/x)v} \otimes \interpretation{(u/x)w}}", from=1-3, to=1-6]
	\arrow["{\alpha \otimes id_{\Gamma_1} \otimes \alpha \otimes id_{\Gamma_2}}"{description}, dashed, from=1-3, to=3-3]
	\arrow["\lambda", from=1-6, to=1-8]
	\arrow["{\interpretation{\elimone(v, w)}}"{description}, dotted, from=3-1, to=1-8, rounded corners, to path={ -- ([yshift=-3ex]\tikztostart.south) -| ([yshift=-3ex,xshift=170pt]\tikztostart.south) \tikztonodes -| (\tikztotarget.south)}]
	\arrow["{\overline d_{(\Upsilon,B),\Gamma_1,\Gamma_2}}"', from=3-1, to=3-3]
	\arrow["{\interpretation v \otimes \interpretation w}"', curve={height=30pt}, from=3-3, to=1-6]
\end{tikzcd}\]

(1) commutes by Corollary~\ref{coro:substitutionaux}, and (2) commutes by the induction hypothesis and functoriality of $\otimes$.
    \item If $t = \lambda y^C.v$, then $A = C \multimap D$ and $\Upsilon, x^B; \Gamma, y^C \vdash v:D$.
    
    By the induction hypothesis, 
    \[
      \interpretation{(u/x)v} = \interpretation{\oc \Upsilon} \otimes \interpretation{\Gamma} \otimes \interpretation{C}
	  \xrightarrow{\alpha \otimes id_\Gamma \otimes id_C} 
	  \interpretation{\oc \Upsilon} \otimes \oc \interpretation{B} \otimes \interpretation{\Gamma} \otimes \interpretation{C} 
	  \xrightarrow{\interpretation{v}} \interpretation{D}
    \]
\[\begin{tikzcd}[ampersand replacement=\&,cramped]
	{\interpretation{\oc\Upsilon} \otimes \interpretation{\Gamma}} \&\& {[\interpretation C \to \interpretation{\oc\Upsilon} \otimes \interpretation{\Gamma} \otimes \interpretation{C}]} \&\& {[\interpretation C \to \interpretation D]} \\
	\& {(1)} \&\& {(2)} \\
	{\interpretation{\oc\Upsilon} \otimes \oc \interpretation{B} \otimes \interpretation{\Gamma}} \&\& {[\interpretation C \to \interpretation{\oc\Upsilon} \otimes \oc \interpretation{B} \otimes \interpretation{\Gamma} \otimes \interpretation{C}]}
	\arrow["{\eta_{\interpretation C}}", from=1-1, to=1-3]
	\arrow["{\interpretation{\lambda y^C.(u/x)v}}"{description}, dotted, from=1-1, to=1-5, rounded corners, to path={ -- ([yshift=2ex]\tikztostart.north) -| ([yshift=2ex,xshift=170pt]\tikztostart.north) \tikztonodes -| (\tikztotarget.north)}]
	\arrow["{\alpha \otimes id_\Gamma}"', from=1-1, to=3-1]
	\arrow["{[\interpretation C \to \interpretation{(u/x)v}]}", from=1-3, to=1-5]
	\arrow["{ [\interpretation C \to \alpha \otimes id_\Gamma \otimes id_C]}"{description}, dashed, from=1-3, to=3-3]
	\arrow["{\interpretation{\lambda y^C.v}}"{description}, dotted, from=3-1, to=1-5, rounded corners, to path={ -- ([yshift=-2ex]\tikztostart.south) -| ([yshift=-2ex,xshift=170pt]\tikztostart.south) \tikztonodes -| (\tikztotarget.south)}]
	\arrow["{\eta_{\interpretation C}}"', from=3-1, to=3-3]
	\arrow["{[\interpretation C \to \interpretation{v}]}"', curve={height=30pt}, from=3-3, to=1-5]
\end{tikzcd}\]

(1) commutes by naturality of $\eta$, and (2) commutes by the induction hypothesis and functoriality of $[\interpretation C \to -]$.
    \item If $t = v~w$, then $\Gamma = \Gamma_1, \Gamma_2$, $\Upsilon, x^B; \Gamma_1 \vdash w:C$ and $\Upsilon, x^B; \Gamma_2 \vdash v:C \multimap A$.
    
    By the induction hypothesis,
    \begin{align*}
	\interpretation{(u/x)w} &= \interpretation{!\Upsilon} \otimes \interpretation{\Gamma_1} \xrightarrow{\alpha \otimes id_{\Gamma_1}} \interpretation{\Upsilon} \otimes \oc\interpretation{B} \otimes \interpretation{\Gamma_1} \xrightarrow{\interpretation{w}} \interpretation{C}\\
    \interpretation{(u/x)v} &= \interpretation{!\Upsilon} \otimes \interpretation{\Gamma_2} \xrightarrow{\alpha \otimes id_{\Gamma_2}} \interpretation{\Upsilon} \otimes \oc\interpretation{B} \otimes \interpretation{\Gamma_2} \xrightarrow{\interpretation{v}} [\interpretation C \to \interpretation A]
    \end{align*}
\[\begin{tikzcd}[ampersand replacement=\&,cramped,column sep=small]
	{\interpretation{\oc\Upsilon} \otimes \interpretation{\Gamma_1} \otimes \interpretation{\Gamma_2}} \&\& \begin{array}{c} \interpretation{\oc\Upsilon} \otimes \interpretation{\Gamma_1} \\\otimes \interpretation{\oc\Upsilon} \otimes \interpretation{\Gamma_2} \end{array} \&\& {\interpretation C \otimes [\interpretation C \to \interpretation A]} \&\& {\interpretation A} \\
	\& {(1)} \&\& {(2)} \\
	{\interpretation{\oc\Upsilon} \otimes \oc \interpretation{B} \otimes \interpretation{\Gamma_1} \otimes \interpretation{\Gamma_2}} \&\& \begin{array}{c} \interpretation{\oc\Upsilon} \otimes \oc \interpretation{B} \otimes \interpretation{\Gamma_1} \\\otimes \interpretation{\oc\Upsilon} \otimes \oc \interpretation{B} \otimes \interpretation{\Gamma_2} \end{array}
	\arrow["{\overline d_{\Upsilon,\Gamma_1,\Gamma_2}}", from=1-1, to=1-3]
	\arrow["{\interpretation{(u/x)v~(u/x)w}}"{description}, dotted, from=1-1, to=1-7, rounded corners, to path={ -- ([yshift=3ex]\tikztostart.north) -| ([yshift=3ex,xshift=200pt]\tikztostart.north) \tikztonodes -| (\tikztotarget.north)}]
	\arrow["{\alpha \otimes id_{\Gamma_1} \otimes id_{\Gamma_2}}"', from=1-1, to=3-1]
	\arrow["{\interpretation{(u/x)w} \otimes \interpretation{(u/x)v}}", from=1-3, to=1-5]
	\arrow["{\alpha \otimes id_{\Gamma_1} \otimes \alpha \otimes id_{\Gamma_2}}"{description}, dashed, from=1-3, to=3-3]
	\arrow["{\varepsilon_{\interpretation A} \circ \sigma}", from=1-5, to=1-7]
	\arrow["{\interpretation{v~w}}"{description}, dotted, from=3-1, to=1-7, rounded corners, to path={ -- ([yshift=-3ex]\tikztostart.south) -| ([yshift=-3ex,xshift=200pt]\tikztostart.south) \tikztonodes -| (\tikztotarget.south)}]
	\arrow["{\overline d_{(\Upsilon,B),\Gamma_1,\Gamma_2}}"', from=3-1, to=3-3]
	\arrow["{\interpretation{w} \otimes \interpretation{v}}"', curve={height=30pt}, from=3-3, to=1-5]
\end{tikzcd}\]

(1) commutes by Corollary~\ref{coro:substitutionaux}, and (2) commutes by the induction hypothesis and functoriality of $\otimes$.
    \item If $t = v \otimes w$, then $A = C \otimes D$, $\Gamma = \Gamma_1, \Gamma_2$, $\Upsilon, x^B; \Gamma_1 \vdash v:C$ and $\Upsilon, x^B; \Gamma_2 \vdash w:D$.
    
    By the induction hypothesis, 
    \begin{align*}
      \interpretation{(u/x)v} &= \interpretation{!\Upsilon} \otimes \interpretation{\Gamma_1} \xrightarrow{\alpha \otimes id_{\Gamma_1}} \interpretation{\Upsilon} \otimes \oc\interpretation{B} \otimes \interpretation{\Gamma_1} \xrightarrow{\interpretation{v}} \interpretation{C} \\
      \interpretation{(u/x)w} &= \interpretation{!\Upsilon} \otimes \interpretation{\Gamma_2} \xrightarrow{\alpha \otimes id_{\Gamma_2}} \interpretation{\Upsilon} \otimes \oc\interpretation{B} \otimes \interpretation{\Gamma_2} \xrightarrow{\interpretation{w}} \interpretation{D}
    \end{align*}

\[\begin{tikzcd}[ampersand replacement=\&,cramped]
	{\interpretation{\oc\Upsilon} \otimes \interpretation{\Gamma_1} \otimes \interpretation{\Gamma_2}} \&\& \begin{array}{c} \interpretation{\oc\Upsilon} \otimes \interpretation{\Gamma_1} \\\otimes \interpretation{\oc\Upsilon} \otimes \interpretation{\Gamma_2} \end{array} \&\&\& {\interpretation C \otimes \interpretation D} \\
	\& {(1)} \&\& {(2)} \\
	{\interpretation{\oc\Upsilon} \otimes \oc \interpretation{B} \otimes \interpretation{\Gamma_1} \otimes \interpretation{\Gamma_2}} \&\& \begin{array}{c} \interpretation{\oc\Upsilon} \otimes \oc \interpretation{B} \otimes \interpretation{\Gamma_1} \\\otimes \interpretation{\oc\Upsilon} \otimes \oc \interpretation{B} \otimes \interpretation{\Gamma_2} \end{array}
	\arrow["{\overline d_{\Upsilon,\Gamma_1,\Gamma_2}}", from=1-1, to=1-3]
	\arrow["{\interpretation{(u/x)v\otimes (u/x)w}}"{description}, dotted, from=1-1, to=1-6, rounded corners, to path={ -- ([yshift=3ex]\tikztostart.north) -| ([yshift=3ex,xshift=170pt]\tikztostart.north) \tikztonodes -| (\tikztotarget.north)}]
	\arrow["{\alpha \otimes id_{\Gamma_1} \otimes id_{\Gamma_2}}"', from=1-1, to=3-1]
	\arrow["{\interpretation{(u/x)v} \otimes \interpretation{(u/x)w}}", from=1-3, to=1-6]
	\arrow["{\alpha \otimes id_{\Gamma_1} \otimes \alpha \otimes id_{\Gamma_2}}"{description}, dashed, from=1-3, to=3-3]
	\arrow["{\interpretation{v \otimes w}}"{description}, dotted, from=3-1, to=1-6, rounded corners, to path={ -- ([yshift=-3ex]\tikztostart.south) -| ([yshift=-3ex,xshift=170pt]\tikztostart.south) \tikztonodes -| (\tikztotarget.south)}]
	\arrow["{\overline d_{(\Upsilon,B),\Gamma_1,\Gamma_2}}"', from=3-1, to=3-3]
	\arrow["{\interpretation v \otimes \interpretation w}"', curve={height=30pt}, from=3-3, to=1-6]
\end{tikzcd}\]
    
$(1)$ commutes by Corollary~\ref{coro:substitutionaux}, and $(2)$ commutes by the induction hypothesis and functoriality of $\otimes$.
    \item If $t = \elimtens(v, y^C z^D.w)$, then $\Gamma = \Gamma_1, \Gamma_2$, $\Upsilon, x^B; \Gamma_1 \vdash v:C \otimes D$ and $\Upsilon, x^B; \Gamma_2, y^C, z^D \vdash w:A$.
    
    By the induction hypothesis, 
    \begin{align*}
      \interpretation{(u/x)v} &= \interpretation{\oc\Upsilon} \otimes \interpretation{\Gamma_1} \xrightarrow{\alpha \otimes id_{\Gamma_1}} \interpretation{\oc\Upsilon} \otimes \oc \interpretation{B} \otimes \interpretation{\Gamma_1} \xrightarrow{\interpretation{v}} \interpretation{C} \otimes \interpretation{D} \\
      \interpretation{(u/x)w} &= \interpretation{\oc\Upsilon} \otimes \interpretation{\Gamma_2} \otimes \interpretation{C} \otimes \interpretation{D} \xrightarrow{\alpha \otimes id_{\Gamma_2} \otimes id_C \otimes id_D} \interpretation{\oc\Upsilon} \otimes \oc \interpretation{B} \otimes \interpretation{\Gamma_2} \otimes \interpretation{C} \otimes \interpretation{D} \xrightarrow{\interpretation{w}} \interpretation{A}
    \end{align*}

\[\resizebox{.85\textwidth}{!}{\begin{tikzcd}[ampersand replacement=\&,cramped, column sep=tiny]
	{\interpretation{\oc\Upsilon} \otimes \interpretation{\Gamma_1} \otimes \interpretation{\Gamma_2}} \&\& \begin{array}{c} \interpretation{\oc\Upsilon} \otimes \interpretation{\Gamma_1} \\\otimes \interpretation{\oc\Upsilon} \otimes \interpretation{\Gamma_2} \end{array} \&\& \begin{array}{c} \interpretation{C} \otimes \interpretation D \otimes\\\interpretation{\oc\Upsilon} \otimes \interpretation{\Gamma_2} \end{array} \&\& \begin{array}{c} \interpretation{\oc\Upsilon} \otimes \interpretation{\Gamma_2}\\ \otimes \interpretation{C} \otimes \interpretation D \end{array} \&\& {\interpretation A} \\
	\& {(1)} \&\& {(2)} \&\& {(3)} \&\& {(4)} \\
	{\interpretation{\oc\Upsilon} \otimes \oc \interpretation{B} \otimes \interpretation{\Gamma_1} \otimes \interpretation{\Gamma_2}} \&\& \begin{array}{c} \interpretation{\oc\Upsilon} \otimes \oc \interpretation{B} \otimes \interpretation{\Gamma_1} \\\otimes \interpretation{\oc\Upsilon} \otimes \oc \interpretation{B} \otimes \interpretation{\Gamma_2} \end{array} \&\& \begin{array}{c} \interpretation{C} \otimes \interpretation D \otimes \\\interpretation{\oc\Upsilon} \otimes \oc \interpretation B \otimes \interpretation{\Gamma_2} \end{array} \&\& \begin{array}{c} \interpretation{\oc\Upsilon} \otimes \oc \interpretation B \otimes \interpretation{\Gamma_2} \\\otimes \interpretation{C} \otimes \interpretation D \end{array}
	\arrow["{\overline d_{\Upsilon,\Gamma_1,\Gamma_2}}", from=1-1, to=1-3]
	\arrow["{\interpretation{\elimtens((u/x)v, y^C z^D.(u/x)w)}}"{description}, dotted, from=1-1, to=1-9, rounded corners, to path={ -- ([yshift=3ex]\tikztostart.north) -| ([yshift=3ex,xshift=230pt]\tikztostart.north) \tikztonodes -| (\tikztotarget.north)}]
	\arrow["{\alpha \otimes id_{\Gamma_1} \otimes id_{\Gamma_2}}"', from=1-1, to=3-1]
	\arrow["{\interpretation{(u/x)v} \otimes id}", from=1-3, to=1-5]
	\arrow["{\alpha \otimes id_{\Gamma_1} \otimes \alpha \otimes id_{\Gamma_2}}"{description}, dashed, from=1-3, to=3-3]
	\arrow["\sigma", from=1-5, to=1-7]
	\arrow["{id_C \otimes id_D \otimes \alpha \otimes id_{\Gamma_2}}"{description}, dashed, from=1-5, to=3-5]
	\arrow["{\interpretation{(u/x)w}}", from=1-7, to=1-9]
	\arrow["{\alpha \otimes id_{\Gamma_2} \otimes id_C \otimes id_D}"{description}, dashed, from=1-7, to=3-7]
	\arrow["{\interpretation{\elimtens(v, y^C z^D.w)}}"{description}, dotted, from=3-1, to=1-9, rounded corners, to path={ -- ([yshift=-3ex]\tikztostart.south) -| ([yshift=-3ex,xshift=230pt]\tikztostart.south) \tikztonodes -| (\tikztotarget.south)}]
	\arrow["{\overline d_{(\Upsilon,B),\Gamma_1,\Gamma_2}}"', from=3-1, to=3-3]
	\arrow["{\interpretation{v} \otimes id}"', from=3-3, to=3-5]
	\arrow["\sigma"', from=3-5, to=3-7]
	\arrow["{\interpretation{w}}"', curve={height=24pt}, from=3-7, to=1-9]
\end{tikzcd}}\]

(1) commutes by Corollary~\ref{coro:substitutionaux}, (2) commutes by the induction hypothesis and functoriality of $\otimes$, (3) commutes by naturality of $\sigma$ and (4) commutes by the induction hypothesis.
    \item If $t = \topintro$, then $A = \top$. The diagram commutes because $\mathbf 0$ is a terminal object.
    
\[\begin{tikzcd}[ampersand replacement=\&,cramped]
	{\interpretation{\oc\Upsilon} \otimes \interpretation{\Gamma}} \&\& {\mathbf 0} \\
	\\
	{\interpretation{\oc\Upsilon} \otimes \oc \interpretation B \otimes \interpretation{\Gamma}}
	\arrow["\oc", from=1-1, to=1-3]
	\arrow["{\interpretation{\langle\rangle}}"{description}, dotted, from=1-1, to=1-3, rounded corners, to path={ -- ([yshift=2ex]\tikztostart.north) -| ([yshift=2ex,xshift=45pt]\tikztostart.north) \tikztonodes -| (\tikztotarget.north)}]
	\arrow["{\alpha \otimes id_\Gamma}"', from=1-1, to=3-1]
	\arrow["\oc"', curve={height=12pt}, from=3-1, to=1-3]
	\arrow["{\interpretation{\langle\rangle}}"{description}, dotted, from=3-1, to=1-3, rounded corners, to path={ -- ([yshift=-2ex]\tikztostart.south) -| ([yshift=-2ex,xshift=45pt]\tikztostart.south) \tikztonodes -| (\tikztotarget.south)}]
\end{tikzcd}\]

    \item If $t = \elimzero(v)$, then $\Gamma = \Gamma_1, \Gamma_2$ and $\Upsilon, x^B; \Gamma_1 \vdash v:\zero$.
    
    By the induction hypothesis, 
    \[
      \interpretation{(u/x)v} = \interpretation{\oc \Upsilon} \otimes \interpretation{\Gamma_1} \xrightarrow{\alpha \otimes id_{\Gamma_1}} \interpretation{\oc \Upsilon} \otimes \oc \interpretation{B} \otimes \interpretation{\Gamma_1} \xrightarrow{\interpretation{v}} \mathbf 0
    \]

\[\begin{tikzcd}[ampersand replacement=\&,cramped]
	{\interpretation{\oc\Upsilon} \otimes \interpretation{\Gamma_1} \otimes \interpretation{\Gamma_2}} \&\& {\mathbf 0 \otimes \interpretation{\Gamma_2}} \&\& {\interpretation{A}} \\
	\\
	{\interpretation{\oc\Upsilon} \otimes \oc \interpretation{B} \otimes \interpretation{\Gamma_1} \otimes \interpretation{\Gamma_2}}
	\arrow["{\interpretation{(u/x)v} \otimes id_{\Gamma_2}}", from=1-1, to=1-3]
	\arrow["{\interpretation{\elimzero((u/x)v)}}"{description}, dotted, from=1-1, to=1-5, rounded corners, to path={ -- ([yshift=2ex]\tikztostart.north) -| ([yshift=2ex,xshift=100pt]\tikztostart.north) \tikztonodes -| (\tikztotarget.north)}]
	\arrow["{\alpha \otimes id_{\Gamma_1} \otimes id_{\Gamma_2}}"', from=1-1, to=3-1]
	\arrow["0", from=1-3, to=1-5]
	\arrow["{\interpretation v \otimes id_{\Gamma_2}}"', curve={height=30pt}, from=3-1, to=1-3]
	\arrow["{\interpretation{\elimzero(v)}}"{description}, dotted, from=3-1, to=1-5, rounded corners, to path={ -- ([yshift=-2ex]\tikztostart.south) -| ([yshift=-2ex,xshift=100pt]\tikztostart.south) \tikztonodes -| (\tikztotarget.south)}]
\end{tikzcd}\]

The diagram commutes by the induction hypothesis and functoriality of $\otimes$.
    \item If $t = \pair vw$, then $A = C \with D$, $\Upsilon, x^B; \Gamma \vdash v:C$ and $\Upsilon, x^B; \Gamma \vdash w:D$.
    
    By the induction hypothesis, 
    \begin{align*}
      \interpretation{(u/x)v} &= \interpretation{!\Upsilon} \otimes \interpretation{\Gamma} \xrightarrow{\alpha \otimes id_\Gamma} \interpretation{\Upsilon} \otimes \oc\interpretation{B} \otimes \interpretation{\Gamma} \xrightarrow{\interpretation{v}} \interpretation{C} \\
      \interpretation{(u/x)w} &= \interpretation{!\Upsilon} \otimes \interpretation{\Gamma} \xrightarrow{\alpha \otimes id_\Gamma} \interpretation{\Upsilon} \otimes \oc\interpretation{B} \otimes \interpretation{\Gamma} \xrightarrow{\interpretation{w}} \interpretation{D}
    \end{align*}

\[\begin{tikzcd}[ampersand replacement=\&,cramped]
	{\interpretation{\oc\Upsilon} \otimes \interpretation{\Gamma}} \&\& \begin{array}{c} (\interpretation{\oc\Upsilon} \otimes \interpretation{\Gamma}) \\\oplus (\interpretation{\oc\Upsilon} \otimes \interpretation{\Gamma}) \end{array} \&\&\& {\interpretation C \oplus \interpretation D} \\
	\& {(1)} \&\& {(2)} \\
	{\interpretation{\oc\Upsilon} \otimes \oc \interpretation B \otimes \interpretation{\Gamma}} \&\& \begin{array}{c} (\interpretation{\oc\Upsilon} \otimes \oc \interpretation B \otimes \interpretation{\Gamma})\\ \oplus (\interpretation{\oc\Upsilon} \otimes \oc \interpretation B \otimes \interpretation{\Gamma}) \end{array}
	\arrow["\Delta", from=1-1, to=1-3]
	\arrow["{\interpretation{(u/x)\langle v,w \rangle}}"{description}, dotted, from=1-1, to=1-6, rounded corners, to path={ -- ([yshift=3ex]\tikztostart.north) -| ([yshift=3ex,xshift=170pt]\tikztostart.north) \tikztonodes -| (\tikztotarget.north)}]
	\arrow["{\alpha \otimes id_\Gamma}"', from=1-1, to=3-1]
	\arrow["{\interpretation{(u/x)v} \oplus \interpretation{(u/x)w}}", from=1-3, to=1-6]
	\arrow["{(\alpha \otimes id_\Gamma) \oplus (\alpha \otimes id_\Gamma)}"{description}, dashed, from=1-3, to=3-3]
	\arrow["{\interpretation{\langle v,w\rangle}}"{description}, dotted, from=3-1, to=1-6, rounded corners, to path={ -- ([yshift=-3ex]\tikztostart.south) -| ([yshift=-3ex,xshift=170pt]\tikztostart.south) \tikztonodes -| (\tikztotarget.south)}]
	\arrow["\Delta"', from=3-1, to=3-3]
	\arrow["{\interpretation v \oplus \interpretation w}"', curve={height=30pt}, from=3-3, to=1-6]
\end{tikzcd}\]

(1) commutes by naturality of $\Delta$ and (2) commutes by the induction hypothesis the universal property of the biproduct.
    \item If $t = \elimwith^1(v, y^C.w)$, then $\Gamma = \Gamma_1, \Gamma_2$, $\Upsilon, x^B; \Gamma_1 \vdash v:C \with D$ and $\Upsilon, x^B; \Gamma_2, y^C \vdash w:A$.
    
	By the induction hypothesis,
	\begin{align*}
		\interpretation{(u/x)v} &= \interpretation{\oc\Upsilon} \otimes \interpretation{\Gamma_1} \xrightarrow{\alpha \otimes id_{\Gamma_1}} \interpretation{\oc\Upsilon} \otimes \oc \interpretation{B} \otimes \interpretation{\Gamma_1}  \xrightarrow{\interpretation{v}} \interpretation{C} \oplus \interpretation{D} \\
		\interpretation{(u/x)w} &= \interpretation{\oc\Upsilon} \otimes \interpretation{\Gamma_2} \otimes \interpretation{C} \xrightarrow{\alpha \otimes id_{\Gamma_2} \otimes id_C} \interpretation{\oc\Upsilon} \otimes \oc \interpretation{B} \otimes \interpretation{\Gamma_2} \otimes \interpretation{C} \xrightarrow{\interpretation{w}} \interpretation{A}
	\end{align*}

\[\begin{tikzcd}[ampersand replacement=\&,cramped,column sep=small]
	{\interpretation{\oc\Upsilon} \otimes \interpretation{\Gamma_1} \otimes \interpretation{\Gamma_2}} \&\& \begin{array}{c} \interpretation{\oc\Upsilon} \otimes \interpretation{\Gamma_1} \\\otimes \interpretation{\oc\Upsilon} \otimes \interpretation{\Gamma_2} \end{array} \&\& \begin{array}{c} (\interpretation{C} \oplus \interpretation{D}) \\\otimes \interpretation{\oc\Upsilon} \otimes \interpretation{\Gamma_2} \end{array} \\
	\& {(1)} \& {(2)} \\
	{\interpretation{\oc\Upsilon} \otimes \oc \interpretation{B} \otimes \interpretation{\Gamma_1} \otimes \interpretation{\Gamma_2}} \& \begin{array}{c} \interpretation{\oc\Upsilon} \otimes \oc \interpretation B\otimes \interpretation{\Gamma_1} \\\otimes \interpretation{\oc\Upsilon} \otimes \interpretation{\Gamma_2} \end{array} \\
	\& {(3)} \&\&\& {\interpretation{C}  \otimes \interpretation{\oc\Upsilon} \otimes \interpretation{\Gamma_2}} \\
	\begin{array}{c} \interpretation{\oc\Upsilon} \otimes \oc \interpretation{B} \otimes \interpretation{\Gamma_1} \\\otimes \interpretation{\oc\Upsilon} \otimes \oc \interpretation{B} \otimes \interpretation{\Gamma_2} \end{array} \&\& {(4)} \\
	\&\&\&\& {\interpretation{\oc\Upsilon} \otimes \interpretation{\Gamma_2} \otimes \interpretation{C}} \\
	\begin{array}{c} (\interpretation{C} \oplus \interpretation{D}) \\\otimes \interpretation{\oc\Upsilon} \otimes \oc \interpretation{B} \otimes \interpretation{\Gamma_2} \end{array} \&\& {(5)} \\
	\&\&\& {(6)} \\
	\begin{array}{c} \interpretation{C} \otimes \interpretation{\oc\Upsilon} \\\otimes \oc \interpretation{B} \otimes \interpretation{\Gamma_2} \end{array} \&\& \begin{array}{c} \interpretation{\oc\Upsilon} \otimes \oc \interpretation{B} \\\otimes \interpretation{\Gamma_2} \otimes \interpretation{C} \end{array} \&\& {\interpretation A}
	\arrow["{\overline{d}_{\Upsilon,\Gamma_1,\Gamma_2}}", from=1-1, to=1-3]
	\arrow["{\alpha \otimes id_{\Gamma_1} \otimes id_{\Gamma_2}}"', from=1-1, to=3-1]
	\arrow["{\interpretation{\elimwith^1((u/x)v,y^C.(u/x)w)}}"{description}, dotted, from=1-1, to=9-5, rounded corners, to path={-- ([yshift=3ex]\tikztostart.north) -- ([yshift=3ex,xshift=350pt]\tikztostart.north) \tikztonodes -| ([xshift=8ex]\tikztotarget.east) -- (\tikztotarget.east)}]
	\arrow["{\interpretation{(u/x)v} \otimes id}", from=1-3, to=1-5]
	\arrow["{\alpha  \otimes id_{\Gamma_1} \otimes id_{\oc \Upsilon}  \otimes id_{\Gamma_2}}"{description}, dashed, from=1-3, to=3-2]
	\arrow["{\pi_1 \otimes id}", from=1-5, to=4-5]
	\arrow["{id_{C \oplus D} \otimes \alpha  \otimes id_{\Gamma_2}}"{description}, curve={height=-24pt}, dashed, from=1-5, to=7-1]
	\arrow["{\overline d_{(\Upsilon,B),\Gamma_1,\Gamma_2}}"', from=3-1, to=5-1]
	\arrow["{\interpretation{\elimwith^1(v,y^C.w)}}"{description}, dotted, from=3-1, to=9-5, rounded corners, to path={-- ([xshift=-1ex]\tikztostart.west) |- ([yshift=-3ex,xshift=-210pt]\tikztotarget.south) -- ([yshift=-3ex,xshift=-200pt]\tikztotarget.south) \tikztonodes -- ([yshift=-3ex]\tikztotarget.south) -- (\tikztotarget.south)}]
	\arrow["{\interpretation v \otimes id}"{description}, curve={height=24pt}, dashed, from=3-2, to=1-5]
	\arrow["{id_{\oc \Upsilon} \otimes id_{\oc B} \otimes id_{\Gamma_1} \otimes \alpha  \otimes id_{\Gamma_2}}"{description}, dashed, from=3-2, to=5-1]
	\arrow["\sigma", from=4-5, to=6-5]
	\arrow["{id_C \otimes \alpha  \otimes id_{\Gamma_2}}"{description}, dashed, from=4-5, to=9-1]
	\arrow["{\interpretation v \otimes id}"', from=5-1, to=7-1]
	\arrow["{\alpha  \otimes id_{\Gamma_2} \otimes id_C}"{description}, dashed, from=6-5, to=9-3]
	\arrow["{\interpretation{(u/x)w}}", from=6-5, to=9-5]
	\arrow["{\pi_1 \otimes id}"', from=7-1, to=9-1]
	\arrow["\sigma"', from=9-1, to=9-3]
	\arrow["{\interpretation w}"', from=9-3, to=9-5]
\end{tikzcd}\]

(1) commutes by Corollary~\ref{coro:substitutionaux}, (2) commutes by the induction hypothesis and functoriality of $\otimes$, (3) and (4) commute by functoriality of $\otimes$, (5) commutes by naturality of $\sigma$ and (6) commutes by the induction hypothesis.
	
    \item If $t = \elimwith^2(v, y^C.w)$, this case is similar to the previous one.

    \item If $t = \inl(v)$, then $A = C \oplus D$ and $\Upsilon, x^B; \Gamma \vdash v:C$.
    
    By the induction hypothesis, 

    \[
      \interpretation{(u/x)v} = \interpretation{!\Upsilon} \otimes \interpretation{\Gamma} \xrightarrow{\alpha \otimes id_\Gamma} \interpretation{\Upsilon} \otimes \oc\interpretation{B} \otimes \interpretation{\Gamma} \xrightarrow{\interpretation{v}} \interpretation{C}
    \]
\[\begin{tikzcd}[ampersand replacement=\&,cramped]
	{\interpretation{\oc\Upsilon} \otimes \interpretation{\Gamma}} \&\& {\interpretation{C}} \&\& {\interpretation{C} \oplus \interpretation{D}} \\
	\\
	{\interpretation{\oc\Upsilon} \otimes \oc \interpretation B \otimes \interpretation{\Gamma}}
	\arrow["{\interpretation{(u/x)v}}", from=1-1, to=1-3]
	\arrow["{\interpretation{inl((u/x)v)}}"{description}, dotted, from=1-1, to=1-5, rounded corners, to path={ -- ([yshift=2ex]\tikztostart.north) -| ([yshift=2ex,xshift=100pt]\tikztostart.north) \tikztonodes -| (\tikztotarget.north)}]
	\arrow["{\alpha \otimes id_\Gamma}"', from=1-1, to=3-1]
	\arrow["{i_1}", from=1-3, to=1-5]
	\arrow["{\interpretation v}"', from=3-1, to=1-3]
	\arrow["{\interpretation{inl(v)}}"{description}, dotted, from=3-1, to=1-5, rounded corners, to path={ -- ([yshift=-2ex]\tikztostart.south) -| ([yshift=-2ex,xshift=100pt]\tikztostart.south) \tikztonodes -| (\tikztotarget.south)}]
\end{tikzcd}\]

The diagram commutes by the induction hypothesis.
    \item If $t = \inr(v)$, this case is similar to the previous one.
    \item If $t = \elimplus(v, y^C.w_1, z^D.w_2)$, then $\Gamma = \Gamma_1, \Gamma_2$, $\Upsilon, x^B; \Gamma_1 \vdash v: C \oplus D$, $\Upsilon, x^B; \Gamma_2, y^C \vdash w_1:A$ and $\Upsilon, x^B; \Gamma_2, z^D \vdash w_2:A$.
    
	By the induction hypothesis,
	\begin{align*}
		\interpretation{(u/x)v} &= \interpretation{\oc\Upsilon} \otimes \interpretation{\Gamma_1} \xrightarrow{\alpha \otimes id_{\Gamma_1}} \interpretation{\oc\Upsilon} \otimes \oc \interpretation{B} \otimes \interpretation{\Gamma_1} \xrightarrow{\interpretation{v}} \interpretation{C} \oplus \interpretation{D} \\
		\interpretation{(u/x)w_1} &= \interpretation{\oc\Upsilon} \otimes \interpretation{\Gamma_2} \otimes \interpretation{C} \xrightarrow{\alpha \otimes id_{\Gamma_2} \otimes id_C} \interpretation{\oc\Upsilon} \otimes \oc \interpretation{B} \otimes \interpretation{\Gamma_2} \otimes \interpretation{C} \xrightarrow{\interpretation{w_1}} \interpretation{A}\\
		\interpretation{(u/x)w_2} &= \interpretation{\oc\Upsilon} \otimes \interpretation{\Gamma_2} \otimes \interpretation{D} \xrightarrow{\alpha \otimes id_{\Gamma_2} \otimes id_D} \interpretation{\oc\Upsilon} \otimes \oc \interpretation{B} \otimes \interpretation{\Gamma_2} \otimes \interpretation{D} \xrightarrow{\interpretation{w_2}} \interpretation{A}
	\end{align*}

\[\begin{tikzcd}[ampersand replacement=\&,cramped,column sep=tiny]
	{\interpretation{\oc\Upsilon} \otimes \interpretation{\Gamma_1} \otimes \interpretation{\Gamma_2}} \&\& \begin{array}{c} \interpretation{\oc\Upsilon} \otimes \interpretation{\Gamma_1} \\\otimes \interpretation{\oc\Upsilon} \otimes \interpretation{\Gamma_2} \end{array} \&\& \begin{array}{c} (\interpretation{C} \oplus \interpretation{D}) \\\otimes \interpretation{\oc\Upsilon} \otimes \interpretation{\Gamma_2} \end{array} \\
	\& {(1)} \\
	\begin{array}{c} \interpretation{\oc\Upsilon} \otimes \oc \interpretation{B} \\\otimes \interpretation{\Gamma_1} \otimes \interpretation{\Gamma_2} \end{array} \&\& {(2)} \\
	\&\&\&\& \begin{array}{c} (\interpretation{C} \otimes \interpretation{\oc\Upsilon} \otimes \interpretation{\Gamma_2}) \\\oplus (\interpretation{D} \otimes \interpretation{\oc\Upsilon} \otimes \interpretation{\Gamma_2}) \end{array} \\
	\begin{array}{c} \interpretation{\oc\Upsilon} \otimes \oc \interpretation{B} \otimes \interpretation{\Gamma_1} \\\otimes \interpretation{\oc\Upsilon} \otimes \oc \interpretation{B} \otimes \interpretation{\Gamma_2} \end{array} \&\& {(3)} \\
	\&\&\&\& \begin{array}{c} (\interpretation{\oc\Upsilon} \otimes \interpretation{\Gamma_2} \otimes \interpretation{C}) \\\oplus (\interpretation{\oc\Upsilon} \otimes \interpretation{\Gamma_2} \otimes \interpretation{D}) \end{array} \\
	\begin{array}{c} (\interpretation{C} \oplus \interpretation{D}) \\\otimes \interpretation{\oc\Upsilon} \otimes \oc \interpretation{B} \otimes \interpretation{\Gamma_2} \end{array} \&\& {(4)} \\
	\&\&\& {(5)} \\
	\begin{array}{c} (\interpretation{C} \otimes \interpretation{\oc\Upsilon} \otimes \oc \interpretation{B} \otimes \interpretation{\Gamma_2}) \\\oplus (\interpretation{D} \otimes \interpretation{\oc\Upsilon} \otimes \oc \interpretation{B} \otimes \interpretation{\Gamma_2}) \end{array} \&\& \begin{array}{c} (\interpretation{\oc\Upsilon} \otimes \oc \interpretation{B} \otimes \interpretation{\Gamma_2} \otimes \interpretation{C}) \\\oplus (\interpretation{\oc\Upsilon} \otimes \oc \interpretation{B} \otimes \interpretation{\Gamma_2} \otimes \interpretation{D}) \end{array} \&\& {\interpretation A}
	\arrow["{\overline d_{\Upsilon,\Gamma_1,\Gamma_2}}", from=1-1, to=1-3]
	\arrow["{\alpha \otimes id_{\Gamma_1} \otimes id_{\Gamma_2}}"', from=1-1, to=3-1]
	\arrow["{\interpretation{\elimplus((u/x)v, y^C.(u/x)w_1, z^D.(u/x)w_2)}}"{description}, dotted, from=1-1, to=9-5, rounded corners, to path={-- ([yshift=3ex]\tikztostart.north) -- ([yshift=3ex,xshift=350pt]\tikztostart.north) \tikztonodes -| ([xshift=10ex]\tikztotarget.east) -- (\tikztotarget.east)}]
	\arrow["{\interpretation{(u/x)v} \otimes id}", from=1-3, to=1-5]
	\arrow["{\alpha \otimes id_{\Gamma_1} \otimes \alpha \otimes id_{\Gamma_2}}"{description}, dashed, from=1-3, to=5-1]
	\arrow["dist", from=1-5, to=4-5]
	\arrow["{(id_C \oplus id_D) \otimes \alpha \otimes id_{\Gamma_2}}"{description}, dashed, from=1-5, to=7-1]
	\arrow["{\overline d_{(\Upsilon,B),\Gamma_1,\Gamma_2}}"', from=3-1, to=5-1]
	\arrow["{\interpretation{\elimplus(v, y^C.w_1, z^D.w_2)}}"{description}, dotted, from=3-1, to=9-5, rounded corners, to path={-- ([xshift=-8ex]\tikztostart.west) |- ([yshift=-3ex,xshift=-210pt]\tikztotarget.south) -- ([yshift=-3ex,xshift=-200pt]\tikztotarget.south) \tikztonodes -- ([yshift=-2ex]\tikztotarget.south) -- (\tikztotarget.south)}]
	\arrow["{\sigma \oplus \sigma}", from=4-5, to=6-5]
	\arrow["{(id_C \otimes \alpha \otimes id_{\Gamma_2}) \oplus (id_D \otimes \alpha \otimes id_{\Gamma_2})}"{description}, dashed, from=4-5, to=9-1]
	\arrow["{\interpretation v \otimes id}"', from=5-1, to=7-1]
	\arrow["{(\alpha \otimes id_{\Gamma_2} \otimes id_C) \oplus (\alpha \otimes id_{\Gamma_2} \otimes id_D)}"{description}, dashed, from=6-5, to=9-3]
	\arrow["{[\interpretation{(u/x)w_1},\interpretation{(u/x)w_2}]}"{description}, from=6-5, to=9-5]
	\arrow["dist"', from=7-1, to=9-1]
	\arrow["{\sigma \oplus \sigma}"', from=9-1, to=9-3]
	\arrow["{[\interpretation{w_1},\interpretation{w_2}]}"', from=9-3, to=9-5]
\end{tikzcd}\]

(1) commutes by Corollary~\ref{coro:substitutionaux}, (2) commutes by the induction hypothesis and functoriality of $\otimes$, (3) commutes by naturality of the distribution map $dist$, (4) commutes by naturality of $\sigma$ and the universal property of the biproduct and (5) commutes by the induction hypothesis and the universal property of the biproduct.
	
    \item If $t = \bang v$, then $A = \oc C$, $\Gamma = \varnothing$ and $\Upsilon, x^B;\varnothing \vdash v:C$.
    
    By the induction hypothesis,
    \[
      \interpretation{(u/x)v} = \interpretation{\oc\Upsilon} \otimes I \xrightarrow{\alpha \otimes id_I} \interpretation{\oc\Upsilon} \otimes \oc \interpretation{B} \otimes I \xrightarrow{\interpretation{v}} \interpretation{C}
    \]

\[\begin{tikzcd}[ampersand replacement=\&,cramped,column sep=tiny]
	{\interpretation{\oc\Upsilon} \otimes I} \&\& {\interpretation{\oc\Upsilon}} \&\& {\oc\interpretation{\oc\Upsilon}} \&\& {\oc(\interpretation{\oc\Upsilon} \otimes I)} \&\& {\oc\interpretation{A}} \\
	\& {(1)} \&\& {(2)} \&\& {(3)} \&\& {(4)} \\
	{\interpretation{\oc\Upsilon} \otimes \oc \interpretation B \otimes I} \&\& {\interpretation{\oc\Upsilon} \otimes \oc \interpretation B} \&\& {\oc(\interpretation{\oc\Upsilon} \otimes \oc \interpretation B)} \&\& {\oc(\interpretation{\oc\Upsilon} \otimes \oc \interpretation B \otimes I)}
	\arrow["\rho", from=1-1, to=1-3]
	\arrow["{\interpretation{\oc (u/x)v}}"{description}, dotted, from=1-1, to=1-9, rounded corners, to path={ -- ([yshift=2ex]\tikztostart.north) -| ([yshift=2ex,xshift=200pt]\tikztostart.north) \tikztonodes -| (\tikztotarget.north)}]
	\arrow["{\alpha \otimes id_I}"', from=1-1, to=3-1]
	\arrow["{\delta_\Upsilon}", from=1-3, to=1-5]
	\arrow["\alpha"{description}, dashed, from=1-3, to=3-3]
	\arrow["{\oc(\rho^{-1})}", from=1-5, to=1-7]
	\arrow["{\oc\alpha}"{description}, dashed, from=1-5, to=3-5]
	\arrow["{\oc\interpretation{(u/x)v}}", from=1-7, to=1-9]
	\arrow["{\oc(\alpha \otimes id_I)}"{description}, dashed, from=1-7, to=3-7]
	\arrow["{\interpretation{\oc v}}"{description}, dotted, from=3-1, to=1-9, rounded corners, to path={ -- ([yshift=-2ex]\tikztostart.south) -| ([yshift=-2ex,xshift=200pt]\tikztostart.south) \tikztonodes -| (\tikztotarget.south)}]
	\arrow["\rho"', from=3-1, to=3-3]
	\arrow["{\delta_{\Upsilon,B}}"', from=3-3, to=3-5]
	\arrow["{\oc(\rho^{-1})}"', from=3-5, to=3-7]
	\arrow["{\oc\interpretation{v}}"', curve={height=30pt}, from=3-7, to=1-9]
\end{tikzcd}\]

(1) commutes by naturality of $\rho$, (2) commutes by Lemma~\ref{lem:deltaprop}, (3) commutes by naturality of $\rho$ and functoriality of $\oc$, and (4) commutes by the induction hypothesis and functoriality of $\oc$.

    \item If $t = \elimbang(v, y^C.w)$, then $\Gamma = \Gamma_1, \Gamma_2$, $\Upsilon, x^B;\Gamma_1 \vdash v:\oc C$ and $\Upsilon,x^B, y^C;\Gamma_2 \vdash w:A$.
    
    By the induction hypothesis, 
    \begin{align*}
      \interpretation{(u/x)v} &= \interpretation{\oc\Upsilon} \otimes \interpretation{\Gamma_1} \xrightarrow{\alpha \otimes id_{\Gamma_1}} \interpretation{\oc\Upsilon} \otimes \oc \interpretation{B} \otimes \interpretation{\Gamma_1} \xrightarrow{\interpretation{v}} \oc\interpretation{C} \\
      \interpretation{(u/x)w} &= \interpretation{\oc\Upsilon} \otimes \oc \interpretation{C} \otimes \interpretation{\Gamma_2} \xrightarrow{\alpha' \otimes id_{\Gamma_2}} \interpretation{\oc\Upsilon} \otimes \oc \interpretation{C} \otimes \oc \interpretation{B} \otimes \interpretation{\Gamma_2} \xrightarrow{\interpretation{w}'} \interpretation{A}
    \end{align*}

	where 
	\begin{align*}
		\alpha'=& (id_{\interpretation{\oc \Upsilon}} \otimes \sigma) \circ (\alpha \otimes id_{\oc \interpretation{C}})\\
	\interpretation{w} =& \interpretation{\Upsilon, x^B,  y^C; \Gamma_2 \vdash w:A}\\
	\interpretation{w}' =& \interpretation{\Upsilon, y^C, x^B; \Gamma_2 \vdash w:A}
	\end{align*}

	Notice that $\interpretation{w}$ and $\interpretation{w}'$ are equal modulo permutations.

\[\begin{tikzcd}[ampersand replacement=\&,cramped,column sep=tiny]
	{\interpretation{\oc\Upsilon} \otimes \interpretation{\Gamma_1} \otimes \interpretation{\Gamma_2}} \&\& \begin{array}{c} \interpretation{\oc\Upsilon} \otimes \interpretation{\Gamma_1} \\\otimes \interpretation{\oc\Upsilon} \otimes \interpretation{\Gamma_2} \end{array} \&\&\&\& {\oc \interpretation C \otimes \interpretation{\oc\Upsilon} \otimes \interpretation{\Gamma_2}} \\
	\& {(1)} \\
	\begin{array}{c} \interpretation{\oc\Upsilon} \otimes \oc \interpretation{B} \\\otimes \interpretation{\Gamma_1} \otimes \interpretation{\Gamma_2} \end{array} \&\& {(2)} \&\& {(3)} \&\& {\interpretation{\oc\Upsilon} \otimes \oc \interpretation C \otimes \interpretation{\Gamma_2}} \\
	\\
	\begin{array}{c} \interpretation{\oc\Upsilon} \otimes \oc \interpretation{B} \otimes \interpretation{\Gamma_1} \\\otimes \interpretation{\oc\Upsilon} \otimes \oc \interpretation{B} \otimes \interpretation{\Gamma_2} \end{array} \&\&\&\& {(4)} \\
	\&\&\&\& \begin{array}{c} \interpretation{\oc\Upsilon} \otimes \oc \interpretation{C} \\\otimes \oc\interpretation{B} \otimes \interpretation{\Gamma_2} \end{array} \& {(6)} \\
	\&\&\&\& {(5)} \\
	\begin{array}{c} \oc\interpretation{C} \otimes \interpretation{\oc\Upsilon} \\\otimes \oc \interpretation{B} \otimes \interpretation{\Gamma_2} \end{array} \&\& \begin{array}{c} \interpretation{\oc\Upsilon} \otimes \oc \interpretation{B} \\\otimes \oc\interpretation{C} \otimes \interpretation{\Gamma_2} \end{array} \&\&\&\& {\interpretation A}
	\arrow["{\overline d_{\Upsilon,\Gamma_1,\Gamma_2}}", from=1-1, to=1-3]
	\arrow["{\alpha \otimes id_{\Gamma_1} \otimes id_{\Gamma_2}}"', from=1-1, to=3-1]
	\arrow["{\interpretation{\elimbang((u/x)v, y^C.(u/x)w)}}"{description}, dotted, from=1-1, to=8-7, rounded corners, to path={-- ([yshift=3ex]\tikztostart.north) -- ([yshift=3ex,xshift=340pt]\tikztostart.north) \tikztonodes -| ([xshift=10ex]\tikztotarget.east) -- (\tikztotarget.east)}]
	\arrow["{\interpretation{(u/x)v} \otimes id}", from=1-3, to=1-7]
	\arrow["{\alpha \otimes id_{\Gamma_1} \otimes \alpha \otimes id_{\Gamma_2}}"{description}, dashed, from=1-3, to=5-1]
	\arrow["{\sigma \otimes id}", from=1-7, to=3-7]
	\arrow["{id_{\oc C} \otimes \alpha \otimes id_{\Gamma_2}}"{description}, curve={height=40pt}, dashed, from=1-7, to=8-1]
	\arrow["{\overline d_{(\Upsilon,B),\Gamma_1,\Gamma_2}}"', from=3-1, to=5-1]
	\arrow["{\interpretation{\elimbang(v, y^C.w)}}"{description}, dotted, from=3-1, to=8-7, rounded corners, to path={-- ([xshift=-4ex]\tikztostart.west) |- ([yshift=-3ex,xshift=-210pt]\tikztotarget.south) -- ([yshift=-3ex,xshift=-200pt]\tikztotarget.south) \tikztonodes -- ([yshift=-3ex]\tikztotarget.south) -- (\tikztotarget.south)}]
	\arrow["{\alpha' \otimes id_{\Gamma_2}}"{description}, curve={height=-18pt}, dashed, from=3-7, to=6-5]
	\arrow["{\alpha \otimes id_{\oc C} \otimes id_{\Gamma_2}}"{description}, curve={height=80pt}, dashed, from=3-7, to=8-3]
	\arrow["{\interpretation{(u/x)w}}", from=3-7, to=8-7]
	\arrow["{\interpretation v \otimes id}"', from=5-1, to=8-1]
	\arrow["{\interpretation{w}'}"{description}, dashed, from=6-5, to=8-7]
	\arrow["{\sigma \otimes id_{\Gamma_2}}"', from=8-1, to=8-3]
	\arrow["{id \otimes \sigma \otimes id_{\Gamma_2}}"{description}, curve={height=18pt}, dashed, from=8-3, to=6-5]
	\arrow["{\interpretation{w}}"', from=8-3, to=8-7]
\end{tikzcd}\]

(1) commutes by Corollary~\ref{coro:substitutionaux}, (2) commutes by the induction hypothesis and functoriality of $\otimes$, (3) commutes by naturality of $\sigma$ and functoriality of $\otimes$, (4) commutes by definition of $\alpha'$, (5) commutes because $\interpretation{w}$ and $\interpretation{w}'$ are equal modulo permutations, and (6) commutes by the induction hypothesis.\qedhere

  \end{itemize}
\end{proof}

\linearsubstitution*
\begin{proof}
  By induction on $u$. This lemma is a straightforward extension of \cite[Lemma 4.5]{DiazcaroMalherbe24}. 
  
  We show some cases as an example:

  \begin{itemize}
      \item If $u = x$, then $B = A$ and $\Delta = \varnothing$.
  \[\begin{tikzcd}[ampersand replacement=\&,cramped,column sep=small]
    {\interpretation{\oc \Upsilon} \otimes \interpretation{\Gamma} \otimes I} \&\&\&\&\&\& {\interpretation{\oc \Upsilon} \otimes \interpretation{\oc \Upsilon} \otimes \interpretation{\Gamma} \otimes I} \\
    \&\&\&\& {(1)} \\
    \&\& {(2)} \&\& {I \otimes \interpretation{\oc \Upsilon} \otimes \interpretation{\Gamma} \otimes I} \\
    \&\&\&\&\& {(3)} \\
    {\interpretation{\oc \Upsilon} \otimes \interpretation{\Gamma}} \&\& {I \otimes \interpretation{\oc \Upsilon} \otimes \interpretation{\Gamma}} \&\& {I \otimes I\otimes \interpretation{\oc \Upsilon} \otimes \interpretation{\Gamma}} \\
    \& {(4)} \&\& {(5)} \&\&\& {\interpretation{\oc \Upsilon} \otimes I\otimes\interpretation{\oc \Upsilon} \otimes \interpretation{\Gamma}} \\
    \&\& {I \otimes\interpretation A} \&\& {I \otimes I\otimes\interpretation A} \& {(6)} \\
    \&\&\& {(7)} \\
    {\interpretation A} \&\&\&\&\&\& {\interpretation{\oc \Upsilon} \otimes I\otimes\interpretation A}
    \arrow["{d_\Upsilon \otimes id}", from=1-1, to=1-7]
    \arrow["{\lambda^{-1}}"{description}, dashed, from=1-1, to=3-5]
    \arrow["\rho"', from=1-1, to=5-1]
    \arrow["{e_\Upsilon \otimes id \otimes id}", dashed, from=1-7, to=3-5]
    \arrow["{id \otimes \sigma}", from=1-7, to=6-7]
    \arrow["{id \otimes \sigma}"{description}, dashed, from=3-5, to=5-5]
    \arrow["{\interpretation{v}}"', from=5-1, to=9-1]
    \arrow["\lambda"{description}, dashed, from=5-3, to=5-1]
    \arrow["{id \otimes \interpretation{v}}"{description}, dashed, from=5-3, to=7-3]
    \arrow["\lambda"{description}, dashed, from=5-5, to=5-3]
    \arrow["{id \otimes id \otimes \interpretation{v}}"{description}, dashed, from=5-5, to=7-5]
    \arrow["{e_\Upsilon \otimes id \otimes id}", dashed, from=6-7, to=5-5]
    \arrow["{id \otimes \interpretation v}", from=6-7, to=9-7]
    \arrow["\lambda"{description}, dashed, from=7-3, to=9-1]
    \arrow["\lambda"{description}, dashed, from=7-5, to=7-3]
    \arrow["{e_\Upsilon \otimes id \otimes id}"{description}, dashed, from=9-7, to=7-5]
    \arrow["{\interpretation x}", from=9-7, to=9-1]
  \end{tikzcd}\]
  
  (1) commutes by Lemma~\ref{lem:comonoidgen}, (2) commutes by coherence, (3) commutes by naturality of $\sigma$, (4) and (5) commute by naturality of $\lambda$, (6) commutes by functoriality of $\otimes$ and (7) commutes by definition of $\interpretation{x}$.
  
  \item If $u = t \plus w$, then $\Upsilon;\Delta,x^A \vdash t:B$ and $\Upsilon;\Delta,x^A \vdash w:B$.
\[\begin{tikzcd}[ampersand replacement=\&,cramped,column sep=2pt]
	{\interpretation{\oc \Upsilon} \otimes \interpretation{\Gamma} \otimes \interpretation{\Delta}} \&\&\&\&\& {\interpretation{\oc \Upsilon} \otimes \interpretation{\oc \Upsilon} \otimes \interpretation{\Gamma} \otimes \interpretation{\Delta}} \\
	\& \begin{array}{c} (\interpretation{\oc \Upsilon} \otimes \interpretation{\Gamma} \otimes \interpretation{\Delta}) \\\oplus (\interpretation{\oc \Upsilon} \otimes \interpretation{\Gamma} \otimes \interpretation{\Delta}) \end{array} \&\&\& \begin{array}{c} (\interpretation{\oc \Upsilon} \otimes \interpretation{\oc \Upsilon} \otimes \interpretation{\Gamma} \otimes \interpretation{\Delta}) \\\oplus (\interpretation{\oc \Upsilon} \otimes \interpretation{\oc \Upsilon} \otimes \interpretation{\Gamma} \otimes \interpretation{\Delta}) \end{array} \\
	\\
	\\
	\&\&\&\& \begin{array}{c} (\interpretation{\oc \Upsilon} \otimes \interpretation{\Delta} \otimes\interpretation{\oc \Upsilon} \otimes \interpretation{\Gamma}) \\ \oplus (\interpretation{\oc \Upsilon} \otimes \interpretation{\Delta} \otimes\interpretation{\oc \Upsilon} \otimes \interpretation{\Gamma}) \end{array} \& {\interpretation{\oc \Upsilon} \otimes \interpretation{\Delta} \otimes\interpretation{\oc \Upsilon} \otimes \interpretation{\Gamma}} \\
	\\
	\\
	\& {\interpretation B \oplus \interpretation B} \&\&\& \begin{array}{c} (\interpretation{\oc \Upsilon} \otimes \interpretation{\Delta}\otimes\interpretation A) \\ \oplus (\interpretation{\oc \Upsilon} \otimes \interpretation{\Delta}\otimes\interpretation A) \end{array} \\
	{\interpretation B} \&\&\&\&\& {\interpretation{\oc \Upsilon} \otimes \interpretation{\Delta}\otimes\interpretation A}
	\arrow[""{name=0, anchor=center, inner sep=0}, "{d_\Upsilon \otimes id}", from=1-1, to=1-6]
	\arrow["\Delta"{description}, dashed, from=1-1, to=2-2]
	\arrow[""{name=1, anchor=center, inner sep=0}, "{\interpretation{(v/x)t \plus (v/x)w}}"{description}, from=1-1, to=9-1]
	\arrow["\Delta"{description}, dashed, from=1-6, to=2-5]
	\arrow["{id \otimes \sigma}", from=1-6, to=5-6]
	\arrow[""{name=2, anchor=center, inner sep=0}, "\begin{array}{c} (d_\Upsilon \otimes id) \\ \oplus \\ (d_\Upsilon \otimes id) \end{array}"{description}, dashed, from=2-2, to=2-5]
	\arrow[""{name=3, anchor=center, inner sep=0}, "{\interpretation{(v/x)t} \oplus \interpretation{(v/x)w}}"{description}, dashed, from=2-2, to=8-2]
	\arrow["\begin{array}{c} (id \otimes \sigma) \\ \oplus \\ (id \otimes \sigma) \end{array}"{description}, dashed, from=2-5, to=5-5]
	\arrow["{(4)}"{description}, draw=none, from=5-5, to=1-6]
	\arrow["\begin{array}{c} (id \otimes \interpretation v) \\\oplus  \\ (id \otimes \interpretation v) \end{array}"{description}, dashed, from=5-5, to=8-5]
	\arrow["{(5)}"{description}, draw=none, from=5-5, to=9-6]
	\arrow["\Delta"{description}, curve={height=-18pt}, dashed, from=5-6, to=5-5]
	\arrow["{id \otimes \interpretation v}", from=5-6, to=9-6]
	\arrow["\nabla"{description}, dashed, from=8-2, to=9-1]
	\arrow[""{name=4, anchor=center, inner sep=0}, "{\interpretation t \oplus \interpretation w}"{description}, dashed, from=8-5, to=8-2]
	\arrow["\Delta"{description}, dashed, from=9-6, to=8-5]
	\arrow[""{name=5, anchor=center, inner sep=0}, "{\interpretation{t \plus w}}", from=9-6, to=9-1]
	\arrow["{(1)}"{description, pos=0.6}, draw=none, from=2, to=0]
	\arrow["{(2)}"{description, pos=0.6}, draw=none, from=3, to=1]
	\arrow["{(3)}"{description, pos=0.4}, draw=none, from=5-5, to=3]
	\arrow["{(6)}"{description}, draw=none, from=4, to=5]
\end{tikzcd}\]
  (1) commutes by naturality of $\Delta$, (2) commutes by definition of $\interpretation{(v/x)t \plus (v/x)w}$, (3) commutes by the induction hypothesis, (4) and (5) commute by naturality of $\Delta$ and (6) commutes by definition of $\interpretation{t \plus w}$.
  
  \item If $u = a \dotprod t$, then $\Upsilon;\Delta,x^A \vdash t:B$.

  \[\begin{tikzcd}[ampersand replacement=\&,cramped]
    {\interpretation{\oc \Upsilon} \otimes \interpretation{\Gamma} \otimes \interpretation{\Delta}} \&\&\& {\interpretation{\oc \Upsilon} \otimes \interpretation{\oc \Upsilon} \otimes \interpretation{\Gamma} \otimes \interpretation{\Delta}} \\
    \\
    \& {\interpretation B} \&\& {\interpretation{\oc \Upsilon} \otimes \interpretation{\Delta} \otimes\interpretation{\oc \Upsilon} \otimes \interpretation{\Gamma}} \\
    \\
    {\interpretation B} \&\&\& {\interpretation{\oc \Upsilon} \otimes \interpretation{\Delta} \otimes\interpretation A}
    \arrow["{d_\Upsilon \otimes id}", from=1-1, to=1-4]
    \arrow["{\interpretation{(v/x)t}}"{description}, dashed, from=1-1, to=3-2]
    \arrow[""{name=0, anchor=center, inner sep=0}, "{\interpretation{a \dotprod (v/x)t}}"', from=1-1, to=5-1]
    \arrow["{id \otimes \sigma}", from=1-4, to=3-4]
    \arrow["{(1)}"{description}, draw=none, from=3-2, to=1-4]
    \arrow["{\widehat{\mono{a}}}"{description}, dashed, from=3-2, to=5-1]
    \arrow["{id \otimes \interpretation v}", from=3-4, to=5-4]
    \arrow["{\interpretation t}"{description}, dashed, from=5-4, to=3-2]
    \arrow[""{name=1, anchor=center, inner sep=0}, "{\interpretation{a \dotprod t}}", from=5-4, to=5-1]
    \arrow["{(3)}"{description}, draw=none, from=3-2, to=1]
    \arrow["{(2)}"{description}, draw=none, from=3-2, to=0]
  \end{tikzcd}\]
  (1) commutes by the induction hypothesis, (2) commutes by definition of $\interpretation{a \dotprod (v/x)t}$ and (3) commutes by interpretation of $\interpretation{a \dotprod t}$.
  \end{itemize}

  The new cases are the following:
  \begin{itemize}
    \item If $u = \elimbang(t,y^C.w)$, then $\Delta = \Delta_1, \Delta_2$.
\begin{itemize}
	\item If $x \in \fv(t)$, then $\Upsilon;\Delta_1,x^A \vdash t:\oc C$ and $\Upsilon,y^C;\Delta_2 \vdash w:B$.
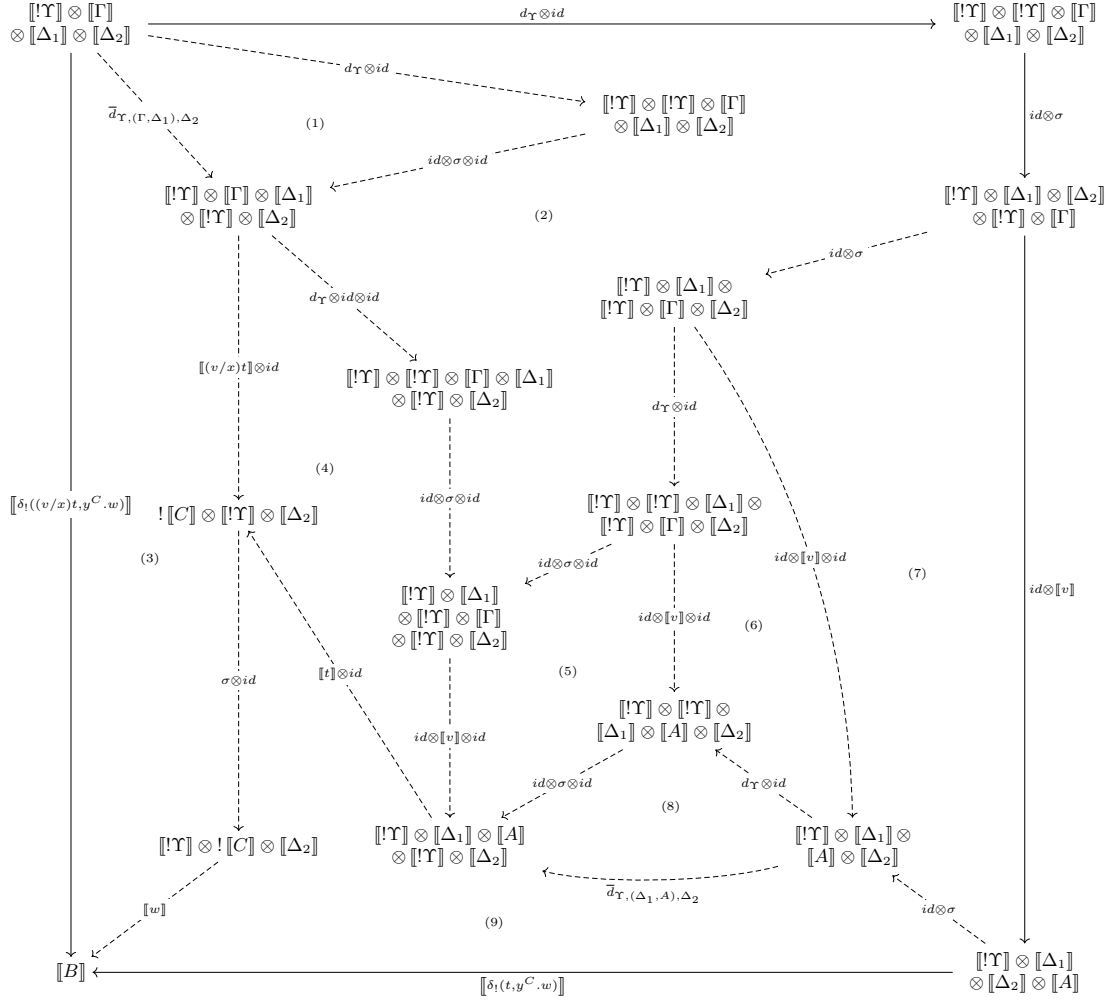
\[\resizebox{.82\textwidth}{!}{\begin{tikzcd}[ampersand replacement=\&,cramped,column sep=0pt]
	\begin{array}{c} \interpretation{\oc \Upsilon} \otimes \interpretation{\Gamma} \\ \otimes \interpretation{\Delta_1} \otimes \interpretation{\Delta_2} \end{array} \&\&\&\&\& \begin{array}{c} \interpretation{\oc \Upsilon} \otimes \interpretation{\oc \Upsilon} \otimes \interpretation{\Gamma} \\ \otimes \interpretation{\Delta_1} \otimes \interpretation{\Delta_2} \end{array} \\
	\&\&\& \begin{array}{c} \interpretation{\oc \Upsilon} \otimes \interpretation{\oc \Upsilon} \otimes \interpretation{\Gamma} \\ \otimes \interpretation{\Delta_1} \otimes \interpretation{\Delta_2} \end{array} \\
	\& \begin{array}{c} \interpretation{\oc \Upsilon} \otimes \interpretation{\Gamma} \otimes \interpretation{\Delta_1}\\ \otimes \interpretation{\oc \Upsilon} \otimes \interpretation{\Delta_2} \end{array} \&\&\&\& \begin{array}{c} \interpretation{\oc \Upsilon} \otimes \interpretation{\Delta_1} \otimes \interpretation{\Delta_2} \\ \otimes\interpretation{\oc \Upsilon} \otimes \interpretation{\Gamma} \end{array} \\
	\&\&\& \begin{array}{c} \interpretation{\oc \Upsilon} \otimes \interpretation{\Delta_1} \otimes \\ \interpretation{\oc \Upsilon} \otimes \interpretation{\Gamma} \otimes \interpretation{\Delta_2} \end{array} \\
	\&\& \begin{array}{c} \interpretation{\oc \Upsilon} \otimes\interpretation{\oc \Upsilon} \otimes \interpretation{\Gamma} \otimes \interpretation{\Delta_1} \\\otimes \interpretation{\oc \Upsilon} \otimes \interpretation{\Delta_2} \end{array} \\
	\\
	\& {\oc\interpretation{C} \otimes \interpretation{\oc \Upsilon} \otimes \interpretation{\Delta_2}} \&\& \begin{array}{c} \interpretation{\oc \Upsilon} \otimes \interpretation{\oc \Upsilon} \otimes \interpretation{\Delta_1} \otimes \\ \interpretation{\oc \Upsilon} \otimes \interpretation{\Gamma} \otimes \interpretation{\Delta_2} \end{array} \\
	\&\& \begin{array}{c} \interpretation{\oc \Upsilon} \otimes \interpretation{\Delta_1} \\\otimes \interpretation{\oc \Upsilon} \otimes \interpretation{\Gamma} \\\otimes \interpretation{\oc \Upsilon} \otimes \interpretation{\Delta_2} \end{array} \\
	\&\&\& \begin{array}{c} \interpretation{\oc \Upsilon} \otimes \interpretation{\oc \Upsilon} \otimes \\ \interpretation{\Delta_1} \otimes\interpretation A \otimes \interpretation{\Delta_2} \end{array} \\
	\\
	\& {\interpretation{\oc \Upsilon} \otimes \oc\interpretation{C} \otimes \interpretation{\Delta_2}} \& \begin{array}{c} \interpretation{\oc \Upsilon} \otimes \interpretation{\Delta_1} \otimes\interpretation A \\ \otimes \interpretation{\oc \Upsilon} \otimes  \interpretation{\Delta_2} \end{array} \&\& \begin{array}{c} \interpretation{\oc \Upsilon} \otimes \interpretation{\Delta_1} \otimes \\ \interpretation A \otimes \interpretation{\Delta_2} \end{array} \\
	\\
	{\interpretation B} \&\&\&\&\& \begin{array}{c} \interpretation{\oc \Upsilon} \otimes \interpretation{\Delta_1} \\ \otimes \interpretation{\Delta_2} \otimes\interpretation A \end{array}
	\arrow["{d_\Upsilon \otimes id}", from=1-1, to=1-6]
	\arrow[""{name=0, anchor=center, inner sep=0}, "{d_\Upsilon \otimes id}"{description}, dashed, from=1-1, to=2-4]
	\arrow["{\overline{d}_{\Upsilon,(\Gamma,\Delta_1),\Delta_2}}"{description}, dashed, from=1-1, to=3-2]
	\arrow[""{name=1, anchor=center, inner sep=0}, "{\interpretation{\elimbang((v/x)t, y^C.w)}}"{description}, from=1-1, to=13-1]
	\arrow["{id \otimes \sigma}", from=1-6, to=3-6]
	\arrow[""{name=2, anchor=center, inner sep=0}, "{id \otimes \sigma \otimes id}"{description}, dashed, from=2-4, to=3-2]
	\arrow["{d_\Upsilon \otimes id \otimes id}"{description}, dashed, from=3-2, to=5-3]
	\arrow["{\interpretation{(v/x)t} \otimes id}"{description}, dashed, from=3-2, to=7-2]
	\arrow["{(4)}"{description, pos=0.4}, draw=none, from=3-2, to=11-3]
	\arrow["{id \otimes \sigma}"{description}, dashed, from=3-6, to=4-4]
	\arrow[""{name=3, anchor=center, inner sep=0}, "{id \otimes \interpretation v}", from=3-6, to=13-6]
	\arrow["{d_\Upsilon \otimes id}"{description}, dashed, from=4-4, to=7-4]
	\arrow[""{name=4, anchor=center, inner sep=0}, "{id \otimes \interpretation v \otimes id}"{description}, curve={height=-30pt}, dashed, from=4-4, to=11-5]
	\arrow["{id \otimes \sigma \otimes id}"{description}, dashed, from=5-3, to=8-3]
	\arrow["{\sigma \otimes id}"{description}, dashed, from=7-2, to=11-2]
	\arrow["{id \otimes \sigma \otimes id}"{description}, dashed, from=7-4, to=8-3]
	\arrow["{id \otimes \interpretation v \otimes id}"{description}, dashed, from=7-4, to=9-4]
	\arrow["{(5)}"{description}, draw=none, from=8-3, to=9-4]
	\arrow["{id \otimes \interpretation v \otimes id}"{description}, dashed, from=8-3, to=11-3]
	\arrow[""{name=5, anchor=center, inner sep=0}, "{id \otimes \sigma \otimes id}"{description}, dashed, from=9-4, to=11-3]
	\arrow["{\interpretation w}"{description}, dashed, from=11-2, to=13-1]
	\arrow["{\interpretation{t} \otimes id}"{description}, dashed, from=11-3, to=7-2]
	\arrow["{d_\Upsilon \otimes id}"{description}, dashed, from=11-5, to=9-4]
	\arrow["{\overline{d}_{\Upsilon, (\Delta_1, A), \Delta_2}}", curve={height=-18pt}, dashed, from=11-5, to=11-3]
	\arrow["{id \otimes \sigma}"{description}, dashed, from=13-6, to=11-5]
	\arrow[""{name=6, anchor=center, inner sep=0}, "{\interpretation{\elimbang(t,y^C.w)}}", from=13-6, to=13-1]
	\arrow["{(1)}"{description}, draw=none, from=3-2, to=0]
	\arrow["{(2)}"{description}, draw=none, from=4-4, to=2]
	\arrow["{(7)}"{description}, draw=none, from=4, to=3]
	\arrow["{(3)}"{description}, curve={height=-30pt}, draw=none, from=7-2, to=1]
	\arrow["{(6)}"{description}, draw=none, from=9-4, to=4]
	\arrow["{(9)}"{description}, draw=none, from=11-3, to=6]
	\arrow["{(8)}"{description}, draw=none, from=11-5, to=5]
\end{tikzcd}}\]
	
	(1) commutes by definition of $\overline{d}$, (2) commutes by Lemma~\ref{lem:comonoidgen}, (3) commutes by definition of\\ $\interpretation{\elimbang((v/x)t, y^C.w)}$, (4) commutes by the induction hypothesis, (5) commutes by naturality of $\sigma$, (6) commutes by functoriality of $\otimes$, (7) commutes by naturality of $\sigma$, (8) commutes by definition of $\overline{d}$ and (9) commutes by definition of $\interpretation{\elimbang(t,y^C.w)}$.

	\item If $x \in \fv(w)$, then $\Upsilon;\Delta_1 \vdash t:\oc C$ and $\Upsilon,y^C;\Delta_2, x^A \vdash w:B$. The commuting diagram is shown in Figure~\ref{fig:lem:linearsubstitutionelimbang} (Appendix~\ref{app:diagramasgrandes}).

(1) and (2) commute by functoriality of $\otimes$, (3) commutes by definition
of $\overline d$, (4) commutes by coherence, (5) and (6) commute by naturality
of $\sigma$, (7) commutes by functoriality of $\otimes$, (8) commutes by
definition of $\interpretation{\elimbang(t,y^C.(v/x)w)}$, (9) commutes by
definition of $d\_$, (10) commutes by the induction hypothesis, (11) commutes
by coherence, (12) commutes by functoriality of $\otimes$, (13) commutes by
naturality of $\sigma$, (14) commutes by coherence, (15) commutes because $(\oc
\interpretation{C}, d_C, e_C)$ is a comonoid, (16), (17), (18), (19) and (20)
commute by naturality of $\sigma$, (21) commutes by
Lemma~\ref{lem:weakeningcat}, where
$\interpretation{v}':\interpretation{\oc \Upsilon} \otimes \oc
\interpretation{C} \otimes \interpretation{\Gamma} \to \interpretation{A}$ is
the interpretation of $v$ with $y^C$ in the intuitionistic context, (22) commutes
by naturality of $\sigma$, (23), (24), (25) and (26) commute by functoriality
of $\otimes$, (27) commutes by definition of $\overline d$ and (28) commutes by
definition of $\interpretation{\elimbang(t,y^C.w)}$.\qedhere
\end{itemize}
  \end{itemize}
\end{proof}

\soundness*
\begin{proof}
  By induction on $\to$.
  \begin{itemize}
    \item If $t = \elimone(a.\star, u)$ and $r = a \dotprod u$, then $\Upsilon;\varnothing \vdash a.\star:\one$ and $\Upsilon;\Gamma \vdash u:A$.
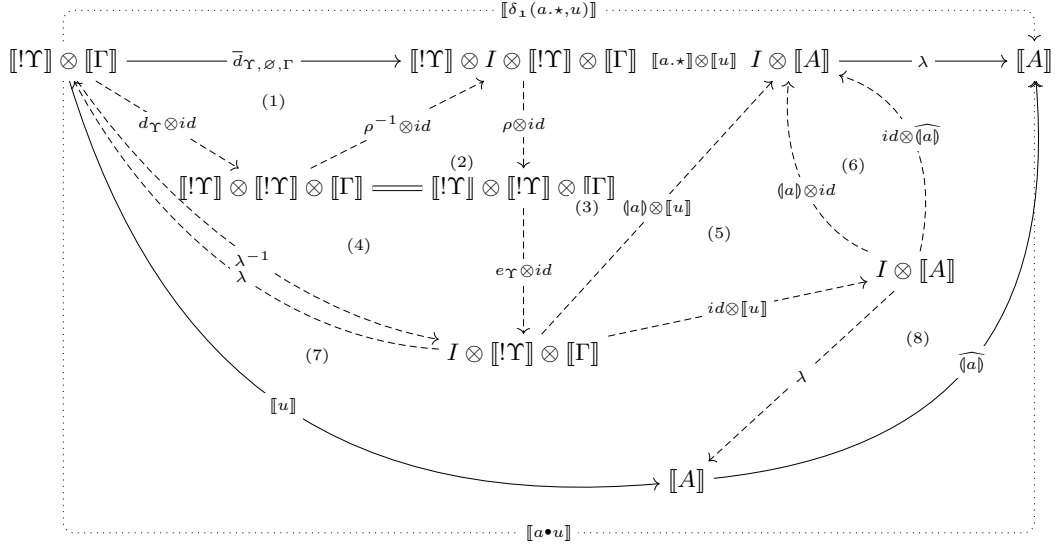
\[\begin{tikzcd}[ampersand replacement=\&,cramped, column sep=tiny]
	{\interpretation{\oc\Upsilon} \otimes \interpretation{\Gamma}} \&\&\&\&\& {\interpretation{\oc\Upsilon} \otimes I\otimes \interpretation{\oc\Upsilon} \otimes \interpretation{\Gamma}} \&\&\& {I \otimes \interpretation{A}} \&\&\&\&\& {\interpretation{A}} \\
	\\
	\&\&\& {\interpretation{\oc\Upsilon} \otimes \interpretation{\oc\Upsilon} \otimes \interpretation{\Gamma}} \&\& {\interpretation{\oc\Upsilon} \otimes \interpretation{\oc\Upsilon} \otimes \interpretation{\Gamma}} \\
	\&\&\&\&\&\&\&\&\&\& {I \otimes \interpretation{A}} \\
	\&\&\&\&\& {I \otimes \interpretation{\oc\Upsilon} \otimes \interpretation{\Gamma}} \\
	\\
	\&\&\&\&\&\& {\interpretation{A}}
	\arrow["{\overline{d}_{\Upsilon, \varnothing, \Gamma}}"{description}, from=1-1, to=1-6]
	\arrow["{\interpretation{\elimone(a.\star, u)}}"{description}, dotted, from=1-1, to=1-14, rounded corners, to path={(\tikztostart.north) -- ([yshift=3ex]\tikztostart.north) -- ([yshift=3ex]\tikztotarget.north) \tikztonodes -- (\tikztotarget.north)}]
	\arrow["{\interpretation{a \dotprod u}}"{description}, dotted, from=1-1, to=1-14, rounded corners, to path={(\tikztostart.south) -- ([yshift=-170pt]\tikztostart.south) -- ([yshift=-170pt]\tikztotarget.south) \tikztonodes -- (\tikztotarget.south)}]
	\arrow["{d_\Upsilon \otimes id}"{description}, dashed, from=1-1, to=3-4]
	\arrow["{(1)}"{description, pos=0.8}, shift left=3, curve={height=-30pt}, draw=none, from=1-1, to=3-4]
	\arrow["{\lambda^{-1}}"{description}, curve={height=30pt}, dashed, from=1-1, to=5-6]
	\arrow["{\interpretation{u}}"{description}, curve={height=80pt}, from=1-1, to=7-7]
	\arrow["{\interpretation{a.\star} \otimes \interpretation{u}}"{description}, from=1-6, to=1-9]
	\arrow["{\rho \otimes id}"{description}, dashed, from=1-6, to=3-6]
	\arrow["\lambda"{description}, from=1-9, to=1-14]
	\arrow["{(6)}"{description}, draw=none, from=1-9, to=4-11]
	\arrow["{\rho^{-1} \otimes id}"{description}, dashed, from=3-4, to=1-6]
	\arrow["{(2)}"{description}, shift left=3, curve={height=24pt}, draw=none, from=3-4, to=1-6]
	\arrow[equals, from=3-4, to=3-6]
	\arrow["{(4)}"{description, pos=0.3}, draw=none, from=3-4, to=5-6]
	\arrow["{e_\Upsilon \otimes id}"{description}, dashed, from=3-6, to=5-6]
	\arrow["{\mono{a} \otimes id}"{description}, curve={height=-24pt}, dashed, from=4-11, to=1-9]
	\arrow["{id \otimes \widehat{\mono{a}}}"{description}, curve={height=30pt}, dashed, from=4-11, to=1-9]
	\arrow["\lambda"{description}, dashed, from=4-11, to=7-7]
	\arrow["{(8)}"{description, pos=0.1}, shift left=3, curve={height=-30pt}, draw=none, from=4-11, to=7-7]
	\arrow["\lambda"{description}, curve={height=-40pt}, dashed, from=5-6, to=1-1]
	\arrow["{(3)}"{description}, curve={height=30pt}, draw=none, from=5-6, to=1-6]
	\arrow["{\mono{a} \otimes \interpretation{u}}"{description}, dashed, from=5-6, to=1-9]
	\arrow["{(5)}"{description}, shift left=3, curve={height=30pt}, draw=none, from=5-6, to=1-9]
	\arrow["{id \otimes \interpretation{u}}"{description}, dashed, from=5-6, to=4-11]
	\arrow["{(7)}"{description}, curve={height=-50pt}, draw=none, from=7-7, to=1-1]
	\arrow["{\widehat{\mono{a}}}"{description}, curve={height=70pt}, from=7-7, to=1-14]
\end{tikzcd}\]

(1) commutes by definition of $\overline d$, (2) commutes trivially, (3) commutes by definition of $\interpretation{a.\star}$ and functoriality of $\otimes$, (4) commutes by Lemma~\ref{lem:comonoidgen}, (5) commutes by functoriality of $\otimes$, (6) commutes by Lemma~\ref{lem:hataprop}, (7) and (8) commute by naturality of $\lambda$.

\item If $t = (\lambda x^B.u)~v$ and $r = (v/x)u$, then $\Gamma = \Gamma_1, \Gamma_2$, $\Upsilon;\Gamma_1, x^B \vdash u:A$ and $\Upsilon;\Gamma_2 \vdash v:B$.
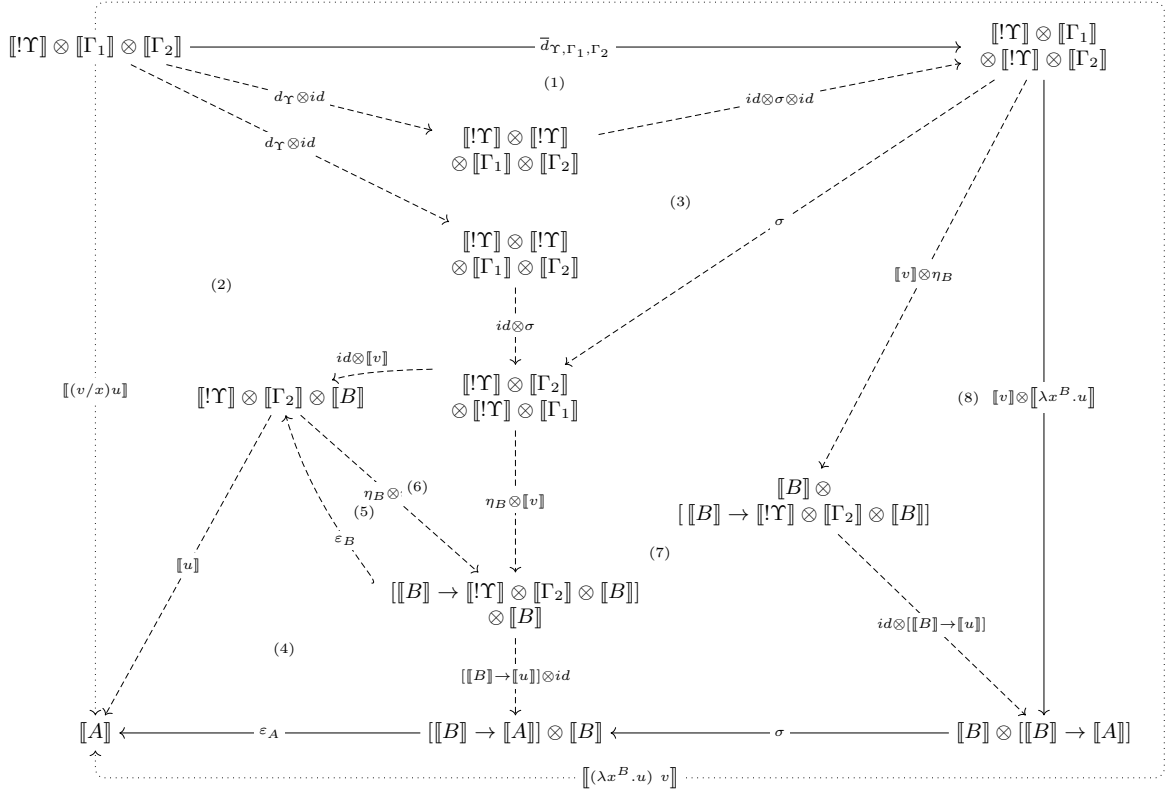
\[\resizebox{.85\textwidth}{!}{\begin{tikzcd}[ampersand replacement=\&,cramped,column sep=0em]
	{\interpretation{\oc \Upsilon} \otimes \interpretation{\Gamma_1} \otimes \interpretation{\Gamma_2}} \&\&\&\&\& \begin{array}{c} \interpretation{\oc \Upsilon} \otimes \interpretation{\Gamma_1} \\\otimes \interpretation{\oc \Upsilon} \otimes \interpretation{\Gamma_2} \end{array} \\
	\&\&\& \begin{array}{c} \interpretation{\oc \Upsilon} \otimes \interpretation{\oc \Upsilon} \\\otimes \interpretation{\Gamma_1} \otimes \interpretation{\Gamma_2} \end{array} \\
	\&\&\& \begin{array}{c} \interpretation{\oc \Upsilon} \otimes \interpretation{\oc \Upsilon} \\\otimes \interpretation{\Gamma_1} \otimes \interpretation{\Gamma_2} \end{array} \\
	\\
	\& {\interpretation{\oc \Upsilon} \otimes \interpretation{\Gamma_2} \otimes\interpretation B} \&\& \begin{array}{c} \interpretation{\oc \Upsilon} \otimes \interpretation{\Gamma_2} \\\otimes\interpretation{\oc \Upsilon} \otimes \interpretation{\Gamma_1} \end{array} \\
	\&\&\&\& \begin{array}{c} \interpretation{B} \otimes \\{[}\interpretation B \to \interpretation{\oc \Upsilon} \otimes \interpretation{\Gamma_2} \otimes \interpretation B ] \end{array} \\
	\&\&\& \begin{array}{c} [\interpretation B \to \interpretation{\oc \Upsilon} \otimes \interpretation{\Gamma_2} \otimes \interpretation B ] \\ \otimes \interpretation B \end{array} \\
	\\
	{\interpretation{A}} \&\&\& {[\interpretation B \to \interpretation A] \otimes \interpretation B} \&\& {\interpretation B \otimes [\interpretation B \to \interpretation A]}
	\arrow[""{name=0, anchor=center, inner sep=0}, "{\overline{d}_{\Upsilon, \Gamma_1, \Gamma_2}}"{description}, from=1-1, to=1-6]
	\arrow["{d_\Upsilon \otimes id}"{description}, dashed, from=1-1, to=2-4]
	\arrow["{d_\Upsilon \otimes id}"{description}, dashed, from=1-1, to=3-4]
	\arrow["{\interpretation{(v/x)u}}"{description}, dotted, from=1-1, to=9-1]
	\arrow["{\interpretation{(\lambda x^B.u)\ v}}"{description}, dotted, from=1-1, to=9-1, rounded corners, to path={-- ([yshift=3ex]\tikztostart.north) -| ([xshift=455pt]\tikztostart.north) -- ([xshift=455pt,yshift=-3ex]\tikztotarget.south) -- ([yshift=-3ex]\tikztotarget.south) \tikztonodes -- (\tikztotarget.south)}]
	\arrow["\sigma"{description}, dashed, from=1-6, to=5-4]
	\arrow["{\interpretation v \otimes \eta_B}"{description}, dashed, from=1-6, to=6-5]
	\arrow["{\interpretation{v} \otimes \interpretation{\lambda x^B.u}}"{description}, from=1-6, to=9-6]
	\arrow["{id \otimes \sigma \otimes id}"{description}, dashed, from=2-4, to=1-6]
	\arrow["{id \otimes \sigma}"{description}, dashed, from=3-4, to=5-4]
	\arrow["{(2)}"{description, pos=0.3}, draw=none, from=5-2, to=1-1]
	\arrow["{\eta_B \otimes id}"{description}, dashed, from=5-2, to=7-4]
	\arrow["{\interpretation u}"{description}, dashed, from=5-2, to=9-1]
	\arrow["{(3)}"{description, pos=0.4}, shift left=3, curve={height=-24pt}, draw=none, from=5-4, to=1-6]
	\arrow["{id \otimes \interpretation v}"', curve={height=12pt}, dashed, from=5-4, to=5-2]
	\arrow["{\eta_B \otimes \interpretation v}"{description}, dashed, from=5-4, to=7-4]
	\arrow["{id \otimes [\interpretation B \to \interpretation u]}"{description}, dashed, from=6-5, to=9-6]
	\arrow["{\varepsilon_B}"{description}, curve={height=-30pt}, dashed, from=7-4, to=5-2]
	\arrow["{(5)}"{description}, shift left=3, curve={height=-6pt}, draw=none, from=7-4, to=5-2]
	\arrow["{(6)}"{description, pos=0.6}, shift left=3, curve={height=-30pt}, draw=none, from=7-4, to=5-4]
	\arrow["{(7)}"{description}, draw=none, from=7-4, to=6-5]
	\arrow["{(4)}"{description}, curve={height=12pt}, draw=none, from=7-4, to=9-1]
	\arrow["{[\interpretation B \to \interpretation u] \otimes id}"{description}, dashed, from=7-4, to=9-4]
	\arrow["{\varepsilon_A}"{description}, from=9-4, to=9-1]
	\arrow["{(8)}"{description}, shift left=3, curve={height=-30pt}, draw=none, from=9-6, to=1-6]
	\arrow["\sigma"{description}, from=9-6, to=9-4]
	\arrow["{(1)}"{description}, draw=none, from=2-4, to=0]
\end{tikzcd}}\]
(1) commutes by definition of $\overline d$, (2) commutes by Lemma~\ref{lem:linearsubstitution}, (3) commutes by Lemma~\ref{lem:comonoidgen}, (4) commutes by naturality of $\varepsilon$, (5) commutes by the adjuction axiom, (6) commutes by functoriality of $\otimes$, (7) commutes by naturality of $\sigma$ and (8) commutes by definition of $\interpretation{\lambda x^B.u}$ and functoriality of $\otimes$.

\item If $t = \elimtens(u \otimes v, x^B y^C.w)$ and $r = (u/x, v/y)w$, then $\Gamma = \Gamma_1, \Gamma_2, \Gamma_3$, $\Upsilon;\Gamma_1\vdash u:B$, $\Upsilon;\Gamma_2\vdash v:C$ and $\Upsilon;\Gamma_3,x^B,y^C \vdash w:A$.
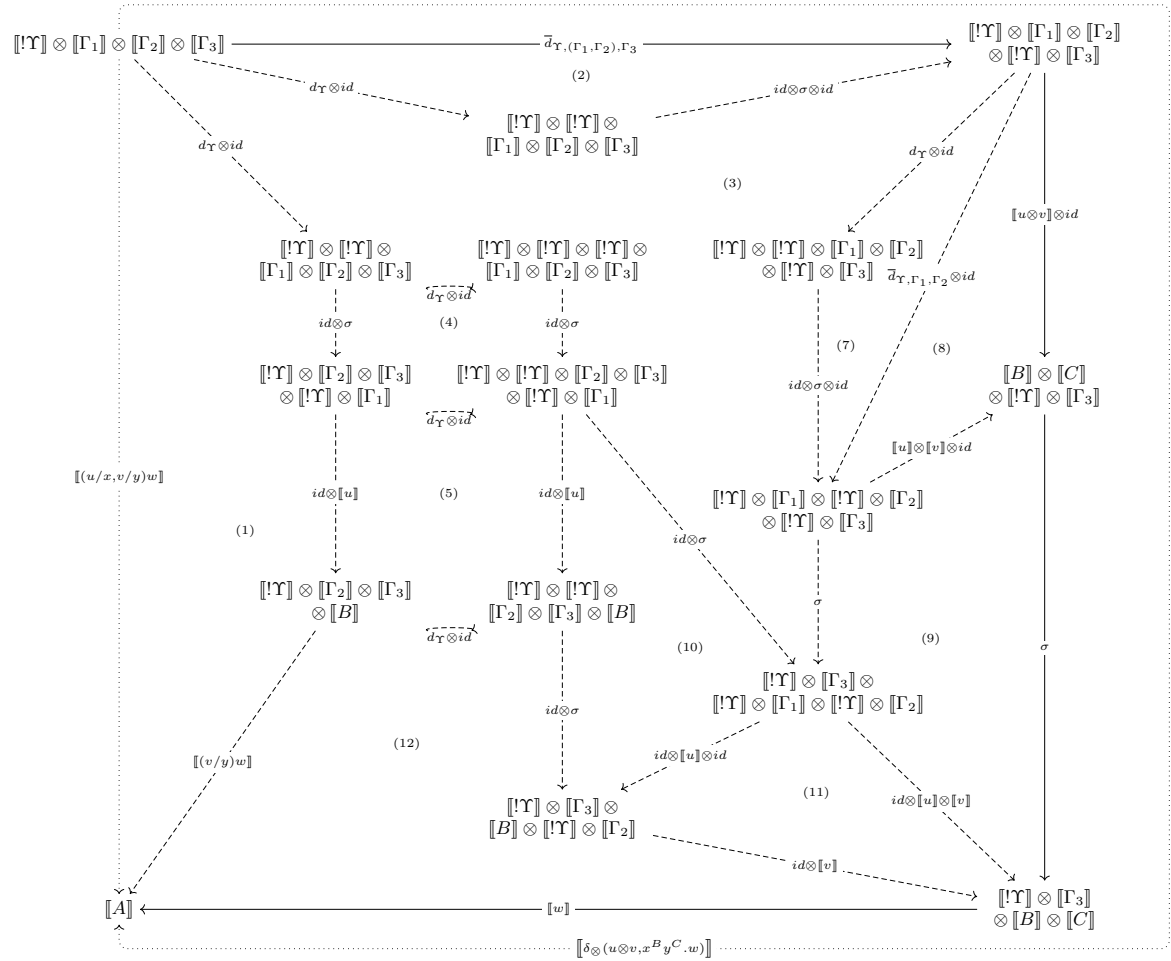
\[\resizebox{.85\textwidth}{!}{\begin{tikzcd}[ampersand replacement=\&,cramped,column sep=tiny]
	{\interpretation{\oc \Upsilon} \otimes \interpretation{\Gamma_1} \otimes \interpretation{\Gamma_2} \otimes \interpretation{\Gamma_3}} \&\&\&\& \begin{array}{c} \interpretation{\oc \Upsilon} \otimes \interpretation{\Gamma_1} \otimes \interpretation{\Gamma_2}\\ \otimes \interpretation{\oc \Upsilon} \otimes \interpretation{\Gamma_3} \end{array} \\
	\&\& \begin{array}{c} \interpretation{\oc \Upsilon} \otimes \interpretation{\oc \Upsilon} \otimes\\ \interpretation{\Gamma_1} \otimes \interpretation{\Gamma_2} \otimes \interpretation{\Gamma_3} \end{array} \\
	\\
	\& \begin{array}{c} \interpretation{\oc \Upsilon} \otimes \interpretation{\oc \Upsilon} \otimes\\ \interpretation{\Gamma_1} \otimes \interpretation{\Gamma_2} \otimes \interpretation{\Gamma_3} \end{array} \& \begin{array}{c} \interpretation{\oc \Upsilon} \otimes \interpretation{\oc \Upsilon} \otimes \interpretation{\oc \Upsilon} \otimes\\ \interpretation{\Gamma_1} \otimes \interpretation{\Gamma_2} \otimes \interpretation{\Gamma_3} \end{array} \& \begin{array}{c} \interpretation{\oc \Upsilon} \otimes \interpretation{\oc \Upsilon} \otimes \interpretation{\Gamma_1} \otimes \interpretation{\Gamma_2} \\\otimes \interpretation{\oc \Upsilon} \otimes \interpretation{\Gamma_3} \end{array} \\
	\\
	\& \begin{array}{c} \interpretation{\oc \Upsilon} \otimes \interpretation{\Gamma_2} \otimes \interpretation{\Gamma_3}\\ \otimes \interpretation{\oc \Upsilon} \otimes \interpretation{\Gamma_1} \end{array} \& \begin{array}{c} \interpretation{\oc \Upsilon} \otimes \interpretation{\oc \Upsilon} \otimes \interpretation{\Gamma_2} \otimes \interpretation{\Gamma_3}\\ \otimes \interpretation{\oc \Upsilon} \otimes \interpretation{\Gamma_1} \end{array} \&\& \begin{array}{c} \interpretation{B} \otimes \interpretation{C} \\\otimes \interpretation{\oc \Upsilon} \otimes \interpretation{\Gamma_3} \end{array} \\
	\\
	\&\&\& \begin{array}{c} \interpretation{\oc \Upsilon} \otimes \interpretation{\Gamma_1} \otimes \interpretation{\oc \Upsilon} \otimes \interpretation{\Gamma_2} \\\otimes \interpretation{\oc \Upsilon} \otimes \interpretation{\Gamma_3} \end{array} \\
	\& \begin{array}{c} \interpretation{\oc \Upsilon} \otimes \interpretation{\Gamma_2} \otimes \interpretation{\Gamma_3} \\\otimes \interpretation B \end{array} \& \begin{array}{c} \interpretation{\oc \Upsilon} \otimes \interpretation{\oc \Upsilon} \otimes \\\interpretation{\Gamma_2} \otimes \interpretation{\Gamma_3} \otimes \interpretation B \end{array} \\
	\&\&\& \begin{array}{c} \interpretation{\oc \Upsilon} \otimes \interpretation{\Gamma_3} \otimes \\\interpretation{\oc \Upsilon} \otimes \interpretation{\Gamma_1} \otimes \interpretation{\oc \Upsilon} \otimes \interpretation{\Gamma_2} \end{array} \\
	\\
	\&\& \begin{array}{c} \interpretation{\oc \Upsilon} \otimes \interpretation{\Gamma_3} \otimes\\ \interpretation B \otimes \interpretation{\oc \Upsilon} \otimes \interpretation{\Gamma_2} \end{array} \\
	{\interpretation A} \&\&\&\& \begin{array}{c} \interpretation{\oc \Upsilon} \otimes \interpretation{\Gamma_3}\\ \otimes \interpretation{B} \otimes \interpretation{C} \end{array}
	\arrow[""{name=0, anchor=center, inner sep=0}, "{\overline{d}_{\Upsilon,(\Gamma_1,\Gamma_2),\Gamma_3}}"{description}, from=1-1, to=1-5]
	\arrow["{d_\Upsilon \otimes id}"{description}, dashed, from=1-1, to=2-3]
	\arrow["{d_\Upsilon \otimes id}"{description}, dashed, from=1-1, to=4-2]
	\arrow["{\interpretation{(u/x,v/y)w}}"{description}, dotted, from=1-1, to=13-1]
	\arrow["{\interpretation{\elimtens(u \otimes v, x^B y^C.w)}}"{description}, dotted, from=1-1, to=13-1, rounded corners, to path={-- ([yshift=3ex]\tikztostart.north) -| ([xshift=510pt]\tikztostart.north) -- ([yshift=-3ex,xshift=510pt]\tikztotarget.south) -- ([yshift=-3ex]\tikztotarget.south) \tikztonodes -- (\tikztotarget.south)}]
	\arrow["{d_\Upsilon \otimes id}"{description}, dashed, from=1-5, to=4-4]
	\arrow[""{name=1, anchor=center, inner sep=0}, "{\interpretation{u \otimes v} \otimes id}"{description}, from=1-5, to=6-5]
	\arrow["{\overline{d}_{\Upsilon, \Gamma_1, \Gamma_2} \otimes id}"{description}, dashed, from=1-5, to=8-4]
	\arrow["{id \otimes \sigma \otimes id}"{description}, dashed, from=2-3, to=1-5]
	\arrow["{d_\Upsilon \otimes id}"', curve={height=12pt}, dashed, from=4-2, to=4-3]
	\arrow["{id \otimes \sigma}"{description}, dashed, from=4-2, to=6-2]
	\arrow["{(3)}"{description, pos=0.3}, draw=none, from=4-3, to=1-5]
	\arrow["{id \otimes \sigma}"{description}, dashed, from=4-3, to=6-3]
	\arrow["{id \otimes \sigma \otimes id}"{description}, dashed, from=4-4, to=8-4]
	\arrow["{(7)}"{description, pos=0.3}, curve={height=-18pt}, draw=none, from=4-4, to=8-4]
	\arrow["{(4)}"{description}, draw=none, from=6-2, to=4-3]
	\arrow["{d_\Upsilon \otimes id}"', curve={height=12pt}, dashed, from=6-2, to=6-3]
	\arrow["{id \otimes \interpretation{u}}"{description}, dashed, from=6-2, to=9-2]
	\arrow["{id \otimes \interpretation{u}}"{description}, dashed, from=6-3, to=9-3]
	\arrow["{id \otimes \sigma}"{description}, dashed, from=6-3, to=10-4]
	\arrow[""{name=2, anchor=center, inner sep=0}, "\sigma"{description}, from=6-5, to=13-5]
	\arrow["{\interpretation{u} \otimes \interpretation{v} \otimes id}"{description}, dashed, from=8-4, to=6-5]
	\arrow["\sigma"{description}, dashed, from=8-4, to=10-4]
	\arrow["{(5)}"{description}, draw=none, from=9-2, to=6-3]
	\arrow["{d_\Upsilon \otimes id}"', curve={height=12pt}, dashed, from=9-2, to=9-3]
	\arrow["{(12)}"{description}, curve={height=30pt}, draw=none, from=9-2, to=12-3]
	\arrow["{\interpretation{(v/y)w}}"{description}, dashed, from=9-2, to=13-1]
	\arrow["{id \otimes \sigma}"{description}, dashed, from=9-3, to=12-3]
	\arrow["{(10)}"{description}, draw=none, from=10-4, to=9-3]
	\arrow["{id \otimes \interpretation u \otimes id}"{description}, dashed, from=10-4, to=12-3]
	\arrow["{id \otimes \interpretation{u} \otimes \interpretation{v}}"{description}, dashed, from=10-4, to=13-5]
	\arrow[""{name=3, anchor=center, inner sep=0}, "{id \otimes \interpretation v}"{description}, dashed, from=12-3, to=13-5]
	\arrow["{(1)}"{description, pos=0.6}, draw=none, from=13-1, to=4-2]
	\arrow["{\interpretation w}"{description}, from=13-5, to=13-1]
	\arrow["{(2)}"{description}, draw=none, from=2-3, to=0]
	\arrow["{(8)}"{description}, draw=none, from=8-4, to=1]
	\arrow["{(9)}"{description}, curve={height=-18pt}, draw=none, from=10-4, to=2]
	\arrow["{(11)}"{description}, draw=none, from=10-4, to=3]
\end{tikzcd}}\]
(1) commutes by Lemma~\ref{lem:linearsubstitution}, (2) commutes by definition of $\overline d$, (3) commutes by Lemma~\ref{lem:comonoidgen}, (4) commutes by functoriality of $\otimes$, (5) commutes by definition of $\overline d$, (6) commutes by definition of $\interpretation{u \otimes v}$, (7) commutes by naturality of $\sigma$, (8) commutes by naturality of $\sigma$ and functoriality of $\otimes$, (9) commutes by functoriality of $\otimes$ and (10) commutes by Lemma~\ref{lem:linearsubstitution}.

\item If $t = \elimwith^1(\pair{u}{v}, x^B.w)$ and $r = (u/x)w$, then $\Gamma = \Gamma_1, \Gamma_2$, $\Upsilon;\Gamma_1 \vdash u:B$, $\Upsilon;\Gamma_1 \vdash v:C$ and $\Upsilon;\Gamma_2,x^B \vdash w:A$.
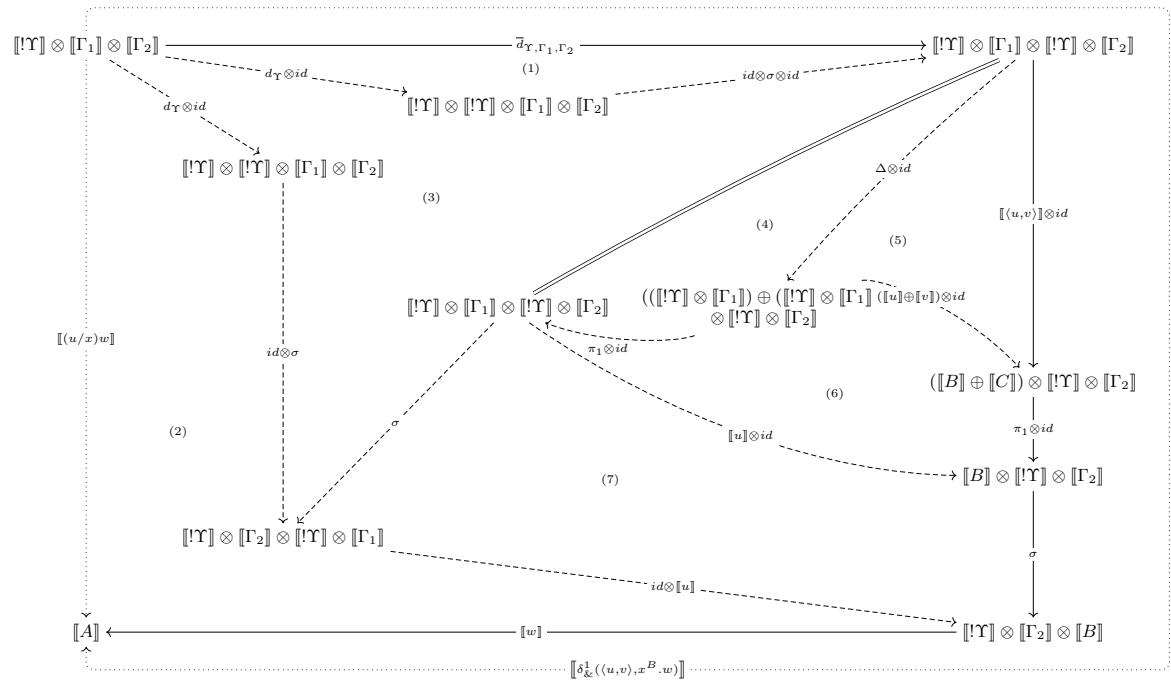
\[\resizebox{.85\textwidth}{!}{\begin{tikzcd}[ampersand replacement=\&,cramped,column sep=tiny]
	{\interpretation{\oc\Upsilon} \otimes \interpretation{\Gamma_1} \otimes \interpretation{\Gamma_2}} \&\&\&\&\& {\interpretation{\oc\Upsilon} \otimes \interpretation{\Gamma_1} \otimes \interpretation{\oc\Upsilon} \otimes \interpretation{\Gamma_2}} \\
	\&\& {\interpretation{\oc\Upsilon} \otimes \interpretation{\oc\Upsilon} \otimes \interpretation{\Gamma_1} \otimes \interpretation{\Gamma_2}} \\
	\& {\interpretation{\oc\Upsilon} \otimes \interpretation{\oc\Upsilon} \otimes \interpretation{\Gamma_1} \otimes \interpretation{\Gamma_2}} \\
	\\
	\\
	\&\& {\interpretation{\oc\Upsilon} \otimes \interpretation{\Gamma_1} \otimes \interpretation{\oc\Upsilon} \otimes \interpretation{\Gamma_2}} \& \begin{array}{c} ((\interpretation{\oc\Upsilon} \otimes \interpretation{\Gamma_1}) \oplus (\interpretation{\oc\Upsilon} \otimes \interpretation{\Gamma_1}))\\ \otimes \interpretation{\oc\Upsilon} \otimes \interpretation{\Gamma_2} \end{array} \\
	\&\&\&\&\& {(\interpretation B \oplus \interpretation C) \otimes \interpretation{\oc\Upsilon} \otimes \interpretation{\Gamma_2}} \\
	\\
	\&\&\&\&\& {\interpretation B \otimes \interpretation{\oc\Upsilon} \otimes \interpretation{\Gamma_2}} \\
	\& {\interpretation{\oc\Upsilon} \otimes \interpretation{\Gamma_2} \otimes \interpretation{\oc\Upsilon} \otimes \interpretation{\Gamma_1}} \\
	\\
	{\interpretation A} \&\&\&\&\& {\interpretation{\oc\Upsilon} \otimes \interpretation{\Gamma_2} \otimes \interpretation B }
	\arrow[""{name=0, anchor=center, inner sep=0}, "{\overline{d}_{\Upsilon, \Gamma_1, \Gamma_2}}"{description}, from=1-1, to=1-6]
	\arrow["{d_\Upsilon \otimes id}"{description}, dashed, from=1-1, to=2-3]
	\arrow["{d_\Upsilon \otimes id}"{description}, dashed, from=1-1, to=3-2]
	\arrow[""{name=1, anchor=center, inner sep=0}, "{\interpretation{(u/x)w}}"{description}, dotted, from=1-1, to=12-1]
	\arrow["{\interpretation{\elimwith^1(\pair{u}{v}, x^B.w)}}"{description}, dotted, from=1-1, to=12-1, rounded corners, to path={-- ([yshift=3ex]\tikztostart.north) -| ([xshift=550pt]\tikztostart.north) -- ([yshift=-3ex,xshift=550pt]\tikztotarget.south) -- ([yshift=-3ex]\tikztotarget.south) \tikztonodes -- (\tikztotarget.south)}]
	\arrow["{\Delta \otimes id}"{description}, curve={height=6pt}, dashed, from=1-6, to=6-4]
	\arrow[""{name=2, anchor=center, inner sep=0}, "{\interpretation{\langle u,v \rangle} \otimes id}"{description}, from=1-6, to=7-6]
	\arrow["{id \otimes \sigma \otimes id}"{description}, dashed, from=2-3, to=1-6]
	\arrow["{id \otimes \sigma}"{description}, dashed, from=3-2, to=10-2]
	\arrow[""{name=3, anchor=center, inner sep=0}, curve={height=-6pt}, equals, from=6-3, to=1-6]
	\arrow["{(3)}"{description}, curve={height=30pt}, draw=none, from=6-3, to=3-2]
	\arrow["{\interpretation{u} \otimes id}"{description}, curve={height=30pt}, dashed, from=6-3, to=9-6]
	\arrow["\sigma"{description}, dashed, from=6-3, to=10-2]
	\arrow["{\pi_1 \otimes id}", curve={height=-18pt}, dashed, from=6-4, to=6-3]
	\arrow["{(\interpretation u \oplus \interpretation v) \otimes id}"{description}, curve={height=-30pt}, dashed, from=6-4, to=7-6]
	\arrow["{\pi_1 \otimes id}"{description}, from=7-6, to=9-6]
	\arrow["{(6)}"{description, pos=0.7}, curve={height=-24pt}, draw=none, from=9-6, to=6-4]
	\arrow["\sigma"{description}, from=9-6, to=12-6]
	\arrow[""{name=4, anchor=center, inner sep=0}, "{id \otimes \interpretation u}"{description}, dashed, from=10-2, to=12-6]
	\arrow["{\interpretation w}"{description}, from=12-6, to=12-1]
	\arrow["{(1)}"{description}, draw=none, from=2-3, to=0]
	\arrow["{(7)}"{description, pos=0.6}, draw=none, from=6-3, to=4]
	\arrow["{(4)}"{description}, draw=none, from=6-4, to=3]
	\arrow["{(5)}"{description}, curve={height=-12pt}, draw=none, from=6-4, to=2]
	\arrow["{(2)}"{description}, draw=none, from=10-2, to=1]
\end{tikzcd}}\]
(1) commutes by definition of $\overline d$, (2) commutes by Lemma~\ref{lem:linearsubstitution}, (3) commutes by Lemma~\ref{lem:comonoidgen}, (4) commutes by definition of $\Delta$, (5) commutes by definition of $\interpretation{\pair{u}{v}}$, (6) commutes by definition of $\interpretation{u} \oplus \interpretation{v}$ and (7) commutes by naturality of $\sigma$.

\item If $t = \elimwith^2(\pair{u}{v}, y^C.w)$ and $r = (v/y)w$, this case is similar to the previous one.

\item If $t = \elimplus(\inl(u), x^B.v, y^C.w)$ and $r = (u/x)v$, then $\Gamma = \Gamma_1,\Gamma_2$, $\Upsilon;\Gamma_1 \vdash u:B$, $\Upsilon;\Gamma_2, x^B \vdash v:A$ and $\Upsilon;\Gamma_2,y^C \vdash w:A$.
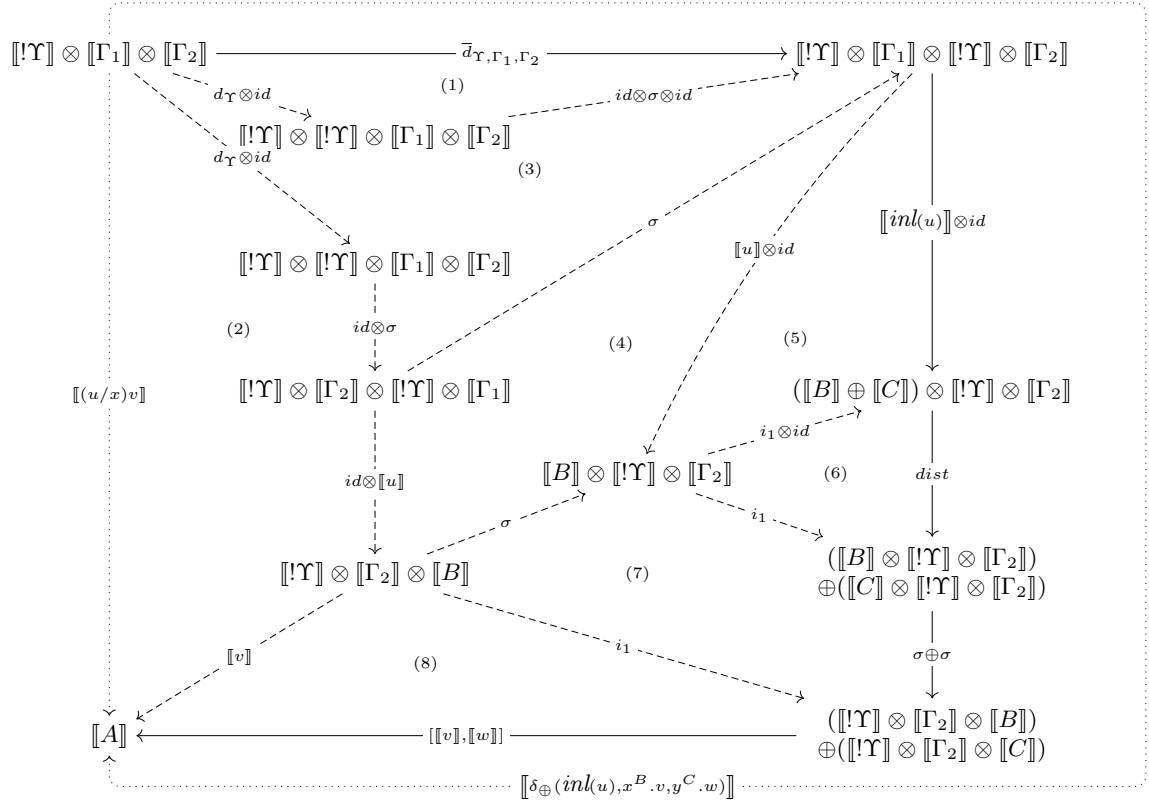
\[\begin{tikzcd}[ampersand replacement=\&,cramped,column sep=tiny]
	{\interpretation{\oc \Upsilon} \otimes \interpretation{\Gamma_1} \otimes \interpretation{\Gamma_2}} \&\&\&\&\& {\interpretation{\oc \Upsilon} \otimes \interpretation{\Gamma_1} \otimes \interpretation{\oc \Upsilon} \otimes \interpretation{\Gamma_2}} \\
	\& {\interpretation{\oc \Upsilon} \otimes \interpretation{\oc \Upsilon} \otimes \interpretation{\Gamma_1} \otimes \interpretation{\Gamma_2}} \\
	\\
	\& {\interpretation{\oc \Upsilon} \otimes \interpretation{\oc \Upsilon} \otimes \interpretation{\Gamma_1} \otimes \interpretation{\Gamma_2}} \\
	\\
	\& {\interpretation{\oc\Upsilon} \otimes \interpretation{\Gamma_2} \otimes \interpretation{\oc\Upsilon} \otimes \interpretation{\Gamma_1}} \&\&\&\& {(\interpretation B \oplus \interpretation C) \otimes \interpretation{\oc \Upsilon} \otimes \interpretation{\Gamma_2}} \\
	\&\& {\interpretation B \otimes \interpretation{\oc \Upsilon} \otimes \interpretation{\Gamma_2}} \\
	\& {\interpretation{\oc\Upsilon} \otimes \interpretation{\Gamma_2} \otimes \interpretation B} \&\&\&\& \begin{array}{c} (\interpretation B \otimes \interpretation{\oc \Upsilon} \otimes \interpretation{\Gamma_2})\\ \oplus (\interpretation C \otimes \interpretation{\oc \Upsilon} \otimes \interpretation{\Gamma_2}) \end{array} \\
	\\
	{\interpretation A} \&\&\&\&\& \begin{array}{c} (\interpretation{\oc \Upsilon} \otimes \interpretation{\Gamma_2} \otimes \interpretation B)\\ \oplus (\interpretation{\oc \Upsilon} \otimes \interpretation{\Gamma_2} \otimes \interpretation C) \end{array}
	\arrow[""{name=0, anchor=center, inner sep=0}, "{\overline{d}_{\Upsilon,\Gamma_1,\Gamma_2}}"{description}, from=1-1, to=1-6]
	\arrow["{d_\Upsilon \otimes id}"{description}, dashed, from=1-1, to=2-2]
	\arrow["{d_\Upsilon \otimes id}"{description}, dashed, from=1-1, to=4-2]
	\arrow[""{name=1, anchor=center, inner sep=0}, "{\interpretation{(u/x)v}}"{description}, dotted, from=1-1, to=10-1]
	\arrow["{\interpretation{\elimplus(\inl(u),x^B.v,y^C.w)}}"{description}, dotted, from=1-1, to=10-1, rounded corners, to path={-- ([yshift=3ex]\tikztostart.north) -| ([xshift=390pt]\tikztostart.north) -- ([yshift=-3ex,xshift=390pt]\tikztotarget.south) -- ([yshift=-3ex]\tikztotarget.south) \tikztonodes -- (\tikztotarget.south)}]
	\arrow[""{name=2, anchor=center, inner sep=0}, "{\interpretation{\inl(u)} \otimes id}"{description}, from=1-6, to=6-6]
	\arrow["{\interpretation u \otimes id}"{description}, curve={height=12pt}, dashed, from=1-6, to=7-3]
	\arrow[""{name=3, anchor=center, inner sep=0}, "{id \otimes \sigma \otimes id}"{description}, dashed, from=2-2, to=1-6]
	\arrow["{id \otimes \sigma}"{description}, dashed, from=4-2, to=6-2]
	\arrow[""{name=4, anchor=center, inner sep=0}, "\sigma"{description}, dashed, from=6-2, to=1-6]
	\arrow["{id \otimes \interpretation{u}}"{description}, dashed, from=6-2, to=8-2]
	\arrow[""{name=5, anchor=center, inner sep=0}, "dist"{description}, from=6-6, to=8-6]
	\arrow["{i_1 \otimes id}"{description}, dashed, from=7-3, to=6-6]
	\arrow["{i_1}"{description}, dashed, from=7-3, to=8-6]
	\arrow["\sigma"{description}, dashed, from=8-2, to=7-3]
	\arrow["{(7)}"{description}, draw=none, from=8-2, to=8-6]
	\arrow["{\interpretation v}"{description}, dashed, from=8-2, to=10-1]
	\arrow["{i_1}"{description}, dashed, from=8-2, to=10-6]
	\arrow["{\sigma \oplus \sigma}"{description}, from=8-6, to=10-6]
	\arrow[""{name=6, anchor=center, inner sep=0}, "{[\interpretation{v},\interpretation w]}"{description}, from=10-6, to=10-1]
	\arrow["{(1)}"{description}, draw=none, from=2-2, to=0]
	\arrow["{(3)}"{description}, draw=none, from=4-2, to=3]
	\arrow["{(2)}"{description}, curve={height=30pt}, draw=none, from=6-2, to=1]
	\arrow["{(4)}"{description}, curve={height=-12pt}, draw=none, from=7-3, to=4]
	\arrow["{(5)}"{description}, draw=none, from=7-3, to=2]
	\arrow["{(6)}"{description}, draw=none, from=7-3, to=5]
	\arrow["{(8)}"{description}, draw=none, from=8-2, to=6]
\end{tikzcd}\]
(1) commutes by definition of $\overline d$, (2) commutes by Lemma~\ref{lem:linearsubstitution}, (3) commutes by Lemma~\ref{lem:comonoidgen}, (4) commutes by naturality of $\sigma$, (5) commutes by definition of $\interpretation{\inl(u)}$, (6) commutes by definition of $dist$, (7) commutes by definition of $\sigma  \oplus \sigma$ and (8) commutes by definition of $[\interpretation v, \interpretation w]$.

\item If $t = \elimplus(\inr(u), x^B.v, y^C.w)$ and $r = (u/y)w$, this case is similar to the previous one.

\item If $t = \elimbang(\bang u, x^B.v)$ and $r = (u/x)v$, then $\Upsilon;\varnothing \vdash u:B$ and $\Upsilon, x^B;\Gamma \vdash v:A$.
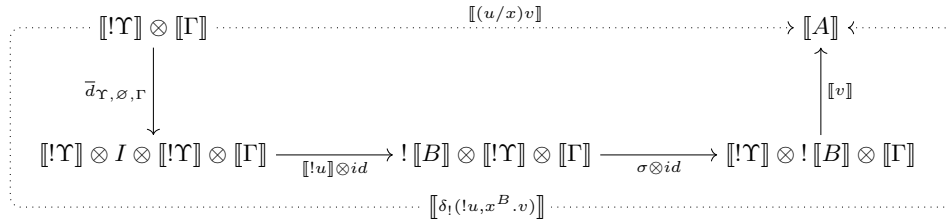
\[\begin{tikzcd}[ampersand replacement=\&,cramped]
	{\interpretation{\oc \Upsilon} \otimes \interpretation{\Gamma}} \&\&\&\& {\interpretation{A}} \\
	\\
	{\interpretation{\oc \Upsilon} \otimes I\otimes \interpretation{\oc \Upsilon} \otimes \interpretation{\Gamma}} \&\& {\oc\interpretation{B} \otimes \interpretation{\oc \Upsilon} \otimes \interpretation{\Gamma}} \&\& {\interpretation{\oc \Upsilon} \otimes \oc\interpretation{B} \otimes \interpretation{\Gamma}}
	\arrow["{\interpretation{(u/x)v}}", dotted, from=1-1, to=1-5]
	\arrow["{\interpretation{\elimbang(\oc u,x^B.v)}}"{description}, dotted, from=1-1, to=1-5, rounded corners, to path = {-- ([xshift=-30pt]\tikztostart.west) |- ([yshift=-60pt]\tikztostart.south) -- ([yshift=-60pt]\tikztotarget.south) \tikztonodes -| ([xshift=40pt]\tikztotarget.east) -- (\tikztotarget.east)}]
	\arrow["{\overline{d}_{\Upsilon,\varnothing,\Gamma}}"', from=1-1, to=3-1]
	\arrow["{\interpretation{\oc u} \otimes id}"', from=3-1, to=3-3]
	\arrow["{\sigma \otimes id}"', from=3-3, to=3-5]
	\arrow["{\interpretation{v}}"', from=3-5, to=1-5]
\end{tikzcd}\]

	The diagram commutes by Lemma~\ref{lem:nonlinearsubstitution}.
    \item If $t = a.\star \plus b.\star$ and $r = (a + b).\star$, then $A = \one$ and $\Gamma = \varnothing$.
\[\begin{tikzcd}[ampersand replacement=\&,cramped]
	{\interpretation{\oc\Upsilon} \otimes I} \&\& {(\interpretation{\oc\Upsilon} \otimes I) \oplus (\interpretation{\oc\Upsilon} \otimes I)} \&\&\& {I \oplus I} \&\& I \\
	\&\&\& {(1)} \&\&\& {(4)} \\
	\& {(2)} \& {\interpretation{\oc\Upsilon} \oplus \interpretation{\oc\Upsilon}} \&\&\& {I \oplus I} \\
	\&\&\& {(3)} \\
	{\interpretation{\oc\Upsilon}} \&\&\&\&\&\&\& I
	\arrow["\Delta"{description}, from=1-1, to=1-3]
	\arrow["{\interpretation{a.\star \plus b.\star}}"{description}, dotted, from=1-1, to=1-8, rounded corners, to path = {-- ([yshift=3ex]\tikztostart.north) -- ([yshift=3ex]\tikztotarget.north) \tikztonodes -- (\tikztotarget)}]
	\arrow["{\interpretation{(a+b).\star}}"{description}, dotted, from=1-1, to=1-8, rounded corners, to path = {(\tikztostart.west) -- ([xshift=-3ex]\tikztostart.west) -- ([xshift=-3ex,yshift=-140pt]\tikztostart.west) -- ([xshift=3ex,yshift=-140pt]\tikztotarget.east) \tikztonodes |- (\tikztotarget.east)}]
	\arrow["\rho"{description}, from=1-1, to=5-1]
	\arrow["{\interpretation{a.\star} \oplus \interpretation{b.\star}}"{description}, from=1-3, to=1-6]
	\arrow["{\rho \oplus \rho}"{description}, dashed, from=1-3, to=3-3]
	\arrow["\nabla"{description}, from=1-6, to=1-8]
	\arrow["{e_\Upsilon \oplus e_\Upsilon}"{description}, dashed, from=3-3, to=3-6]
	\arrow["{\mono{a} \oplus \mono{b}}"{description}, dashed, from=3-6, to=1-6]
	\arrow["\Delta"{description}, dashed, from=5-1, to=3-3]
	\arrow["{e_\Upsilon}"{description}, from=5-1, to=5-8]
	\arrow["{\mono{a+b}}"{description}, from=5-8, to=1-8]
	\arrow["\Delta"{description}, dashed, from=5-8, to=3-6]
\end{tikzcd}\]

(1) commutes by definition of $\interpretation{a.\star}$ and $\interpretation{b.\star}$, (2) and (3) commute by naturality of $\Delta$, and (4) commutes because $\nabla \circ (\mono{a} \oplus \mono{b}) \circ \Delta = \mono{a} + \mono{b} = \mono{a+b}$, since $\mono{\cdot}$ is a semiring homomorphism.

    \item If $t = (\lambda x^B.u) \plus (\lambda x^B.v)$ and $r =\lambda x^B.(u \plus v)$, then $A = B \multimap C$, $\Upsilon;\Gamma, x^B \vdash u:C$ and $\Upsilon;\Gamma, x^B \vdash v:C$.
\[\resizebox{.85\textwidth}{!}{\begin{tikzcd}[ampersand replacement=\&,cramped,column sep=0pt]
	{\interpretation{\oc\Upsilon} \otimes \interpretation{\Gamma}} \&\& {(\interpretation{\oc\Upsilon} \otimes \interpretation{\Gamma}) \oplus (\interpretation{\oc\Upsilon} \otimes \interpretation{\Gamma})} \\
	\& \begin{array}{c} [\interpretation B \to \interpretation{\oc\Upsilon} \otimes \interpretation{\Gamma} \otimes \interpretation{B}] \\ \oplus \\ {[}\interpretation B \to \interpretation{\oc\Upsilon} \otimes \interpretation{\Gamma} \otimes \interpretation{B}] \end{array} \\
	\\
	\& {[\interpretation B \to (\interpretation{\oc\Upsilon} \otimes \interpretation{\Gamma} \otimes \interpretation{B}) \oplus (\interpretation{\oc\Upsilon} \otimes \interpretation{\Gamma} \otimes \interpretation{B})]} \\
	{[\interpretation B \to \interpretation{\oc\Upsilon} \otimes \interpretation{\Gamma} \otimes \interpretation{B}]} \\
	\& {[\interpretation B \to \interpretation{C} \oplus \interpretation{C}]} \\
	\\
	{[\interpretation B \to \interpretation C]} \&\& {[\interpretation B \to \interpretation C] \oplus [\interpretation B \to \interpretation C]}
	\arrow["\Delta"{description}, from=1-1, to=1-3]
	\arrow["{(1)}"{description}, draw=none, from=1-1, to=2-2]
	\arrow["{\eta_B}"{description}, from=1-1, to=5-1]
	\arrow["{\interpretation{\lambda x^B.u \plus \lambda x^B.v}}"{description}, dotted, from=1-1, to=8-1, rounded corners, to path={-- ([yshift=3ex]\tikztostart.north) -- ([yshift=3ex, xshift=380pt]\tikztostart.north) \tikztonodes |- ([yshift=-3ex]\tikztotarget.south) -- (\tikztotarget)}]
	\arrow["{\interpretation{\lambda x^B.u \plus v}}"{description}, dotted, from=1-1, to=8-1, rounded corners, to path={-- ([xshift=-8ex]\tikztostart.west) -- ([xshift=-8ex,yshift=-200pt]\tikztostart.west) \tikztonodes |- (\tikztotarget.west)}]
	\arrow["{\eta_B \oplus \eta_B}"{description}, dashed, from=1-3, to=2-2]
	\arrow["{\interpretation{\lambda x^B.u} \oplus \interpretation{\lambda x^B.v}}"{description}, from=1-3, to=8-3]
	\arrow[""{name=0, anchor=center, inner sep=0}, "{[\interpretation B \to \interpretation u] \oplus [\interpretation B \to \interpretation v]}"{description}, curve={height=-24pt}, dashed, from=2-2, to=8-3]
	\arrow["\gamma"{description}, dashed, from=4-2, to=2-2]
	\arrow["{[\interpretation B \to \interpretation u \oplus \interpretation v]}"{description}, dashed, from=4-2, to=6-2]
	\arrow[""{name=1, anchor=center, inner sep=0}, "\Delta"{description}, curve={height=-18pt}, dashed, from=5-1, to=2-2]
	\arrow["{[\interpretation B \to \Delta]}"{description}, dashed, from=5-1, to=4-2]
	\arrow["{(4)}"{description}, draw=none, from=5-1, to=6-2]
	\arrow["{[\interpretation B \to \interpretation{u \plus v}]}"{description}, from=5-1, to=8-1]
	\arrow["{[\interpretation B \to \nabla]}"{description}, dashed, from=6-2, to=8-1]
	\arrow["\gamma"{description}, dashed, from=6-2, to=8-3]
	\arrow[""{name=2, anchor=center, inner sep=0}, "\nabla"{description}, from=8-3, to=8-1]
	\arrow["{(2)}"{description}, draw=none, from=1-3, to=0]
	\arrow["{(3)}"{description}, draw=none, from=4-2, to=1]
	\arrow["{(6)}"{description}, draw=none, from=6-2, to=2]
	\arrow["{(5)}"{description}, draw=none, from=6-2, to=0]
\end{tikzcd}}\]

(1) commutes by naturality of $\Delta$, (2) commutes by definition of $\interpretation{\lambda x^B.u}$ and $\interpretation{\lambda x^B.v}$, (3) commutes by Corollary~\ref{cor:deltasemiadditive}, (4) commutes by definition of $\interpretation{u \plus v}$, (5) commutes by naturality of $\gamma$ (Lemma~\ref{lem:semiadditiveiso}) and (6) commutes by Corollary~\ref{cor:nablasemiadditive}.

    \item If $t = \elimtens(u \plus v,x^B y^C.w)$ and $r = \elimtens(u, x^B y^C.w) \plus \elimtens(v,x^B y^C.w)$, then $\Gamma = \Gamma_1,\Gamma_2$, $\Upsilon;\Gamma_1 \vdash u:B \otimes C$, $\Upsilon;\Gamma_1 \vdash v:B \otimes C$ and $\Upsilon;\Gamma_2,x^B,y^C \vdash w:A$.

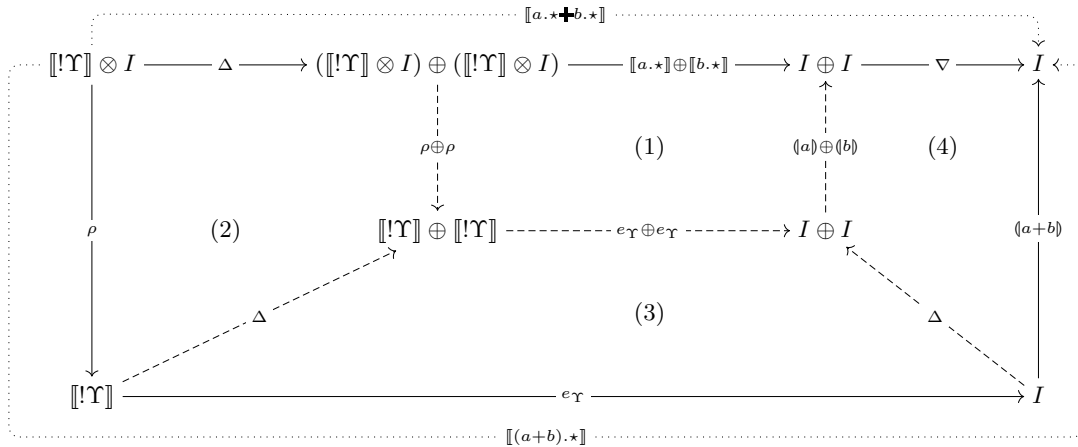
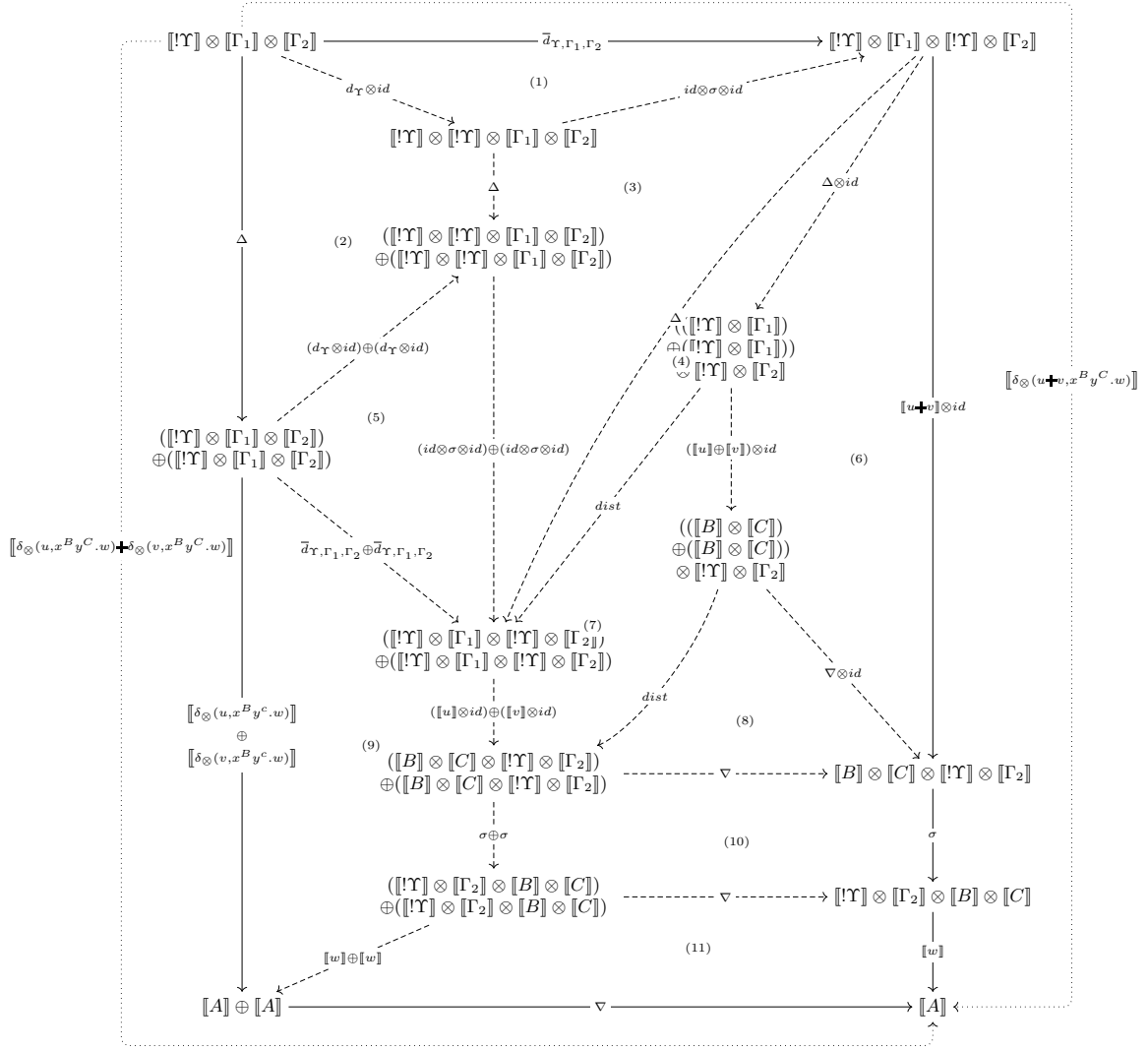
\[\resizebox{.85\textwidth}{!}{\begin{tikzcd}[ampersand replacement=\&,cramped,column sep=tiny]
	{\interpretation{\oc\Upsilon} \otimes \interpretation{\Gamma_1} \otimes \interpretation{\Gamma_2}} \&\&\&\& {\interpretation{\oc\Upsilon} \otimes \interpretation{\Gamma_1} \otimes \interpretation{\oc\Upsilon} \otimes \interpretation{\Gamma_2}} \\
	\\
	\& {\interpretation{\oc\Upsilon} \otimes \interpretation{\oc\Upsilon} \otimes \interpretation{\Gamma_1} \otimes \interpretation{\Gamma_2}} \\
	\\
	\& \begin{array}{c} (\interpretation{\oc\Upsilon} \otimes \interpretation{\oc\Upsilon} \otimes \interpretation{\Gamma_1} \otimes \interpretation{\Gamma_2}) \\\oplus (\interpretation{\oc\Upsilon} \otimes \interpretation{\oc\Upsilon} \otimes \interpretation{\Gamma_1} \otimes \interpretation{\Gamma_2}) \end{array} \\
	\&\&\& \begin{array}{c} ((\interpretation{\oc\Upsilon} \otimes \interpretation{\Gamma_1} )\\\oplus (\interpretation{\oc\Upsilon} \otimes \interpretation{\Gamma_1}))\\ \otimes \interpretation{\oc\Upsilon} \otimes \interpretation{\Gamma_2} \end{array} \\
	\begin{array}{c} (\interpretation{\oc\Upsilon} \otimes \interpretation{\Gamma_1} \otimes \interpretation{\Gamma_2})\\ \oplus (\interpretation{\oc\Upsilon} \otimes \interpretation{\Gamma_1} \otimes \interpretation{\Gamma_2}) \end{array} \\
	\&\&\& \begin{array}{c} ((\interpretation B \otimes \interpretation C)\\\oplus (\interpretation B \otimes \interpretation C))\\ \otimes \interpretation{\oc\Upsilon} \otimes \interpretation{\Gamma_2} \end{array} \\
	\& \begin{array}{c} (\interpretation{\oc\Upsilon} \otimes \interpretation{\Gamma_1} \otimes \interpretation{\oc\Upsilon} \otimes \interpretation{\Gamma_2})\\\oplus (\interpretation{\oc\Upsilon} \otimes \interpretation{\Gamma_1} \otimes \interpretation{\oc\Upsilon} \otimes \interpretation{\Gamma_2}) \end{array} \\
	\\
	\& \begin{array}{c} (\interpretation B \otimes \interpretation C \otimes \interpretation{\oc\Upsilon} \otimes \interpretation{\Gamma_2})\\\oplus (\interpretation B \otimes \interpretation C \otimes \interpretation{\oc\Upsilon} \otimes \interpretation{\Gamma_2}) \end{array} \&\&\& {\interpretation B \otimes \interpretation C \otimes \interpretation{\oc\Upsilon} \otimes \interpretation{\Gamma_2}} \\
	\\
	\& \begin{array}{c} (\interpretation{\oc\Upsilon} \otimes \interpretation{\Gamma_2} \otimes \interpretation B \otimes \interpretation C) \\\oplus (\interpretation{\oc\Upsilon} \otimes \interpretation{\Gamma_2} \otimes \interpretation B \otimes \interpretation C) \end{array} \&\&\& {\interpretation{\oc\Upsilon} \otimes \interpretation{\Gamma_2} \otimes \interpretation B \otimes \interpretation C} \\
	\\
	{\interpretation A \oplus \interpretation A} \&\&\&\& {\interpretation A}
	\arrow[""{name=0, anchor=center, inner sep=0}, "{\overline{d}_{\Upsilon, \Gamma_1, \Gamma_2}}"{description}, from=1-1, to=1-5]
	\arrow[""{name=1, anchor=center, inner sep=0}, "{d_\Upsilon \otimes id}"{description}, dashed, from=1-1, to=3-2]
	\arrow["\Delta"{description}, from=1-1, to=7-1]
	\arrow["{\interpretation{\elimtens(u,x^B y^C.w) \plus \elimtens(v,x^B y^C.w)}}"{description}, dotted, from=1-1, to=15-5, rounded corners, to path={-- ([xshift=-5ex]\tikztostart.west) |- ([xshift=-5ex, yshift=-250pt]\tikztostart.west) \tikztonodes |- ([yshift=-3ex]\tikztotarget.south) -- (\tikztotarget.south)}]
	\arrow["{\interpretation{\elimtens(u \plus v,x^B y^C.w)}}"{description}, dotted, from=1-1, to=15-5, rounded corners, to path={-- ([yshift=3ex]\tikztostart.north) -| ([xshift=410pt]\tikztostart.north) -- ([xshift=410pt,yshift=-350pt]\tikztostart.north) \tikztonodes |- (\tikztotarget.east)}]
	\arrow["{\Delta \otimes id}"{description}, dashed, from=1-5, to=6-4]
	\arrow[""{name=2, anchor=center, inner sep=0}, "\Delta"{description}, curve={height=30pt}, dashed, from=1-5, to=9-2]
	\arrow[""{name=3, anchor=center, inner sep=0}, "{\interpretation{u \plus v} \otimes id}"{description}, from=1-5, to=11-5]
	\arrow[""{name=4, anchor=center, inner sep=0}, "{id \otimes \sigma \otimes id}"{description}, dashed, from=3-2, to=1-5]
	\arrow["\Delta"{description}, dashed, from=3-2, to=5-2]
	\arrow[""{name=5, anchor=center, inner sep=0}, "{(id \otimes \sigma \otimes id) \oplus (id \otimes \sigma \otimes id)}"{description}, dashed, from=5-2, to=9-2]
	\arrow["{(\interpretation u \oplus \interpretation v) \otimes id}"{description}, dashed, from=6-4, to=8-4]
	\arrow["dist"{description}, dashed, from=6-4, to=9-2]
	\arrow["{(d_\Upsilon \otimes id) \oplus (d_\Upsilon \otimes id)}"{description}, dashed, from=7-1, to=5-2]
	\arrow["{\overline{d}_{\Upsilon, \Gamma_1, \Gamma_2} \oplus \overline{d}_{\Upsilon, \Gamma_1, \Gamma_2}}"{description}, dashed, from=7-1, to=9-2]
	\arrow["\begin{array}{c} \interpretation{\elimtens(u, x^B y^c.w)} \\ \oplus \\ \interpretation{\elimtens(v, x^B y^c.w)} \end{array}"{description}, from=7-1, to=15-1]
	\arrow[""{name=6, anchor=center, inner sep=0}, "dist"{description}, curve={height=-30pt}, dashed, from=8-4, to=11-2]
	\arrow["{\nabla \otimes id}"{description}, dashed, from=8-4, to=11-5]
	\arrow["{(\interpretation u \otimes id) \oplus (\interpretation v \otimes id)}"{description}, dashed, from=9-2, to=11-2]
	\arrow["{(9)}"{description, pos=0.3}, curve={height=30pt}, draw=none, from=9-2, to=15-1]
	\arrow["\nabla"{description}, dashed, from=11-2, to=11-5]
	\arrow["{\sigma \oplus \sigma}"{description}, dashed, from=11-2, to=13-2]
	\arrow["{(10)}"{description}, draw=none, from=11-2, to=13-5]
	\arrow["\sigma"{description}, from=11-5, to=13-5]
	\arrow["\nabla"{description}, dashed, from=13-2, to=13-5]
	\arrow["{\interpretation w \oplus \interpretation w}"{description}, dashed, from=13-2, to=15-1]
	\arrow["{(11)}"{description, pos=0.3}, draw=none, from=13-2, to=15-5]
	\arrow["{\interpretation w}"{description}, from=13-5, to=15-5]
	\arrow["\nabla"{description}, from=15-1, to=15-5]
	\arrow["{(1)}"{description}, draw=none, from=3-2, to=0]
	\arrow["{(3)}"{description}, curve={height=18pt}, draw=none, from=5-2, to=4]
	\arrow["{(4)}"{description}, curve={height=-18pt}, draw=none, from=6-4, to=2]
	\arrow["{(2)}"{description, pos=0.6}, curve={height=18pt}, draw=none, from=7-1, to=1]
	\arrow["{(5)}"{description}, curve={height=-18pt}, draw=none, from=7-1, to=5]
	\arrow["{(6)}"{description}, draw=none, from=8-4, to=3]
	\arrow["{(7)}"{description}, curve={height=-30pt}, draw=none, from=9-2, to=6]
	\arrow["{(8)}"{description, pos=0.6}, draw=none, from=11-5, to=6]
\end{tikzcd}}\]

(1) commutes by definition of $\overline d$, (2) and (3) commute by naturality of $\Delta$, (4) commutes by Corollary~\ref{cor:deltasemiadditive}, (5) commutes by definition of $\overline d$, (6) commutes by definition of $\interpretation{u \plus v}$, (7) commutes by naturality of $dist$ (Lemma~\ref{lem:semiadditiveiso}), (8) commutes by Corollary~\ref{cor:nablasemiadditive}, (9) commutes by definition of $\interpretation{\elimtens(u,x^B y^C.w)}$ and $\interpretation{\elimtens(v,x^B y^C.w)}$, (10) and (11) commute by naturality of $\nabla$.

    \item If $t = \langle \rangle \plus \langle \rangle$ and $r = \langle \rangle$, then $A = \top$.
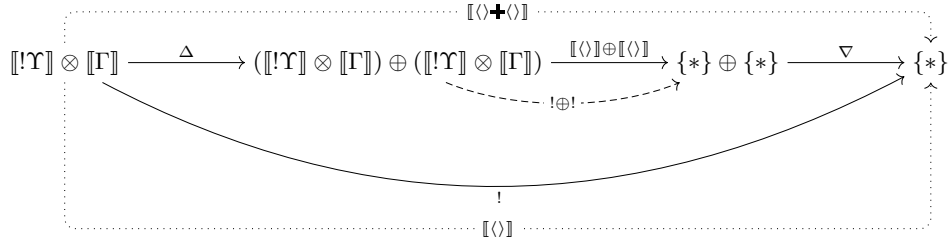
\[\begin{tikzcd}[cramped]
	{\interpretation{\oc\Upsilon}\otimes\interpretation{\Gamma}} && {(\interpretation{\oc\Upsilon}\otimes\interpretation{\Gamma}) \oplus (\interpretation{\oc\Upsilon}\otimes\interpretation{\Gamma})} && {\{\ast\} \oplus \{\ast\}} && {\{\ast\}}
	\arrow["\Delta", from=1-1, to=1-3]
	\arrow["\oc"', curve={height=60pt}, from=1-1, to=1-7]
	\arrow["{\interpretation{\langle\rangle \plus \langle\rangle}}"{description}, rounded corners, dotted, from=1-1, to=1-7, to path={-- ([yshift=3ex]\tikztostart.north) -- ([yshift=3ex]\tikztotarget.north) \tikztonodes -- (\tikztotarget)}]
	\arrow["{\interpretation{\langle\rangle}}"{description}, rounded corners, dotted, from=1-1, to=1-7, to path={-- ([yshift=-55pt]\tikztostart.south) -- ([yshift=-55pt]\tikztotarget.south) \tikztonodes -- (\tikztotarget)}]
	\arrow["{\interpretation{\langle\rangle} \oplus \interpretation{\langle\rangle}}", from=1-3, to=1-5]
	\arrow["{\oc \oplus \oc}"{description}, curve={height=18pt}, dashed, from=1-3, to=1-5]
	\arrow["\nabla", from=1-5, to=1-7]
\end{tikzcd}\]
The diagram commutes because $\{\ast\}$ is a terminal object.
    \item If $t = \pair{u_1}{u_2} \plus \pair{v_1}{v_2}$ and $r = \pair{u_1 \plus v_1}{u_2 \plus v_2}$, then $A = B \with C$, $\Upsilon;\Gamma \vdash u_1:B$, $\Upsilon;\Gamma \vdash u_2:C$, $\Upsilon;\Gamma \vdash v_1:B$ and $\Upsilon;\Gamma \vdash v_2:C$.

\[\resizebox{.85\textwidth}{!}{\begin{tikzcd}[cramped,column sep=small, ampersand replacement=\&]
	{\interpretation{\oc\Upsilon} \otimes \interpretation{\Gamma}} \&\& {(\interpretation{\oc\Upsilon} \otimes \interpretation{\Gamma}) \oplus (\interpretation{\oc\Upsilon} \otimes \interpretation{\Gamma})} \&\&\& {\interpretation B \oplus \interpretation C \oplus \interpretation B \oplus \interpretation C} \\
	\&\&\& {(2)} \\
	\&\&\& \begin{array}{c} (\interpretation{\oc\Upsilon} \otimes \interpretation{\Gamma}) \oplus (\interpretation{\oc\Upsilon} \otimes \interpretation{\Gamma}) \\ \oplus (\interpretation{\oc\Upsilon} \otimes \interpretation{\Gamma}) \oplus (\interpretation{\oc\Upsilon} \otimes \interpretation{\Gamma}) \end{array} \& {(3)} \\
	\&\& |[xshift=30pt]| \smash{(1)} \\
	\&\&\& {\interpretation B \oplus \interpretation B \oplus \interpretation C \oplus \interpretation C} \& {(4)} \\
	\&\&\&\&\& {\interpretation B \oplus \interpretation C}
	\arrow["\Delta"{description}, from=1-1, to=1-3]
	\arrow["{\interpretation{\langle u_1, u_2\rangle \plus \langle v_1, v_2\rangle}}"{description},  dotted, from=1-1, to=6-6, rounded corners, to path={-- ([yshift=3ex]\tikztostart.north) -- ([yshift=3ex,xshift=400pt]\tikztostart.north) \tikztonodes -| ([xshift=7ex]\tikztotarget.east) -- (\tikztotarget)}]
	\arrow["{\interpretation{\langle u_1 \plus v_1, u_2 \plus v_2\rangle}}"{description}, dotted, from=1-1, to=6-6, rounded corners, to path={ -- ([yshift=-150pt]\tikztostart.south) \tikztonodes |- ([yshift=-3ex]\tikztotarget.south) -- (\tikztotarget)}]
	\arrow["{\interpretation{\langle u_1, u_2\rangle} \oplus \interpretation{\langle v_1, v_2\rangle}}"{description}, from=1-3, to=1-6]
	\arrow["{\Delta \oplus \Delta}"{description}, dashed, from=1-3, to=3-4]
	\arrow["{\interpretation{u_1 \plus v_1} \oplus \interpretation{u_2 \plus v_2}}"{description}, from=1-3, to=6-6, rounded corners, to path={-- ([yshift=-130pt]\tikztostart.south) \tikztonodes |- (\tikztotarget)}]
	\arrow["\nabla"{description}, from=1-6, to=6-6]
	\arrow["{\interpretation{u_1} \oplus \interpretation{u_2} \oplus \interpretation{v_1} \oplus \interpretation{v_2}}"{description}, dashed, from=3-4, to=1-6]
	\arrow["{\interpretation{u_1} \oplus \interpretation{v_1} \oplus \interpretation{u_2} \oplus \interpretation{v_2}}"{description}, dashed, from=3-4, to=5-4]
	\arrow["{\nabla \oplus \nabla}"{description}, dashed, from=5-4, to=6-6]
	\arrow["{id \oplus \sigma' \oplus id}"{description}, dashed, from=1-6, to=5-4, curve={height=-40pt}]
\end{tikzcd}}\]
(1) commutes by definition of $\interpretation{u_1 \plus v_1}$ and $\interpretation{u_2 \plus v_2}$, (2) commutes by definition of $\interpretation{\pair{u_1}{u_2}}$ and $\interpretation{\pair{v_1}{v_2}}$, (3) commutes by naturality of $\sigma'$ and (4) commutes by Lemma~\ref{lem:sigmanabladelta}.

    \item If $t = \elimplus(u \plus v,x^B.w_1,y^C.w_2)$ and $r = \elimplus(u,x^B.w_1,y^C.w_2) \plus \elimplus(v,x^B.w_1,y^C.w_2)$, then $\Gamma = \Gamma_1, \Gamma_2$, $\Upsilon;\Gamma_1 \vdash u:B \oplus C$, $\Upsilon;\Gamma_1 \vdash v:B \oplus C$, $\Upsilon;\Gamma_2, x^B \vdash w_1:A$ and $\Upsilon;\Gamma_2, y^C \vdash w_2:A$.

\[\resizebox{.85\textwidth}{!}{\begin{tikzcd}[ampersand replacement=\&,cramped,column sep=tiny]
	{\interpretation{\oc\Upsilon} \otimes \interpretation{\Gamma_1} \otimes \interpretation{\Gamma_2}} \&\&\&\&\& {\interpretation{\oc\Upsilon} \otimes \interpretation{\Gamma_1} \otimes \interpretation{\oc\Upsilon} \otimes \interpretation{\Gamma_2}} \\
	\\
	\& {\interpretation{\oc\Upsilon} \otimes \interpretation{\oc\Upsilon} \otimes \interpretation{\Gamma_1} \otimes \interpretation{\Gamma_2}} \\
	\\
	\begin{array}{c} (\interpretation{\oc\Upsilon} \otimes \interpretation{\Gamma_1} \otimes \interpretation{\Gamma_2}) \\\oplus (\interpretation{\oc\Upsilon} \otimes \interpretation{\Gamma_1} \otimes \interpretation{\Gamma_2}) \end{array} \& \begin{array}{c} (\interpretation{\oc\Upsilon} \otimes \interpretation{\oc\Upsilon} \otimes \interpretation{\Gamma_1} \otimes \interpretation{\Gamma_2}) \\\oplus (\interpretation{\oc\Upsilon} \otimes \interpretation{\oc\Upsilon} \otimes \interpretation{\Gamma_1} \otimes \interpretation{\Gamma_2}) \end{array} \\
	\&\&\&\& \begin{array}{c} ((\interpretation{\oc\Upsilon} \otimes \interpretation{\Gamma_1} )\\\oplus (\interpretation{\oc\Upsilon} \otimes \interpretation{\Gamma_1}))\\ \otimes \interpretation{\oc\Upsilon} \otimes \interpretation{\Gamma_2} \end{array} \\
	\\
	\& \begin{array}{c} (\interpretation{\oc\Upsilon} \otimes \interpretation{\Gamma_1} \otimes \interpretation{\oc\Upsilon} \otimes \interpretation{\Gamma_2}) \\\oplus (\interpretation{\oc\Upsilon} \otimes \interpretation{\Gamma_1} \otimes \interpretation{\oc\Upsilon} \otimes \interpretation{\Gamma_2}) \end{array} \&\&\& \begin{array}{c} (\interpretation B \oplus \interpretation C\\\oplus \interpretation B \oplus \interpretation C)\\ \otimes \interpretation{\oc\Upsilon} \otimes \interpretation{\Gamma_2} \end{array} \\
	\\
	\& \begin{array}{c} ((\interpretation B \oplus \interpretation C) \otimes \interpretation{\oc\Upsilon} \otimes \interpretation{\Gamma_2}) \\\oplus ((\interpretation B \oplus \interpretation C) \otimes \interpretation{\oc\Upsilon} \otimes \interpretation{\Gamma_2}) \end{array} \&\&\&\& {(\interpretation B \oplus \interpretation C) \otimes \interpretation{\oc\Upsilon} \otimes \interpretation{\Gamma_2}} \\
	\\
	\& \begin{array}{c} (\interpretation B \otimes \interpretation{\oc\Upsilon} \otimes \interpretation{\Gamma_2}) \\\oplus (\interpretation C \otimes \interpretation{\oc\Upsilon} \otimes \interpretation{\Gamma_2}) \\ \oplus (\interpretation B \otimes \interpretation{\oc\Upsilon} \otimes \interpretation{\Gamma_2}) \\\oplus (\interpretation C \otimes \interpretation{\oc\Upsilon} \otimes \interpretation{\Gamma_2}) \end{array} \&\&\&\& \begin{array}{c} (\interpretation B \otimes \interpretation{\oc\Upsilon} \otimes \interpretation{\Gamma_2}) \\\oplus (\interpretation C \otimes \interpretation{\oc\Upsilon} \otimes \interpretation{\Gamma_2}) \end{array} \\
	\\
	\& \begin{array}{c} (\interpretation{\oc\Upsilon} \otimes \interpretation{\Gamma_2} \otimes \interpretation B) \\\oplus (\interpretation{\oc\Upsilon} \otimes \interpretation{\Gamma_2} \otimes \interpretation C) \\ \oplus (\interpretation{\oc\Upsilon} \otimes \interpretation{\Gamma_2} \otimes \interpretation B) \\\oplus (\interpretation{\oc\Upsilon} \otimes \interpretation{\Gamma_2} \otimes \interpretation C) \end{array} \&\&\&\& \begin{array}{c} (\interpretation{\oc\Upsilon} \otimes \interpretation{\Gamma_2} \otimes \interpretation B) \\\oplus (\interpretation{\oc\Upsilon} \otimes \interpretation{\Gamma_2} \otimes \interpretation C) \end{array} \\
	\\
	{\interpretation A \oplus \interpretation A} \&\&\&\&\& {\interpretation A}
	\arrow[""{name=0, anchor=center, inner sep=0}, "{\overline{d}_{\Upsilon, \Gamma_1, \Gamma_2}}"{description}, from=1-1, to=1-6]
	\arrow["{d_\Upsilon \otimes id}"{description}, dashed, from=1-1, to=3-2]
	\arrow["\Delta"{description}, from=1-1, to=5-1]
	\arrow["{\interpretation{\elimplus(u \plus v,x^B.w_1,y^C.w_2)}}"{description}, dotted, from=1-1, to=16-6, rounded corners, to path={-- ([yshift=3ex]\tikztostart.north) -| ([xshift=415pt]\tikztostart.north) -- ([xshift=415pt,yshift=-350pt]\tikztostart.north) \tikztonodes |- (\tikztotarget.east)}]
	\arrow["{\interpretation{\elimplus(u,x^B.w_1,y^C.w_2) \plus \elimplus(v,x^B.w_1,y^C.w_2)}}"{description}, dotted, from=1-1, to=16-6, rounded corners, to path={-- ([xshift=-3ex]\tikztostart.west) |- ([xshift=-3ex, yshift=-250pt]\tikztostart.west) \tikztonodes |- ([yshift=-3ex]\tikztotarget.south) -- (\tikztotarget.south)}]
	\arrow["{\Delta \otimes id}"{description}, dashed, from=1-6, to=6-5]
	\arrow[""{name=1, anchor=center, inner sep=0}, "\Delta"{description}, curve={height=30pt}, dashed, from=1-6, to=8-2]
	\arrow[""{name=2, anchor=center, inner sep=0}, "{\interpretation{u \plus v} \otimes id}"{description}, from=1-6, to=10-6]
	\arrow["{id \otimes \sigma \otimes id}"{description}, dashed, from=3-2, to=1-6]
	\arrow["\Delta"{description}, dashed, from=3-2, to=5-2]
	\arrow[""{name=3, anchor=center, inner sep=0}, "{(d_\Upsilon \otimes id) \oplus (d_\Upsilon \otimes id)}", curve={height=-24pt}, dashed, from=5-1, to=5-2]
	\arrow[""{name=4, anchor=center, inner sep=0}, "{\overline{d}_{\Upsilon, \Gamma_1, \Gamma_2} \oplus \overline{d}_{\Upsilon, \Gamma_1, \Gamma_2}}"{description}, curve={height=18pt}, dashed, from=5-1, to=8-2]
	\arrow[""{name=5, anchor=center, inner sep=0}, "\begin{array}{c} \interpretation{\elimplus(u,x^B.w_1,y^C.w_2)} \\\oplus \\ \interpretation{\elimplus(v,x^B.w_1,y^C.w_2)} \end{array}"{description}, from=5-1, to=16-1]
	\arrow["{(id \otimes \sigma \otimes id) \oplus (id \otimes \sigma \otimes id)}"{description}, dashed, from=5-2, to=8-2]
	\arrow["dist"{description}, dashed, from=6-5, to=8-2]
	\arrow["{(\interpretation u \oplus \interpretation v) \otimes id}"{description}, dashed, from=6-5, to=8-5]
	\arrow["{(7)}"{description}, draw=none, from=6-5, to=10-2]
	\arrow["{(\interpretation{u} \otimes id) \oplus (\interpretation{v} \otimes id)}"{description}, dashed, from=8-2, to=10-2]
	\arrow["dist"{description}, dashed, from=8-5, to=10-2]
	\arrow["{\nabla \otimes id}"{description}, dashed, from=8-5, to=10-6]
	\arrow[""{name=6, anchor=center, inner sep=0}, "\nabla"{description}, dashed, from=10-2, to=10-6]
	\arrow["{dist \oplus dist}"{description}, dashed, from=10-2, to=12-2]
	\arrow["dist"{description}, from=10-6, to=12-6]
	\arrow["{(10)}"{description}, draw=none, from=12-2, to=10-6]
	\arrow["\nabla"{description}, dashed, from=12-2, to=12-6]
	\arrow["{\sigma \oplus \sigma \oplus \sigma \oplus \sigma}"{description}, dashed, from=12-2, to=14-2]
	\arrow["{\sigma \oplus \sigma}"{description}, from=12-6, to=14-6]
	\arrow["{(11)}"{description}, draw=none, from=14-2, to=12-6]
	\arrow["\nabla"{description}, dashed, from=14-2, to=14-6]
	\arrow["{[\interpretation{w_1}, \interpretation{w_2}] \oplus [\interpretation{w_1}, \interpretation{w_2}]}"{description}, dashed, from=14-2, to=16-1]
	\arrow["{[\interpretation{w_1}, \interpretation{w_2}]}"{description}, from=14-6, to=16-6]
	\arrow["{(12)}"{description, pos=0.6}, draw=none, from=16-1, to=14-6]
	\arrow["\nabla"{description}, from=16-1, to=16-6]
	\arrow["{(2)}"{description, pos=0.6}, draw=none, from=1-1, to=3]
	\arrow["{(1)}"{description}, draw=none, from=3-2, to=0]
	\arrow["{(4)}"{description}, draw=none, from=3-2, to=1]
	\arrow["{(3)}"{description}, draw=none, from=5-2, to=4]
	\arrow["{(5)}"{description}, draw=none, from=6-5, to=1]
	\arrow["{(6)}"{description}, curve={height=24pt}, draw=none, from=6-5, to=2]
	\arrow["{(8)}"{description}, draw=none, from=8-5, to=6]
	\arrow["{(9)}"{description, pos=0.3}, curve={height=-30pt}, draw=none, from=12-2, to=5]
\end{tikzcd}}\]

(1) commutes by definition of $\overline d$, (2) commutes by naturality of $\Delta$, (3) commutes by definition of $\overline d$, (4) commutes by naturality of $\Delta$, (5) commutes by Corollary~\ref{cor:deltasemiadditive}, (6) commutes by definition of $\interpretation{u \plus v}$, (7) commutes by naturality of $dist$ (Lemma~\ref{lem:semiadditiveiso}), (8) commutes by Corollary~\ref{cor:nablasemiadditive}, (9) commutes by definition of $\interpretation{\elimplus(u,x^B.w_1,y^C.w_2)}$ and $\interpretation{\elimplus(v,x^B.w_1,y^C.w_2)}$, (10), (11) and (12) commute by naturality of $\nabla$.

    \item If $t = \elimbang(u \plus v, x^B.s)$ and $r = \elimbang(u, x^B.s) \plus \elimbang(v, x^B.s)$, then $\Gamma = \Gamma_1, \Gamma_2$, $\Upsilon;\Gamma_1 \vdash u:\oc B$, $\Upsilon;\Gamma_1 \vdash v:\oc B$ and $\Upsilon, x^B; \Gamma_2 \vdash s:A$.
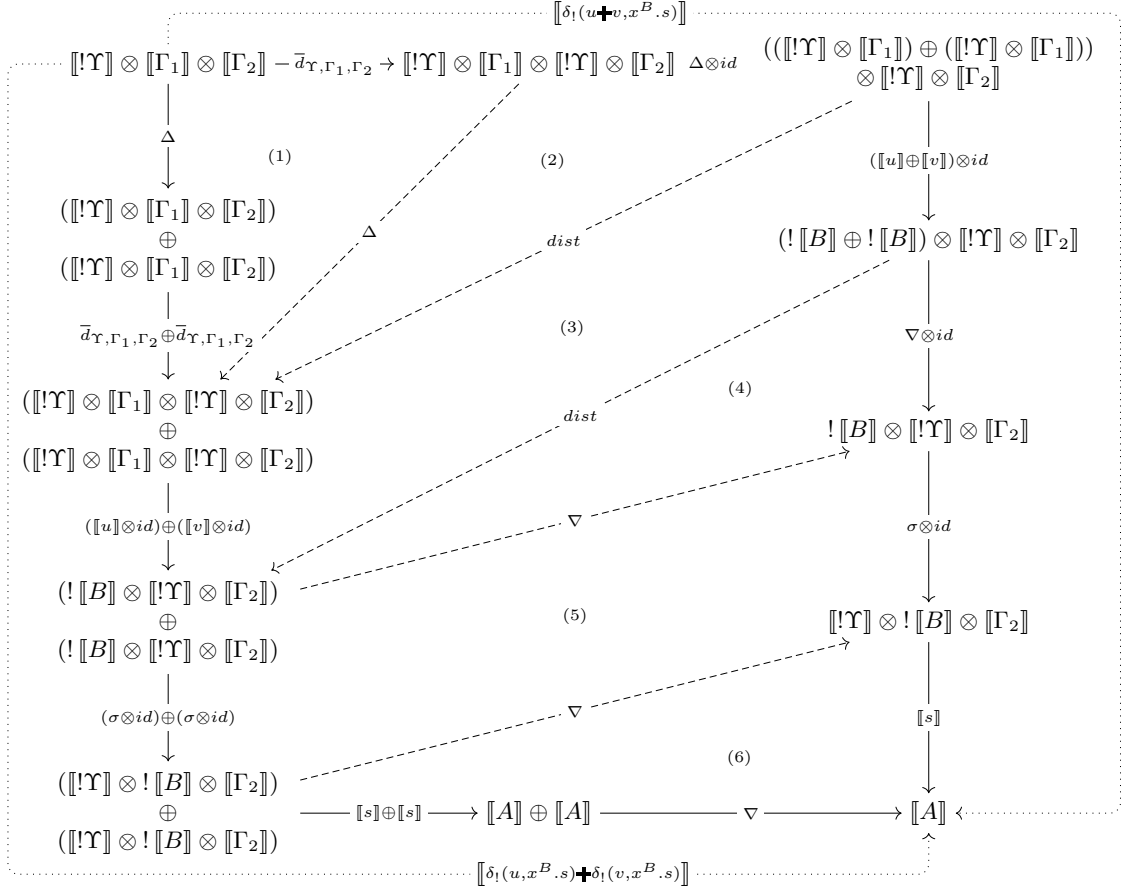
\[\begin{tikzcd}[ampersand replacement=\&,cramped,column sep=small]
	{\interpretation{\oc \Upsilon} \otimes \interpretation{\Gamma_1} \otimes \interpretation{\Gamma_2}} \&\& {\interpretation{\oc \Upsilon} \otimes \interpretation{\Gamma_1} \otimes \interpretation{\oc \Upsilon} \otimes \interpretation{\Gamma_2}} \&\& \begin{array}{c} ((\interpretation{\oc \Upsilon} \otimes \interpretation{\Gamma_1}) \oplus (\interpretation{\oc \Upsilon} \otimes \interpretation{\Gamma_1})) \\ \otimes \interpretation{\oc \Upsilon} \otimes \interpretation{\Gamma_2} \end{array} \\
	\\
	\begin{array}{c} (\interpretation{\oc \Upsilon} \otimes \interpretation{\Gamma_1} \otimes \interpretation{\Gamma_2})\\ \oplus \\ (\interpretation{\oc \Upsilon} \otimes \interpretation{\Gamma_1} \otimes \interpretation{\Gamma_2}) \end{array} \&\&\&\& {(\oc \interpretation{B} \oplus \oc \interpretation{B}) \otimes \interpretation{\oc \Upsilon} \otimes \interpretation{\Gamma_2}} \\
	\\
	\begin{array}{c} (\interpretation{\oc \Upsilon} \otimes \interpretation{\Gamma_1} \otimes \interpretation{\oc \Upsilon} \otimes \interpretation{\Gamma_2}) \\ \oplus \\ (\interpretation{\oc \Upsilon} \otimes \interpretation{\Gamma_1} \otimes \interpretation{\oc \Upsilon} \otimes \interpretation{\Gamma_2}) \end{array} \&\&\&\& {\oc \interpretation{B} \otimes \interpretation{\oc \Upsilon} \otimes \interpretation{\Gamma_2}} \\
	\\
	\begin{array}{c} (\oc \interpretation{B} \otimes \interpretation{\oc \Upsilon} \otimes \interpretation{\Gamma_2}) \\ \oplus \\ (\oc \interpretation{B} \otimes \interpretation{\oc \Upsilon} \otimes \interpretation{\Gamma_2}) \end{array} \&\&\&\& {\interpretation{\oc \Upsilon} \otimes \oc \interpretation{B} \otimes \interpretation{\Gamma_2}} \\
	\\
	\begin{array}{c} (\interpretation{\oc \Upsilon} \otimes \oc \interpretation{B} \otimes \interpretation{\Gamma_2}) \\ \oplus \\ (\interpretation{\oc \Upsilon} \otimes \oc \interpretation{B} \otimes \interpretation{\Gamma_2}) \end{array} \&\& {\interpretation{A} \oplus \interpretation{A}} \&\& {\interpretation A}
	\arrow["{\overline d_{\Upsilon, \Gamma_1, \Gamma_2}}"{description}, from=1-1, to=1-3]
	\arrow["\Delta"{description}, from=1-1, to=3-1]
	\arrow["{\interpretation{\elimbang(u \plus v, x^B.s)}}"{description}, dotted, from=1-1, to=9-5, rounded corners, to path={-- ([yshift=3ex]\tikztostart.north) -- ([yshift=3ex,xshift=340pt]\tikztostart.north) \tikztonodes -| ([xshift=16ex]\tikztotarget.east) -- (\tikztotarget)}]
	\arrow["{\interpretation{\elimbang(u, x^B.s) \plus \elimbang(v, x^B.s)}}"{description}, dotted, from=1-1, to=9-5, rounded corners, to path={-- ([xshift=-20pt]\tikztostart.west) |- ([xshift=-260pt, yshift=-4ex]\tikztotarget.south) -- ([yshift=-4ex]\tikztotarget.south) \tikztonodes -- (\tikztotarget)}]
	\arrow["{\Delta \otimes id}"{description}, from=1-3, to=1-5]
	\arrow[""{name=0, anchor=center, inner sep=0}, "\Delta"{description}, dashed, from=1-3, to=5-1]
	\arrow["{(\interpretation u \oplus \interpretation v) \otimes id}"{description}, from=1-5, to=3-5]
	\arrow[""{name=1, anchor=center, inner sep=0}, "dist"{description}, dashed, from=1-5, to=5-1]
	\arrow["{\overline d_{\Upsilon, \Gamma_1, \Gamma_2} \oplus \overline d_{\Upsilon, \Gamma_1, \Gamma_2}}"{description}, from=3-1, to=5-1]
	\arrow["{\nabla \otimes id}"{description}, from=3-5, to=5-5]
	\arrow[""{name=2, anchor=center, inner sep=0}, "dist"{description}, dashed, from=3-5, to=7-1]
	\arrow["{(\interpretation{u} \otimes id) \oplus (\interpretation{v} \otimes id)}"{description}, from=5-1, to=7-1]
	\arrow["{\sigma \otimes id}"{description}, from=5-5, to=7-5]
	\arrow[""{name=3, anchor=center, inner sep=0}, "\nabla"{description}, dashed, from=7-1, to=5-5]
	\arrow["{(\sigma \otimes id) \oplus (\sigma \otimes id)}"{description}, from=7-1, to=9-1]
	\arrow["{\interpretation s}"{description}, from=7-5, to=9-5]
	\arrow[""{name=4, anchor=center, inner sep=0}, "\nabla"{description}, dashed, from=9-1, to=7-5]
	\arrow["{\interpretation{s} \oplus \interpretation{s}}"{description}, from=9-1, to=9-3]
	\arrow["\nabla"{description}, from=9-3, to=9-5]
	\arrow["{(1)}"{description}, draw=none, from=1-1, to=0]
	\arrow["{(2)}"{description}, draw=none, from=1-3, to=1]
	\arrow["{(3)}"{description}, draw=none, from=1, to=2]
	\arrow["{(4)}"{description}, draw=none, from=3-5, to=3]
	\arrow["{(5)}"{description}, draw=none, from=3, to=4]
	\arrow["{(6)}"{description}, draw=none, from=9-5, to=4]
\end{tikzcd}\]

(1) commutes by naturality of $\Delta$, (2) commutes by Corollary~\ref{cor:deltasemiadditive}, (3) commutes by naturality of $dist$, (4) commutes by Corollary~\ref{cor:nablasemiadditive}, (5) and (6) commute by naturality of $\nabla$.

    \item If $t = a \dotprod b.\star$ and $r = (a \times b).\star$, then $A = \one$ and $\Gamma = \varnothing$.

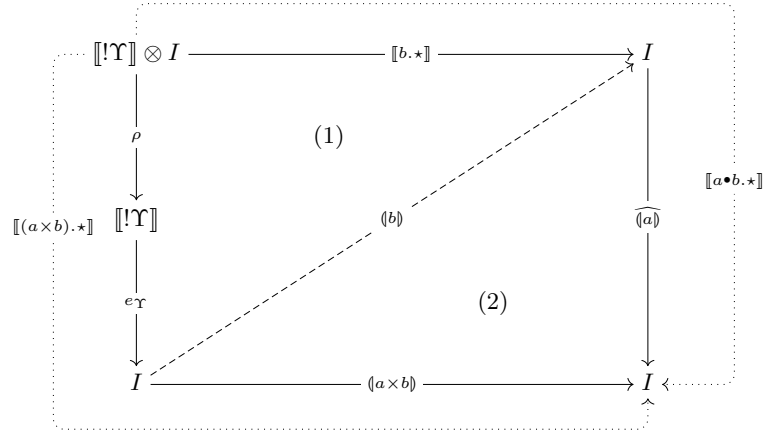
\[\begin{tikzcd}[ampersand replacement=\&,cramped]
	{\interpretation{\oc\Upsilon} \otimes I} \&\&\&\&\&\& I \\
	\&\& {(1)} \\
	{\interpretation{\oc\Upsilon}} \\
	\&\&\&\& {(2)} \\
	I \&\&\&\&\&\& I
	\arrow["{\interpretation{b.\star}}"{description}, from=1-1, to=1-7]
	\arrow["\rho"{description}, from=1-1, to=3-1]
	\arrow["{\interpretation{a \dotprod b.\star}}"{description}, dotted, from=1-1, to=5-7, rounded corners, to path={-- ([yshift=3ex]\tikztostart.north) -| ([xshift=225pt]\tikztostart.north) -- ([xshift=225pt,yshift=-110pt]\tikztostart.north) \tikztonodes |- (\tikztotarget.east)}]
	\arrow["{\interpretation{(a \times b).\star}}"{description}, dotted, from=1-1, to=5-7, rounded corners, to path={-- ([xshift=-3ex]\tikztostart.west) -- ([xshift=-3ex,yshift=-130pt]\tikztostart.west) \tikztonodes |- ([yshift=-3ex]\tikztotarget.south) -- (\tikztotarget.south)}]
	\arrow["{\widehat{\mono{a}}}"{description}, from=1-7, to=5-7]
	\arrow["{e_\Upsilon}"{description}, from=3-1, to=5-1]
	\arrow["\mono{b}"{description}, dashed, from=5-1, to=1-7]
	\arrow["{\mono{a \times b}}"{description}, from=5-1, to=5-7]
\end{tikzcd}\]

(1) commutes by definition of $\interpretation{b.\star}$ and (2) commutes because $\widehat{\mono{a}} = \mono{a}$ by Lemma~\ref{lem:hataprop} and because $\mono{a \times b} = \mono{a} \circ \mono{b}$, since $\mono{\cdot}$ is a semiring homomorphism.

    \item If $t = a \dotprod \lambda x^B.u$ and $r = \lambda x^B. a \dotprod u$, then $A = B \multimap C$ and $\Upsilon;\Gamma,x^B \vdash u:C$.
    
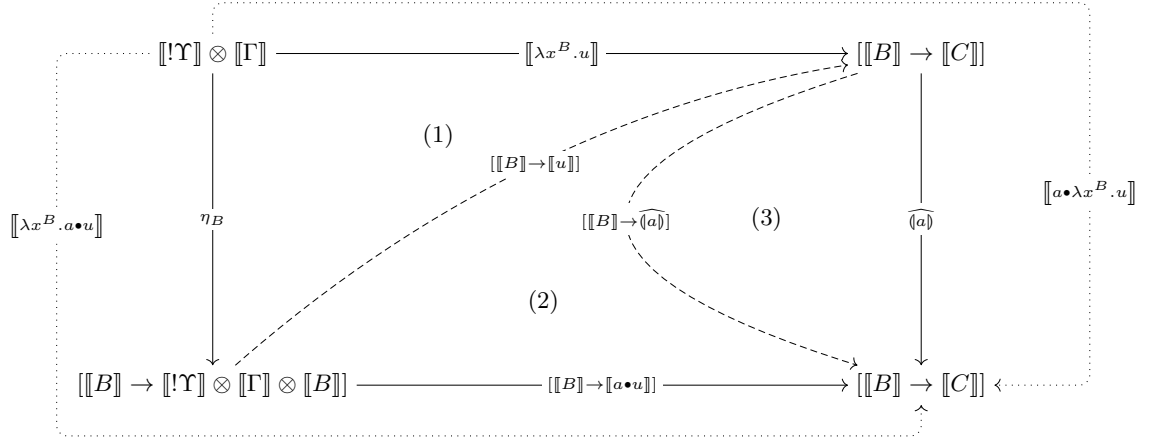
\[\begin{tikzcd}[cramped]
	{\interpretation{\oc\Upsilon} \otimes \interpretation{\Gamma}} &&&&&& {[\interpretation B \to \interpretation C]} \\
	& {(1)} \\
	&&&&& {(3)} \\
	&& {(2)} \\
	{[\interpretation B \to \interpretation{\oc\Upsilon} \otimes \interpretation{\Gamma} \otimes \interpretation B]} &&&&&& {[\interpretation B \to \interpretation C]}
	\arrow["{\interpretation{\lambda x^B.u}}"{description}, from=1-1, to=1-7]
	\arrow["{\eta_B}"{description}, from=1-1, to=5-1]
	\arrow["{\interpretation{a \dotprod \lambda x^B.u}}"{description}, dotted, from=1-1, to=5-7, rounded corners, to path={-- ([yshift=3ex]\tikztostart.north) -| ([xshift=330pt]\tikztostart.north) -- ([xshift=330pt,yshift=-120pt]\tikztostart.north) \tikztonodes |- (\tikztotarget.east)}]
	\arrow["{\interpretation{\lambda x^B. a \dotprod u}}"{description}, dotted, from=1-1, to=5-7, rounded corners, to path={-- ([xshift=-9ex]\tikztostart.west) -- ([xshift=-9ex,yshift=-130pt]\tikztostart.west) \tikztonodes |- ([yshift=-3ex]\tikztotarget.south) -- (\tikztotarget.south)}]
	\arrow["{\widehat{\mono{a}}}"{description}, from=1-7, to=5-7]
	\arrow["{[\interpretation B \to \widehat{\mono{a}}]}"{description}, curve={height=140pt}, dashed, from=1-7, to=5-7]
	\arrow["{[\interpretation B \to \interpretation{u}]}"{description}, curve={height=-30pt}, dashed, from=5-1, to=1-7]
	\arrow["{[\interpretation B \to \interpretation{a \dotprod u}]}"{description}, from=5-1, to=5-7]
\end{tikzcd}\]

(1) commutes by definition of $\interpretation{\lambda x^B.u}$, (2) commutes by definition of $\interpretation{a \dotprod u}$ and functoriality of $[\interpretation{B} \to -]$ and (3) commutes by Lemma~\ref{lem:ahathom}.

    \item If $t = \elimtens(a \dotprod u,x^B y^C.v)$ and $r = a \dotprod \elimtens(u,x^B y^C.v)$, then $\Gamma = \Gamma_1, \Gamma_2$, $\Upsilon;\Gamma_1\vdash u: B \otimes C$ and $\Upsilon;\Gamma_2,x^B,y^C \vdash v:A$.

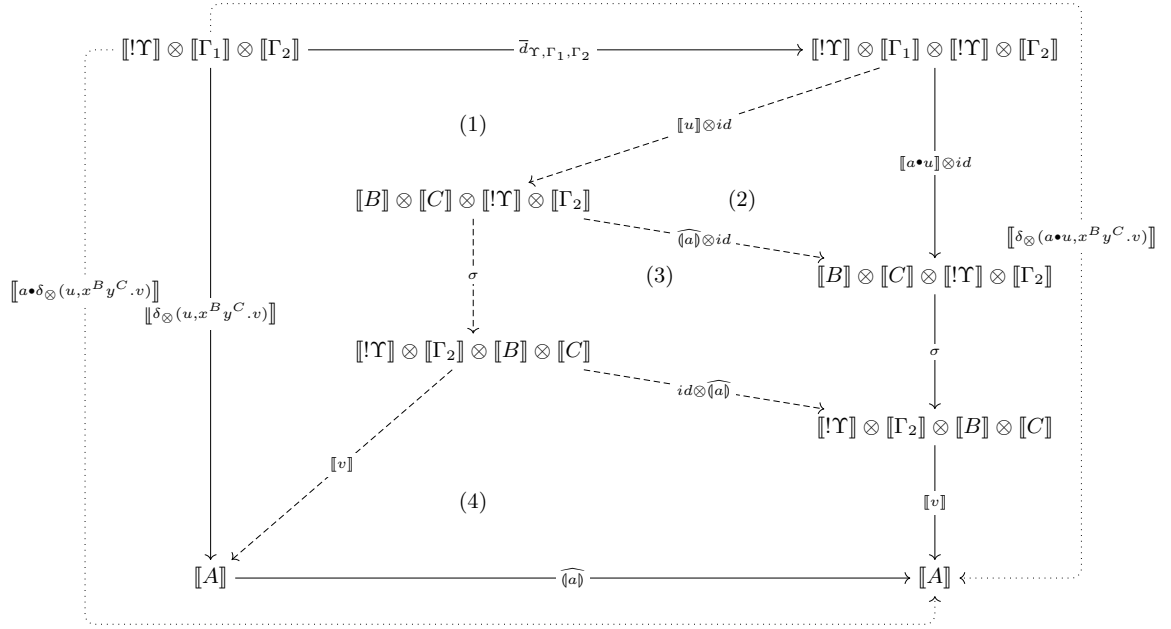
\[\resizebox{.85\textwidth}{!}{\begin{tikzcd}[ampersand replacement=\&,cramped, column sep=scriptsize]
	{\interpretation{\oc \Upsilon} \otimes \interpretation{\Gamma_1} \otimes \interpretation{\Gamma_2}} \&\&\&\& {\interpretation{\oc \Upsilon} \otimes \interpretation{\Gamma_1} \otimes \interpretation{\oc \Upsilon} \otimes \interpretation{\Gamma_2}} \\
	\& {(1)} \\
	\& {\interpretation B \otimes \interpretation C \otimes \interpretation{\oc \Upsilon} \otimes \interpretation{\Gamma_2}} \&\& {(2)} \\
	\&\& {(3)} \&\& {\interpretation B \otimes \interpretation C \otimes \interpretation{\oc \Upsilon} \otimes \interpretation{\Gamma_2}} \\
	\& {\interpretation{\oc \Upsilon} \otimes \interpretation{\Gamma_2} \otimes \interpretation B \otimes \interpretation C} \\
	\&\&\&\& {\interpretation{\oc \Upsilon} \otimes \interpretation{\Gamma_2} \otimes \interpretation B \otimes \interpretation C} \\
	\& {(4)} \\
	{\interpretation A} \&\&\&\& {\interpretation A}
	\arrow["{\overline{d}_{\Upsilon, \Gamma_1, \Gamma_2}}"{description}, from=1-1, to=1-5]
	\arrow["{\interpretation{\elimtens(u,x^B y^C.v)}}"{description}, from=1-1, to=8-1]
	\arrow["{\interpretation{\elimtens(a \dotprod u,x^B y^C.v)}}"{description}, dotted, from=1-1, to=8-5, rounded corners, to path={-- ([yshift=3ex]\tikztostart.north) -| ([xshift=360pt]\tikztostart.north) -- ([xshift=360pt,yshift=-170pt]\tikztostart.north) \tikztonodes |- (\tikztotarget.east)}]
	\arrow["{\interpretation{a \dotprod \elimtens(u,x^B y^C.v)}}"{description}, dotted, from=1-1, to=8-5, rounded corners, to path={-- ([xshift=-3ex]\tikztostart.west) -- ([xshift=-3ex,yshift=-200pt]\tikztostart.west) \tikztonodes |- ([yshift=-3ex]\tikztotarget.south) -- (\tikztotarget.south)}]
	\arrow["{\interpretation{u} \otimes id}"{description}, dashed, from=1-5, to=3-2]
	\arrow["{\interpretation{a \dotprod u} \otimes id}"{description}, from=1-5, to=4-5]
	\arrow["{\widehat{\mono{a}} \otimes id}"{description}, dashed, from=3-2, to=4-5]
	\arrow["\sigma"{description}, dashed, from=3-2, to=5-2]
	\arrow["\sigma"{description}, from=4-5, to=6-5]
	\arrow["{id \otimes \widehat{\mono{a}}}"{description}, dashed, from=5-2, to=6-5]
	\arrow["{\interpretation v}"{description}, dashed, from=5-2, to=8-1]
	\arrow["{\interpretation v}"{description}, from=6-5, to=8-5]
	\arrow["{\widehat{\mono{a}}}"{description}, from=8-1, to=8-5]
\end{tikzcd}}\]

(1) commutes by definition of $\interpretation{\elimtens(u,x^B y^C.v)}$, (2) commutes by definition of $\interpretation{a \dotprod u}$, (3) commutes by naturality of $\sigma$ and (4) commutes by Lemma~\ref{lem:hatanat}.

    \item If $t = a \dotprod \langle\rangle$ and $r = \langle\rangle$, then $A = \top$.

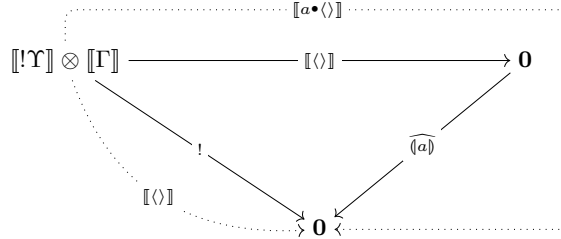
\[\begin{tikzcd}[ampersand replacement=\&,cramped]
	{\interpretation{\oc\Upsilon} \otimes \interpretation{\Gamma}} \&\&\&\&\&\& {\mathbf 0} \\
	\\
	\\
	\&\&\& {\mathbf 0}
	\arrow["{\interpretation{\langle\rangle}}"{description}, from=1-1, to=1-7]
	\arrow["\oc"{description}, from=1-1, to=4-4]
	\arrow["{\interpretation{a \dotprod \langle\rangle}}"{description}, dotted, from=1-1, to=4-4, rounded corners, to path={-- ([yshift=3ex]\tikztostart.north) -- ([yshift=3ex,xshift=190pt]\tikztostart.north) \tikztonodes |- (\tikztotarget.east)}]
	\arrow["{\interpretation{\langle\rangle}}"{description}, curve={height=30pt}, dotted, from=1-1, to=4-4]
	\arrow["{\widehat{\mono{a}}}"{description}, from=1-7, to=4-4]
\end{tikzcd}\]

The diagram commutes because $\mathbf 0$ is a terminal object.

    \item If $t = a \dotprod \pair{u}{v}$ and $r = \pair{a \dotprod u}{a \dotprod v}$, then $A = B \with C$, $\Upsilon;\Gamma \vdash u:B$ and $\Upsilon;\Gamma \vdash v:C$.
    
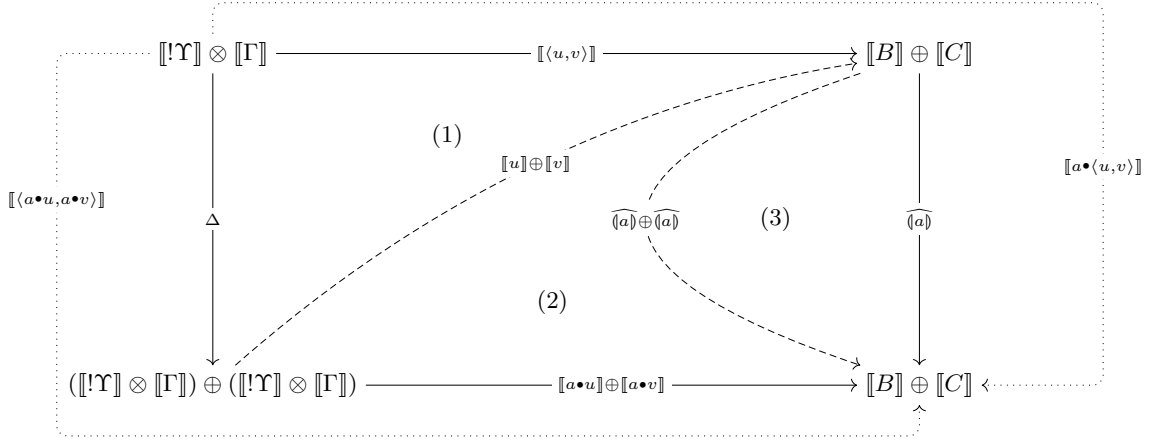
\[\begin{tikzcd}[cramped]
	{\interpretation{\oc\Upsilon} \otimes \interpretation{\Gamma}} &&&&&& {\interpretation B \oplus \interpretation C} \\
	& {(1)} \\
	&&&&& {(3)} \\
	&& {(2)} \\
	{(\interpretation{\oc\Upsilon} \otimes \interpretation{\Gamma}) \oplus (\interpretation{\oc\Upsilon} \otimes \interpretation{\Gamma})} &&&&&& {\interpretation B \oplus \interpretation C}
	\arrow["{\interpretation{\pair{u}{v}}}"{description}, from=1-1, to=1-7]
	\arrow["\Delta"{description}, from=1-1, to=5-1]
	\arrow["{\interpretation{a \dotprod \pair{u}{v}}}"{description}, dotted, from=1-1, to=5-7, rounded corners, to path={-- ([yshift=3ex]\tikztostart.north) -| ([xshift=335pt]\tikztostart.north) -- ([xshift=335pt,yshift=-100pt]\tikztostart.north) \tikztonodes |- (\tikztotarget.east)}]
	\arrow["{\interpretation{\pair{a \dotprod u}{a \dotprod v}}}"{description}, dotted, from=1-1, to=5-7, rounded corners, to path={-- ([xshift=-9ex]\tikztostart.west) -- ([xshift=-9ex,yshift=-110pt]\tikztostart.west) \tikztonodes |- ([yshift=-3ex]\tikztotarget.south) -- (\tikztotarget.south)}]
	\arrow["{\widehat{\mono{a}}}"{description}, from=1-7, to=5-7]
	\arrow["{\widehat{\mono{a}} \oplus \widehat{\mono{a}}}"{description}, curve={height=130pt}, dashed, from=1-7, to=5-7]
	\arrow["{\interpretation u \oplus \interpretation v}"{description}, curve={height=-30pt}, dashed, from=5-1, to=1-7]
	\arrow["{\interpretation{a \dotprod u} \oplus \interpretation{a \dotprod v}}"{description}, from=5-1, to=5-7]
\end{tikzcd}\]

(1) commutes by definition of $\interpretation{\pair{u}{v}}$, (2) commutes by definition of $\interpretation{a \dotprod u}$ and $\interpretation{a \dotprod v}$, and (3) commutes by Lemma~\ref{lem:hataprop}.

    \item If $t = \elimplus(a \dotprod u,x^B.v,y^C.w)$ and $r = a \dotprod \elimplus(u,x^B.v,y^C.w)$, then $\Gamma = \Gamma_1, \Gamma_2$, $\Upsilon;\Gamma_1 \vdash u:B \oplus C$, $\Upsilon;\Gamma_2,x^B \vdash v:A$ and $\Upsilon;\Gamma_2,y^C \vdash w:A$.

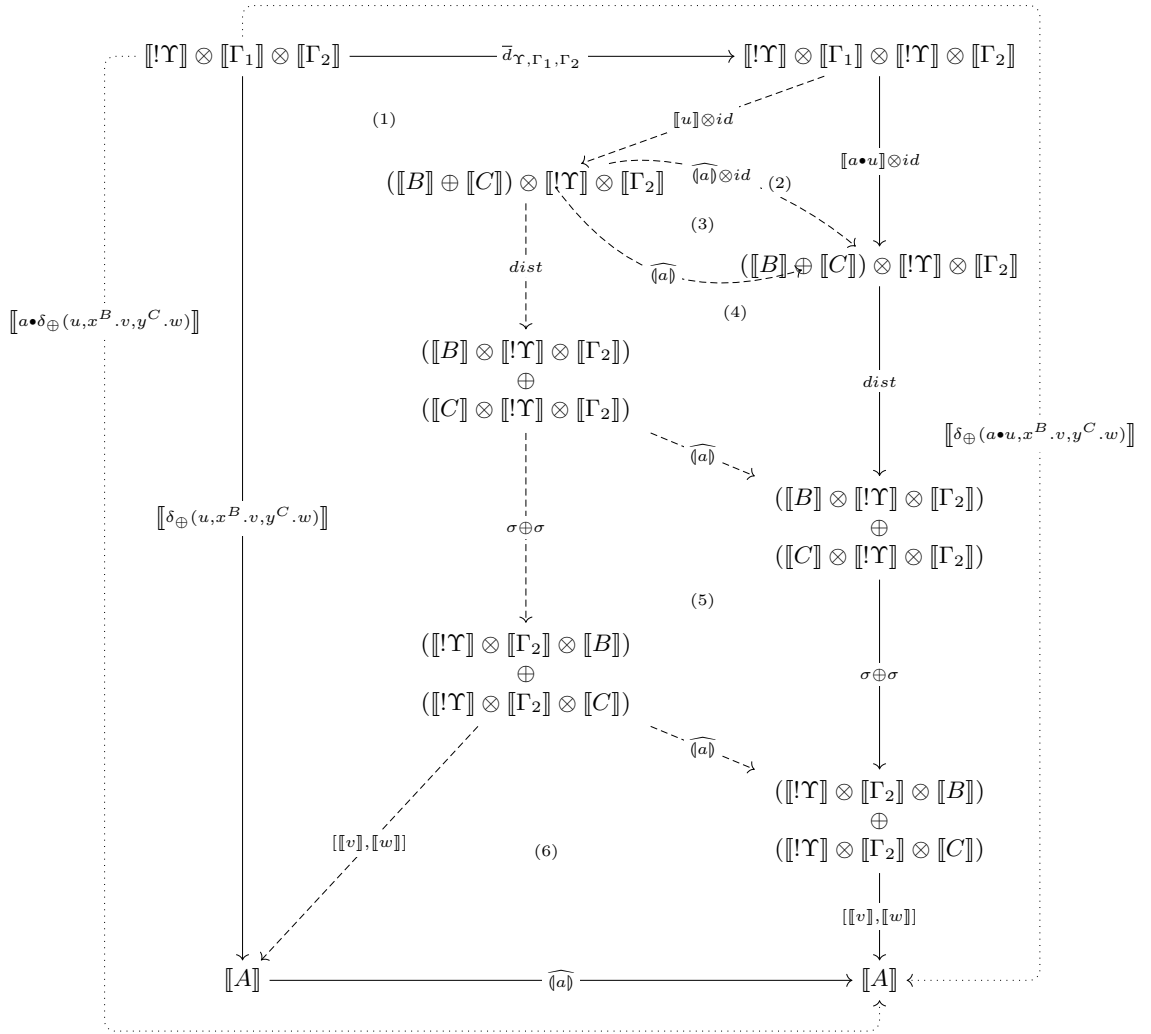
\[\begin{tikzcd}[ampersand replacement=\&,cramped,column sep=small]
	{\interpretation{\oc\Upsilon} \otimes \interpretation{\Gamma_1} \otimes \interpretation{\Gamma_2}} \&\&\& {\interpretation{\oc\Upsilon} \otimes \interpretation{\Gamma_1} \otimes \interpretation{\oc\Upsilon} \otimes \interpretation{\Gamma_2}} \\
	\\
	\& {(\interpretation B \oplus \interpretation C) \otimes \interpretation{\oc\Upsilon} \otimes \interpretation{\Gamma_2}} \\
	\&\&\& {(\interpretation B \oplus \interpretation C) \otimes \interpretation{\oc\Upsilon} \otimes \interpretation{\Gamma_2}} \\
	\& \begin{array}{c} (\interpretation B \otimes \interpretation{\oc\Upsilon} \otimes \interpretation{\Gamma_2})\\ \oplus \\(\interpretation C \otimes \interpretation{\oc\Upsilon} \otimes \interpretation{\Gamma_2}) \end{array} \\
	\&\&\& \begin{array}{c} (\interpretation B \otimes \interpretation{\oc\Upsilon} \otimes \interpretation{\Gamma_2})\\ \oplus \\(\interpretation C \otimes \interpretation{\oc\Upsilon} \otimes \interpretation{\Gamma_2}) \end{array} \\
	\& \begin{array}{c} (\interpretation{\oc\Upsilon} \otimes \interpretation{\Gamma_2} \otimes \interpretation B)\\ \oplus \\(\interpretation{\oc\Upsilon} \otimes \interpretation{\Gamma_2} \otimes \interpretation C) \end{array} \\
	\&\&\& \begin{array}{c} (\interpretation{\oc\Upsilon} \otimes \interpretation{\Gamma_2} \otimes \interpretation B)\\ \oplus \\(\interpretation{\oc\Upsilon} \otimes \interpretation{\Gamma_2} \otimes \interpretation C) \end{array} \\
	\\
	{\interpretation A} \&\&\& {\interpretation A}
	\arrow["{\overline d_{\Upsilon,\Gamma_1,\Gamma_2}}"{description}, from=1-1, to=1-4]
	\arrow["{\interpretation{\elimplus(u,x^B.v,y^C.w)}}"{description}, from=1-1, to=10-1]
	\arrow["{\interpretation{\elimplus(a \dotprod u,x^B.v,y^C.w)}}"{description}, dotted, from=1-1, to=10-4, rounded corners, to path={-- ([yshift=3ex]\tikztostart.north) -| ([xshift=300pt]\tikztostart.north) -- ([xshift=300pt,yshift=-300pt]\tikztostart.north) \tikztonodes |- (\tikztotarget.east)}]
	\arrow["{\interpretation{a \dotprod \elimplus(u,x^B.v,y^C.w)}}"{description}, dotted, from=1-1, to=10-4, rounded corners, to path={-- ([xshift=-3ex]\tikztostart.west) -- ([xshift=-3ex,yshift=-200pt]\tikztostart.west) \tikztonodes |- ([yshift=-3ex]\tikztotarget.south) -- (\tikztotarget.south)}]
	\arrow[""{name=0, anchor=center, inner sep=0}, "{\interpretation{u} \otimes id}"{description}, dashed, from=1-4, to=3-2]
	\arrow["{\interpretation{a \dotprod u} \otimes id}"{description}, from=1-4, to=4-4]
	\arrow["{(1)}"{description}, draw=none, from=3-2, to=1-1]
	\arrow["{\widehat{\mono{a}} \otimes id}"{description}, curve={height=-24pt}, dashed, from=3-2, to=4-4]
	\arrow["{\widehat{\mono{a}}}"{description}, shift left=3, curve={height=24pt}, dashed, from=3-2, to=4-4]
	\arrow["{(3)}"{description}, draw=none, from=3-2, to=4-4]
	\arrow["dist"{description}, dashed, from=3-2, to=5-2]
	\arrow["dist"{description}, from=4-4, to=6-4]
	\arrow["{(4)}"{description}, draw=none, from=5-2, to=4-4]
	\arrow["{\widehat{\mono{a}}}"{description}, dashed, from=5-2, to=6-4]
	\arrow["{\sigma \oplus \sigma}"{description}, dashed, from=5-2, to=7-2]
	\arrow["{\sigma \oplus \sigma}"{description}, from=6-4, to=8-4]
	\arrow["{(5)}"{description}, draw=none, from=7-2, to=6-4]
	\arrow["{\widehat{\mono{a}}}"{description}, dashed, from=7-2, to=8-4]
	\arrow["{[\interpretation v, \interpretation w]}"{description}, dashed, from=7-2, to=10-1]
	\arrow["{[\interpretation v, \interpretation w]}"{description}, from=8-4, to=10-4]
	\arrow[""{name=1, anchor=center, inner sep=0}, "{\widehat{\mono{a}}}"{description}, from=10-1, to=10-4]
	\arrow["{(2)}"{description}, draw=none, from=4-4, to=0]
	\arrow["{(6)}"{description}, draw=none, from=7-2, to=1]
\end{tikzcd}\]

(1) commutes by definition of $\interpretation{\elimplus(u,x^B.v,y^C.w)}$, (2) commutes by definition of $\interpretation{a \dotprod u}$, (3) commutes by Lemma~\ref{lem:hataprop}, (4), (5) and (6) commute by naturality of $\widehat{\mono{a}}$ (Lemma~\ref{lem:hatanat}).

    \item If $t = \elimbang(a \dotprod u, x^B.s)$ and $r = a \dotprod \elimbang(u, x^B.s)$, then $\Gamma = \Gamma_1, \Gamma_2$, $\Upsilon;\Gamma_1 \vdash u:\oc B$ and $\Upsilon, x^B;\Gamma_2 \vdash s:A$.

\[\begin{tikzcd}[ampersand replacement=\&,cramped,column sep=scriptsize]
	{\interpretation{\oc \Upsilon} \otimes \interpretation{\Gamma_1} \otimes \interpretation{\Gamma_2}} \&\& {\interpretation{\oc \Upsilon} \otimes \interpretation{\Gamma_1} \otimes \interpretation{\oc \Upsilon} \otimes \interpretation{\Gamma_2}} \&\& {\oc \interpretation B \otimes \interpretation{\oc \Upsilon} \otimes \interpretation{\Gamma_2}} \\
	\\
	{\interpretation{\oc \Upsilon} \otimes \interpretation{\Gamma_1} \otimes \interpretation{\oc \Upsilon} \otimes \interpretation{\Gamma_2}} \&\&\&\& {\oc \interpretation B \otimes \interpretation{\oc \Upsilon} \otimes \interpretation{\Gamma_2}} \\
	\\
	{\oc \interpretation B \otimes \interpretation{\oc \Upsilon} \otimes \interpretation{\Gamma_2}} \&\&\&\& {\interpretation{\oc \Upsilon} \otimes \oc \interpretation B \otimes \interpretation{\Gamma_2}} \\
	\\
	{\interpretation{\oc \Upsilon} \otimes \oc \interpretation B \otimes \interpretation{\Gamma_2}} \&\& {\interpretation A} \&\& {\interpretation A}
	\arrow["{\overline d_{\Upsilon, \Gamma_1, \Gamma_2}}"{description}, from=1-1, to=1-3]
	\arrow["{\overline d_{\Upsilon, \Gamma_1, \Gamma_2}}"{description}, from=1-1, to=3-1]
	\arrow["{\interpretation{\elimbang(a \dotprod u, x^B.s)}}"{description}, dotted, from=1-1, to=7-5, rounded corners, to path={-- ([yshift=3ex]\tikztostart.north) -- ([yshift=3ex,xshift=310pt]\tikztostart.north) \tikztonodes -| ([xshift=10ex]\tikztotarget.east) -- (\tikztotarget)}]
	\arrow["{\interpretation{a \dotprod \elimbang(u, x^B.s)}}"{description}, dotted, from=1-1, to=7-5, rounded corners, to path={-- ([xshift=-20pt]\tikztostart.west) |- ([xshift=-260pt, yshift=-4ex]\tikztotarget.south) -- ([yshift=-4ex]\tikztotarget.south) \tikztonodes -- (\tikztotarget)}]
	\arrow["{\interpretation u \otimes id}"{description}, from=1-3, to=1-5]
	\arrow["{\widehat{\mono{a}} \otimes id}"{description}, from=1-5, to=3-5]
	\arrow["{\interpretation u \otimes id}"{description}, from=3-1, to=5-1]
	\arrow["{\sigma \otimes id}"{description}, from=3-5, to=5-5]
	\arrow[""{name=0, anchor=center, inner sep=0}, equals, dashed, from=5-1, to=1-5]
	\arrow[""{name=1, anchor=center, inner sep=0}, "{\widehat{\mono{a}}}"{description}, dashed, from=5-1, to=3-5]
	\arrow["{\sigma \otimes id}"{description}, from=5-1, to=7-1]
	\arrow["{\interpretation s}"{description}, from=5-5, to=7-5]
	\arrow[""{name=2, anchor=center, inner sep=0}, "{\widehat{\mono{a}}}"{description}, dashed, from=7-1, to=5-5]
	\arrow["{\interpretation s}"{description}, from=7-1, to=7-3]
	\arrow["{\widehat{\mono{a}}}"{description}, from=7-3, to=7-5]
	\arrow["{(1)}"{description}, draw=none, from=1-1, to=0]
	\arrow["{(2)}"{description}, draw=none, from=3-5, to=0]
	\arrow["{(3)}"{description}, draw=none, from=1, to=2]
	\arrow["{(4)}"{description}, draw=none, from=7-5, to=2]
\end{tikzcd}\]

(1) commutes trivially, (2) commutes by Lemma~\ref{lem:hataprop}, (3) and (4) commute by naturality of $\widehat{\mono{a}}$ (Lemma~\ref{lem:hatanat}).

  \end{itemize} 
  Inductive cases are trivial by composition.
\end{proof}

\subsection{Adequacy}\label{proof:adequacy}
\semanticsadequacy*
\begin{proof}
  By induction on the size of $A$.
  \begin{itemize}
    \item Let $A = \one$. By Theorems~\ref{thm:SR}, \ref{thm:introductions}, and \ref{thm:ST}, we have $t \lras a.\star$ and $r \lras b.\star$. By Theorem~\ref{thm:soundness}, $\interpretation{t} = \interpretation{a.\star}$ and $\interpretation{r} = \interpretation{b.\star}$. Hence, since $\interpretation{t} = \interpretation{r}$, the following diagram commutes:

\[\begin{tikzcd}[cramped]
	{\oc I \otimes I} && {\oc I} && I && I
	\arrow["\rho", from=1-1, to=1-3]
	\arrow["{e_I}", from=1-3, to=1-5]
	\arrow["{\mono{a}}", curve={height=-18pt}, from=1-5, to=1-7]
	\arrow["{\mono{b}}"', curve={height=18pt}, from=1-5, to=1-7]
\end{tikzcd}\]

Since $e_I \circ \rho \circ \rho^{-1} \circ m_I = id_I$, we have that $\mono{a} = \mono{b}$. Therefore $a=b$, because $\mono{\cdot}$ is a monomorphism, and then $t \equiv r$.

    \item Let $A = B \multimap C$. By Theorems~\ref{thm:SR}, \ref{thm:introductions} and \ref{thm:ST}, we have $t \lras \lambda x^B.t'$ and $r \lras \lambda x^B.r'$. By Theorem~\ref{thm:soundness}, $\interpretation{t} = \interpretation{\lambda x^B.t'}$ and $\interpretation{r} = \interpretation{\lambda x^B.r'}$. Hence, since $\interpretation{t} =\interpretation{r}$, the following diagram commutes:

\[\begin{tikzcd}[cramped]
	{\oc I \otimes I} && {[\interpretation B \to \oc I \otimes I \otimes \interpretation B]} && {[\interpretation B \to \interpretation C]}
	\arrow["{\eta_{\interpretation B}}", from=1-1, to=1-3]
	\arrow["{[\interpretation B \to \interpretation{t'}]}", curve={height=-18pt}, from=1-3, to=1-5]
	\arrow["{[\interpretation B \to \interpretation{r'}]}"', curve={height=18pt}, from=1-3, to=1-5]
\end{tikzcd}\]

Therefore, we have $\interpretation{t'} = \interpretation{r'}$, by the bijection arising from the adjunction. By Lemma~\ref{lem:linearsubstitution}, we have that for all $w$ such that $\vdash w:B$:
\begin{align*}
  \interpretation{(w/x)t'} &= \interpretation{t'} \circ (id \otimes \interpretation{w}) \circ (id \otimes \sigma) \circ (d_I \otimes id)\\
  \interpretation{(w/x)r'} &= \interpretation{r'} \circ (id \otimes \interpretation{w}) \circ (id \otimes \sigma) \circ (d_I \otimes id)
\end{align*}

Then, $\interpretation{(w/x)t'} = \interpretation{(w/x)r'}$. Hence, by the induction hypothesis, $(w/x)t' \equiv (w/x)r'$. Since this is true for all $w$, we have $\lambda x^B.t' \equiv \lambda x^B.r'$. Therefore, $t \equiv \lambda x^B.t' \equiv \lambda x^B.r' \equiv r$.

    \item Let $A = B \otimes C$. Any elimination context for $t$ and $r$ has the shape $K = K'[\elimtens(\_, x^B y^C.u)]$, where $\varnothing;x^B,y^C \vdash u:D$ with $|D| < |B \otimes C|$. By definition, we have that
    \begin{align*}
      \interpretation{\elimtens(t,x^B y^C.u)} &= \interpretation{u} \circ \sigma \circ (\interpretation{t} \otimes id) \circ \overline d\\
      \interpretation{\elimtens(r,x^B y^C.u)} &= \interpretation{u} \circ \sigma \circ (\interpretation{r} \otimes id) \circ \overline d
    \end{align*}

    $\interpretation{t} = \interpretation{r}$ by hypothesis, therefore $\interpretation{\elimtens(t,x^B y^C.u)} = \interpretation{\elimtens(r,x^B y^C.u)}$. Thus, by the induction hypothesis, we have $\elimtens(t,x^B y^C.u) \equiv \elimtens(r,x^B y^C.u)$.

      Hence,
      \(
	K[t]=K'[\elimtens(t,x^B y^C.u)]\equiv  K'[\elimtens(r,x^B y^C.u)]=K[r]
      \).

      Then, $t\equiv r$.

    \item Let $A = \top$. 
      There is no elimination context ${\_}^\top\vdash K:\one$, then the observational equivalence is vacuously true.

    \item Let $A = \zero$. By Theorems~\ref{thm:SR}, \ref{thm:introductions} and \ref{thm:ST}, this case is not possible.

    \item Let $A = B \with C$. By Theorems~\ref{thm:SR}, \ref{thm:introductions} and \ref{thm:ST}, we have $t \lras \pair{t_1}{t_2}$ and $r \lras \pair{r_1}{r_2}$. By Theorem~\ref{thm:soundness}, $\interpretation{t} = \interpretation{\pair{t_1}{t_2}}$ and $\interpretation{r} = \interpretation{\pair{r_1}{r_2}}$. Hence, since $\interpretation{t} = \interpretation{r}$, the following diagram commutes:
    
\[\begin{tikzcd}[cramped]
	{\oc I \otimes I} && {(\oc I \otimes I) \oplus (\oc I \otimes I)} && {\interpretation B \oplus \interpretation C}
	\arrow["\Delta", from=1-1, to=1-3]
	\arrow["{\interpretation{t_1} \oplus \interpretation{t_2}}", curve={height=-18pt}, from=1-3, to=1-5]
	\arrow["{\interpretation{r_1} \oplus \interpretation{r_2}}"', curve={height=18pt}, from=1-3, to=1-5]
\end{tikzcd}\]

Since $\pi_1 \circ (f \oplus g) \circ \Delta = f$, and $\pi_2 \circ (f \oplus g) \circ \Delta = g$, we have that $\interpretation{t_1} = \interpretation{r_1}$ and $\interpretation{t_2} = \interpretation{r_2}$. Hence, by the induction hypothesis, we have $t_1 \equiv r_1$ and $t_2 \equiv r_2$. Therefore, $t \equiv \pair{t_1}{t_2} \equiv \pair{r_1}{r_2} \equiv r$.

    \item Let $A = B \oplus C$. Any elimination context for $t$ and $r$ has the shape $K = K'[\elimplus(\_, x^B.u, y^C.v)]$, where $\varnothing;x^B \vdash u:D$ and $\varnothing;y^C \vdash v:D$ with $|D| < |B \oplus C|$. By definition, we have that
    \begin{align*}
      \interpretation{\elimplus(t, x^B.u, y^C.v)} &= [\interpretation{u},\interpretation{v}] \circ (\sigma \oplus \sigma) \circ dist \circ (\interpretation{t} \otimes id) \circ \overline{d}\\
      \interpretation{\elimplus(r, x^B.u, y^C.v)} &= [\interpretation{u},\interpretation{v}] \circ (\sigma \oplus \sigma) \circ dist \circ (\interpretation{r} \otimes id) \circ \overline{d}
    \end{align*}

    $\interpretation{t} = \interpretation{r}$ by hypothesis, therefore $\interpretation{\elimplus(t, x^B.u, y^C.v)} = \interpretation{\elimplus(r, x^B.u, y^C.v)}$. Thus, by the induction hypothesis, we have $\elimplus(t, x^B.u, y^C.v) \equiv \elimplus(r, x^B.u, y^C.v)$.

      Hence,
      \(
	K[t]=K'[\elimplus(t, x^B.u, y^C.v)]\equiv  K'[\elimplus(r, x^B.u, y^C.v)]=K[r]
      \).

      Then, $t\equiv r$.
    \item Let $A = \oc B$. Any elimination context for $t$ and $r$ has the shape $K = K'[\elimbang(\_, x^B.u)]$, where $x^B; \varnothing \vdash u:C$ with $|C| < |\oc B|$. By definition, we have that
    \begin{align*}
		\interpretation{\elimbang(t, x^B.u)} &= \interpretation{u} \circ (\sigma \otimes id) \circ (\interpretation{t} \otimes id) \circ \overline d\\
		\interpretation{\elimbang(r, x^B.u)} &= \interpretation{u} \circ (\sigma \otimes id) \circ (\interpretation{r} \otimes id) \circ \overline d
	\end{align*}

	$\interpretation{t} = \interpretation{r}$ by hypothesis, therefore $\interpretation{\elimbang(t, x^B.u)} = \interpretation{\elimbang(r, x^B.u)}$. Thus, by the induction hypothesis, we have $\elimbang(t, x^B.u) \equiv \elimbang(r, x^B.u)$.

	Hence,
	\(
		K[t] = K'[\elimbang(t, x^B.u)] \equiv K'[\elimbang(r, x^B.u)] = K[r]
	\).

	Then, $t \equiv r$.\qedhere
  \end{itemize}
\end{proof}

\subsection{Diagrams of Section~\ref{sec:denotationalsemantics}}\label{app:diagramasgrandes}
This appendix contains diagrams corresponding to different proofs:
\begin{itemize}
	\item Figure \ref{fig:lem:dcomonoidaxiom} corresponds to the proof of Lemma~\ref{lem:dcomonoidaxiom} (Appendix~\ref{proof:genprop}).
	\item Figure \ref{fig:lem:dcoalgebramorph} corresponds to the proof of Lemma~\ref{lem:dcoalgebramorph} (Appendix~\ref{proof:genprop}).
	\item Figures \ref{fig:lem:nonlinearsubstitutionx} and \ref{fig:lem:linearsubstitutionelimbang} correspond to the proof of Lemma~\ref{lem:nonlinearsubstitution} (Appendix~\ref{proof:soundness}).
\end{itemize}

\begin{figure}
\[\resizebox{!}{.8\paperheight}{\rotatebox{90}{\begin{tikzcd}[ampersand replacement=\&,cramped]
	{(\bigotimes_{i=1}^n \oc A_i)^{\otimes 2}} \&\&\&\& {\bigotimes_{i=1}^n \oc A_i} \&\&\&\& {(\bigotimes_{i=1}^n \oc A_i)^{\otimes 2}} \\
	\\
	\&\&\& {(\bigotimes_{i=1}^{n-1} \oc A_i)^{\otimes 2} \otimes \oc A_n} \&\& {(\bigotimes_{i=1}^{n-1} \oc A_i)^{\otimes 2} \otimes \oc A_n} \\
	\\
	\&\& {(\bigotimes_{i=1}^{n-1} \oc A_i)^{\otimes 2} \otimes \oc A_n \otimes  \oc A_n} \&\&\&\& {(\bigotimes_{i=1}^{n-1} \oc A_i)^{\otimes 2} \otimes \oc A_n \otimes  \oc A_n} \\
	\& {\bigotimes_{i=1}^n \oc A_i \otimes (\bigotimes_{i=1}^{n-1} \oc A_i)^{\otimes 2} \otimes \oc A_n} \&\&\&\&\&\& {\bigotimes_{i=1}^n \oc A_i \otimes (\bigotimes_{i=1}^{n-1} \oc A_i)^{\otimes 2} \otimes \oc A_n} \\
	\&\&\& {(\bigotimes_{i=1}^{n-1} \oc A_i)^{\otimes 3} \otimes \oc A_n} \&\& {(\bigotimes_{i=1}^{n-1} \oc A_i)^{\otimes 3} \otimes \oc A_n} \\
	\& {\bigotimes_{i=1}^n \oc A_i \otimes (\bigotimes_{i=1}^{n-1} \oc A_i)^{\otimes 2} \otimes (\oc A_n)^{\otimes 2}} \&\&\&\&\&\& {\bigotimes_{i=1}^n \oc A_i \otimes (\bigotimes_{i=1}^{n-1} \oc A_i)^{\otimes 2} \otimes (\oc A_n)^{\otimes 2}} \\
	\&\& {(\bigotimes_{i=1}^{n-1} \oc A_i)^{\otimes 3} \otimes  (\oc A_n)^{\otimes 2}} \&\&\&\& {(\bigotimes_{i=1}^{n-1} \oc A_i)^{\otimes 3} \otimes  (\oc A_n)^{\otimes 2}} \\
	\\
	\&\& {(\bigotimes_{i=1}^{n-1} \oc A_i)^{\otimes 3} \otimes  (\oc A_n)^{\otimes 3}} \&\&\&\& {(\bigotimes_{i=1}^{n-1} \oc A_i)^{\otimes 3} \otimes  (\oc A_n)^{\otimes 3}} \\
	{(\bigotimes_{i=1}^n \oc A_i)^{\otimes 3}} \&\&\&\&\&\&\&\& {(\bigotimes_{i=1}^n \oc A_i)^{\otimes 3}}
	\arrow["{id \otimes d_{A_1, \dots, A_{n-1}} \otimes id}"{description}, dashed, from=1-1, to=6-2]
	\arrow[""{name=0, anchor=center, inner sep=0}, "{id \otimes d_{A_1, \dots, A_n}}"{description}, from=1-1, to=12-1]
	\arrow[""{name=1, anchor=center, inner sep=0}, "{d_{A_1, \dots, A_n}}"', from=1-5, to=1-1]
	\arrow[""{name=2, anchor=center, inner sep=0}, "{d_{A_1, \dots, A_n}}", from=1-5, to=1-9]
	\arrow["{d_{A_1, \dots, A_{n-1}} \otimes id}"{description}, dashed, from=1-5, to=3-4]
	\arrow["{d_{A_1, \dots, A_{n-1}} \otimes id}"{description}, dashed, from=1-5, to=3-6]
	\arrow["{id \otimes d_{A_1, \dots, A_{n-1}} \otimes id}"{description}, dashed, from=1-9, to=6-8]
	\arrow[""{name=3, anchor=center, inner sep=0}, "{d_{A_1, \dots, A_n} \otimes id}"{description}, from=1-9, to=12-9]
	\arrow["{id \otimes d_{A_n}}"{description}, dashed, from=3-4, to=5-3]
	\arrow["{id \otimes d_{A_1, \dots, A_{n-1}} \otimes id}"{description}, dashed, from=3-4, to=7-4]
	\arrow["{id \otimes d_{A_n}}"{description}, dashed, from=3-6, to=5-7]
	\arrow["{d_{A_1, \dots, A_{n-1}} \otimes id \otimes id}"{description}, dashed, from=3-6, to=7-6]
	\arrow["{id \otimes \sigma \otimes id}"{description}, dashed, from=5-3, to=1-1]
	\arrow["{(4)}"{description}, draw=none, from=5-3, to=7-4]
	\arrow["{id \otimes d_{A_1, \dots, A_{n-1}} \otimes id}"{description}, dashed, from=5-3, to=9-3]
	\arrow["{id \otimes \sigma \otimes id}"{description}, dashed, from=5-7, to=1-9]
	\arrow["{(7)}"{description}, draw=none, from=5-7, to=6-8]
	\arrow["{id \otimes d_{A_1, \dots, A_{n-1}} \otimes id}"{description}, dashed, from=5-7, to=9-7]
	\arrow["{(6)}"{description}, draw=none, from=6-2, to=5-3]
	\arrow[""{name=4, anchor=center, inner sep=0}, "{id \otimes id \otimes d_{A_n}}"{description}, dashed, from=6-2, to=8-2]
	\arrow[""{name=5, anchor=center, inner sep=0}, "{id \otimes id \otimes d_{A_n}}"{description}, dashed, from=6-8, to=8-8]
	\arrow[""{name=6, anchor=center, inner sep=0}, equals, from=7-4, to=7-6]
	\arrow["{id \otimes d_{A_n}}"{description}, dashed, from=7-4, to=9-3]
	\arrow["{(5)}"{description}, draw=none, from=7-6, to=5-7]
	\arrow["{id \otimes d_{A_n}}"{description}, dashed, from=7-6, to=9-7]
	\arrow["{(10)}"{description}, draw=none, from=8-2, to=9-3]
	\arrow["{id \otimes \sigma \otimes id}"{description}, dashed, from=8-2, to=12-1]
	\arrow["{id \otimes \sigma \otimes id}"{description}, dashed, from=8-8, to=12-9]
	\arrow["{id \otimes \sigma \otimes id}"{description}, dashed, from=9-3, to=6-2]
	\arrow[""{name=7, anchor=center, inner sep=0}, equals, from=9-3, to=9-7]
	\arrow["{id \otimes d_{A_n}}"{description}, dashed, from=9-3, to=11-3]
	\arrow["{id \otimes \sigma \otimes id}"{description}, dashed, from=9-7, to=6-8]
	\arrow["{(11)}"{description}, draw=none, from=9-7, to=8-8]
	\arrow["{id \otimes d_{A_n}}"{description}, dashed, from=9-7, to=11-7]
	\arrow["{id \otimes \sigma \otimes id}"{description}, dashed, from=11-3, to=8-2]
	\arrow[""{name=8, anchor=center, inner sep=0}, equals, from=11-3, to=11-7]
	\arrow["{id \otimes \sigma \otimes id}"{description}, dashed, from=11-7, to=8-8]
	\arrow[""{name=9, anchor=center, inner sep=0}, equals, from=12-1, to=12-9]
	\arrow["{(8)}"{description}, draw=none, from=0, to=4]
	\arrow["{(1)}"{description}, draw=none, from=1, to=5-3]
	\arrow["{(3)}"{description}, draw=none, from=1-5, to=6]
	\arrow["{(9)}"{description}, draw=none, from=3, to=5]
	\arrow["{(2)}"{description}, draw=none, from=5-7, to=2]
	\arrow["{(12)}"{description}, draw=none, from=6, to=7]
	\arrow["{(13)}"{description}, draw=none, from=7, to=8]
	\arrow["{(14)}"{description}, draw=none, from=8, to=9]
\end{tikzcd}}}\]
\caption{Diagram for case $n > 1$ of Lemma~\ref{lem:dcomonoidaxiom}.}\label{fig:lem:dcomonoidaxiom}
\end{figure}
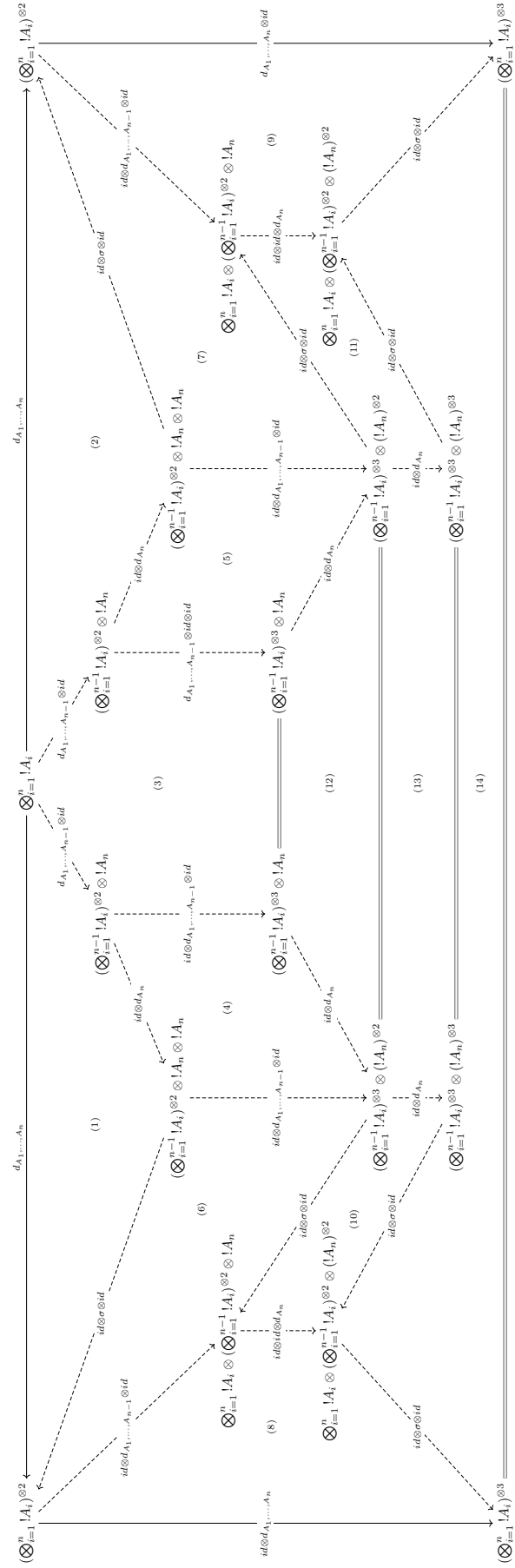

\begin{figure}
  \[\resizebox{!}{.8\paperheight}{\rotatebox{90}{\begin{tikzcd}[ampersand replacement=\&,cramped,row sep=small, column sep=0pt]
    {\bigotimes_{i=1}^n \oc A_i} \&\&\&\&\&\& {\oc (\bigotimes_{i=1}^n \oc A_i)} \\
    \&\&\& {\oc (\bigotimes_{i=1}^{n-1} \oc A_i) \otimes \oc A_n} \&\& \begin{array}{c} \oc (\bigotimes_{i=1}^{n-1} \oc A_i)\\ \otimes \oc \oc A_n \end{array} \\
    \& \begin{array}{c} \bigotimes_{i=1}^{n-1} \oc A_i \otimes \bigotimes_{i=1}^{n-1} \oc A_i\\\otimes \oc A_n \end{array} \\
    \&\&\& \begin{array}{c} \oc (\bigotimes_{i=1}^{n-1} \oc A_i \otimes \bigotimes_{i=1}^{n-1} \oc A_i)\\ \otimes \oc A_n \end{array} \& \begin{array}{c} \oc (\bigotimes_{i=1}^{n-1} \oc A_i \otimes \bigotimes_{i=1}^{n-1} \oc A_i)\\ \otimes \oc \oc A_n \end{array} \\
    \&\& \begin{array}{c} \oc (\bigotimes_{i=1}^{n-1} \oc A_i) \otimes \oc (\bigotimes_{i=1}^{n-1} \oc A_i)\\\otimes \oc A_n \end{array} \\
    {\bigotimes_{i=1}^n \oc A_i \otimes \bigotimes_{i=1}^n \oc A_i} \& \begin{array}{c} \bigotimes_{i=1}^{n-1} \oc A_i \otimes \bigotimes_{i=1}^{n-1} \oc A_i\\\otimes \oc A_n \otimes \oc A_n \end{array} \&\& \begin{array}{c} \oc (\bigotimes_{i=1}^{n-1} \oc A_i \otimes \bigotimes_{i=1}^{n-1} \oc A_i)\\ \otimes \oc A_n \otimes \oc A_n \end{array} \\
    \&\&\&\&\& \begin{array}{c} \oc (\bigotimes_{i=1}^{n-1} \oc A_i \otimes\\ \bigotimes_{i=1}^{n-1} \oc A_i \otimes \oc A_n) \end{array} \\
    \&\& \begin{array}{c} \oc (\bigotimes_{i=1}^{n-1} \oc A_i) \otimes \oc (\bigotimes_{i=1}^{n-1} \oc A_i)\\\otimes \oc A_n \otimes \oc A_n \end{array} \\
    \& \begin{array}{c} \oc (\bigotimes_{i=1}^{n-1} \oc A_i) \otimes \oc A_n\\ \otimes \oc (\bigotimes_{i=1}^{n-1} \oc A_i) \otimes \oc A_n \end{array} \\
    \&\&\& \begin{array}{c} \oc (\bigotimes_{i=1}^{n-1} \oc A_i \otimes \bigotimes_{i=1}^{n-1} \oc A_i)\\ \otimes \oc \oc A_n \otimes \oc \oc A_n \end{array} \& \begin{array}{c} \oc (\bigotimes_{i=1}^{n-1} \oc A_i \otimes \bigotimes_{i=1}^{n-1} \oc A_i)\\ \otimes \oc (\oc A_n \otimes \oc A_n) \end{array} \& \begin{array}{c} \oc (\bigotimes_{i=1}^{n-1} \oc A_i \otimes \bigotimes_{i=1}^{n-1} \oc A_i\\ \otimes \oc A_n \otimes \oc A_n) \end{array} \\
    \& \begin{array}{c} \oc (\bigotimes_{i=1}^{n-1} \oc A_i) \otimes \oc \oc A_n\\ \otimes \oc (\bigotimes_{i=1}^{n-1} \oc A_i) \otimes \oc \oc A_n \end{array} \& \begin{array}{c} \oc (\bigotimes_{i=1}^{n-1} \oc A_i) \otimes \oc (\bigotimes_{i=1}^{n-1} \oc A_i)\\\otimes \oc \oc A_n \otimes \oc \oc A_n \end{array} \\
    {\oc (\bigotimes_{i=1}^n \oc A_i) \otimes \oc (\bigotimes_{i=1}^n \oc A_i)} \&\&\&\&\&\& {\oc (\bigotimes_{i=1}^n \oc A_i \otimes \bigotimes_{i=1}^n \oc A_i)}
    \arrow[""{name=0, anchor=center, inner sep=0}, "{\delta_{A_1, \dots, A_n} }", from=1-1, to=1-7]
    \arrow["{\delta_{A_1, \dots, A_{n-1}} \otimes id}"{description}, dashed, from=1-1, to=2-4]
    \arrow["{d_{A_1, \dots, A_{n-1}} \otimes id }"{description}, dashed, from=1-1, to=3-2]
    \arrow[""{name=1, anchor=center, inner sep=0}, "{d_{A_1, \dots, A_n} }"{description}, from=1-1, to=6-1]
    \arrow["{\oc (d_{A_1, \dots, A_{n-1}} \otimes id)}"{description}, dashed, from=1-7, to=7-6]
    \arrow[""{name=2, anchor=center, inner sep=0}, "{\oc (d_{A_1, \dots, A_n})}"{description}, from=1-7, to=12-7]
    \arrow["{id \otimes \delta_{A_n}}"{description}, dashed, from=2-4, to=2-6]
    \arrow["{\oc (d_{A_1, \dots, A_{n-1}}) \otimes id}"{description}, dashed, from=2-4, to=4-4]
    \arrow["m"{description}, dashed, from=2-6, to=1-7]
    \arrow["{\oc (d_{A_1, \dots, A_{n-1}}) \otimes id}"{description}, dashed, from=2-6, to=4-5]
    \arrow["{(5)}"{description}, draw=none, from=2-6, to=7-6]
    \arrow[""{name=3, anchor=center, inner sep=0}, "{\delta_{A_1, \dots, A_{n-1}} \otimes \delta_{A_1, \dots, A_{n-1}} \otimes id}"{description}, dashed, from=3-2, to=5-3]
    \arrow["{id \otimes d_{A_n}}"{description}, dashed, from=3-2, to=6-2]
    \arrow["{(4)}"{description}, draw=none, from=4-4, to=2-6]
    \arrow["{id \otimes \delta_{A_n}}"', curve={height=12pt}, dashed, from=4-4, to=4-5]
    \arrow["{id \otimes d_{A_n}}"{description}, dashed, from=4-4, to=6-4]
    \arrow["m"{description}, dashed, from=4-5, to=7-6]
    \arrow[""{name=4, anchor=center, inner sep=0}, "{id \otimes \oc d_{A_n}}"{description}, dashed, from=4-5, to=10-5]
    \arrow["{m \otimes id}"{description}, dashed, from=5-3, to=4-4]
    \arrow["{(7)}"{description}, draw=none, from=5-3, to=6-4]
    \arrow["{id \otimes d_{A_n}}"{description}, dashed, from=5-3, to=8-3]
    \arrow[""{name=5, anchor=center, inner sep=0}, "{\delta_{A_1, \dots, A_{n-1}} \otimes id \otimes \delta_{A_1, \dots, A_{n-1}} \otimes id}"{description}, dashed, from=6-1, to=9-2]
    \arrow["{\delta_{A_1, \dots, A_n} \otimes \delta_{A_1, \dots, A_n}}"{description}, from=6-1, to=12-1]
    \arrow["{(6)}"{description}, draw=none, from=6-2, to=5-3]
    \arrow["{id \otimes \sigma \otimes id}"', curve={height=18pt}, dashed, from=6-2, to=6-1]
    \arrow["{\delta_{A_1, \dots, A_{n-1}} \otimes \delta_{A_1, \dots, A_{n-1}} \otimes id}"{description}, dashed, from=6-2, to=8-3]
    \arrow["{(9)}"{description}, curve={height=-12pt}, draw=none, from=6-2, to=9-2]
    \arrow["{id \otimes \delta_{A_n} \otimes \delta_{A_n}}"{description}, dashed, from=6-4, to=10-4]
    \arrow["{(13)}"{description}, draw=none, from=7-6, to=10-5]
    \arrow["{\oc (id \otimes d_{A_n})}"{description}, dashed, from=7-6, to=10-6]
    \arrow["{m \otimes id}"{description}, dashed, from=8-3, to=6-4]
    \arrow["{id \otimes \sigma \otimes id}", curve={height=-18pt}, dashed, from=8-3, to=9-2]
    \arrow["{(12)}"{description}, draw=none, from=8-3, to=10-4]
    \arrow["{id \otimes id \otimes \delta_{A_n} \otimes \delta_{A_n}}"{description}, dashed, from=8-3, to=11-3]
    \arrow["{id \otimes \delta_{A_n} \otimes id \otimes \delta_{A_n}}"{description}, dashed, from=9-2, to=11-2]
    \arrow["{id \otimes m}"', curve={height=18pt}, dashed, from=10-4, to=10-5]
    \arrow["m"', curve={height=18pt}, dashed, from=10-5, to=10-6]
    \arrow["{\oc (id \otimes \sigma \otimes id)}"{description}, dashed, from=10-6, to=12-7]
    \arrow["{(11)}"{description, pos=0.4}, draw=none, from=11-2, to=8-3]
    \arrow["{m \otimes m}"{description}, dashed, from=11-2, to=12-1]
    \arrow["{m \otimes id}"', curve={height=18pt}, dashed, from=11-3, to=10-4]
    \arrow["{id \otimes \sigma \otimes id}", curve={height=-18pt}, dashed, from=11-3, to=11-2]
    \arrow[""{name=6, anchor=center, inner sep=0}, "m"', from=12-1, to=12-7]
    \arrow["{(1)}"{description}, draw=none, from=2-4, to=0]
    \arrow["{(3)}"{description}, draw=none, from=2-4, to=3]
    \arrow["{(2)}"{description}, draw=none, from=3-2, to=1]
    \arrow["{(8)}"{description}, draw=none, from=6-4, to=4]
    \arrow["{(14)}"{description}, draw=none, from=7-6, to=2]
    \arrow["{(15)}"{description}, draw=none, from=10-4, to=6]
    \arrow["{(10)}"{description}, draw=none, from=12-1, to=5]
  \end{tikzcd}}}\]
  \caption{Diagram for case $n > 1$ of Lemma~\ref{lem:dcoalgebramorph}.}\label{fig:lem:dcoalgebramorph}
  \end{figure}

  \begin{figure}
  \[\resizebox{!}{.8\paperheight}{\rotatebox{90}{\begin{tikzcd}[ampersand replacement=\&,cramped]
    {\interpretation{\oc \Upsilon} \otimes I} \&\&\&\&\&\&\&\&\& {\interpretation{A}} \\
    \\
    \& {\interpretation{\oc \Upsilon} \otimes \interpretation{\oc \Upsilon} \otimes I} \&\& {\interpretation{\oc \Upsilon} \otimes \interpretation{\oc \Upsilon} \otimes I} \&\& {\interpretation{\oc \Upsilon} \otimes I \otimes \interpretation{\oc \Upsilon} \otimes I} \&\&\& {\interpretation{A} \otimes \interpretation{\oc \Upsilon} \otimes I} \\
    \\
    \& {\interpretation{\oc \Upsilon} \otimes I\otimes \interpretation{\oc \Upsilon} \otimes I} \\
    \&\&\&\&\& {\oc \interpretation{\oc \Upsilon} \otimes I \otimes \interpretation{\oc \Upsilon} \otimes I} \&\& {\interpretation{\oc \Upsilon} \otimes I \otimes \interpretation{\oc \Upsilon} \otimes I} \\
    \& {\oc \interpretation{A} \otimes \interpretation{\oc \Upsilon} \otimes I} \&\& {\oc \interpretation{\oc \Upsilon} \otimes \interpretation{\oc \Upsilon} \otimes I} \\
    \&\&\&\&\& {\oc \interpretation{\oc \Upsilon} \otimes \oc I \otimes \interpretation{\oc \Upsilon} \otimes I} \&\&\& {\interpretation{\oc \Upsilon} \otimes \interpretation{A} \otimes I} \\
    \&\&\& {\oc (\interpretation{\oc \Upsilon} \otimes I) \otimes \interpretation{\oc \Upsilon} \otimes I} \\
    \&\&\&\&\& {\interpretation{\oc \Upsilon} \otimes \oc \interpretation{\oc \Upsilon} \otimes \oc I \otimes I} \\
    \&\&\&\&\&\&\& {\interpretation{\oc \Upsilon} \otimes \interpretation{\oc \Upsilon} \otimes \oc \otimes I} \\
    \&\&\& {\interpretation{\oc \Upsilon} \otimes \oc (\interpretation{\oc \Upsilon} \otimes I) \otimes I} \&\&\&\&\& {I\otimes \interpretation{A} \otimes I} \\
    \&\&\&\&\& {\interpretation{\oc \Upsilon} \otimes \oc (\interpretation{\oc \Upsilon} \otimes I)} \&\& {\interpretation{\oc \Upsilon} \otimes \interpretation{\oc \Upsilon} \otimes I} \\
    \&\&\&\&\& {\interpretation{\oc \Upsilon} \otimes \interpretation{A}} \\
    {\interpretation{\oc \Upsilon} \otimes \oc \interpretation{A} \otimes I} \&\&\& {\interpretation{\oc \Upsilon} \otimes \oc \interpretation{A}} \&\&\&\&\&\& {I \otimes \interpretation{A}}
    \arrow["{\interpretation{u}}", from=1-1, to=1-10]
    \arrow["{d_\Upsilon \otimes id}"{description}, dashed, from=1-1, to=3-2]
    \arrow[""{name=0, anchor=center, inner sep=0}, "{\alpha \otimes id_I}"', from=1-1, to=15-1]
    \arrow[equals, from=3-2, to=3-4]
    \arrow["{\rho^{-1} \otimes id}"{description}, dashed, from=3-2, to=5-2]
    \arrow["{\rho^{-1} \otimes id}"{description}, dashed, from=3-4, to=3-6]
    \arrow["{\delta_\Upsilon \otimes id}"{description}, dashed, from=3-4, to=7-4]
    \arrow["{\interpretation{u} \otimes id \otimes id}"{description}, dashed, from=3-6, to=3-9]
    \arrow["{\delta_\Upsilon \otimes id}"{description}, dashed, from=3-6, to=6-6]
    \arrow[""{name=1, anchor=center, inner sep=0}, curve={height=-30pt}, equals, from=3-6, to=6-8]
    \arrow["{(17)}"{description}, draw=none, from=3-9, to=1-10]
    \arrow["{\sigma \otimes id}"{description}, dashed, from=3-9, to=8-9]
    \arrow["{\rho \otimes id}"{description}, dashed, from=5-2, to=3-4]
    \arrow["{\interpretation{\oc u} \otimes id}"{description}, dashed, from=5-2, to=7-2]
    \arrow["{\varepsilon_{\interpretation{\oc \Upsilon}} \otimes id}"{description}, curve={height=-18pt}, dashed, from=6-6, to=6-8]
    \arrow["{id \otimes m_I \otimes id}"{description}, dashed, from=6-6, to=8-6]
    \arrow["{\interpretation{u} \otimes id \otimes id}"{description}, dashed, from=6-8, to=3-9]
    \arrow["{(12)}"{description}, draw=none, from=6-8, to=8-9]
    \arrow[""{name=2, anchor=center, inner sep=0}, "{\sigma \otimes id}"{description}, dashed, from=6-8, to=11-8]
    \arrow["{(2)}"{description}, draw=none, from=7-2, to=3-4]
    \arrow[""{name=3, anchor=center, inner sep=0}, "{\sigma \otimes id}"{description}, dashed, from=7-2, to=15-1]
    \arrow["{(4)}"{description}, draw=none, from=7-4, to=3-6]
    \arrow["{\rho^{-1} \otimes id}"{description}, dashed, from=7-4, to=6-6]
    \arrow["{(5)}"{description}, draw=none, from=7-4, to=8-6]
    \arrow["{\oc (\rho^{-1}) \otimes id}"{description}, dashed, from=7-4, to=9-4]
    \arrow["{\varepsilon_{\interpretation{\oc \Upsilon}} \otimes \varepsilon_I \otimes id}"{description}, curve={height=18pt}, dashed, from=8-6, to=6-8]
    \arrow["{(8)}"{description, pos=0.6}, curve={height=-18pt}, draw=none, from=8-6, to=6-8]
    \arrow["{m \otimes id}"{description}, dashed, from=8-6, to=9-4]
    \arrow["{\sigma \otimes id}"{description}, dashed, from=8-6, to=10-6]
    \arrow["{e_\Upsilon \otimes id \otimes id}"{description}, dashed, from=8-9, to=12-9]
    \arrow["{\oc \interpretation{u} \otimes id \otimes id}"{description}, dashed, from=9-4, to=7-2]
    \arrow["{\sigma \otimes id}"{description}, dashed, from=9-4, to=12-4]
    \arrow["{id \otimes \varepsilon_{\interpretation{\oc \Upsilon}} \otimes \varepsilon_I \otimes id}", curve={height=-30pt}, dashed, from=10-6, to=11-8]
    \arrow["{id \otimes m \otimes id}"{description}, dashed, from=10-6, to=12-4]
    \arrow["{id \otimes \interpretation{u} \otimes id}"{description}, dashed, from=11-8, to=8-9]
    \arrow["{(6)}"{description}, draw=none, from=12-4, to=8-6]
    \arrow[""{name=4, anchor=center, inner sep=0}, "{id \otimes \varepsilon_{\interpretation{\oc \Upsilon} \otimes I} \otimes id}"{description}, dashed, from=12-4, to=11-8]
    \arrow["\rho"{description}, dashed, from=12-4, to=13-6]
    \arrow["{id \otimes \oc \interpretation{u} \otimes id}"{description}, dashed, from=12-4, to=15-1]
    \arrow["{(13)}"{description}, draw=none, from=12-4, to=15-4]
    \arrow["\rho"{description}, dashed, from=12-9, to=15-10]
    \arrow["{id \otimes \varepsilon_{\interpretation{\oc \Upsilon} \otimes I}}"', curve={height=-18pt}, dashed, from=13-6, to=13-8]
    \arrow["{(14)}"{description}, draw=none, from=13-6, to=14-6]
    \arrow["{id \otimes \oc \interpretation{u}}"{description}, dashed, from=13-6, to=15-4]
    \arrow["{(15)}"{description}, draw=none, from=13-8, to=11-8]
    \arrow["{id \otimes \interpretation{u}}", dashed, from=13-8, to=14-6]
    \arrow["{e_\Upsilon \otimes id}"{description}, dashed, from=14-6, to=15-10]
    \arrow["\rho"', from=15-1, to=15-4]
    \arrow["{id \otimes \varepsilon_A}"{description}, dashed, from=15-4, to=14-6]
    \arrow[""{name=5, anchor=center, inner sep=0}, "{e_\Upsilon \otimes \varepsilon_A}"', from=15-4, to=15-10]
    \arrow["\lambda"', from=15-10, to=1-10]
    \arrow["{(11)}"{description}, draw=none, from=3-9, to=1]
    \arrow["{(7)}"{description}, draw=none, from=6-6, to=1]
    \arrow["{(9)}"{description}, draw=none, from=2, to=8-6]
    \arrow["{(1)}"{description}, draw=none, from=7-2, to=0]
    \arrow["{(10)}"{description}, draw=none, from=10-6, to=4]
    \arrow["{(3)}"{description}, draw=none, from=12-4, to=3]
    \arrow["{(16)}"{description}, draw=none, from=14-6, to=5]
  \end{tikzcd}}}\]
  \caption{Diagram for case $t = x$ of Lemma~\ref{lem:nonlinearsubstitution}.}\label{fig:lem:nonlinearsubstitutionx}
  \end{figure}

  \begin{figure}
\[\resizebox{!}{.8\paperheight}{\rotatebox{90}{\begin{tikzcd}[ampersand replacement=\&,cramped]
	\begin{array}{c} \interpretation{\oc \Upsilon} \otimes \interpretation{\Gamma} \\ \otimes \interpretation{\Delta_1} \otimes \interpretation{\Delta_2} \end{array} \&\&\&\&\&\&\&\&\&\&\&\&\& \begin{array}{c} \interpretation{\oc \Upsilon} \otimes \interpretation{\oc \Upsilon} \otimes \interpretation{\Gamma} \\ \otimes \interpretation{\Delta_1} \otimes \interpretation{\Delta_2} \end{array} \\
	\\
	\& \begin{array}{c} \interpretation{\oc \Upsilon} \otimes \interpretation{\Delta_1} \\ \otimes \interpretation{\Gamma} \otimes \interpretation{\Delta_2} \end{array} \&\& \begin{array}{c} \interpretation{\oc \Upsilon} \otimes \interpretation{\oc \Upsilon} \otimes \interpretation{\Delta_1} \\ \otimes \interpretation{\Gamma} \otimes \interpretation{\Delta_2} \end{array} \&\&\&\&\&\&\&\& \begin{array}{c}  \interpretation{\oc \Upsilon} \otimes\interpretation{\oc \Upsilon} \otimes \interpretation{\oc \Upsilon} \otimes \interpretation{\Gamma} \\ \otimes \interpretation{\Delta_1} \otimes \interpretation{\Delta_2} \end{array} \\
	\&\&\&\&\&\& \begin{array}{c} \interpretation{\oc \Upsilon} \otimes \interpretation{\oc \Upsilon} \otimes \interpretation{\oc \Upsilon} \otimes \interpretation{\Delta_1} \\ \otimes \interpretation{\Gamma} \otimes \interpretation{\Delta_2} \end{array} \&\&\&\&\&\&\& \begin{array}{c} \interpretation{\oc \Upsilon} \otimes \interpretation{\Delta_1} \otimes \interpretation{\Delta_2} \\ \otimes\interpretation{\oc \Upsilon} \otimes \interpretation{\Gamma} \end{array} \\
	\& \begin{array}{c} \interpretation{\oc \Upsilon} \otimes \interpretation{\Delta_1} \\ \otimes \interpretation{\oc \Upsilon} \otimes \interpretation{\Gamma} \otimes \interpretation{\Delta_2} \end{array} \&\& \begin{array}{c} \interpretation{\oc \Upsilon} \otimes \interpretation{\oc \Upsilon} \otimes \interpretation{\Delta_1} \\ \otimes \interpretation{\Gamma} \otimes \interpretation{\Delta_2} \end{array} \&\& \begin{array}{c} \interpretation{\oc \Upsilon} \otimes \interpretation{\oc \Upsilon} \otimes \\ \interpretation{\oc \Upsilon} \otimes \interpretation{\Delta_1} \\ \otimes \interpretation{\Gamma} \otimes \interpretation{\Delta_2} \end{array} \&\&\&\&\&\&\& \begin{array}{c} \interpretation{\oc \Upsilon} \otimes \interpretation{\oc \Upsilon} \otimes \\ \interpretation{\Delta_1} \otimes \interpretation{\Delta_2} \\ \otimes\interpretation{\oc \Upsilon} \otimes \interpretation{\Gamma} \end{array} \\
	\\
	\& \begin{array}{c} \oc \interpretation C \otimes \interpretation{\oc \Upsilon} \\ \otimes \interpretation{\Gamma} \otimes \interpretation{\Delta_2} \end{array} \&\& \begin{array}{c} \interpretation{\oc \Upsilon} \otimes \interpretation{\oc \Upsilon} \otimes \oc \interpretation C \\ \otimes \interpretation{\Gamma} \otimes \interpretation{\Delta_2} \end{array} \&\&\&\&\&\&\& \begin{array}{c} \interpretation{\oc \Upsilon} \otimes \interpretation{\oc \Upsilon} \\\otimes \interpretation{\Delta_1}  \otimes\interpretation{\oc \Upsilon} \\ \otimes \interpretation{\Gamma} \otimes \interpretation{\Delta_2} \end{array} \\
	\&\&\&\&\&\& \begin{array}{c} \interpretation{\oc \Upsilon} \otimes \interpretation{\oc \Upsilon} \\ \otimes \oc \interpretation C \otimes I \\ \otimes \interpretation{\Gamma} \otimes \interpretation{\Delta_2} \end{array} \&\& \begin{array}{c} \interpretation{\oc \Upsilon} \otimes \interpretation{\oc \Upsilon} \otimes \oc \interpretation C \\ \otimes \interpretation{\Gamma} \otimes \interpretation{\Delta_2} \end{array} \&\&\& \begin{array}{c} \interpretation{\oc \Upsilon} \otimes \interpretation{\Delta_1} \otimes\\ \interpretation{\oc \Upsilon} \otimes \interpretation{\Delta_2} \\ \otimes\interpretation{\oc \Upsilon} \otimes \interpretation{\Gamma} \end{array} \\
	\&\&\& \begin{array}{c} \interpretation{\oc \Upsilon} \otimes \interpretation{\oc \Upsilon} \\ \otimes \oc \interpretation C \otimes \oc \interpretation C \\ \otimes \interpretation{\Gamma} \otimes \interpretation{\Delta_2} \end{array} \&\&\&\&\&\&\& \begin{array}{c} \interpretation{\oc \Upsilon} \otimes \interpretation{\oc \Upsilon} \\\otimes \interpretation{\Delta_1}  \otimes \interpretation{\Delta_2} \\ \otimes\interpretation{\oc \Upsilon} \otimes \interpretation{\Gamma} \end{array} \\
	\& \begin{array}{c} \interpretation{\oc \Upsilon} \otimes \oc \interpretation C \\ \otimes \interpretation{\Gamma} \otimes \interpretation{\Delta_2} \end{array} \&\&\&\&\& \begin{array}{c} \interpretation{\oc \Upsilon} \otimes \oc \interpretation C \\ \otimes \interpretation{\oc \Upsilon} \otimes I \\ \otimes \interpretation{\Gamma} \otimes \interpretation{\Delta_2} \end{array} \&\& \begin{array}{c} \interpretation{\oc \Upsilon} \otimes \oc \interpretation C \otimes \interpretation{\oc \Upsilon} \\ \otimes \interpretation{\Gamma} \otimes \interpretation{\Delta_2} \end{array} \&\&\&\& \begin{array}{c} \interpretation{\oc \Upsilon} \otimes \interpretation{\oc \Upsilon} \otimes \interpretation{\Delta_1} \\ \otimes \interpretation{\Delta_2} \otimes\interpretation A \end{array} \\
	\\
	\&\&\& \begin{array}{c} \interpretation{\oc \Upsilon} \otimes \oc \interpretation C \\ \otimes \interpretation{\oc \Upsilon} \otimes \oc \interpretation C \\ \otimes \interpretation{\Gamma} \otimes \interpretation{\Delta_2} \end{array} \&\&\& \begin{array}{c} \interpretation{\oc \Upsilon} \otimes \oc \interpretation C \otimes \interpretation{\Delta_2}\\\otimes \interpretation{\oc \Upsilon} \otimes I \otimes \interpretation{\Gamma} \end{array} \&\& \begin{array}{c} \interpretation{\oc \Upsilon} \otimes\oc\interpretation{C} \otimes \interpretation{\Delta_2} \\ \otimes\interpretation{\oc \Upsilon} \otimes \interpretation{\Gamma} \end{array} \&\& \begin{array}{c} \oc\interpretation{C} \otimes \interpretation{\oc \Upsilon} \otimes \interpretation{\Delta_2} \\ \otimes\interpretation{\oc \Upsilon} \otimes \interpretation{\Gamma} \end{array} \\
	\\
	\&\& \begin{array}{c} \interpretation{\oc \Upsilon} \otimes \oc \interpretation C \otimes \interpretation{\Delta_2}\\\otimes \interpretation{\oc \Upsilon} \otimes \oc \interpretation C \otimes \interpretation{\Gamma} \end{array} \\
	\&\&\&\&\&\&\&\&\&\& \begin{array}{c} \oc \interpretation C \otimes \interpretation{\oc \Upsilon} \\ \otimes \interpretation{\Delta_2} \otimes\interpretation A \end{array} \&\& \begin{array}{c} \interpretation{\oc \Upsilon} \otimes \interpretation{\Delta_1} \otimes \\ \interpretation{\oc \Upsilon} \otimes \interpretation{\Delta_2} \otimes\interpretation A \end{array} \\
	\&\& \begin{array}{c} \interpretation{\oc \Upsilon} \otimes \oc \interpretation C \\ \otimes \interpretation{\Delta_2} \otimes \interpretation A \end{array} \\
	{\interpretation B} \&\&\&\&\&\&\&\&\&\&\&\&\& \begin{array}{c} \interpretation{\oc \Upsilon} \otimes \interpretation{\Delta_1} \\ \otimes \interpretation{\Delta_2} \otimes\interpretation A \end{array}
	\arrow[""{name=0, anchor=center, inner sep=0}, "{d_\Upsilon \otimes id}", from=1-1, to=1-14]
	\arrow["{id \otimes \sigma \otimes id}"{description}, dashed, from=1-1, to=3-2]
	\arrow[""{name=1, anchor=center, inner sep=0}, "{\interpretation{\elimbang(t, y^C.(v/x)w)}}"{description}, from=1-1, to=17-1]
	\arrow[""{name=2, anchor=center, inner sep=0}, "{id \otimes \sigma \otimes id}"{description}, dashed, from=1-14, to=3-4]
	\arrow[""{name=3, anchor=center, inner sep=0}, "{d_\Upsilon \otimes id}"{description}, dashed, from=1-14, to=3-12]
	\arrow["{id \otimes \sigma}", from=1-14, to=4-14]
	\arrow[""{name=4, anchor=center, inner sep=0}, "{d_\Upsilon \otimes id}", dashed, from=3-2, to=3-4]
	\arrow["{\overline d_{\Upsilon,\Delta_1,(\Gamma,\Delta_2)}}"{description}, dashed, from=3-2, to=5-2]
	\arrow["{id \otimes d_\Upsilon \otimes id}"{description}, dashed, from=3-4, to=4-7]
	\arrow[""{name=5, anchor=center, inner sep=0}, "{id \otimes \sigma \otimes id}"{description}, dashed, from=3-4, to=5-2]
	\arrow["{\sigma \otimes id}"{description}, dashed, from=3-4, to=5-4]
	\arrow["{(5)}"{description}, draw=none, from=3-4, to=5-6]
	\arrow[""{name=6, anchor=center, inner sep=0}, "{id \otimes \sigma \otimes id}"{description}, dashed, from=3-12, to=4-7]
	\arrow["{id \otimes \sigma}"{description}, dashed, from=3-12, to=5-13]
	\arrow["{\sigma \otimes id}"{description}, curve={height=24pt}, dashed, from=4-7, to=5-6]
	\arrow[""{name=7, anchor=center, inner sep=0}, "{d_\Upsilon \otimes id}"{description}, curve={height=18pt}, dashed, from=4-14, to=5-13]
	\arrow[""{name=8, anchor=center, inner sep=0}, "{id \otimes \interpretation v}", from=4-14, to=17-14]
	\arrow["{\sigma \otimes id}"{description}, dashed, from=5-2, to=5-4]
	\arrow["{\interpretation t \otimes id}"{description}, dashed, from=5-2, to=7-2]
	\arrow["{d_\Upsilon \otimes id}"{description}, dashed, from=5-4, to=5-6]
	\arrow[""{name=9, anchor=center, inner sep=0}, "{id \otimes \interpretation t \otimes id}"{description}, curve={height=30pt}, dashed, from=5-4, to=10-2]
	\arrow["{id \otimes \interpretation t \otimes id}"{description}, curve={height=18pt}, dashed, from=5-6, to=7-4]
	\arrow[""{name=10, anchor=center, inner sep=0}, "{id \otimes \sigma \otimes id}"{description}, dashed, from=5-6, to=7-11]
	\arrow["{id \otimes \sigma \otimes id}"{description}, curve={height=12pt}, dashed, from=5-13, to=8-12]
	\arrow[""{name=11, anchor=center, inner sep=0}, "{id \otimes \interpretation v}"{description}, dashed, from=5-13, to=10-13]
	\arrow["{\sigma \otimes id}"{description}, dashed, from=7-2, to=10-2]
	\arrow[""{name=12, anchor=center, inner sep=0}, "{id \otimes \rho^{-1} \otimes id}"{description}, dashed, from=7-4, to=8-7]
	\arrow[""{name=13, anchor=center, inner sep=0}, curve={height=-30pt}, equals, from=7-4, to=8-9]
	\arrow["{id \otimes d_C \otimes id}"{description}, dashed, from=7-4, to=9-4]
	\arrow["{(11)}"{description}, draw=none, from=7-11, to=3-12]
	\arrow["{id \otimes \sigma}"{description}, dashed, from=7-11, to=9-11]
	\arrow["{id \otimes \interpretation t \otimes id}"{description}, dashed, from=7-11, to=10-9]
	\arrow["{id \otimes \lambda \otimes id}"{description}, dashed, from=8-7, to=8-9]
	\arrow["{id \otimes \sigma \otimes id}"{description}, dashed, from=8-7, to=10-7]
	\arrow["{(17)}"{description}, draw=none, from=8-9, to=10-7]
	\arrow["{id \otimes \sigma \otimes id}"{description}, dashed, from=8-9, to=10-9]
	\arrow["{\sigma \otimes id}"{description}, dashed, from=8-12, to=9-11]
	\arrow["{\interpretation t \otimes id}"{description}, dashed, from=8-12, to=12-11]
	\arrow[""{name=14, anchor=center, inner sep=0}, "{id \otimes \interpretation v}"{description}, curve={height=24pt}, dashed, from=8-12, to=15-13]
	\arrow["{id \otimes e_C \otimes id}"{description}, dashed, from=9-4, to=8-7]
	\arrow["{(16)}"{description}, draw=none, from=9-4, to=10-7]
	\arrow["{id \otimes \sigma \otimes id}"{description}, dashed, from=9-4, to=12-4]
	\arrow["{id \otimes \interpretation t \otimes id}"{description}, dashed, from=9-11, to=12-9]
	\arrow["{(22)}"{description}, draw=none, from=9-11, to=12-11]
	\arrow["{d_\Upsilon \otimes id}"{description}, dashed, from=10-2, to=7-4]
	\arrow["{(9)}"{description, pos=0.6}, shift left, curve={height=-12pt}, draw=none, from=10-2, to=9-4]
	\arrow["{d_{\Upsilon,C} \otimes id}"{description}, dashed, from=10-2, to=12-4]
	\arrow["{(10)}"{description}, draw=none, from=10-2, to=16-3]
	\arrow["{\interpretation{(v/x)w}}"{description}, dashed, from=10-2, to=17-1]
	\arrow["{id \otimes \lambda \otimes id}"{description}, dashed, from=10-7, to=10-9]
	\arrow["{id \otimes \sigma}"{description}, dashed, from=10-7, to=12-7]
	\arrow["{(20)}"{description}, draw=none, from=10-9, to=9-11]
	\arrow["{id \otimes \sigma}"{description}, dashed, from=10-9, to=12-9]
	\arrow["{id \otimes \sigma \otimes id}"{description}, dashed, from=10-13, to=15-13]
	\arrow["{id \otimes e_C \otimes id}"{description}, dashed, from=12-4, to=10-7]
	\arrow["{(18)}"{description}, draw=none, from=12-4, to=12-7]
	\arrow["{id \otimes \sigma}"{description}, dashed, from=12-4, to=14-3]
	\arrow["{(19)}"{description}, draw=none, from=12-7, to=10-9]
	\arrow["{id \otimes \lambda}"{description}, dashed, from=12-7, to=12-9]
	\arrow["{(23)}"{description}, draw=none, from=12-9, to=15-11]
	\arrow[""{name=15, anchor=center, inner sep=0}, "{id \otimes \interpretation v}"{description}, curve={height=-12pt}, dashed, from=12-9, to=16-3]
	\arrow["{\sigma \otimes id}"{description}, dashed, from=12-11, to=12-9]
	\arrow[""{name=16, anchor=center, inner sep=0}, "{id \otimes \interpretation v}"{description}, dashed, from=12-11, to=15-11]
	\arrow[""{name=17, anchor=center, inner sep=0}, "{id \otimes e_C \otimes id}"{description}, dashed, from=14-3, to=12-7]
	\arrow["{id \otimes \interpretation{v}'}"{description}, dashed, from=14-3, to=16-3]
	\arrow["{\sigma \otimes id}", curve={height=24pt}, dashed, from=15-11, to=16-3]
	\arrow["{\interpretation t \otimes id}"', curve={height=24pt}, dashed, from=15-13, to=15-11]
	\arrow["{\interpretation w}"{description}, dashed, from=16-3, to=17-1]
	\arrow[""{name=18, anchor=center, inner sep=0}, "{d_\Upsilon \otimes id}"{description}, curve={height=18pt}, dashed, from=17-14, to=10-13]
	\arrow["{\overline d_{\Upsilon,\Delta_1,(\Delta_2,A)}}"{description}, dashed, from=17-14, to=15-13]
	\arrow[""{name=19, anchor=center, inner sep=0}, "{\interpretation{\elimbang(t,y^C.w)}}", from=17-14, to=17-1]
	\arrow["{(1)}"{description}, draw=none, from=0, to=4]
	\arrow["{(2)}"{description}, draw=none, from=2, to=6]
	\arrow["{(12)}"{description}, draw=none, from=3, to=7]
	\arrow["{(3)}"{description}, draw=none, from=3-2, to=5]
	\arrow["{(4)}"{description}, draw=none, from=5, to=5-4]
	\arrow["{(6)}"{description}, draw=none, from=5-2, to=9]
	\arrow["{(8)}"{description}, draw=none, from=5-2, to=1]
	\arrow["{(26)}"{description}, draw=none, from=11, to=8]
	\arrow["{(13)}"{description}, draw=none, from=13, to=10]
	\arrow["{(7)}"{description}, draw=none, from=7-4, to=9]
	\arrow["{(14)}"{description}, draw=none, from=8-7, to=13]
	\arrow["{(25)}"{description}, draw=none, from=14, to=11]
	\arrow["{(15)}"{description}, draw=none, from=9-4, to=12]
	\arrow["{(24)}"{description}, draw=none, from=16, to=14]
	\arrow["{(21)}"{description}, draw=none, from=17, to=15]
	\arrow["{(28)}"{description}, draw=none, from=15-11, to=19]
	\arrow["{(27)}"{description}, curve={height=12pt}, draw=none, from=15-13, to=18]
\end{tikzcd}}}\]
\caption{Diagram for case $u = \elimbang(t, y^C.w)$ when $x \in \fv(w)$ of Lemma~\ref{lem:linearsubstitution}.}\label{fig:lem:linearsubstitutionelimbang}
  \end{figure}

\end{document}